\documentclass[10pt,A4paper]{article}
\usepackage{amssymb}
\usepackage{amsfonts}
\usepackage{amsmath}
\usepackage{amsthm}
\usepackage{slashed}
\usepackage{latexsym}
\usepackage[dvips]{epsfig}
\usepackage{enumerate}
\usepackage{mathrsfs}
\usepackage{eufrak}
\usepackage{bm}
\usepackage{tikz}
\usepackage{authblk}
\usepackage{cancel}
\usepackage{hyperref}
\usepackage{color}

\theoremstyle{plain}
\newtheorem{proposition}{Proposition}
\newtheorem{lemma}{Lemma}
\newtheorem{theorem}{Theorem}

\newtheorem{corollary}{Corollary}

\newtheorem{remark}{Remark}
\newtheorem{gauge}{Gauge Choice}

\def\bme{{\bm e}}

\def\bmg{{\bm g}}

\def\bml{{\bm l}}
\def\bmn{{\bm n}}
\def\bmm{{\bm m}}

\def\bmvarphi{{\bm \varphi}}

\def\bmGamma{{\bm \Gamma}}

\def\bmDelta{{\bm \Delta}}

\def\bmpartial{{\bm \partial}}

\def\nablasl{/\kern-0.58em\nabla}
\def\Deltasl{/\kern-0.7em\Delta}
\def\Dsl{/\kern-0.7em D}

\def\TiPsi{{\tilde \Psi}}
\def\Titau{{\tilde \tau}}
\def\Tipi{{\tilde \pi}}
\def\Timu{{\tilde \mu}}
\def\Tivarphi{{\tilde \varphi}}

\def\ulchi{{\underline{\chi}}}
\def\ulomega{{\underline{\omega}}}

\font\teniOne=ecti1000
\font\seveniOne=ecti0900
\font\fiveiOne=ecti0500

\newfam\tOne
\textfont\tOne=\teniOne
\scriptfont\tOne=\seveniOne
\scriptscriptfont\tOne=\fiveiOne

\def\definetOnesymbol#1#2{%
  \mathchardef#1=\numexpr\tOne*256+#2\relax
}
\definetOnesymbol{\meth}{"F0}
\definetOnesymbol{\mthorn}{"FE}

\allowdisplaybreaks

\begin{document}

\title{\textbf{Trapped surface formation for the Einstein-Scalar system}}

\author[,1]{Peng
  Zhao \footnote{E-mail address:{\tt p.zhao@bnu.edu.cn}}}
\author[,2]{David Hilditch \footnote{E-mail address:{\tt
      david.hilditch@tecnico.ulisboa.pt}}}
\author[,3]{Juan A. Valiente
    Kroon \footnote{E-mail address:{\tt
        j.a.valiente-kroon@qmul.ac.uk}}}
\affil[1]{College of Education for the Future, Beijing Normal University
 at Zhuhai, No.18, Jinfeng Road, Tangjiawan, Zhuhai City, Guangdong
  Province, 519087, P.R.China.}        
\affil[2]{CENTRA, Departamento de F\'isica, Instituto Superior
  T\'ecnico – IST, Universidade de Lisboa – UL, Avenida Rovisco Pais
  1, 1049 Lisboa, Portugal.}
\affil[3]{School of Mathematical Sciences, Queen Mary, University of
    London, Mile End Road, London E1 4NS, United Kingdom.}

\maketitle
  
\begin{abstract}
  We consider the formation of trapped surfaces in the evolution of
  the Einstein-scalar field system without symmetries. To this end, we
  follow An's strategy to analyse the formation of trapped surfaces in
  vacuum and for the Einstein-Maxwell system. Accordingly, we set up a
  characteristic initial value problem (CIVP) for the Einstein-Scalar
  system with initial data given on two intersecting null
  hypersurfaces such that on the incoming slince the data is
  Minkowskian whereas on the outgoing side no symmetries are
  assumed. We obtain a scale-critical semi-global existence result by
  assigning a signature for decay rates to both the geometric
  quantities and the scalar field. The analysis makes use of a gauge
  due to J. Stewart and an adjustment of the Geroch-Held-Penrose (GHP)
  formalism, the T-weight formalism, which renders the connection
  between the Newman-Penrose (NP) quantities and the PDE analysis
  systematic and transparent.
\end{abstract}

\tableofcontents

\section{Introduction}

In this paper we study the Einstein-scalar system
\begin{subequations}
\begin{align}
&R_{ab}-\frac{1}{2}Rg_{ab}=T_{ab}, \label{EFEMasslessscalar0}\\
&\nabla^a\nabla_a\varphi=0, \label{EOMMasslessscalar0}
\end{align}
\end{subequations}
where~$R_{ab}$ and~$R$ denote, respectively, the Ricci tensor and
scalar of the Levi-Civita covariant derivative of a Lorentzian
metric~$g_{ab}$, and
\begin{align*}
  T_{ab}=\nabla_a\varphi\nabla_b\varphi
  -\frac{1}{2}g_{ab}\nabla_c\varphi\nabla^c\varphi
\end{align*}
is the energy-momentum tensor associated to
equation~\eqref{EOMMasslessscalar0}. It then follows from
equation~\eqref{EFEMasslessscalar0} that
\begin{align}
\label{EFEMasslessscalar1}
R_{ab}=\nabla_a\varphi\nabla_b\varphi.
\end{align}

Black holes are one of the most fascinating predictions of general
relativity, and evidence for their existence is overwhelmingly
supported by many recent observations (see for
instance~\cite{Abb16,EHT2019}). Theoretically, a black hole is defined
as a region in spacetime which has no causal relation to future null
infinity. However, this global definition is not ideal for the study
of the dynamical formation of black holes. An alternative approach is
to track the appearance of \emph{trapped surfaces} ---see
e.g.~\cite{Wal84,Chr20}. Trapped surfaces are closed, spacelike
topological 2-spheres such that both the outgoing and ingoing light
rays emanating from it converge.  This local perspective was first
introduced by R. Penrose in his famous \emph{incompleteness
  theorem}~\cite{Penrose1965}. This theorem shows that the existence
of trapped surfaces in spacetime leads to the existence of incomplete
future causal geodesics.  Hence, the study of trapped surfaces is
central to the understanding of the formation of black holes.

A drawback of Penrose's incompleteness theorem is that it does not
provide an explanation of how trapped surfaces can form. In
particular, it is clear that the formation of trapped surfaces
requires appropriate initial conditions. Following the work of
Christodoulou and Klainerman~\cite{ChrKla93} on the
non-linear stability of the Minkowski spacetime, initial conditions
close (in a suitable sense) to initial conditions for the Minkowski
spacetime develop into a geodesically complete spacetime with a global
structure similar to that of the Minkowski spacetime. This suggests
that to form trapped surfaces one needs data which is \emph{in some
  sense} large. In view of Birkhoff's theorem, this data needs to be
non-spherical in order to allow for the formation of trapped surfaces
in vacuum.

The question of the dynamical formation of trapped surfaces in vacuum
was first solved by Christodoulou in~\cite{Chr00}. In doing so, he
considered initial data having a special hierarchical structure called
the \emph{short pulse ansatz}. In this setting, various quantities
have different decay behaviours parameterised by~$\delta$, and these
properties are maintained under long time nonlinear evolution.
Subsequently, Klainerman and Rodnianski~\cite{KlaRod09} simplified
Christodoulou's main theorem by introducing an index~$s_1$ called
\emph{signature} to track the~$\delta$-weight of the various geometry
quantities. In parallel to the latter, they defined scale invariant
norms for which most of the non-linear terms are
small. In~\cite{An2012} An has extended Klainerman and Rodnianski's
work so that initial conditions are prescribed on past null infinity
thus recovering Christodoulou's original results. As part of his
analysis An introduced a new signature~$s_2$ corresponding to the
powers of a retarded time~$u$ which describes the decay rates near
past null infinity.

Building on the work described in the previous paragraph,
in~\cite{AnLuk2017}, An and Luk proved the first scale-critical result
from data which can be regarded as \emph{mild} in comparison to that
used in the original work by Christodoulou and relaxed the lower bound
required to form a trapped surface.  Based on the idea of
scale-critical arguments, in~\cite{An2022} An provided a simple proof
of the formation of trapped surfaces in vacuum which sharpens the
previous results both of Christodoulou and An-Luk.  This argument
connects Christodoulou’s short-pulse method and
Klainerman-Rodnianski’s signature counting argument to the peeling
properties of the gravitational field at past null infinity.

The results on gravitational collapse in the presence of matter
predate those in the vacuum with the first result being the analysis
by Christodoulou of the spherically symmetric self-gravitating
massless scalar field in~\cite{Chr91}.  Christodoulou goes to show
that trapped surfaces form if the ratio of the mass contained in an
annular region to the largest radius is large in some specific sense.
Li and Liu~\cite{LiLiu2017} have revisited Christodoulou's model and
obtained an almost scale-critical trapped surface formation criterion
without assuming symmetry on the initial outgoing lightcone but with
spherical singular initial data on the incoming lightcone.

The case of the Einstein-Maxwell field system, has been analysed by Yu
in~\cite{PYu2011}. This work also makes use of the short pulse Ansatz
to study the formation of trapped surfaces. In particular, it is shown
that the formation of a trapped surface can be purely due to the
concentration of the Maxwell field on the initial lightcone where
there is no incoming gravitational energy. Building on these ideas, An
and Athanasiou~\cite{An202209} have obtained an extended
scale-critical trapped surface formation criterion from past null
infinity. This work generalises the An's work on vacuum case
in~\cite{An2022} by employing additional elliptic estimates and
geometric renormalisations.

Motivated by~\cite{An2022} and~\cite{An202209}, the goal of the
present article is to obtain a scale-critical trapped surface
formation criterion for the Einstein-scalar system from near past null
infinity. To this end we study the behaviour of both the scalar field
and geometric quantities and develop an understanding of the mechanism
of collapse.

\subsection{New insights}

The analysis in the present article make use of the following new
insights and ideas:

\smallskip
\noindent
\textbf{~(1)} We systematically study the asymptotic behaviour of
scalar field by applying An's strategy ---i.e. by assigning
signatures~$s_2$ to the geometric quantities, the scalar
field~$\varphi$ as well as the auxiliary
field~$\varphi_a\equiv\frac{1}{3}\nabla_a\varphi$. (The
factor~$\frac{1}{3}$ is added for convenience when working in the NP
formalism). The signature works well in keeping consistency with the
equations of the scalar field. For example,
for~$\varphi_0\equiv l^a\varphi_a$, $\varphi_1\equiv m^a\varphi_a$
and~$\varphi_2\equiv n^a\varphi_a$, their signature is the same as
that of the Faraday tensor components~$\phi_0\equiv F_{ab}l^am^b$,
$\phi_1\equiv \frac{1}{2}F_{ab}(l^an^b+\bar{m}^am^b)$ and
$\phi_2\equiv F_{ab}\bar{m}^an^b$ respectively. The structure of the
equations for~$\varphi_0$, $\varphi_1$ and~$\varphi_2$ is also similar
to that of the Maxwell equations, see~\cite{An202209}, so that energy
estimates can be constructed along the same lines.

\smallskip
\noindent
\textbf{~(2)} In addition to the renormalised components of the Weyl
tensor
\begin{align*}
  \tilde\Psi_1\equiv\Psi_1-3\varphi_{0}\varphi_{1}, \quad
  \tilde\Psi_3\equiv\Psi_3-3\varphi_{2}\bar\varphi_{1}, 
\end{align*}
our analysis requires yet another renormalised Weyl component, namely
\begin{align*}
\tilde\Psi_2\equiv\Psi_2-\varphi_{0}\varphi_{2}+\varphi_{1}\bar\varphi_{1}, 
\end{align*}
to close the energy estimates for Weyl curvature.  This term does not
arise in the analysis of the Einstein-Maxwell system as in that case
the energy-momentum tensor is tracefree.

\smallskip
\noindent
\textbf{~(3)} Due to the intrinsic relations between~$\varphi$ and
$\varphi_a$, the asymptotic behaviour of~$\varphi_2$ does not satisfy
the same peeling property as the Maxwell component $\phi_2$, and this
difference is inherited also by the associated stress-energy tensor.
Moreover, as the first order equations for~$\varphi_0$, $\varphi_1$
and~$\varphi_2$ arise from the use of commutation relations, one faces
obstructions in the analysis of~$(\varphi_1,\varphi_2)$
and~$(\TiPsi_3,\Psi_4)$. This difficulty can be overcome by
introducing a new variable~$\Tivarphi_2$
\begin{align*}
\Tivarphi_2\equiv\meth'\varphi_2+\mu\bar\varphi_1.
\end{align*}
In the above expression~$\meth'$ denotes an improved \emph{T-weight
  operator} arising from the NP
operator~$\bar\delta=\bar{m}^a\nabla_a$ on the tangent bundle of a
2-sphere, see~\textbf{(4)} below.

\smallskip
\noindent
\textbf{~(4)} We make use of a new NP formalism, which we dub the
\emph{T-weight formalism}. This is an adjusted version of the GHP
formalism which takes into account the requirements of PDE analysis
and the use of double null foliations. In addition to the usual
spin-weighted quantities $\kappa$, $\nu$, $\tau$, $\pi$, $\rho$,
$\sigma$, $\mu$ and~$\lambda$, we introduce another three
spin-weighted quantities $\vartheta\equiv\beta+\bar\alpha$,
$\ulomega\equiv\gamma+\bar\gamma$
and~$\omega\equiv\epsilon+\bar\epsilon$ in the analysis. These
quantities are related to definite tensors on a topological
2-sphere~$\mathcal{S}\approx\mathbb{S}^2$ and are characterised by one
integer ---their \emph{T-weight}, within the present formalism. We
also redefine the covariant operators $\mthorn$, $\mthorn'$, $\meth$
and~$\meth'$ by appropriately incorporating the
combinations~$\epsilon-\bar\epsilon$, $\gamma-\bar\gamma$ and
$\beta-\bar\alpha$, into their definitions. These operators have a
definite correspondence with the NP directional derivative
operators~$l^a\nabla_a$, $n^a\nabla_a$, $m^a\nabla_a$ and
$\bar{m}^a\nabla_a$ and preserve the T-weight property.  This new
formalism is a bridge connecting the derivatives of NP quantities and
the covariant derivative of tensors.  Another advantage of this
formalism is that the divergence, curl and the Stokes' theorem used in
the discussion of elliptic estimates and the Hodge system have clear,
concise, expressions.

\subsection{Overview and main results}

We will study the formation of trapped surfaces by means of double
null foliation. In the following~$\mathcal{N}_u$ and~$\mathcal{N}'_v$
will denote, respectively, outgoing and ingoing null hypersurfaces
with~$u$ and~$v$ optical functions satisfying the eikonal equations
\begin{align*}
g^{\mu\nu}\partial_{\mu}u\partial_{\nu}u=0 \quad{and} \quad 
g^{\mu\nu}\partial_{\mu}v\partial_{\nu}v=0.
\end{align*}
For given values of~$u$ and~$v$ the
intersection~$\mathcal{S}_{u,v}\equiv \mathcal{N}_u \cap
\mathcal{N}'_v$ has the topology of~$\mathbb{S}^2$.

Our analysis will focus on the region of spacetime given by
\begin{align*}
\mathbb{D}=\{(u,v)|u_{\infty}\leq u\leq -a/4, \ \ 0\leq v\leq 1\}.
\end{align*}
where~$a$ is a suitable (large) positive number. The initial outgoing
null hypersurface
is~$\mathcal{N}_{\star}\equiv\mathcal{N}_{u=u_{\infty}}$ and the
initial ingoing null hypersurface
is~$\mathcal{N}'_{\star}\equiv\mathcal{N}_{v=0}$.  On~$\mathbb{D}$ one
can choose appropriate coordinates and specify an adapted NP
frame~$\{\bml,\bmn,\bm{m},\bar{\bm{m}}\}$. The expansions of outgoing
and ingoing null hypersurfaces, respectively, can be expressed in
terms of the NP spin connection coefficients as
\begin{align*}
\theta_{\bml}&\equiv\sigma^{ab}\nabla_a\l_b=-\rho-\bar\rho=-2\rho, \\
\theta_{\bmn}&\equiv\sigma^{ab}\nabla_an_b=\mu+\bar\mu=2\mu,
\end{align*}
where~$\sigma_{ab}$ is the induced metric on~$\mathcal{S}_{u,v}$.  The
norm of the shears of the outgoing and ingoing null hypersurface are
given by
\begin{align*}
\hat{\chi}_{\bml}&\equiv-\sigma-\bar\sigma=-2|\sigma| \\
\hat{\underline{\chi}}_{\bmn}&\equiv\lambda+\bar\lambda=2|\lambda|.
\end{align*}
The four coefficients~$\rho$, $\mu$, $\sigma$ and~$\lambda$ are the
components of null second fundamental forms
\begin{align*}
\hat{\chi}_{ab}\equiv-\Pi_{\mathcal{S}}(\nabla_al_b), \quad
\hat{\underline{\chi}}_{ab}\equiv\Pi_{\mathcal{S}}(\nabla_an_b).
\end{align*}
where~$\Pi_{\mathcal{S}}$ is projection on~$\mathcal{S}$. We establish
the following existence theorem:

\begin{theorem}
[\textbf{\em Existence result}]
Given a positive number~$\mathcal{I}$, there exists a sufficiently
large~$a_0=a_0(\mathcal{I})$ such that for~$a\geq a_0\geq0$ and
initial data such that
\begin{align*}
\mathcal{I}_0\equiv\sum_{j=0}^1\sum_{i=0}^{15}\frac{1}{a^{\frac{1}{2}}}
||\mthorn^j(|u_{\infty}|\mathcal{D})^i(\sigma,\varphi_0)||_{L^{2}(\mathcal{S}_{u_{\infty},v})}\leq\mathcal{I}
\end{align*}
along the outgoing initial null hypersurface~$u=u_{\infty}$, and
Minkowskian initial data along ingoing initial null
hypersurface~$v=0$, then the Einstein-scalar field system admits a
unique solution in
\begin{align*}
\mathbb{D}=\{(u,v)|u_{\infty}\leq u\leq -a/4, \ \ 0\leq v\leq 1\}.
\end{align*}
\end{theorem}

\begin{remark} {\em In the above expression~$\mathcal{D}$ denotes the
    covariant derivative of the metric~$\sigma_{ab}$ on the
    sphere~$\mathcal{S}_{u,v}$ acting on the NP quantities.}
\end{remark}

 Based on the existence result, we have the following trapped
surface formation theorem:
\begin{theorem}
[\textbf{\em Trapped surface formation}]
Given~$\mathcal{I}$, there exists a sufficiently large
$a=a(\mathcal{I})$ such that for~$a\geq a_0\geq0$ an initial data set
with
\begin{align*}
\mathcal{I}_0\equiv\sum_{j=0}^1\sum_{i=0}^{15}\frac{1}{a^{\frac{1}{2}}}
  ||\mthorn^j(|u_{\infty}|\mathcal{D})^i(\sigma,\varphi_0)||_{L^{2}(\mathcal{S}_{u_{\infty},v})}
  \leq\mathcal{I}
\end{align*}
along the outgoing initial null hypersurface~$u=u_{\infty}$ and
Minkowskian initial data along the ingoing initial null hypersurface
$v=0$ such that
\begin{align*}
  \int_0^1\left(
  |u_{\infty}|^2(\sigma\bar\sigma+3\varphi_0^2)|\right)(u_{\infty},v') \mathrm{d}v'\geq a
\end{align*}
holds uniformly for any point on the initial outgoing null
hypersurface $u=u_{\infty}$, it follows that~$\mathcal{S}_{-a/4,1}$ is
a trapped surface.
\end{theorem}

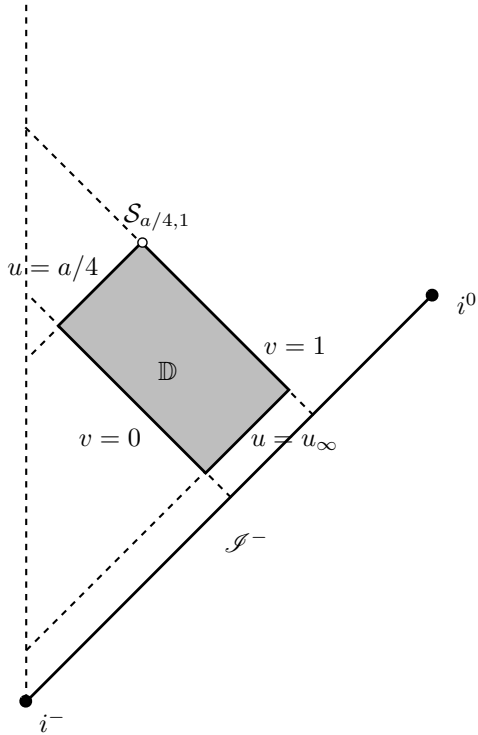
\begin{figure}[t]
\centering
\includegraphics[width=0.8
\textwidth]{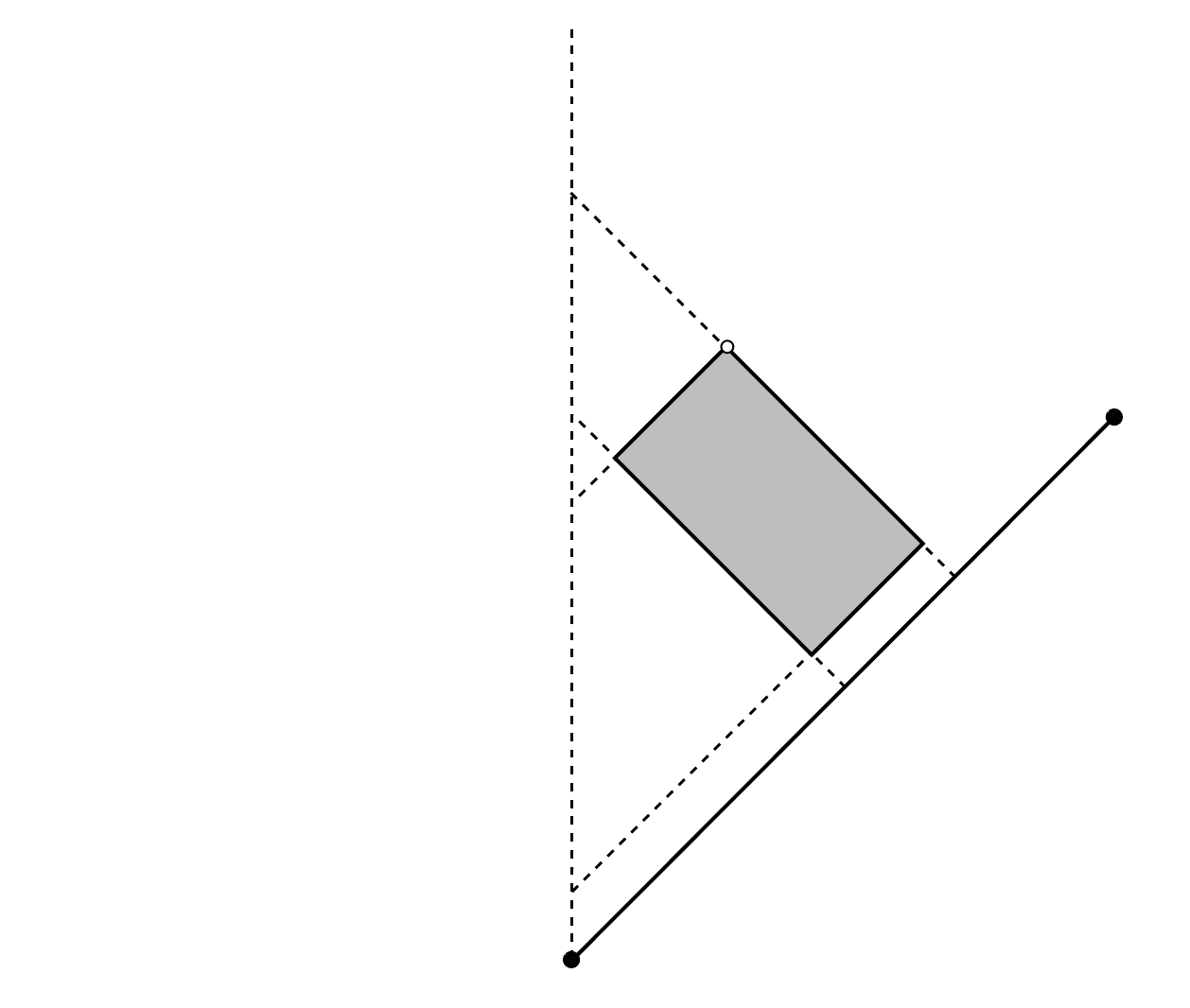}
\put(-100,70){$\mathscr{I}^-$}
\put(-170,2){$i^-$}
\put(-12,160){$i^0$}
\put(-125,135){$\mathbb{D}$}
\put(-155,110){$v=0$}
\put(-85,145){$v=1$}
\put(-90,110){$u=u_\infty$}
\put(-182,175){$u=a/4$}
\put(-138,195){$\mathcal{S}_{a/4,1}$}
\caption{Existence domain~$\mathbb{D}$ for the main theorem. The setup
consists of a characteristic initial value problem where Minkoswkian
initial data is prescribed on the incoming null hypersurface~$v=0$ and
suitable data on the outgoing hypersurface~$u=u_\infty$ lying in the
asymptotic region. By choosing~$a$ appropriately one can ensure that
the 2-surface~$\mathcal{S}_{a/4,1}$ (upper vertex of the existence
diamond) is a trapped surface.}
\label{Fig:ExistenceDomain}
\end{figure}

\subsection{A heuristic argument}

In this section we assume the results of the preceding section and
provide a heuristic argument of the dynamic formation of a trapped
surface.

A trapped surface is a closed 2-sphere~$\mathcal{S}_{u,v}$ such that
both the outgoing expansion~$\theta_{\bml}$ and the ingoing expansion
$\theta_{\bmn}$ are negative on the whole of~$\mathcal{S}_{u,v}$. On
the initial ingoing null hypersurface~$\mathcal{N}'_{\star}$ we
prescribe Minkowskian data and the expansions can be computed to be
given by
\begin{align*}
\theta_{\bml}=\frac{2}{|u|}, \qquad
\theta_{\bmn}=-\frac{2}{|u|}.
\end{align*}
Now, it is possible to show that the ingoing expansion
on~$\mathcal{N}_{u_{\infty}}$ is given by
$-2/|u_{\infty}|+l.o.t$. Here and in the following~$l.o.t$ means lower
order terms compared with the other terms in the equation. Applying
the ingoing structure equation
\begin{align*}
\mthorn'\mu=-\mu^2+l.o.t.,
\end{align*}
where~$\mthorn'$ is the ingoing null derivative along~$\bmn$ and one
can show that~$\theta_{\bmn}=2\mu$ is a decreasing function and, thus,
it becomes even more negative.

For the outgoing expansion we first make use of the ingoing equations
\begin{align*}
\mthorn'\sigma&=-\mu\sigma+l.o.t , \\
\mthorn'\varphi_{0}&=-\mu\varphi_{0}+l.o.t,
\end{align*}
to show that 
\begin{align}
\label{HeuArgu1}
  |u|^2(\sigma\bar\sigma+3\varphi_0^2)|_{u,v}
  =|u_{\infty}|^2(\sigma\bar\sigma+3\varphi_0^2)|_{u=u_{\infty},v}+l.o.t.
\end{align}
Then we apply the transport equation
\begin{align*}
\mthorn\rho&=\rho^2+\sigma\bar\sigma+3\varphi_{0}^2
\end{align*}
to show that 
\begin{align*}
  \rho(-\frac{a}{4},1,x^2,x^3)&=\rho(-\frac{a}{4},0,x^2,x^3)
  +\int_0^1\frac{\partial\rho}{\partial v} \mathrm{d}v 
  =-\frac{4}{a}+\int_0^1\mthorn\rho \mathrm{d}v' \\
  &\geq -\frac{4}{a}+\int_0^1(\sigma\bar\sigma+3\varphi_{0}^2)
        (-\frac{a}{4},v',x^2,x^3) \mathrm{d}v'  \\
&\geq-\frac{4}{a}+\frac{8}{a}=\frac{4}{a}>0,
\end{align*}
where in the last step we have made use of~\eqref{HeuArgu1} and the
condition on the initial data on~$\mathcal{N}_{u_{\infty}}$
\begin{align*}
  \int_0^1\left( |u_{\infty}|^2(\sigma\bar\sigma+3\varphi_0^2)|\right)
  (u_{\infty},v') \mathrm{d}v'\geq a.
\end{align*}
Hence~$\theta_{\bml}=-2\rho<0$ on~$\mathcal{S}_{-a/4,1}$ which
mean~$\mathcal{S}_{-a/4,1}$ is a trapped surface.  More details can be
found in Section~\ref{TrappedSurface}.

\subsection{Outline of the article}

In Section~\ref{GeoSettingCIVP} we give a discussion of the geometric
setting being considered and of the characteristic initial value
problem for the Einstein-scalar field system.
Section~\ref{T-weightFormalism} provides a full introduction to the
T-weight formalism. In particular, we discuss the relation between the
T-weight formalism and the PDE analysis and give the structure
equations, Bianchi identities and the scalar field equations in this
formalism. In Section~\ref{Preliminary} we define the scale invariant
norms used in the analysis and discuss the bootstrap assumptions. We
also give basic estimates for the components of the metric and a
Gr\"onwall type estimate for solutions to null transport
equations. Section~\ref{L2estimate} discusses the $L^2(\mathcal{S})$
estimate for non-top-order connection coefficients, the auxiliary
field~$\varphi_a$ and the Weyl curvature.  In
Section~\ref{EllipticEstimate} we provide elliptic estimates of the
top-order connection coefficients. Section~\ref{EnergyEstimate}
discusses energy estimates of~$\varphi_a$ and the Weyl curvature.
Section~\ref{TrappedSurface} provides the proof of the formation of
trapped surfaces.

\subsection{Acknowledgements}

The Authors are grateful to Xinliang An for helpful discussions. Many
of our derivations were performed in xAct~\cite{xAct} for
Mathematica. This work was partially supported by FCT (Portugal) grant
UIDB/00099/2020. PZ was supported by the start-up fund of Beijing
Normal University at Zhuhai.

\section{Geometric setting and formulation of the characteristic
  initial value problem}
\label{GeoSettingCIVP}

In this section we provides details of the construction of Stewart's
gauge and discuss the basic local existence result for the
characteristic initial value problem underlining our analysis. More
discussions can be found in~\cite{HilValZha19} and~\cite{HilValZha20}

\subsection{Coordinates and Stewart's gauge}
\label{Subsection:StewartGauge}

The basic geometric setting of this paper consists of a 4-dimensional
manifold~$(\mathcal{M},\bmg)$ with boundary and an edge.  The boundary
consists of two null hypersurfaces:~$\mathcal{N}_{\star}$, the
outgoing null hypersurface; and~$\mathcal{N}_{\star}^{\prime}$, the
incoming null hypersurface.  We denote the non-empty intersection of
the initial hypersurfaces by
$\mathcal{S}_{\star}\equiv\mathcal{N}_{\star}\cap\mathcal{N}_{\star}^{\prime}$.
Given a neighbourhood~$\mathcal{U}$ of~$\mathcal{S}_{\star}$, the
future of~$\mathcal{S}_{\star}$ can be foliated by two families of
null hypersurfaces:~$\mathcal{N}_{u}$ (the outgoing null
hypersurfaces) and~$\mathcal{N}_{v}'$ (the ingoing null
hypersurfaces). The functions~$u$ and~$v$ satisfy the eikonal equation
and~$\mathcal{N}_{u_{\infty}}=\mathcal{N}_{\star}$,
$\mathcal{N}_{0}'=\mathcal{N}_{\star}'$.  Given suitable data on
$(\mathcal{N}_{\star}\cup\mathcal{N}_{\star}')\cap\mathcal{U}$, one
can make statement of existence and uniqueness of solutions to the
Einstein-Scalar system on some open set
 \begin{align*}
\mathcal{V}\subset\{p\in\mathcal{U}| u(p)>u_{\infty}, v(p)>0 \}.
 \end{align*}
The relevant area for our analysis is given by the conditions
 \begin{align*}
u_{\infty}\leq u\leq -a/4, \ \ 0\leq v\leq 1,
\end{align*}
where~$a$ is a large positive number.

To complete the coordinate system, consider arbitrary coordinate
charts~$(\mathcal{U}^{\star}_{\alpha},x_{\alpha}^{\mathcal{A}})$
on~$\mathcal{S}_{\star}$, with the index~$^{\mathcal{A}}$ taking the
values~$2,\; 3$ and where~$\mathcal{U}^{\star}_{\alpha}$ denote a
finite cover of~$\mathcal{S}_{\star}$. For convenience we can make use
of the stereographic coordinates given by the diffeomorphism of
$\mathcal{S}_\star$ with~$\mathbb{S}^2$. In this case the chart label
$\alpha$ takes the values~$1,\;2$ corresponding, respectively, to the
North and South polar charts.  Now, define
$(\mathcal{U}^{\star}_{\alpha})_{u,0}$ be the coordinate patch
generated from~$\mathcal{U}^{\star}_{\alpha}$ along the generators of
$\mathcal{N}_{\star}'$, and~$(\mathcal{U}_{\alpha})_{u,v}$ be the
coordinate patch generated from~$(\mathcal{U}^{\star}_{\alpha})_{u,0}$
along along the generators of each hypersurface~$\mathcal{N}_{u}$. For
brevity we write~$\mathcal{U}_{\alpha}$ to denote
$(\mathcal{U}_{\alpha})_{u,v}$. Coordinates~$x_{\alpha}^{\mathcal{A}}$
on~$\mathcal{U}^{\star}_{\alpha}$ are then first propagated into
$\mathcal{N}_{\star}'$ by requiring them to be constant along the
generators of~$\mathcal{N}_{\star}'$. Once coordinates have been
defined on~$\mathcal{N}_{\star}'$, one can propagate them into
$\mathcal{V}$, a suitable neighbourhood of
$\mathcal{N}_\star\cap\mathcal{N}_\star'$, by requiring them to be
constant along the generators of each~$\mathcal{N}_{u}$. In this
manner one obtains a coordinate system
$(x^{\mu})=(v,\ u,\ x^{\mathcal{A}})$ in
$D_{\mathcal{U}_{\alpha}}\equiv\cup_{u,v}(\mathcal{U}_{\alpha})_{u,v}$.

The above coordinate construction leads in a natural way to a
Newman-Penrose (NP) null tetrad~$\{\bml,\bmn,\bm{m},\bar{\bm{m}}\}$
with the vectors~$\bml$ and~$\bmn$ tangent to the generators of the
null hypersurfaces~$\mathcal{N}_{u}$ and~$\mathcal{N}_{v}'$
respectively, and~$m^a$ lying in the tangent space to the
intersection~$\mathcal{S}_{u,v}$. The auxiliary field
$\varphi_a\equiv\frac{1}{3}\nabla_a\varphi$ can be expanded in terms
of the NP frame and we write
\begin{align*}
\varphi_0&\equiv l^a\varphi_a=\frac{1}{3} l^a\nabla_a\varphi=\frac{1}{3}D\varphi, \\
\varphi_2&\equiv n^a\varphi_a=\frac{1}{3}n^a\nabla_a\varphi=\frac{1}{3}\Delta\varphi, \\
\varphi_1&\equiv m^a\varphi_a=\frac{1}{3}m^a\nabla_a\varphi=\frac{1}{3}\delta\varphi.
\end{align*}

Following the same discussion of~\cite{HilValZha19} we make the
following:

\begin{gauge}[\textbf{\em Stewart's choice of the components of the
      frame}]\label{Assumption:Stewarts_Frame}
  {\em On~$\mathcal{V}$ we consider a NP frame of the form
\begin{align}
  \bml=Q\bmpartial_v,\qquad
  \bmn=\bmpartial_u+C^{\mathcal{A}}\bmpartial_{\mathcal{A}}, \qquad
  \bmm=P^{\mathcal{A}} \bmpartial_{\mathcal{A}}, \label{framem}
\end{align}}
where~$C^{\mathcal{A}}=0$ on~$\mathcal{N}_{\star}'$, $\bmm$ and
$\bar\bmm$ span the tangent space
of~$\mathcal{S}_{u,v}$. On~$\mathcal{N}_{\star}$ one has that
$\bml=Q\bmpartial_v$. As the coordinates~$(x^{\mathcal{A}})$ are
constant along the generators of~$\mathcal{N}_{\star}$ and
$\mathcal{N}_{\star}^{\prime}$, it follows that on
$\mathcal{N}_{\star}$ the coefficient~$Q$ is only a function of
$v$. Thus, without loss of generality one can parameterise~$u$ so as
to set~$Q=1$ on~$\mathcal{N}_{\star}$.
\end{gauge}

A direct application of the NP commutators to the
coordinates~$(v,u, x^2, x^3)$ leads to the following:

\begin{lemma}[\textbf{\em conditions on the connection coefficients}]
\label{Lemma1}
The NP frame of Gauge Choice~\ref{Assumption:Stewarts_Frame} can
be chosen such that
\begin{subequations}
\begin{align}
  \kappa&=\nu= \epsilon=0, \label{spinconnection1}\\
  \rho&=\bar{\rho},\ \ \mu=\bar{\mu}, \label{spinconnection2}\\
  \tau&=\bar\alpha+\beta, \label{spinconnection3}
\end{align}
\end{subequations}
on~$\mathcal{V}$ and, furthermore, with
\begin{align*}
\gamma-\bar{\gamma}=0\ \ \ on\ \ \
\mathcal{V}\cap\mathcal{N}_{\star}'. 
\end{align*}
\end{lemma}

\begin{remark}
  {\em Additional commutator relations can be used to obtain equations
    for frame coefficient~$Q$, $P^{\mathcal{A}}$
    and~$C^{\mathcal{A}}$. We have
\begin{subequations}
\begin{align}
DC^{\mathcal{A}}&=(\bar{\tau}+\pi)P^{\mathcal{A}}+(\tau+\bar{\pi})\bar{P}^{\mathcal{A}}, \label{framecoefficient1} \\
DP^{\mathcal{A}}&=\bar\rho P^{\mathcal{A}}+\sigma\bar{P}^{\mathcal{A}}, \label{framecoefficient2} \\
\Delta P^{\mathcal{A}}-\delta C^{\mathcal{A}}&=-(\mu-\gamma+\bar{\gamma})P^{\mathcal{A}}-\bar\lambda\bar{P}^{\mathcal{A}}, \label{framecoefficient3} \\
\Delta Q&=(\gamma+\bar{\gamma})Q,  \label{framecoefficient4}\\
\bar{\delta}P^{\mathcal{A}}-\delta\bar{P}^{\mathcal{A}}&=(\alpha-\bar{\beta})P^{\mathcal{A}}-(\bar{\alpha}-\beta)\bar{P}^{\mathcal{A}}, \label{framecoefficient5} \\
\delta Q&=(\tau-\bar{\pi})Q. \label{framecoefficient6}
\end{align}
\end{subequations}}
\end{remark}

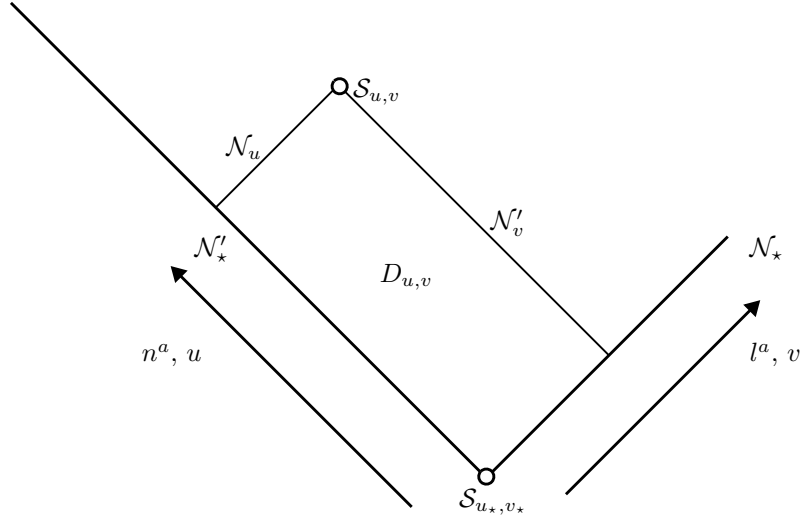
\begin{figure}[t]
\centering
\includegraphics[width=0.8\textwidth]{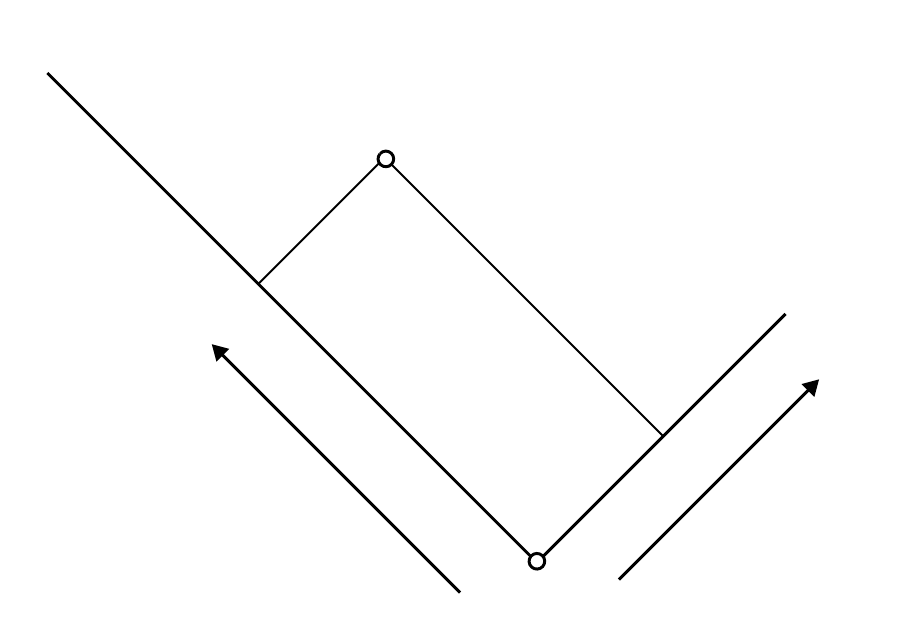}
\put(-40,70){$l^a$, $v$}
\put(-270,70){$n^a$, $u$}
\put(-40,110){$\mathcal{N}_\star$}
\put(-250,110){$\mathcal{N}'_\star$}
\put(-150,15){$\mathcal{S}_{u_\star,v_\star}$}
\put(-190,170){$\mathcal{S}_{u,v}$}
\put(-180,100){${D}_{u,v}$}
\put(-238,148){$\mathcal{N}_u$}
\put(-138,120){$\mathcal{N}'_v$}
\caption{Setup for Stewart's gauge. The construction makes use of a
  double null foliation of the future domain of dependence of the
  initial hypersurface~$\mathcal{N}_\star\cup\mathcal{N}_\star'$. The
  coordinates and NP null tetrad are adapted to this geometric
  setting. See the main text for the definitions of
  the various regions and objects.}
\label{Fig:ExistenceDomain}
\end{figure}

\subsubsection{The renormalised Weyl tensor and Bianchi identities}

Once the NP frame has been specified one can obtain the structure
equations, the Bianchi Identities and the matter equations associated
to the Einstein-scalar field system.  For the scalar field we
introduce an auxiliary field, $\varphi_a$ corresponding to its
derivative. Then the scalar field equation can be recast as a
symmetric hyperbolic system for~$\varphi$ and~$\varphi_a$
---see~\eqref{EOMMasslessScalarNP1}-\eqref{EOMMasslessScalarNP5} in
Appendix~\ref{Appendix:MasslessScalarEquations}. In fact, as it will
be seen in the main part of the article, these equations can be
formulated as a Hodge system (as defined for instance
in~\cite{ChrKla93}).

Recall that from the Einstein field equations one has that the Ricci
curvature equals~$\nabla_a\varphi\nabla_b\varphi$. To expand the
Bianchi Identity with respect to NP frame, we need derivatives of
curvature. It can the be seen that the equations for the Weyl
curvature components~$\Psi_1$, $\Psi_2$ and~$\Psi_3$ have the form
\begin{align*}
(D,\Delta)\bm\Psi-\delta\bm\Psi+\bm\varphi(D,\Delta)\bm\varphi-\bm\varphi\delta\bm\varphi=...
\end{align*}
where~$\bm\varphi$ represents the auxiliary fields~$\varphi_0$,
$\varphi_1$ and~$\varphi_2$. Equations of this form cannot be used to
formulate a Hodge system from which energy estimates can be
extracted. For this, one can only have one null derivative in the
equation. This problem can be solved by introducing the
\emph{renormalised quantities}
\begin{subequations}
\begin{align}
\tilde\Psi_1&\equiv\Psi_1-\Phi_{01}=\Psi_1-3\varphi_0\varphi_1,  \label{NewPsi1}\\
\tilde\Psi_2&\equiv\Psi_2+2\Lambda=\Psi_2-\varphi_0\varphi_2+\varphi_1\bar\varphi_1,  \label{NewPsi2} \\
\tilde\Psi_3&\equiv\Psi_3-\Phi_{21}=\Psi_3-3\varphi_2\bar\varphi_1. \label{NewPsi3}
\end{align}
\end{subequations}
With these three renormalised Weyl curvature components, the terms
$D\bm\varphi$ and~$\Delta\varphi$ in the Bianchi identities are
absorbed in~$D\TiPsi$, $\Delta\TiPsi$ and~$\delta\TiPsi$. Hence, it is
possible to formulate a Hodge system for the renormalised components
of the Weyl curvature. The resulting equations are
\eqref{BianchiMasslessScalarNP1}-\eqref{BianchiMasslessScalarNP8} in
Appendix~\ref{Appendix:MasslessScalarEquations}.

It is worth comparing the above construction with that for the
Einstein-Maxwell system, where the renormalisation
excludes~$\TiPsi_2$. The reason for this difference can be traced back
to the fact that for the Einstein-Maxwell system one
has~$\Lambda=R/24=0$. More details can be found in~\cite{An2022}.

\subsection{The formulation of the characteristic initial value
  problem for the Einstein-scalar system}

In this section we discuss the general aspects of the characteristic
initial value problem (CIVP) for the Einstein-scalar (ES) system with
data on the null hypersurfaces~$\mathcal{N}_{\star}$
and~$\mathcal{N}_{\star}'$. The existence of a hierarchical structure
in the \emph{reduced ES system} leads to a \emph{reduced initial data
  set} from which the full initial data on
$\mathcal{N}_{\star}\cup\mathcal{N}_{\star}'$ for the ES system can be
computed. In a situation where the initial data on
$\mathcal{N}_{\star}'$ is Minkowskian, due to the underlying
symmetries, there is more freedom on the choice of the reduced initial
data.

\begin{lemma}[\textbf{\em freely specifiable data for the CIVP for the
    ES system}]
\label{Lemma2}
Assume that the gauge conditions given by
Lemma~\ref{Assumption:Stewarts_Frame} are satisfied in a
neighbourhood~$\mathcal{V}$ of~$\mathcal{S}_{\star}$. Initial data for
the ES system on~$\mathcal{N}_{\star}\cup\mathcal{N}'_{\star}$ can be
computed from the reduced data set~$\mathbf{r}_\star$ consisting of:
\begin{align*}
  &\Psi_0, \ \ \varphi_0 \quad \mbox{on}\quad
 \mathcal{N}_{\star},\nonumber\\
  &\Psi_4, \ \ \varphi_2, \ \ \gamma+\bar\gamma \quad \mbox{on}\quad \mathcal{N}'_{\star}, \nonumber\\
  &\rho,\ \ \sigma,\ \ \mu,\ \ \lambda, \ \ \tau, \ \ \varphi_1, \ \
  P^{\mathcal{A}}\quad \mbox{on}\quad \mathcal{S}_{\star}. \nonumber
\end{align*}
\end{lemma}
\begin{proof}
  
\smallskip
\noindent
\textbf{Data on~$\mathcal{S}_\star$.} From the frame coefficients
$P^{\mathcal{A}}$ one can compute the intrinsic 2-metric on
$\mathcal{S}_\star$ and the associated NP operators~$\delta$,
$\bar\delta$.  Then from~\eqref{framecoefficient6} one has
$\pi=\tau=\bar\alpha+\beta$. From the 2-metric one can compute
$\beta-\bar\alpha$ and hence we can obtain~$\alpha$
and~$\beta$. From~\eqref{MasslessScalarstructureeq14},
\eqref{MasslessScalarstructureeq17}
and~\eqref{MasslessScalarstructureeq10} we can compute the values
of~$\Psi_1$, $\Psi_2$ and~$\Psi_3$.

\smallskip
\noindent
\textbf{Data on~$\mathcal{N}_\star$.} From~$Q=1$ we have
$D=\partial_v$ and~$\tau=\bar\pi$. From the coefficients~$\Psi_0$ and
$\varphi_0$ we can solve~\eqref{MasslessScalarstructureeq6} and
\eqref{MasslessScalarstructureeq13} together and obtain the value of
$\rho$ and~$\sigma$. Then one can solve the equation for
$P^{\mathcal{A}}$~\eqref{framecoefficient2} and obtain the operators
$\delta$ and~$\bar\delta$ along~$\mathcal{N}_\star$. Combining
\eqref{BianchiMasslessScalarNP2}, \eqref{EOMMasslessScalarNP3} and
\eqref{MasslessScalarstructureeq2} we can compute the values of
$\TiPsi_1$, $\varphi_1$ and~$\tau$. Next, we can solve
\eqref{framecoefficient1} and obtain~$C^{\mathcal{A}}$. From
\eqref{MasslessScalarstructureeq5} and
\eqref{MasslessScalarstructureeq12} we can obtain the value of~$\beta$
and~$\alpha$. Solving~\eqref{MasslessScalarstructureeq1},
\eqref{MasslessScalarstructureeq8},
\eqref{MasslessScalarstructureeq10} and
\eqref{MasslessScalarstructureeq16} together one can obtain the values
of~$\gamma$, $\mu$, $\TiPsi_2$ and~$\lambda$. With the results above
in hand, from~\eqref{BianchiMasslessScalarNP6},
\eqref{BianchiMasslessScalarNP8} and~\eqref{EOMMasslessScalarNP4} we
can obtain the coefficients~$\TiPsi_3$, $\TiPsi_4$ and~$\varphi_2$
on~$\mathcal{N}_\star$.

\smallskip
\noindent
\textbf{Data on~$\mathcal{N}'_\star$.} The condition
$C^{\mathcal{A}}=0$ leads to~$\Delta=\partial_u$. With the value
of~$\gamma+\bar\gamma$ on~$\mathcal{N}'_\star$ we can solve
\eqref{framecoefficient4} and obtain~$Q$. Solving
\eqref{MasslessScalarstructureeq7} and
\eqref{MasslessScalarstructureeq15} together we have~$\mu$
and~$\lambda$. Next, solving~\eqref{framecoefficient3} and we
have~$P^{\mathcal{A}}$ which gives us the operators~$\delta$
and~$\bar\delta$
along~$\mathcal{N}_\star'$. Combining~\eqref{MasslessScalarstructureeq4},
\eqref{MasslessScalarstructureeq11}, \eqref{EOMMasslessScalarNP2},
\eqref{BianchiMasslessScalarNP7}, \eqref{MasslessScalarstructureeq3}
and~\eqref{framecoefficient5} and solving as a system we can
determine~$\beta$, $\alpha$, $\varphi_1$, $\TiPsi_3$, $\pi$
and~$\tau$. Solving~\eqref{MasslessScalarstructureeq9},
\eqref{MasslessScalarstructureeq18} and
\eqref{MasslessScalarstructureeq10} gives the values of~$\sigma$,
$\rho$ and~$\TiPsi_2$. With the results above we can obtain
$\varphi_0$ from~\eqref{EOMMasslessScalarNP1}, obtain~$\TiPsi_1$
from~\eqref{BianchiMasslessScalarNP3} and obtain~$\Psi_0$
from~\eqref{BianchiMasslessScalarNP1}.
\end{proof}

In the present analysis we are interested in a situation in which the
data for the ingoing lightcone is Minkowskian. In this case we have
another choice for the reduced initial data. Namely, one has:

\begin{lemma}[\textbf{\em freely specifiable data of the
    CIVP for ES system on~$\mathcal{N}_{\star}$}]
\label{Lemma3}
 Assume the gauge choice~\ref{Assumption:Stewarts_Frame} and the gauge
conditions implied by Lemma~\ref{Lemma1} are satisfied in a
neighbourhood~$\mathcal{V}$ of~$\mathcal{S}_{\star}$.  In addition
assume that the data on~$\mathcal{N}_{\star}^{\prime}$ is Minkowskian.
Then the freely specifiable initial data for the ES system on
$\mathcal{N}_{\star}$ is given~$\sigma$ and~$\varphi_{0}$.
\end{lemma}

\begin{proof}
  On~$\mathcal{N}_{\star}$ we have~$Q=1$ and this leads to
  $\tau=\bar\pi$ and~$D=\partial_v$.  With the given data~$\sigma$ and
  $\varphi_0$, the value of~$\rho$ can be obtained using the structure
  equation~\eqref{MasslessScalarstructureeq13}.  Then, the value of
  $\Psi_0$ can be directly computed by
  \eqref{MasslessScalarstructureeq6} ---this requires~$D\sigma$.  In
  addition, with the value of~$\rho$ and~$\sigma$, we can obtain
  $P^{\mathcal{A}}$ via~\eqref{framecoefficient2} and, hence, the
  operators~$\delta$ and~$\bar\delta$.  From
  \eqref{EOMMasslessScalarNP3} we can obtain the value of~$\varphi_1$
  ---$\delta\varphi_0$ is needed for
  this. Combining~\eqref{MasslessScalarstructureeq2},
  \eqref{MasslessScalarstructureeq5},
  \eqref{MasslessScalarstructureeq12}
  and~\eqref{BianchiMasslessScalarNP2} and solving together we can
  obtain the values of~$\tau$, $\alpha$, $\beta$ and~$\TiPsi_1$ ---the
  value of~$\delta\Psi_0$ is needed for
  this. From~\eqref{MasslessScalarstructureeq10}, with the value
  of~$\delta\alpha$ and~$\bar\delta\beta$ we can compute~$\TiPsi_2$.
  Solving~\eqref{MasslessScalarstructureeq8}
  and~\eqref{MasslessScalarstructureeq16} together we have~$\mu$
  and~$\lambda$ ---$\delta\pi$ is needed for this. In addition, we can
  obtain~$\gamma$ form~\eqref{MasslessScalarstructureeq1}. Equation
  \eqref{EOMMasslessScalarNP4} leads to~$\varphi_2$
  where~$\bar\delta\varphi_1$ is
  needed. From~\eqref{MasslessScalarstructureeq14} we can
  obtain~$\TiPsi_3$ after $\delta\lambda$ and~$\bar\delta\mu$ have
  been obtained.  With the information above at hand~$\Psi_4$ can be
  obtained via~\eqref{BianchiMasslessScalarNP8}
  where~$\bar\delta\TiPsi_3$ is needed to do this.
\end{proof}

\begin{remark} {\em From the analysis in the previous proof it follows
    that the regularity requirements on~$\sigma$ is are of least four
    angular derivatives and one extra outgoing derivative. We need to
    know at least two angular derivatives of~$\varphi_0$.}
\end{remark}

\begin{remark} {\em The coefficient~$\sigma$ corresponds to the shear
    of the outgoing lightcone. The coefficient~$\varphi_0$ describes
    the rate of change of the scalar field along the
    hypersurface~$\mathcal{N}_{\star}$. This interpretation allows us
    to understand the role of scalar field in the process of formation
    of trapped surfaces. In this work the formation of trapped
    surfaces follows from the interaction of gravity and matter.
    When~$\varphi_0$ is zero, the pure gravity case, our result
    collapses to that of~\cite{An2022}.}
\end{remark}

To apply the standard theory of the CIVP we extract a symmetric
hyperbolic system (SHS) from the ES system. Following the strategy
in~\cite{Hue95}, in order to avoid the derivatives of~$\varphi_a$ in
the Bianchi Identity, to close the equations we introduce the
trace-free symmetric tensor~$\hat{\varphi}_{ab}$ and the
scalar~$\hat{\varphi}$
\begin{align*}
  \hat{\varphi}_{ab}\equiv\nabla_a\varphi_b-\frac{1}{4}\hat{\varphi}g_{ab},
  \qquad
\hat{\varphi}\equiv\nabla_c\varphi^c.
\end{align*}
It follows then that thee ES system is equivalent to the first order
system
\begin{align}
\label{firstorderzeroquantities}
\left\{
\begin{aligned}
&[\bme_{a},\bme_{b}]-(\Gamma_{a\phantom{c}b}^{\phantom{a}c}
-\Gamma_{b\phantom{c}a}^{\phantom{b}c})\bme_{c}=0, \\
&R_{abc}^{\phantom{abc}d}-\rho_{abc}^{\phantom{abc}d}=0,\\
&\nabla_{[a}R_{bc]de}=0,\\
&\nabla_a\varphi-3\varphi_a=0, \\
&\nabla_a\varphi_b-\hat{\varphi}_{ab}-\frac{1}{4}\hat{\varphi}g_{ab}=0, \\
&\hat{\varphi}=0, \\
&\nabla_a\hat{\varphi}=0, \\
&\nabla_{[a}\hat{\varphi}_{b]c}-\frac{1}{2}R_{abc}^{\phantom{abc}d}\varphi_d=0, \\
&R_{ab}-6\varphi_a\varphi_b=0,
\end{aligned}
\right.
\end{align}
where~$\rho_{abc}^{\phantom{abc}d}$ is the \emph{algebraic curvature}
as defined by irreducible decomposition of the Riemann curvature in
terms of the Weyl and Ricci tensors. We denote the above first order
system as~$\mathscr{N}=0$.

We now apply the same 
strategy as in~\cite{HilValZha19} and~\cite{HilValZha20} and define 
\begin{align*}
\bm{e}^t&\equiv\{C^{\mathcal{A}}, P^{\mathcal{A}}, Q\}, \\
\bm{\Gamma}^t&\equiv\{\gamma,\pi,\beta,\mu,\alpha,\lambda,\tau,\sigma,\rho\}, \\
\bm{\Psi}^t&\equiv\{\Psi_0,\Psi_1,\Psi_2,\Psi_3,\Psi_4\}, \\
\bm{\varphi}^t&\equiv\{\varphi_0,\varphi_1,\varphi_2\}, \\
\hat{\bm{\varphi}}^t&\equiv\{\hat{\varphi}_{00},\hat{\varphi}_{01},\hat{\varphi}_{02},\hat{\varphi}_{11},\hat{\varphi}_{12},\hat{\varphi}_{22}\},
\end{align*}
where~$\hat{\varphi}_{00}\equiv\hat{\varphi}_{ab}l^al^b$,
$\hat{\varphi}_{01}\equiv\hat{\varphi}_{ab}l^am^b$,
$\hat{\varphi}_{02}\equiv\hat{\varphi}_{ab}m^am^b$,
$\hat{\varphi}_{11}\equiv\hat{\varphi}_{ab}l^an^b=\hat{\varphi}_{ab}m^a\bar{m}^b$,
$\hat{\varphi}_{12}\equiv\hat{\varphi}_{ab}n^am^b$,
$\hat{\varphi}_{22}\equiv\hat{\varphi}_{ab}n^an^b$. The naming here is
given to coincide with the components in spinorial language. By
selecting particular NP equations
for~$\bm{u}^t=\{\varphi,\bm{e}^t,\bm{\Gamma}^t,\bm{\Psi}^t,\bm{\varphi}^t,\hat{\bm{\varphi}}^t\}$
out of the first order system~$\mathscr{N}=0$, one can extract the SHS
\begin{align}
\label{ESSHS}
\bm{\mathcal{D}}\bm{u}=\bm{B}(\bm{e},\bm{\Gamma},\bm{\varphi},\hat{\bm{\varphi}},\bm{\Psi}), 
\end{align}
where 
\begin{align*}
\bm{\mathcal{D}}\bm{u}=\bm{A}^{\mu}\partial_{\mu}\bm{u},
\end{align*}
and~$\bm{A}^{\mu}$ is Hermitian and~$\bm{A}^{\mu}(l_{\mu}+n_{\mu})$ is positive. 

To complete the analysis one needs to show that a solution to the SHS
\eqref{ESSHS} implies a solution to~$\mathscr{N}=0$ ---that is, one
needs to prove the \emph{propagation of the constraints}. The idea
behind this argument is that one can first define the constraints in
$\mathscr{N}=0$ as zero-quantities and then show that the latter
satisfies an homogeneous SHS under the assumption that~\eqref{ESSHS}
holds. Although straightforward the explicit steps are technical and
lengthy, and so are not given explicitly. A similar example can be
found in~\cite{Fri81b}.

\section{The T-weight formalism}
\label{T-weightFormalism}

In this section we introduce the \emph{T-weight formalism} ---an
extension of the GHP formalism used in our analysis.

\subsection{From the NP to the T-weight formalism}

As a part of the PDE analysis in the present article we routinely need
to integrate the norm of NP quantities globally over a topological
2-sphere~$\mathcal{S}\approx \mathbb{S}^2$. However, strictly
speaking, NP coefficients are only defined on single coordinate
patches. Frames on the overlap region between coordinate patches are
related by rotation~$m^a\mapsto e^{i\theta}m^a$. For a quantity which
transforms like~$f\mapsto e^{is\theta}f$, namely with a definite
spin-weight, the norm~$|f|^2=f\bar{f}$ can be well defined all over
the 2-spheres. Therefore, the integral of~$|f|^2$ over $\mathcal{S}$
makes sense. At the same time, such quantities are connected directly
to tensors over~$\mathcal{S}$ which depend only on the choice
of~$\bm{l}$ and~$\bm{n}$ and do not need a particular choice
of~$\bm{m}$ and~$\bar{\bm{m}}$. Hence, for fixed~$\bm{l}$
and~$\bm{n}$, one can define tensors on a topological
2-sphere~$\mathcal{S}$ associated to these quantities via the
relations
\begin{subequations}
\begin{align}
&\sigma_{ab}\equiv\sigma\bar{m}_a\bar{m_b}, \quad
\lambda_{ab}\equiv\lambda m_am_b, \quad
\tau_a\equiv\tau\bar{m}_a,\quad
\pi_a\equiv\pi m_a, \label{T-weightTensor1}\\
&\kappa_a\equiv\kappa\bar{m}_a, \quad
\nu_a\equiv\nu m_a, \quad
\vartheta_a\equiv\vartheta\bar{m}_a, \label{T-weightTensor2}\\
&(\Psi_0)_{ab}\equiv\Psi_0\bar{m}_a\bar{m_b},\quad
(\Psi_1)_a\equiv\Psi_1\bar{m}_a,\quad
(\Psi_3)_a\equiv\Psi_3m_a,\quad
(\Psi_4)_{ab}\equiv\Psi_4m_am_b, \label{T-weightTensor3}\\
&(\Phi_{01})_a\equiv\Phi_{01}\bar{m}_a,\quad
(\Phi_{02})_{ab}\equiv\Phi_{02}\bar{m}_a\bar{m}_b, \quad
(\Phi_{12})_a\equiv\Phi_{12}\bar{m}_a, \label{T-weightTensor4} \\
&(\Phi_{10})_a\equiv\Phi_{10}m_a,\quad
(\Phi_{20})_{ab}\equiv\Phi_{20}m_am_b, \quad
(\Phi_{21})_a\equiv\Phi_{21}m_a. \label{T-weightTensor5} 
\end{align}
\end{subequations}
where~$\vartheta\equiv\beta+\bar\alpha$. In addition, we consider the
following scalars with~$0$ spin-weight:
\begin{align}
&\ulomega\equiv\gamma+\bar\gamma=l^an^b\nabla_bn_a, \quad
\omega\equiv\epsilon+\bar\epsilon=-n^al^b\nabla_bl_a, \quad
\mu, \quad \rho, \nonumber\\
&\Psi_2,\quad \Phi_{00},\quad\Phi_{11}, \quad \Phi_{22}, \quad
\Lambda=R/24.  \label{T-weightTensor6} 
\end{align}

It is readily verified that the above tensors and scalars are
invariant under rotations~$m^a\mapsto e^{i\theta}m^a$. Therefore, the
tensors and scalars above are defined globally on~$\mathcal{S}$. It
also follows from the above definitions that the norms of the tensors
coincide with the norm of the associated NP quantities.

Taking into account the transformation laws for the above objects, we
define, for fixed~$\bml$ and~$\bmn$, a new weight
(\emph{T-weight})~$s(f)$ to the associated NP scalars:

\begin{table}[h!]
\large
\begin{center}
\caption{T-weight of NP quantities}
\label{QuantityT-weight}
\begin{tabular}{|c|c|}
\hline
$s=-2$ & $\sigma$, $\Psi_0$, $\Phi_{02}$ \\
\hline
$s=-1$&$\tau$, $\kappa$, $\vartheta$, $\Psi_1$, $\Phi_{01}$, $\Phi_{12}$\\
\hline
$s=0$&$\rho$, $\mu$, $\ulomega$, $\omega$, $\Psi_2$, $\Phi_{00}$, $\Phi_{11}$, $\Phi_{22}$, $\Lambda$\\
\hline
$s=1$&$\pi$, $\nu$, $\Psi_3$, $\Phi_{10}$, $\Phi_{21}$\\
\hline
$s=2$&$\lambda$, $\Psi_4$, $\Phi_{20}$\\
\hline
\end{tabular}
\end{center}
\end{table}

Under a rotation~$m_a\mapsto e^{i\theta}m_a$ 
a quantity~$f$ with T-weight~$s$ transforms as
 \begin{align*}
f\mapsto e^{-is\theta}f.
\end{align*}
Hence, for any T-weighted quantities~$f$ and~$g$, we have
 \begin{align*}
s(\bar{f})=-s(f), \qquad
s(fg)=s(f)+s(g).
\end{align*}
 
The coordinates~$\overline{x}=(v,u,x^{\mathcal{A}})$ are defined as
having~$0$ T-weight.  For any quantity~$f$ in
Table~\ref{QuantityT-weight} with T-weight~$s(f)$, denote~$T(f)$ be
its corresponding tensor
in~\eqref{T-weightTensor1}-\eqref{T-weightTensor5} and
\eqref{T-weightTensor6}. Up to this point the T-weight of an objects
coincides with its standard spin-weight up to a conventional minus
sign. \emph{The main differences arise in the following discussions
  about operators.}

\begin{remark}
  {\em Combining the fields
    in~\eqref{T-weightTensor1}-\eqref{T-weightTensor5},
    and~\eqref{T-weightTensor6} and their complex conjugates one can
    construct the following real tensors: 
\begin{subequations}
\begin{align}
&\hat{\chi}_{ab}\equiv\sigma_{ab}+\bar\sigma_{ab}=-\Pi_{\mathcal{S}}(\nabla_{\{b}l_{a\}}), \quad
\underline{\hat{\chi}}_{ab}\equiv\lambda_{ab}+\bar\lambda_{ab}=\Pi_{\mathcal{S}}(\nabla_{\{b}n_{a\}}), \\
&\eta_a\equiv\frac{1}{2}(\tau_a+\bar\tau_a)=-\frac{1}{2}\Pi_{\mathcal{S}}(n^b\nabla_bl_a),\quad
\underline{\eta}_a\equiv\frac{1}{2}(\pi_a+\bar\pi_a)=\frac{1}{2}\Pi_{\mathcal{S}}(l^b\nabla_bn_a),\\
&\zeta_a\equiv\frac{1}{2}\left(\beta+\bar\alpha)\bar{m}_a +(\bar\beta+\alpha)m_a\right)=\frac{1}{2}\Pi_{\mathcal{S}}(l^b\nabla_an_b), \\
&\varkappa_a\equiv\kappa_a+\bar\kappa_a=-\Pi_{\mathcal{S}}(l^b\nabla_bl_a), \quad
\upsilon_a\equiv\nu_a+\bar\nu_a=\Pi_{\mathcal{S}}(n^b\nabla_bn_a),
\end{align}
\end{subequations}
and
\begin{subequations}
\begin{align}
&\alpha_{ab}\equiv(\Psi_0)_{ab}+(\bar\Psi_0)_{ab}=\Pi_{\mathcal{S}}(C_{lalb}),\quad
\beta_a\equiv\frac{1}{2}[(\Psi_1)_a+(\bar\Psi_1)_a]=\frac{1}{2}\Pi_{\mathcal{S}}(C_{lnla}),\\
&\underline{\beta}_a\equiv\frac{1}{2}[(\Psi_3)_a+(\bar\Psi_3)_a]=\frac{1}{2}\Pi_{\mathcal{S}}(C_{nlna}), \quad
\underline{\alpha}_{ab}\equiv(\Psi_4)_{ab}+(\bar\Psi_4)_{ab}=\Pi_{\mathcal{S}}(C_{nanb}).
\end{align}
\end{subequations}
In the above expressions~$\Pi_{\mathcal{S}}$ denotes the projection on
to~$\mathcal{S}_{u,v}$ and indices~${\{ab \}}$ denote the operation of
taking the traceless part. The notation on the left-hand-side of these
expressions is consistent with that of Christodoulou in~\cite{Chr00}.}
\end{remark}

\subsection{T-weight Operators and Norms}

Next, we introduce differential operators which preserve the T-weight
as well as reflect the relation with tensor and covariant
derivative. A direct computation shows that the standard
operators~$D$,$\Delta$, $\delta$ and~$\bar\delta$ do not interact
simply with the T-weight.

The differential operators of the GHP formalism preserve the
spin-weight property but cannot be applied to quantities with
non-vanishing boost weight. To overcome this we overwrite the usual
GHP notation, and define four new differential operators~$\meth$,
$\meth'$, $\mthorn$ and~$\mthorn'$
\begin{subequations}
\begin{align}
\meth f&\equiv\delta f+s(\beta-\bar\alpha)f, \label{DefT-weighteth}\\
\meth' f&\equiv\bar\delta f-s(\bar\beta-\alpha)f, \label{DefT-weighteth'}\\
\mthorn f&\equiv Df+s(\epsilon-\bar\epsilon)f, \label{DefT-weightthorn}\\
\mthorn' f&\equiv\Delta f+s(\gamma-\bar\gamma)f,\label{DefT-weightthorn'}
\end{align}
\end{subequations}
acting on any quantity~$f$ with defined T-weight~$s$.  One can check
that under a rotation~$m^a\mapsto e^{i\theta}m^a$, these operators
transform as
\begin{align*}
&\meth f\mapsto e^{-i(s-1)\theta}\meth f, \qquad
\meth' f\mapsto e^{-i(s+1)\theta}\meth' f, \\
&\mthorn f\mapsto e^{-is\theta}\mthorn f, \qquad
\mthorn' f\mapsto e^{-is\theta}\mthorn' f.
\end{align*}
Thus we set
\begin{align}
&s(\meth f)=s(f)-1, \quad
s(\meth' f)=s(f)+1, \quad
s(\mthorn f)=s(f)=s(\mthorn'f).
\end{align}
Moreover, one can check that these four operators follow the Leibnitz
rule.  Making use of~$s(\bar{f})=-s(f)$, one has
\begin{align}
\meth'\bar{f}&=\overline{\meth f}, \quad
\meth\bar{f}=\overline{\meth'f}, \\
\overline{\mthorn f}&=\mthorn\bar{f}, \quad
\overline{\mthorn' f}=\mthorn'\bar{f}.
\end{align}

In our analysis we will also need to consider NP coefficients which do
not transform homogeneously under a rotation~$f\mapsto
e^{is\theta}f$. These include the combinations
\begin{align*}
  \beta-\bar\alpha,\qquad \epsilon-\bar\epsilon, \qquad \gamma-\bar\gamma.
\end{align*} 
These quantities are related to the covariant derivative of the frame
on the 2-sphere along~$\bml$ and~$\bmn$. This connection can be seen
by introducing two new operators~$\mathop{D}\limits_{\gets}$
and~$\mathop{\Delta}\limits_{\gets}$ which are the projections
to~$\mathcal{S}$ of the covariant derivatives~$l^a\nabla_a$
and~$n^a\nabla_a$. Denoting by~$\nablasl_a$ the induced connection,
i.e. the covariant derivative on~$\mathcal{S}$ we have that
\begin{align*}
\nablasl_am_b&=(\beta-\bar\alpha)\bar{m}_am_b-(\bar\beta-\alpha)m_am_b, \\
\mathop{D}\limits_{\gets}m_a&=(\epsilon-\bar\epsilon)m_a, \quad
\mathop{\Delta}\limits_{\gets}m_a=(\gamma-\bar\gamma)m_a.
\end{align*}

Now, given a T-weighted quantity~$f$, denote its associated tensor
(cf. equations~\eqref{T-weightTensor1}-\eqref{T-weightTensor5})
by~$T(f)$. It follows from the preceding discussion that~$\mthorn f$
is the single non-vanishing component
of~$\mathop{D}\limits_{\gets} T(f)$; $\mthorn' f$ is the single
non-vanishing component of~$\mathop{\Delta}\limits_{\gets} T(f)$;
$\meth^kf$, $\meth^{k-1}\meth' f$, ..., $\meth'^kf$ are the components
of~$\nablasl_{a_1}...\nablasl_{a_k}T(f)$ and so on. Taking the NP
scalar~$\lambda$ as an example one readily verifies
that~$s(\lambda)=2$. Moreover, we have that
\begin{align*}
\nablasl_{a_1}\lambda_{ab}&=[\bar\delta\lambda-2(\bar\beta-\alpha)\lambda]m_{a_1}m_am_b+
[\delta\lambda+2(\beta-\bar\alpha)\lambda]\bar{m}_{a_1}m_am_b \\
&=(\meth'\lambda)m_{a_1}m_am_b+(\meth\lambda)\bar{m}_{a_1}m_am_b, \\
\mathop{D}\limits_{\gets}\lambda_{ab}&=[D\lambda+2(\epsilon-\bar\epsilon)\lambda]m_am_b
=(\mthorn\lambda)m_am_b, \\
\mathop{\Delta}\limits_{\gets}\lambda_{ab}&=[\Delta\lambda+2(\gamma-\bar\gamma)\lambda]m_am_b
=(\mthorn'\lambda)m_am_b.
\end{align*}
One can also verify the transformation rules
\begin{align*}
\delta\lambda+2(\beta-\bar\alpha)\lambda&\mapsto e^{-i\theta}[\delta\lambda+2(\beta-\bar\alpha)\lambda], \\
\bar\delta\lambda-2(\bar\beta-\alpha)\lambda&\mapsto e^{-3i\theta}[\bar\delta\lambda-2(\bar\beta-\alpha)\lambda], \\
D\lambda+2(\epsilon-\bar\epsilon)\lambda&\mapsto e^{-2i\theta}[D\lambda+2(\epsilon-\bar\epsilon)\lambda], \\
\Delta\lambda+2(\gamma-\bar\gamma)\lambda&\mapsto e^{-2i\theta}[\Delta\lambda+2(\gamma-\bar\gamma)\lambda].
\end{align*}
One also has that
\begin{align*}
|\nablasl_{a_1}\lambda_{ab}|^2=|\meth'\lambda|^2+|\meth\lambda|^2.
\end{align*}
Similarly, one can compute~$\nablasl_{a_2}\nablasl_{a_1}\lambda_{ab}$
and has
\begin{align*}
|\nablasl_{a_2}\nablasl_{a_1}\lambda_{ab}|^2=
|\meth'^2\lambda|^2+|\meth\meth'\lambda|^2+|\meth'\meth\lambda|^2+|\meth^2\lambda|^2.
\end{align*}
Hence, the norm of~$\nablasl^kT(f)$ can be computed in terms of the
norm of all its components~$...\meth...\meth'...f$, i.e. we have
\begin{align*}
  |\nablasl^kT(f)|^2
  =\sum_{\alpha}|\mathcal{D}^{k_i}f|^2,
\end{align*}
where~$\mathcal{D}^{k_i}f$ is a string of order~$k$ of the
operators~$\meth$ and~$\meth'$, and the sum over~$\alpha$ denotes all
such strings. The set of all possible strings~$\mathcal{D}^{k_i}$ is
contains all the possible combinations of~$\meth$ and~$\meth'$ of
order~$k$. As~$\meth$ and~$\meth'$ do not, in general, commute, the
order of~$\meth$ and~$\meth'$ in the string matters. For example
when~$k=3$ we have the combinations
\begin{align*}
\meth^3f, \quad \meth'\meth^2f, \quad \meth\meth'\meth f, \quad \meth'^2\meth f, \quad
\meth'\meth\meth'f, \quad \meth^2\meth'f, \quad \meth\meth'^2f,\quad \meth'^3 f.
\end{align*}
To obtain all the fourth order combinations one just needs to
apply~$\meth$ and~$\meth'$ to each of the terms above. We see that, in
general, there are~$2^k$ terms in the collection of strings
in~$\mathcal{D}^{k_i}f$.
Proceeding in this
way leads to the norm on $\mathcal{S}$
\begin{align*}
  ||\nablasl^kT(f)||^2_{L^2(\mathcal{S})}
  =\sum_{\alpha}\int_{\mathcal{S}}|\mathcal{D}^{k_i}f|^2.
\end{align*}

Now we can introduce the norms of T-weighted quantities ---namely, we
define
\begin{align}
\label{T-weightL2Norm}
||\mathcal{D}^kf||^2_{L^2(\mathcal{S})}\equiv\int_{\mathcal{S}}|\mathcal{D}^{k}f|^2,
\end{align}
where
\begin{align}
\label{T-weightL2Norm01}
|\mathcal{D}^{k}f|^2\equiv \sum_{\alpha}|\mathcal{D}^{k_i}f|^2.
\end{align}
We use the notation~$\mthorn\mathcal{D}^{k_i}f$
and~$\mthorn'\mathcal{D}^{k_i}f$ to denote the~$\mthorn$
and~$\mthorn'$-derivatives of~$\mathcal{D}^{k_i}f$
respectively. Making use of the commutator
relations~\eqref{T-weightCommutatorAlt2} and
\eqref{T-weightCommutatorAlt4}, one readily finds that the structure
of the equations for the various
quantities~$\mthorn\mathcal{D}^{k_i}f$ is similar. The only
differences arise in the terms involving~$\meth$, $\meth'$ and the
conjugates of T-weighted quantities.  These differences disappear once
the norm and the sum in~\eqref{T-weightL2Norm} are evaluated ---this
is an important observation in the analysis of the transport equations
and the construction of energy estimates. A similar discussion applies
to~$\mthorn'\mathcal{D}^{k_i}f$. The differences between~$\meth$
and~$\meth'$ will be addressed in the construction of elliptic
estimates.

Similarly to the above we define norms
\begin{align}
\label{T-weightLpNorm}
||\mathcal{D}^kf||^p_{L^p(\mathcal{S})}\equiv\int_{\mathcal{S}}|\mathcal{D}^kf|^p,
\end{align}
and 
\begin{align}
\label{T-weightInftyNorm}
||\mathcal{D}^kf||_{L^{\infty}(\mathcal{S})}\equiv\sup_{\mathcal{S}}|\mathcal{D}^kf|.
\end{align}

\begin{remark}
{\em In the previous discussion, the notation~$T(f)$ has been used to
denote the tensor associated to~$f$ according to Table
\ref{QuantityT-weight}. When we apply the derivatives~$\meth$ or~$\meth'$
to~$f$, its associated tensor is defined by
\begin{align*}
...\meth...\meth'...f \mapsto (...\meth...\meth'...f)...\bar{m}_a...m_b...
\end{align*}
In the following discussions, we will only consider quantities~$f$ and
associated tensors~$T(f)$ appearing in Table~\ref{QuantityT-weight}
and in equations~\eqref{T-weightTensor1}-\eqref{T-weightTensor5}
and~\eqref{T-weightTensor6}. From the definition
in~\eqref{T-weightL2Norm01}, for any two quantities~$f$ and~$g$,
the norm of the tensor product $\nablasl^kT(f)\nablasl^{k'}T(g)$ 
can be naturally defined as follows:
\begin{align*}
  |\nablasl^kT(f)\nablasl^{k'}T(g)|^2
  =(\sum_{\alpha}|\mathcal{D}^{k_i}f|^2)(\sum_{\beta}|\mathcal{D}^{k'_i}g|^2)
  =|\mathcal{D}^{k}f|^2|\mathcal{D}^{k'}g|^2.
\end{align*}
Integral on $\mathcal{S}$ we have 
\begin{align*}
||\mathcal{D}^kf\mathcal{D}^{k'}g||^2_{L^2(\mathcal{S})}&\equiv
  ||\nablasl^kT(f)\nablasl^{k'}T(g)||^2_{L^2(\mathcal{S})} \\
  &=\int_{\mathcal{S}}(\sum_{\alpha}|\mathcal{D}^{k_i}f|^2)(\sum_{\beta}|\mathcal{D}^{k'_i}g|^2) \\
  &=\int_{\mathcal{S}}|\mathcal{D}^{k}f|^2|\mathcal{D}^{k'}g|^2.
\end{align*}
Hence for any product $\mathcal{D}^{k_i}f\mathcal{D}^{k'_j}g$ we have
\begin{align*}
|\mathcal{D}^{k_i}f\mathcal{D}^{k'_j}g|\leq|\mathcal{D}^{k}f||\mathcal{D}^{k'}g|,
\end{align*}
which leads to
\begin{align*}
  ||\mathcal{D}^{k_i}f\mathcal{D}^{k'_j}g||_{L^2(\mathcal{S})}
  \leq||\mathcal{D}^{k}f\mathcal{D}^{k'}g||_{L^2(\mathcal{S})}.
\end{align*}
This follows from the definition given in~\eqref{T-weightLpNorm}. For
expressions involving the product of more than two terms we have
similar results. These observations supplement the previous
discussions on the relation between the norms for terms
involving~$\mathcal{D}^{k_i}$ and~$\mathcal{D}^{k}$.}
\end{remark}

\begin{remark}[Relation to the standard GHP formalism]
  {\em The connection coefficients, $\kappa$, $\rho$, $\sigma$,
    $\tau$, $\nu$, $\mu$, $\lambda$ and~$\pi$ have GHP
    weights~$\{3,1\}$, $\{1,1\}$, $\{3,-1\}$, $\{1,-1\}$, $\{-3,-1\}$,
    $\{-1,-1\}$, $\{-3,1\}$ and~$\{-1,1\}$ respectively. In addition
    to these coefficients which also are considered in the GHP
    formalism, we consider further three coefficients with definite
    T-weight ---namely, $\ulomega=\gamma+\bar\gamma$,
    $\omega=\epsilon+\bar\epsilon$
    and~$\vartheta=\beta+\bar\alpha$. Now, the Weyl curvature
    coefficients~$\Psi_0$, $\Psi_1$, $\Psi_2$, $\Psi_3$ and~$\Psi_4$
    have weight $\{4,0\}$, $\{2,0\}$, $\{0,0\}$, $\{-2,0\}$
    and~$\{-4,0\}$ respectively.  The Ricci curvature
    coefficients~$\Phi_{00}$, $\Phi_{10}$, $\Phi_{01}$, $\Phi_{11}$,
    $\Phi_{20}$, $\Phi_{21}$, $\Phi_{02}$, $\Phi_{12}$ and~$\Phi_{22}$
    have weight~$\{2,2\}$, $\{0,2\}$, $\{2,0\}$, $\{0,0\}$,
    $\{-2,2\}$, $\{-2,0\}$, $\{2,-2\}$, $\{0,-2\}$ and~$\{-2,-2\}$
    respectively. For any quantity~$f$ with GHP weight~$\{w,z\}$ we
    have that
\begin{align*}
\mthorn_Gf&=Df+(-w\epsilon-z\bar\epsilon)f, \qquad
\mthorn_G'f=\Delta f+(-w\gamma-z\bar\gamma)f, \\
\meth_G f&=\delta f+(-w\beta-z\bar\alpha)f, \qquad
\meth_G'f=\bar\delta f+(-w\alpha-z\bar\beta)f,
\end{align*}
where we have used the subscript~$G$ to denote the GHP
operators. Similarly, we use the subscript~$T$ to denote the T-weight
operators. Then, for quantity~$f$ we have that
\begin{align*}
\meth_Tf-\meth_Gf&=[(w+s)\beta+(z-s)\bar\alpha]f, \qquad
\meth'_Tf-\meth'_Gf=[(z-s)\bar\beta+(w+s)\alpha]f, \\
\mthorn_Tf-\mthorn_Gf&=[(w+s)\epsilon+(z-s)\bar\epsilon]f \qquad
\mthorn'_Tf-\mthorn'_Gf=[(w+s)\gamma+(z-s)\bar\gamma]f.
\end{align*}
One can check that for quantities with zero boost-weight, the GHP
operators coincide with the T-weight operators. The difference comes
arises for non-zero boost-weighted quantities. To translate between
the formalism we can use the maps~$w\mapsto s=\frac{w-z}{2}$
and~$z\mapsto-s=-\frac{w-z}{2}$.}
\end{remark}

The NP equations in the T-weight formalism are given in
Appendix~\ref{NPT-weightEq}.

We conclude this section by noting the following lemma:

\begin{lemma}
  Suppose~$f_1$ and~$f_2$ are two real scalar fields with zero
  T-weight. Then
\begin{align*}
  ||\mathcal{D}^k(f_1+if_2)||^2_{L^2(\mathcal{S})}
  =||\mathcal{D}^kf_1||^2_{L^2(\mathcal{S})}
  +||\mathcal{D}^kf_2||^2_{L^2(\mathcal{S})}.
\end{align*}
\end{lemma}
\begin{proof}
  The proof follows from a direct computation. Namely, one has that
\begin{align*}
||\mathcal{D}^k(f_1+if_2)||^2_{L^2(\mathcal{S})}&=||\nablasl^k(f_1+if_2)||^2_{L^2(\mathcal{S})} \\
&=\int_{\mathcal{S}}[\nablasl^{a_1,...,a_k}(f_1+if_2)][\nablasl_{a_1,...,a_k}(f_1-if_2)] \\
&=\int_{\mathcal{S}}\nablasl^{a_1,...,a_k}f_1\nablasl_{a_1,...,a_k}f_1
+\int_{\mathcal{S}}\nablasl^{a_1,...,a_k}f_2\nablasl_{a_1,...,a_k}f_2 \\
&=||\mathcal{D}^kf_1||^2_{L^2(\mathcal{S})}
+||\mathcal{D}^kf_2||^2_{L^2(\mathcal{S})}.
\end{align*}
\end{proof}

\begin{remark} {\em The above discussion shows that the
    terms~$\beta-\bar\alpha$, $\epsilon-\bar\epsilon$
    and~$\gamma-\bar\gamma$ allow us to connect the derivative of NP
    quantities~$f$ with T-weight and the derivative of
    tensors~$T(f)$. In what follows we will not distinguish~$f$
    and~$T(f)$ where there is no source for confusion.}
\end{remark}

\subsection{The T-weight formalism and PDE analysis}
\label{T-weightPDEanalysis}

In this section we discuss applications of the T-weight formalism to
the construction of estimates to solutions to differential equations
in a number of scenarios.

\subsubsection{Transport estimates}

Given a quantity~$f$ such that~$s(f)=s$ we have that
\begin{align*}
  \mthorn(f\bar{f})=D(f\bar{f})=f\mthorn\bar{f}+\bar{f}\mthorn f, \qquad
  \mthorn'(f\bar{f})=\Delta(f\bar{f})=f\mthorn'\bar{f}+\bar{f}\mthorn' f.
\end{align*}
It follows that in order to obtain~$L^2$-estimates of~$\nablasl_kT(f)$
one can apply the transport identities
\begin{align}
&\frac{\mathrm{d}}{\mathrm{d}v}\int_{\mathcal{S}_{u,v}}
\phi=\int_{\mathcal{S}_{u,v}}Q^{-1}\left(D\phi-(\rho+\bar\rho)\phi\right),\label{DerivativeTwoSphere1}\\
&\frac{\mathrm{d}}{\mathrm{d}u}\int_{\mathcal{S}_{u,v}}
\phi=\int_{\mathcal{S}_{u,v}}\left(\Delta\phi+(\mu+\bar\mu)\phi\right) \label{DerivativeTwoSphere2}
\end{align}
to~$\phi\equiv f\bar{f}$, where we have made use of the expressions
for~$\bml$ and~$\bmn$ in Stewart's gauge ---namely
\begin{align*}
\bml=Q\bmpartial_v, \qquad  \bmn=\bmpartial_u+C^{\mathcal{A}}\bmpartial_{\mathcal{A}}.
\end{align*}
As an example of the above discussion consider~$\phi=|\meth^kf|^2$. In
this case the derivatives~$\mthorn\meth^kf$ and~$\mthorn'\meth^kf$
encode the typical structural property
of~$\mathop{D}\limits_{\gets}\nablasl^kT(f)$ and
$\mathop{\Delta}\limits_{\gets}\nablasl^kT(f)$. It follows then that
\begin{align*}
\frac{\mathrm{d}}{\mathrm{d}v}\int_{\mathcal{S}_{u,v}}
|\meth^kf|^2=\int_{\mathcal{S}_{u,v}}Q^{-1}\left(D|\meth^kf|^2-(\rho+\bar\rho)|\meth^kf|^2\right).
\end{align*}
The left-hand side can be expanded as
\begin{align*}
\frac{\mathrm{d}}{\mathrm{d}v}||\meth^kf||^2_{L^2(\mathcal{S}_{u,v})}
=2||\meth^kf||_{L^2(\mathcal{S}_{u,v})}\frac{\mathrm{d}}{\mathrm{d}v}||\meth^kf||_{L^2(\mathcal{S}_{u,v})},
\end{align*}
while the right hand side can be estimated as
\begin{align*}
&\int_{\mathcal{S}_{u,v}}Q^{-1}\left(\mthorn|\meth^kf|^2-(\rho+\bar\rho)|\meth^kf|^2\right) \\
&=\int_{\mathcal{S}_{u,v}}Q^{-1}\left(\overline{\meth^k f}\mthorn\meth^kf+\meth^kf\mthorn\overline{\meth^kf}-(\rho+\bar\rho)|\meth^kf|^2 \right) \\
&\leq\int_{\mathcal{S}_{u,v}}2Q^{-1}\left(|\meth^kf|\dot|\mthorn\meth^kf|+|\rho||\meth^kf|^2 \right) \\
&\leq2||\meth^kf||_{L^2(\mathcal{S}_{u,v})}\left(||Q^{-1}\mthorn\meth^kf||_{L^2(\mathcal{S}_{u,v})}
+||\rho||_{L^{\infty}(\mathcal{S}_{u,v})}||Q^{-1}\meth^kf||_{L^2(\mathcal{S}_{u,v})} \right).
\end{align*}
Combining the above expressions one concludes that
\begin{align*}
\frac{\mathrm{d}}{\mathrm{d}v}||\meth^kf||_{L^2(\mathcal{S}_{u,v})}\leq
\left(||Q^{-1}\mthorn\meth^kf||_{L^2(\mathcal{S}_{u,v})}
+||\rho||_{L^{\infty}(\mathcal{S}_{u,v})}||Q^{-1}\meth^kf||_{L^2(\mathcal{S}_{u,v})} \right).
\end{align*}
Similarly, one can show the inequality
\begin{align*}
\frac{\mathrm{d}}{\mathrm{d}u}||\meth^kf||_{L^2(\mathcal{S}_{u,v})}\leq
\left(||\mthorn'\meth^kf||_{L^2(\mathcal{S}_{u,v})}
+||\mu||_{L^{\infty}(\mathcal{S}_{u,v})}||\meth^kf||_{L^2(\mathcal{S}_{u,v})} \right).
\end{align*}

\subsubsection{Elliptic estimates}
\label{T-weightElliptic}

We are particularly interested in constructing estimates for the
divergence and curl operators on the 2-spheres~$\mathcal{S}_{u,v}$.
We note that the volume element~$\varepsilon_{ab}$ can be written
as~$\varepsilon_{ab}=i(\bar{m}_am_b-m_a\bar{m}_b)$. For any NP
quantity $f$ related tensor~$T(f)$ in Table~\ref{QuantityT-weight}, we
define
\begin{align*}
(\mathrm{div}T(f))_{a_1...a_r}\equiv\nablasl^aT(f)_{aa_1...a_r}, \qquad
(\mathrm{curl}T(f))_{a_1...a_r}\equiv\varepsilon^{ab}\nablasl_aT(f)_{ba_1...a_r},
\end{align*}
where~$r=0,1$. One can readily verify that
\begin{align*}
\mathrm{curl}T(f)&=\mathrm{i}(\mathrm{div}T(f))=\mathrm{i}T(\meth f), \quad &\mbox{for}\quad s(f)>0, \\
\mathrm{curl}T(f)&=-\mathrm{i}(\mathrm{div}T(f))=-\mathrm{i}T(\meth'f),
                    \quad &\mbox{for} \quad s(f)<0.
\end{align*}
It follows that the curl of the tensor~$T(f)$ is essentially encoded
in its divergence. For quantities in Table~\ref{QuantityT-weight} with
nonzero T-weight, we will use the T-weighted scalar~$\mathscr{D}_f$ to
denote its divergence:
\begin{align}
\mathscr{D}_f=\left\{
\begin{aligned}
\label{T-weightDiv}
\meth f &\quad s(f)>0, \\
\meth' f &\quad s(f)<0.
\end{aligned}
\right.
\end{align}

When~$f$ has~$s=0$, we have that 
\begin{align*}
\bm{\Delta}f=(\meth\meth'+\meth'\meth)f,
\end{align*} 
where~$\bm{\Delta}$ is the Laplacian~$\nablasl^a\nablasl_a$.  Now, for
a double null foliation we have~$\mu=\bar\mu$
and~$\rho=\bar\rho$. This leads to
\begin{align*}
\bm{\Delta}f=2\meth'\meth f=2\meth\meth' f.
\end{align*} 

\subsubsection{Stokes theorem}
\label{T-weightStokes}

As a topological 2-sphere~$\mathcal{S}$ has no boundary, it follows
then from Stokes theorem that for any 1-form~$h_a$ on~$\mathcal{S}$,
the integral of a divergence vanishes ---that is
\begin{align*}
\int_{\mathcal{S}}\nablasl^ah_a=0.
\end{align*}

Now, suppose we have two T-weighted quantities~$f$ and~$g$ satisfying
$s(f)+s(g)-1=0$ and define~$h\equiv fg$. One can verify that
$s(h)=1$. It follows then that the associated tensor is of the form
\begin{align*}
h_a=hm_a
\end{align*} 
which is globally defined on~$\mathcal{S}$. One can check that
\begin{align*}
\mathrm{div}\bm{h}=\meth h=\meth(fg)=f\meth g+g\meth f, \qquad
s(\mathrm{div}h)=0,
\end{align*} 
from which it follows that~$f\meth g+g\meth f$ is a globally defined
scalar on~$\mathcal{S}$. Thus, applying Stokes Theorem one has
\begin{align*}
0=\int_{\mathcal{S}}\mathrm{div}h
=\int_{\mathcal{S}}f\meth g+g\meth f \Rightarrow 
\int_{\mathcal{S}}f\meth g=-\int_{\mathcal{S}}g\meth f.
\end{align*}  
Similarly, when~$f$ and~$g$ satisfy 
\begin{align*}
s(f)+s(g)+1=0,
\end{align*} 
we have that
\begin{align*}
\int_{\mathcal{S}}f\meth' g=-\int_{\mathcal{S}}g\meth' f.
\end{align*} 

We summarise the above discussion in the following:

\begin{lemma}
  \label{IntegralbyPartT-weight}
$\phantom{X}$  
\begin{itemize}
\item[(i)] For a T-weighted quantity~$h$ with~$s(h)=1$ we have that
\begin{align*}
\int_{\mathcal{S}}\meth h=0.
\end{align*} 
If, on the other hand~$s(h)=-1$, we have that
\begin{align*}
\int_{\mathcal{S}}\meth'h=0.
\end{align*} 

\item[(ii)] For T-weighted quantities~$f$ and~$g$ on~$\mathcal{S}$ satisfying
\begin{align*}
s(f)+s(g)-1=0,
\end{align*}
we have that
\begin{align*}
\int_{\mathcal{S}}f\meth g=-\int_{\mathcal{S}}g\meth f.
\end{align*} 
 If, on the other hand, we have that 
\begin{align*}
s(f)+s(g)+1=0,
\end{align*}
then we conclude that 
\begin{align*}
\int_{\mathcal{S}}f\meth' g=-\int_{\mathcal{S}}g\meth' f.
\end{align*} 
\end{itemize}
\end{lemma}

The above lemma can be used to deal with a system of the form
\begin{align*}
\left\{\begin{aligned}
\mthorn'\Psi_{I}-\meth\Psi_{II}&=P, \\
\mthorn\Psi_{II}-\meth'\Psi_{I}&=Q,
\end{aligned} 
\right.
\end{align*}
where the T-weights of~$\Psi_{I}$ and~$\Psi_{II}$ satisfy
\begin{align*}
s(\Psi_{I})+s(\bar\Psi_{II})+1=0, \qquad
s(\Psi_{II})+s(\bar\Psi_{I})-1=0.
\end{align*} 
From Lemma~\ref{IntegralbyPartT-weight} it follows then that
\begin{align*}
\int_{\mathcal{S}}\Psi_{I}\meth'\bar\Psi_{II}&=-\int_{\mathcal{S}}\bar\Psi_{II}\meth'\Psi_{I} ,\\
\int_{\mathcal{S}}\Psi_{II}\meth\bar\Psi_{I}&=-\int_{\mathcal{S}}\bar\Psi_{I}\meth\Psi_{II}.
\end{align*} 
Therefore, the above system leads to
\begin{align*}
\int_{\mathcal{S}}\Delta(|\Psi_{I}|^2)+\int_{\mathcal{S}}D(|\Psi_{II}|^2)=
\int_{\mathcal{S}}\left(\Psi_{I}\bar{P}+\bar\Psi_{I}P+\Psi_{II}\bar{Q}+\bar\Psi_{II}Q \right).
\end{align*} 

Recall now that on a double null foliation one has that
$\bar\rho=\rho$ and~$\bar\mu=\mu$. Hence, for a T-weight quantity~$f$
with T-weight~$s$ the commutation relation~\eqref{T-weightCommutator6}
simplifies to
\begin{align*}
(\meth\meth'-\meth'\meth)f=sKf,
\end{align*}
with~$K$ the Gaussian curvature. It follows then that 
\begin{align*}
\int_{\mathcal{S}}|\meth f|^2&=\int_{\mathcal{S}}\overline{\meth f}\meth f
=\int_{\mathcal{S}}\meth'\bar{f}\meth f=-\int_{\mathcal{S}}\bar{f}\meth'\meth f
=-\int_{\mathcal{S}}\bar{f}\meth\meth'f+\int_{\mathcal{S}}sK|f|^2 \\
&=\int_{\mathcal{S}}\meth\bar{f}\meth'f+\int_{\mathcal{S}}sK|f|^2
=\int_{\mathcal{S}}|\meth'f|^2+\int_{\mathcal{S}}sK|f|^2.
\end{align*}
In the particular case of a scalar~$f$ zero T-weight, we have that 
\begin{align*}
\int_{\mathcal{S}}|\meth f|^2
=\int_{\mathcal{S}}|\meth'f|^2.
\end{align*} 

\subsubsection{Sobolev embeddings} 

In this subsection we briefly recall results from~\cite{Chr00},
Chapter 5.2, but adapted here to the T-weight formalism. Given any
T-weighted quantity~$f$, we have that
\begin{align}
  \label{SobolevLp1}
\left(\mbox{Area}(\mathcal{S})\right)^{-1/p}||f||_{L^p(\mathcal{S})}\leq
 C_p\sqrt{I'(\mathcal{S})}\left(\left(\mbox{Area}(\mathcal{S})\right)^{-1/2}||f||_{L^2(\mathcal{S})}
 +||\mathcal{D}f||_{L^2(\mathcal{S})}\right),
\end{align}
for~$2<p<\infty$, and 
\begin{align}
  \label{SobolevLinfty1}
||f||_{L^{\infty}(\mathcal{S})}\leq
 C_p\sqrt{I'(\mathcal{S})}\left(\mbox{Area}(\mathcal{S})\right)^{1/2-1/p}\left(||f||_{L^p(\mathcal{S})}
 +\left(\mbox{Area}(\mathcal{S})\right)^{-1/2}||\mathcal{D}f||_{L^2(\mathcal{S})}\right),
\end{align}
for~$p>2$, where~$C_p$ is a numerical constant depending only on~$p$
and
\begin{align*}
I'(S)\equiv \max\{I(\mathcal{S}),1\},
\end{align*}
where~$I(\mathcal{S})$ is the isoperimetric constant of~$\mathcal{S}$.
 
\subsection{Local existence results in terms of the T-weight
  formalism}

In this section we make use of the T-weight formalism to reformulate
the local existence results for the characteristic initial value
problem given in~\cite{HilValZha19}. This reformulation can be
directly obtained by a replacement of the spin connection coefficients
and differential operators by their T-weighted
counterpart. Accordingly, we only provide the main statements of the
theorems.

The article~\cite{HilValZha19} is concerned with the local existence
of CIVP for the vacuum Einstein field equations (EFE). The analysis is
done in a coordinate system and gauge choice which are slightly
different from those considered in the previous section. In this case
the statement for the freely specifiable data reads as follows:

\begin{lemma}[\textbf{\em Freely specifiable data for the CIVP}]
  \label{Lemma:FreeDataCIVP} Working in the gauge discussed in
  Section~2 in~\cite{HilValZha19}, initial data for the vacuum
  Einstein field equations
  on~$\mathcal{N}_{\star}\cup\mathcal{N}_{\star}^{\prime}$ can be
  computed (near~$\mathcal{S}_{\star}$) from a reduced data
  set~$\mathbf{r}_\star$ consisting of:
\begin{align*}
  &\Psi_0, \; \omega \quad \mbox{\textrm{on}}
  \quad \mathcal{N}_{\star},\nonumber\\
  &\Psi_4\ \ \mbox{\textrm{on}}\ \ \mathcal{N}_{\star}^{\prime}, \nonumber\\
  &\lambda,\ \ \sigma,\ \ \mu,\ \ \rho, \ \ \pi, \ \
  P^{\mathcal{A}}\ \ \mbox{\textrm{on}}\ \ \mathcal{S}_{\star}. \nonumber
\end{align*}
\end{lemma}

The existence result can be reformulated as:

\begin{theorem}[\textbf{\em Improved local existence for the CIVP for
    the EFE}]
\label{MainTheoremEFE}
Given regular initial data for the vacuum Einstein field equations as
constructed in Lemma~\ref{Lemma:FreeDataCIVP} on the null
hypersurfaces~$\mathcal{N}_\star\cup\mathcal{N}'_\star$
for~$I\equiv\{0\leq v\leq v_\bullet\}$, there exists~$\varepsilon>0$
such that a unique smooth solution to the vacuum Einstein field
equations exists in the region where~$v\in I$
and~$0\leq u\leq \varepsilon$ defined by the null
coordinates~$(u,v)$. The number~$\varepsilon$ can be chosen to depend
only on~$I$, and the initial data
 \begin{align*}
\Delta_{e_\star}&\equiv\sup_{\mathcal{N}_\star,\mathcal{N}'_\star}\sup_{U_{\alpha}}
\left(|Q|,|Q^{-1}|,|C^{\mathcal{A}}|,|P^{\mathcal{A}}|\right), \\
  \Delta_{\Gamma_\star} &\equiv\sup_{\mathcal{S}_{u,v}
    \subset\mathcal{N}_\star,\mathcal{N}'_\star}
  \sup_{\Gamma\in\{\mu,\lambda,\rho,\sigma,\pi,\tau,\omega,\chi\}}
  \max\{1,||\Gamma||_{L^{\infty}(\mathcal{S}_{u,v})},
  \sum_{i=0}^1||\mathcal{D}^i\Gamma||_{L^4(\mathcal{S}_{u,v})},
  \sum_{i=0}^2||\mathcal{D}^i\Gamma||_{L^2(\mathcal{S}_{u,v})}\},\\
  \Delta_{\Psi_\star}&\equiv\sup_{\mathcal{S}_{u,v}\subset\mathcal{N}_\star,\mathcal{N}'_\star}
  \sup_{\Psi\in\{\Psi_0,\Psi_1,\Psi_2,\Psi_3,\Psi_4\}}
  \max\{1,\sum_{i=0}^1||\mathcal{D}^i\Psi||_{L^4(\mathcal{S}_{u,v})},
  \sum_{i=0}^2||\mathcal{D}^i\Psi||_{L^2(\mathcal{S}_{u,v})}\} \\
  &+\sum_{i=0}^3\sup_{\Psi\in\{\Psi_0,\Psi_1,\Psi_2,\Psi_3\}}
  ||\mathcal{D}^i\Psi||_{L^2(\mathcal{N}_\star)}
  +\sup_{\Psi\in\{\Psi_1,\Psi_2,\Psi_3,\Psi_4\}}||\mathcal{D}^i\Psi||_{L^2(\mathcal{N}'_\star)}.
\end{align*}
Furthermore, in this existence domain one has that
  \begin{align*}
  &\sup_{u,v}\sup_{\Gamma\in\{\mu,\lambda,\rho,\sigma,\omega,\chi,\tau,\pi\}}
    \max\bigg\{\sum_{i=0}^1
    ||\mathcal{D}^i\Gamma||_{L^{\infty}(\mathcal{S}_{u,v})},\sum_{i=0}^2
  ||\mathcal{D}^i\Gamma||_{L^4(\mathcal{S}_{u,v})},
    \sum_{i=0}^3||\mathcal{D}^i\Gamma||_{L^2(\mathcal{S}_{u,v})}\bigg\}\\
  &+\sum_{i=0}^3\sup_{\Psi\in\{\Psi_0,\Psi_1,\Psi_2,\Psi_3\}}\sup_u
  ||\mathcal{D}^i\Psi||_{L^2(\mathcal{N}_u)}+\sup_{\Psi\in\{\Psi_1,\Psi_2,\Psi_3,\Psi_4\}}
  \sup_v||\mathcal{D}^i\Psi||_{L^2(\mathcal{N}'_v)}\leq C(I,\Delta_{e_\star},
  \Delta_{\Gamma_\star},\Delta_{\Psi_\star}).
\end{align*}
\end{theorem}

\begin{remark}
  These three quantities~$ \beta-\bar\alpha$, $\epsilon-\bar\epsilon$,
  $\gamma-\bar\gamma$ depend on the choice of frame on~$\mathcal{S}$
  and thus are not T-weight quantities. For these three we need to
  choose an appropriate coordinate system to make sure that their
  initial values do not blow up.  Then we can solve the ODE
  systems~(\eqref{MasslessScalarstructureeq1},
  \eqref{MasslessScalarstructureeq4}
  and~\eqref{MasslessScalarstructureeq11}) for these three quantities
  in the target area.
\end{remark}

\subsection{The Einstein-massless scalar field system in the T-weight
  formalism}
\label{ESstGauge}

The use of Stewart's gauge as discussed in
Section~\ref{Subsection:StewartGauge} together with the T-weight
formalism gives rise to a simpler version of the Einstein-scalar field
equations.

As already discussed, in Stewart's gauge one has that
\begin{align*}
  &\kappa=0,\quad \nu=0,\quad \epsilon=0, \\
  &\bar\rho=\rho,\quad \bar\mu=\mu,\quad \vartheta=\tau.
\end{align*}

It follows then that the equations for the scalar field take the form
\begin{subequations}
\begin{align}
\mthorn'\varphi_{0}-\meth\bar\varphi_{1}&=(\ulomega-\mu)\varphi_{0}
+\rho\varphi_{2}-\bar\tau\varphi_{1}-\tau\bar\varphi_{1}, \label{EOMMasslessScalarStT-weight1} \\
\mthorn\bar\varphi_{1}-\meth'\varphi_{0}&=(\pi-\bar\tau)\varphi_{0}
+\bar\sigma\varphi_{1}+\rho\bar\varphi_{1}, \label{EOMMasslessScalarStT-weight2} \\
\mthorn'\varphi_{1}-\meth\varphi_{2}&=
-\mu\varphi_{1}-\bar\lambda\bar\varphi_{1}, \label{EOMMasslessScalarStT-weight3} \\
\mthorn\varphi_{2}-\meth'\varphi_{1}&=-\mu\varphi_{0}
+\rho\varphi_{2}+\pi\varphi_{1}+\bar\pi\bar\varphi_{1}, \label{EOMMasslessScalarStT-weight4} \\
\meth\bar\varphi_{1}-\meth'\varphi_{1}&=0. \label{EOMMasslessScalarStT-weight5} 
\end{align}
\end{subequations}
Taking into account the definition of renormalised components of the
Weyl curvature as given by
\begin{align*}
\tilde\Psi_1\equiv\Psi_1-3\varphi_{0}\varphi_{1}, \quad
\tilde\Psi_2\equiv\Psi_2-\varphi_{0}\varphi_{2}+\varphi_{1}\bar\varphi_{1},  \quad
\tilde\Psi_3\equiv\Psi_3-3\varphi_{2}\bar\varphi_{1},
\end{align*}
it follows that the structure equations are
\begin{subequations}
\begin{align}
\mthorn\rho&=\rho^2+\sigma\bar\sigma+3\varphi_{0}^2, \label{T-weightMasslessStructureStEq1}\\
\mthorn\sigma&=2\rho\sigma+\Psi_0, \label{T-weightMasslessStructureStEq2}\\
\mthorn\mu-\meth\pi&=\rho\mu+\sigma\lambda+\pi\bar\pi+\TiPsi_2, \label{T-weightMasslessStructureStEq3}\\
\mthorn\lambda-\meth'\pi&=\rho\lambda+\bar\sigma\mu+\pi^2+3\bar\varphi_1^2, \label{T-weightMasslessStructureStEq4}\\
\mthorn'\rho-\meth'\tau&=-\rho\mu-\sigma\lambda-\bar\tau\tau+\ulomega\rho
-\TiPsi_2, \label{T-weightMasslessStructureStEq5}\\
\mthorn'\sigma-\meth\tau&=-\mu\sigma-\bar\lambda\rho-\tau^2+\ulomega\sigma
-3\varphi_{1}^2, \label{T-weightMasslessStructureStEq6}\\
\mthorn'\mu&=-\mu^2-\lambda\bar\lambda-\ulomega\mu-3\varphi_{2}^2, \label{T-weightMasslessStructureStEq7}\\
\mthorn'\lambda&=-2\mu\lambda-\ulomega\lambda-\Psi_4, \label{T-weightMasslessStructureStEq8}\\
\mthorn\tau&=(\tau+\bar\pi)\rho+(\bar\tau+\pi)\sigma
+\TiPsi_1+6\varphi_0\varphi_1, \label{T-weightMasslessStructureStEq9}\\
\mthorn'\pi&=-(\pi+\bar\tau)\mu-(\bar\pi+\tau)\lambda
-\TiPsi_3-6\varphi_{2}\bar\varphi_{1},  \label{T-weightMasslessStructureStEq10}\\
\mthorn\ulomega&=2\tau\bar\tau+2\tau\pi+2\bar\tau\bar\pi
+\TiPsi_2+\bar\TiPsi_2+6\varphi_{0}\varphi_{2}, \label{T-weightMasslessStructureStEq11}\\
\mthorn'\tau-\meth\ulomega&=-2\bar\tau\bar\lambda-2\tau\mu
-\bar\TiPsi_3-6\varphi_{2}\varphi_{1}, \label{T-weightMasslessStructureStEq13}\\
\meth\rho-\meth'\sigma&=\tau\rho-\bar\tau\sigma-\TiPsi_1,  \label{T-weightMasslessStructureStEq14}\\
\meth\lambda-\meth'\mu&=\bar\tau\mu-\tau\lambda-\TiPsi_3. \label{T-weightMasslessStructureStEq15}
\end{align}
\end{subequations}
The Bianchi Identities take the form
\begin{subequations}
\begin{align}
\mthorn\TiPsi_1-\meth'\Psi_0&=-6\varphi_0\meth\varphi_0+(\pi-2\bar\tau)\Psi_0
+4\rho\TiPsi_1
 \nonumber\\
&+3\varphi_0^2(2\tau-\bar\pi)+6\varphi_0\varphi_1\rho-6\varphi_0\bar\varphi_1\sigma, \label{T-weightMasslessStBianchi1}\\
\mthorn'\Psi_0-\meth\TiPsi_1&=6\varphi_0\meth\varphi_1+(2\ulomega-\mu)\Psi_0
-5\tau\TiPsi_1+3\sigma\TiPsi_2
\nonumber \\
&+6\varphi_0\varphi_2\sigma-3\varphi_0^2\bar\lambda-3\varphi_1^2\rho
-6\varphi_1\bar\varphi_1\sigma-12\varphi_0\varphi_1\tau, \label{T-weightMasslessStBianchi2}\\
\mthorn\TiPsi_2-\meth'\TiPsi_1&=6\varphi_1\meth’\varphi_0-\lambda\Psi_0
+(2\pi-\bar\tau)\TiPsi_1+3\rho\TiPsi_2
 \nonumber\\
&+3\varphi_0^2\mu+6\varphi_0\varphi_1\pi+3\varphi_1^2\bar\sigma
-6\varphi_0\varphi_1\bar\tau, \label{T-weightMasslessStBianchi3}\\
\mthorn'\TiPsi_1-\meth\TiPsi_2&=-6\varphi_1\meth’\varphi_1
+(\ulomega-2\mu)\TiPsi_1-3\tau\TiPsi_2+2\sigma\TiPsi_3
\nonumber \\
&+6\varphi_2\bar\varphi_1\sigma+6\varphi_1\bar\varphi_1\tau
+3\varphi_1^2\bar\tau-6\varphi_2\varphi_1\rho, \label{T-weightMasslessStBianchi4}\\
\mthorn\TiPsi_3-\meth'\TiPsi_2&=-6\bar\varphi_1\meth\bar\varphi_1-2\lambda\TiPsi_1
+3\pi\TiPsi_2+2\rho\TiPsi_3 \nonumber \\
&-3\bar\varphi_1^2\bar\pi+6\varphi_0\bar\varphi_1\mu
-6\varphi_0\varphi_1\lambda-6\varphi_1\bar\varphi_1\pi, \label{T-weightMasslessStBianchi5} \\
\mthorn'\TiPsi_2-\meth\TiPsi_3&=6\bar\varphi_1\meth\varphi_2-3\mu\TiPsi_2
-\tau\TiPsi_3+\sigma\Psi_4-3\bar\varphi_1^2\bar\lambda-3\varphi_2^2\bar\rho,  \label{T-weightMasslessStBianchi6}\\
\mthorn\Psi_4-\meth'\TiPsi_3&=6\varphi_2\meth’\bar\varphi_1-3\lambda\TiPsi_2
+(4\pi+\bar\tau)\TiPsi_3+\rho\Psi_4
 \nonumber \\
&+12\varphi_2\bar\varphi_1\pi+6\varphi_1\bar\varphi_1\lambda+3\bar\varphi_1^2\mu
+3\varphi_2^2\bar\sigma-6\varphi_0\varphi_2\lambda, \label{T-weightMasslessStBianchi7}\\
\mthorn'\TiPsi_3-\meth\Psi_4&=-6\varphi_2\meth’\varphi_2
-(4\mu+\ulomega)\TiPsi_3+\tau\Psi_4
 \nonumber \\
&+6\varphi_2\varphi_1\lambda-6\varphi_2\bar\varphi_1\mu
-3\varphi_2^2\bar\tau. \label{T-weightMasslessStBianchi8}
\end{align}
\end{subequations}
Finally, we observe that the commutators are given by
\begin{subequations}
\begin{align}
(\mthorn\mthorn'-\mthorn'\mthorn)f&=s(\TiPsi_2-\bar\TiPsi_2+\pi\tau-\bar\pi\bar\tau)f 
-\ulomega\mthorn f+(\pi+\bar\tau)\meth f+(\bar\pi+\tau)\meth'f, \label{T-weightMasslessCommutatorSt1}\\
(\mthorn\meth-\meth\mthorn)f&=s(-\meth\rho+\meth'\sigma+\tau\rho-\bar\tau\sigma-\bar\pi\rho+\pi\sigma)f 
+(\bar\pi-\tau)\mthorn f+\rho\meth f+\sigma\meth'f \label{T-weightMasslessCommutatorSt2}\\
(\mthorn'\meth-\meth\mthorn')f&=s(-\meth'\bar\lambda+\meth\mu)f 
-\mu\meth f-\bar\lambda\meth'f, \label{T-weightMasslessCommutatorSt3}\\
(\mthorn\meth'-\meth'\mthorn)f&=s(\meth'\rho-\meth\bar\sigma-\bar\tau\rho+\tau\bar\sigma+\pi\rho-\bar\pi\bar\sigma)f 
+(\bar\tau-\pi)\mthorn f+\bar\sigma\meth f+\rho\meth'f, \label{T-weightMasslessCommutatorSt4}\\
(\mthorn'\meth'-\meth'\mthorn')f&=s(\meth\lambda-\meth'\mu)f 
-\lambda\meth f-\bar\mu\meth'f, \label{T-weightMasslessCommutatorSt5} \\
(\meth\meth'-\meth'\meth)f&=sKf, \label{T-weightMasslessCommutatorSt6}
\end{align}
\end{subequations}
where~$K$ is the \emph{Gaussian curvature} of the 2-surfaces
$\mathcal{S}$ ---namely,
\begin{align*}
K\equiv 2\mu\rho-\lambda\sigma-\bar\lambda\bar\sigma-\TiPsi_2-\bar\TiPsi_2+6\varphi_1\bar\varphi_1.
\end{align*}

\section{Preliminaries}
\label{Preliminary}

In this section we provide a first discussion of the general structure
and strategy for the analysis of the formation of trapped surfaces.

\subsection{Further ancillary fields}

We begin by observing that for Minkowski data
on~$\mathcal{N}'_{0}\equiv\mathcal{N}'_{\star}$ the two expansions are
given by
  \begin{align*}
&\theta_{\bml}(u,v=0)=-\rho-\bar\rho=-2\rho=-2\cdot\frac{1}{u-v}=-\frac{2}{u}; \\
&\theta_{\bmn}(u,v=0)=\mu+\bar\mu=2\mu=2\cdot\frac{1}{u-v}=\frac{2}{u},
  \end{align*}
  where we have made use of the null coordinates~$(u,v)$. For
  convenience, we define a new quantity 
 \begin{align*}
\tilde\mu\equiv\mu-\frac{1}{u},
\end{align*}
to analyse the behaviour of~$\mu$. This renormalized expansion
satisfies transport equations given by
\begin{subequations}
\begin{align}
D\tilde\mu-\rho\tilde\mu&=-\rho\tilde\mu+\mu\rho+\TiPsi_2+\lambda\sigma+\delta\pi+\pi\bar\pi-\bar\alpha\pi+\beta\pi \label{Dtildemu},\\
\Delta\tilde\mu+2\mu\tilde\mu&=\tilde\mu^2-(\gamma+\bar\gamma)\mu-\lambda\bar\lambda-3\varphi_{2}^2 \label{Deltatildemu}.
\end{align}
\end{subequations}
In the T-weight formalism these equations can be rewritten as
\begin{subequations}
\begin{align}
\mthorn\tilde\mu-\meth\pi-\rho\tilde\mu&=-\rho\tilde\mu+\mu\rho+\TiPsi_2+\lambda\sigma+\pi\bar\pi \label{Ttildemu},\\
\mthorn'\tilde\mu+2\mu\tilde\mu&=\tilde\mu^2-\ulomega\mu-\lambda\bar\lambda-3\varphi_{2}^2 \label{Deltatildemu}.
\end{align}
\end{subequations}

In view of the subsequent analysis it convenient to define the scalar
 \begin{align*}
\ulchi\equiv D\log Q,
\end{align*}
which, as it is the derivative of a component of frame, is at the same
level of connection coefficients. This scalar is the direct analogue
of~$\chi$ as defined in~\cite{HilValZha20} with a derivative in the
incoming direction. It provides a component of the connection which
does not arise in the original NP formalism, but is needed to obtain a
complete set of outgoing equations for the frame. From the definition
of~$\ulchi$ and the NP Ricci identities we readily obtain
\begin{align}
  \mthorn'\ulchi=2\tau\pi+2\bar\tau\bar\pi+2\pi\bar\pi+\TiPsi_2
  +\bar\TiPsi_2+6\varphi_0\varphi_2+\ulomega\ulchi.
\label{Tpulchi}
\end{align}
Now, as a consequence of the gauge choice~$Q=1$
on~$\mathcal{N}_{\star}$, the initial data for~$\ulchi$ is 0
on~$\mathcal{N}_{\star}$.

\begin{remark}
  {\em From the definitions we have that
\begin{align*}
  s(\Timu)=s(\ulchi)=0.
\end{align*}}
\end{remark}

\subsection{Signatures and scale-invariant norms}

Following the assumptions in~\cite{An2022} and~\cite{An202209}, we
assume throughout that on the initial outgoing null
slice~$\mathcal{N}_{\star}$, the connection coefficients satisfy
asymptotic behaviour
\begin{align*}
&|\rho|\lesssim\frac{1}{|u_{\infty}|}, \quad
|\sigma|\lesssim\frac{a^{\frac{1}{2}}}{|u_{\infty}|}, \quad
|\mu|\lesssim\frac{1}{|u_{\infty}|}, \quad
|\Timu|=|\mu+\frac{1}{|u_{\infty}|}|\lesssim\frac{1}{|u_{\infty}|^2}, \quad
|\lambda|\lesssim\frac{a^{\frac{1}{2}}}{|u_{\infty}|^2}, \\
&|\ulomega|\lesssim\frac{a}{|u_{\infty}|^3}, \quad
|\tau|\lesssim\frac{a^{\frac{1}{2}}}{|u_{\infty}|^2},\quad
|\pi|\lesssim\frac{a^{\frac{1}{2}}}{|u_{\infty}|^2},
\end{align*}
and that for the components of the Weyl tensor one has
\begin{align*}
|\Psi_0|\lesssim\frac{a^{\frac{1}{2}}}{|u_{\infty}|}, \quad
|\Psi_1|\lesssim\frac{a^{\frac{1}{2}}}{|u_{\infty}|^2}, \quad
|\Psi_2|\lesssim\frac{a}{|u_{\infty}|^3}, \quad
|\Psi_3|\lesssim\frac{a^{\frac{3}{2}}}{|u_{\infty}|^4}, \quad
|\Psi_4|\lesssim\frac{a^2}{|u_{\infty}|^5}.
\end{align*}

For the scalar field~$\varphi$ and the auxiliary field associated to
its derivative~$\varphi_a$, we assume that on~$\mathcal{N}_{\star}$
one has that
\begin{align*}
|\varphi|\lesssim\frac{a^{\frac{1}{2}}}{|u_{\infty}|}, \quad 
|\varphi_0|\lesssim\frac{a^{\frac{1}{2}}}{|u_{\infty}|}, \quad
|\varphi_2|\lesssim\frac{a^{\frac{1}{2}}}{|u_{\infty}|^2}, \quad
|\varphi_1|\lesssim\frac{a^{\frac{1}{2}}}{|u_{\infty}|^2}.
\end{align*}
From the above assumption on~$\varphi_a$ it follows that 
\begin{align*}
|\tilde\Psi_1|&\equiv|\Psi_1-3\varphi_{0}\varphi_{1}|\lesssim\frac{a^{\frac{1}{2}}}{|u_{\infty}|^2}, \\
|\tilde\Psi_2|&\equiv|\Psi_2-\varphi_{0}\varphi_{2}+\varphi_{1}\bar\varphi_{1}|\lesssim\frac{a}{|u_{\infty}|^3},  \\
|\tilde\Psi_3|&\equiv|\Psi_3-3\varphi_{2}\bar\varphi_{1}|\lesssim\frac{a^{\frac{3}{2}}}{|u_{\infty}|^4},
\end{align*}
which show that the renormalised Weyl components can be assumed to
keep the same behaviour as that of~$\Psi_1$, $\Psi_2$ and~$\Psi_3$,
respectively. The equivalence between~$\Psi$ and~$\TiPsi$ is given in
the main proof.

For future use is convenient to introduce another auxiliary field
---namely,
\begin{align*}
\Tivarphi_2\equiv\meth'\varphi_2+\mu\bar\varphi_1, 
\end{align*}
and assume on~$\mathcal{N}_{\star}$ that
\begin{align*}
  |\Tivarphi_2|\lesssim\frac{a^{\frac{3}{2}}}{|u_{\infty}|^4}.
\end{align*}

\begin{remark}
  \label{remark:etheth'decay}
  {\em Note that applying~$\meth$ or~$\meth'$ to our variables results
    in additional decay of~$1/|u|$ to the asymptotic decay so
    that~$\meth'\varphi_2$ decays
    like~$a^{\frac{1}{2}}/|u_{\infty}|^3$. The term~$\mu\bar\varphi_1$
    shares the same decay. The leading (negative) piece of~$\mu$
    eliminates the leading term of these two expressions. A more
    detailed discussions on the behaviour of the
    quantity~$\Tivarphi_2$ is given in Section~\ref{EnergyEstimate} on
    the energy estimates.}
\end{remark}

In our analysis we follow the basic design in An's work~\cite{An2022}
for the definition of scale
invariant~$L^{\infty}_{sc}(\mathcal{S})$-norms with~$|u|$ and~$a$
weights such that
\begin{align*}
||\phi||_{L^{\infty}_{sc}(\mathcal{S})}\lesssim1.
\end{align*}
To this end associate a number (\emph{signature})~$s_2(f)$ to
characterise the decay rate of the various geometric quantities: \\

\noindent
\begin{tabular}{|c|c|c|c|c|c|c|c|c|c|c|c|c|c|c|c|c|}
\hline
~& $\Psi_0$ & $\Psi_1$ & $\Psi_2$ & $\Psi_3$ & $\Psi_4$ 
& $\rho,\sigma$ & $\omega$ & $\vartheta$ &$\tau$ &$\pi$ 
&$ \mu,\lambda, \tilde\mu$ & $\ulomega$ & $\ulchi$ 
&$\varphi,\varphi_{0}$ &$\varphi_{1}$&$\varphi_{2}$\\
\hline
$s_2$ &0&$1/2$ &1&$3/2$&2&0&0&$1/2$&$1/2$&$1/2$&1&1&0&0&$1/2$&1\\
\hline
\end{tabular}\\

The signature is required to satisfy
\begin{align*}
&s_2(\mthorn^i\mthorn'^j\{\meth,\meth'\}^k\phi)=s_2(\phi)+0\times i+1\times j+\tfrac{1}{2}\times k , \\
&s_2(\phi_1\cdot\phi_2)=s_2(\phi_1)+s_2(\phi_2).
\end{align*}
We make use of the signature to define the following
\emph{scale-invariant norms}:
 \begin{align*}
&||\phi||_{L^{\infty}_{sc}(\mathcal{S}_{u,v})}\equiv a^{-s_2(\phi)}|u|^{2s_2(\phi)+1}||\phi||_{L^{\infty}(\mathcal{S}_{u,v})}, \\
&||\phi||_{L^2_{sc}(\mathcal{S}_{u,v})}\equiv a^{-s_2(\phi)}|u|^{2s_2(\phi)}||\phi||_{L^2(\mathcal{S}_{u,v})},\\
&||\phi||_{L^1_{sc}(\mathcal{S}_{u,v})}\equiv a^{-s_2(\phi)}|u|^{2s_2(\phi)-1}||\phi||_{L^1(\mathcal{S}_{u,v})}.
 \end{align*}
 The above norms satisfy the H\"older inequalities
 \begin{align*}
||\phi_1\cdot\phi_2||_{L^2_{sc}(\mathcal{S}_{u,v})}&\leq\frac{1}{|u|}||\phi_1||_{L^{\infty}_{sc}(\mathcal{S}_{u,v})}||\phi_2||_{L^2_{sc}(\mathcal{S}_{u,v})},\\
||\phi_1\cdot\phi_2||_{L^1_{sc}(\mathcal{S}_{u,v})}&\leq\frac{1}{|u|}||\phi_1||_{L^{\infty}_{sc}(\mathcal{S}_{u,v})}||\phi_2||_{L^1_{sc}(\mathcal{S}_{u,v})},\\
||\phi_1\cdot\phi_2||_{L^1_{sc}(\mathcal{S}_{u,v})}&\leq\frac{1}{|u|}||\phi_1||_{L^2_{sc}(\mathcal{S}_{u,v})}||\phi_2||_{L^2_{sc}(\mathcal{S}_{u,v})}.
\end{align*}
Thus it follows that, making use of these norms, that if the norms
themselves are order~$1$, then the nonlinear terms can be treated as
lower order terms when~$|u|\gg1$. Finally, we define the following
scale-invariant norms along the lightcones:
\begin{align*}
||\phi||^2_{L^2_{sc}(\mathcal{N}_u(0,v))}&
\equiv\int_0^v||\phi||^2_{L^2_{sc}(\mathcal{S}_{u,v'})}\mathrm{d}v',\\
||\phi||^2_{L^2_{sc}(\mathcal{N}'_v(u_{\infty},u))}&
\equiv\int_{u_{\infty}}^u\frac{a}{|u'|^2}||\phi||^2_{L^2_{sc}(\mathcal{S}_{u',v})}\mathrm{d}u'.
\end{align*}

\begin{remark}
  In the analysis of main proof, from Section~\ref{L2estimate} to
  Section~\ref{EnergyEstimate}, we frequently use the H\"older
  inequalities to divide the $L^2_{sc}(\mathcal{S}_{u,v})$ of multiple
  terms to one $L^2_{sc}(\mathcal{S}_{u,v})$ norm and rest
  $L^{\infty}_{sc}(\mathcal{S}_{u,v})$ norms.  Take multiple terms
  $A\cdot B\cdot C$ as an example:
\begin{align*}
||A\cdot B\cdot C||_{L^2_{sc}(\mathcal{S}_{u,v})}\leq\frac{1}{|u|^2}
||A||_{L^2_{sc}(\mathcal{S}_{u,v})}||B||_{L^{\infty}_{sc}(\mathcal{S}_{u,v})}||C||_{L^{\infty}_{sc}(\mathcal{S}_{u,v})}.
\end{align*}
The total number of spherical derivatives in each such term is less or
equal to 11. The reason $A$ takes the $L^2$ norm is that $A$ is the
term containing the highest spherical derivative among $A$, $B$ and
$C$. According to the bootstrap assumptions and the field equations
themselves, no terms which are quadratic (or higher order) in the top
(11th) derivative occur. Therefore, for convenience we use
\begin{align*}
||A\cdot B\cdot C||_{L^2_{sc}(\mathcal{S}_{u,v})}\leq\frac{1}{|u|^2}
||A||\times||B||\times||C||
\end{align*}
in the discussion when there is no ambiguity.
\end{remark}

\subsection{Bootstrap norms}

In this section we introduce the norms that will be used to set up our
main bootstrap argument. In the following,
let~$\hat\Gamma\equiv \{\rho,\tau,\pi,\ulomega,\ulchi\}$,
$\hat\Psi=\{\TiPsi_1,\TiPsi_2,\TiPsi_3,\Psi_4\}$. For~$0\leq i\leq6$,
we define
 \begin{align*}
\Gamma_{i,\infty}(u,v)&\equiv\frac{1}{a^{\frac{1}{2}}}||(a^{\frac{1}{2}}\mathcal{D})^i\sigma||_{L^{\infty}_{sc}(\mathcal{S}_{u,v})}
+||(a^{\frac{1}{2}}\mathcal{D})^i\hat\Gamma||_{L^{\infty}_{sc}(\mathcal{S}_{u,v})}
+\frac{a^{\frac{1}{2}}}{|u|}||(a^{\frac{1}{2}}\mathcal{D})^i\lambda||_{L^{\infty}_{sc}(\mathcal{S}_{u,v})} \\
&+\frac{a}{|u|^2}||(a^{\frac{1}{2}}\mathcal{D})^i\mu||_{L^{\infty}_{sc}(\mathcal{S}_{u,v})}
+\frac{a}{|u|}||(a^{\frac{1}{2}}\mathcal{D})^i\tilde\mu||_{L^{\infty}_{sc}(\mathcal{S}_{u,v})},
\end{align*}
 \begin{align*}
\Psi_{i,\infty}(u,v)\equiv\frac{1}{a^{\frac{1}{2}}}||(a^{\frac{1}{2}}\mathcal{D})^i\Psi_0||_{L^{\infty}_{sc}(\mathcal{S}_{u,v})}
+||(a^{\frac{1}{2}}\mathcal{D})^i\hat\Psi||_{L^{\infty}_{sc}(\mathcal{S}_{u,v})},
\end{align*}
 \begin{align*}
\bmvarphi_{i,\infty}(u,v)\equiv\frac{1}{a^{\frac{1}{2}}}||(a^{\frac{1}{2}}\mathcal{D})^i\varphi_{0}||_{L^{\infty}_{sc}(\mathcal{S}_{u,v})}
+||(a^{\frac{1}{2}}\mathcal{D})^i\varphi_{1},\Tivarphi_2||_{L^{\infty}_{sc}(\mathcal{S}_{u,v})}
+\frac{a^{\frac{1}{2}}}{|u|}||(a^{\frac{1}{2}}\mathcal{D})^i\varphi_{2}||_{L^{\infty}_{sc}(\mathcal{S}_{u,v})}.
 \end{align*}
For~$0\leq i\leq10$, we define
 \begin{align*}
\Gamma_{i,2}(u,v)&\equiv
\frac{1}{a^{\frac{1}{2}}}||(a^{\frac{1}{2}}\mathcal{D})^i\sigma||_{L^2_{sc}(\mathcal{S}_{u,v})}
+||(a^{\frac{1}{2}}\mathcal{D})^i\hat\Gamma||_{L^2_{sc}(\mathcal{S}_{u,v})}
+\frac{a^{\frac{1}{2}}}{|u|}||(a^{\frac{1}{2}}\mathcal{D})^i\lambda||_{L^2_{sc}(\mathcal{S}_{u,v})} \\
&+\frac{a}{|u|^2}||(a^{\frac{1}{2}}\mathcal{D})^i\mu||_{L^2_{sc}(\mathcal{S}_{u,v})}
+\frac{a}{|u|}||(a^{\frac{1}{2}}\mathcal{D})^i\tilde\mu||_{L^2_{sc}(\mathcal{S}_{u,v})},
\end{align*}
\begin{align*}
\bmvarphi_{i,2}(u,v)\equiv\frac{1}{a^{\frac{1}{2}}}||(a^{\frac{1}{2}}\mathcal{D})^i\varphi_{0}||_{L^2_{sc}(\mathcal{S}_{u,v})}
+||(a^{\frac{1}{2}}\mathcal{D})^i\varphi_{1}||_{L^2_{sc}(\mathcal{S}_{u,v})}
+\frac{a^{\frac{1}{2}}}{|u|}||(a^{\frac{1}{2}}\mathcal{D})^i\varphi_{2}||_{L^{2}_{sc}(\mathcal{S}_{u,v})}.
\end{align*}
For~$0\leq i\leq9$, we define
\begin{align*}
\tilde{\bmvarphi}_{i,2}(u,v)&\equiv||(a^{\frac{1}{2}}\mathcal{D})^i\Tivarphi_2||_{L^2_{sc}(\mathcal{S}_{u,v})}, \\
\Psi_{i,2}(u,v)&\equiv\frac{1}{a^{\frac{1}{2}}}||(a^{\frac{1}{2}}\mathcal{D})^i\Psi_0||_{L^2_{sc}(\mathcal{S}_{u,v})}
+||(a^{\frac{1}{2}}\mathcal{D})^i\hat\Psi||_{L^2_{sc}(\mathcal{S}_{u,v})}.
\end{align*}
For~$0\leq i\leq10$, we define
\begin{align*}
\Psi_{i}(u,v)\equiv\frac{1}{a^{\frac{1}{2}}}||(a^{\frac{1}{2}}\mathcal{D})^i\Psi_0||_{L^2_{sc}(\mathcal{N}_u(0,v))}
+||(a^{\frac{1}{2}}\mathcal{D})^i\{\TiPsi_1,\TiPsi_2,\TiPsi_3\}||_{L^2_{sc}(\mathcal{N}_u(0,v))},
\end{align*}
\begin{align*}
\underline\Psi_{i}(u,v)\equiv\frac{1}{a^{\frac{1}{2}}}||(a^{\frac{1}{2}}\mathcal{D})^i\TiPsi_1||_{L^2_{sc}(\mathcal{N}'_v(u_{\infty},u))}
+||(a^{\frac{1}{2}}\mathcal{D})^i\{\TiPsi_2,\TiPsi_3,\Psi_4\}||_{L^2_{sc}(\mathcal{N}'_v(u_{\infty},u))}.
\end{align*}
For~$0\leq i\leq11$, we define
\begin{align*}
\bmvarphi_{i}(u,v)\equiv\frac{1}{a^{\frac{1}{2}}}||(a^{\frac{1}{2}})^{i-1}\mathcal{D}^i\varphi_{0}||_{L^2_{sc}(\mathcal{N}_u(0,v))}
+||(a^{\frac{1}{2}})^{i-1}\mathcal{D}^i\varphi_{1}||_{L^2_{sc}(\mathcal{N}_u(0,v))},
\end{align*}
\begin{align*}
\underline\bmvarphi_{i}(u,v)\equiv\frac{1}{a^{\frac{1}{2}}}||(a^{\frac{1}{2}})^{i-1}\mathcal{D}^i\varphi_{1}||_{L^2_{sc}(\mathcal{N}'_v(u_{\infty},u))}
+||\frac{a^{\frac{1}{2}}}{|u'|}(a^{\frac{1}{2}})^{i-1}\mathcal{D}^i\varphi_{2}||_{L^2_{sc}(\mathcal{N}'_v(u_{\infty},u))},
\end{align*}
where 
\begin{align*}
||\frac{a^{\frac{1}{2}}}{|u'|}(a^{\frac{1}{2}})^{i-1}\mathcal{D}^i\varphi_{2}||^2_{L^2_{sc}(\mathcal{N}'_v(u_{\infty},u))}
\equiv\int_{u_{\infty}}^u\frac{a^2}{|u'|^4}
||(a^{\frac{1}{2}})^{i-1}\mathcal{D}^i\varphi_{2}||^2_{L^2_{sc}(\mathcal{S}_{u',v})}.
\end{align*}
Finally, for~$0\leq i\leq10$, we define
\begin{align*}
\underline{\tilde{\bmvarphi}}_{i}(u,v)\equiv
||(a^{\frac{1}{2}}\mathcal{D})^i\Tivarphi_{2}||_{L^2_{sc}(\mathcal{N}'_v(u_{\infty},u))}.
\end{align*}

\begin{remark}
  {\em From the choice of signature for~$\TiPsi$ and~$\Psi$, it
    follows that the norm of~$\TiPsi$ is equivalent to that of~$\Psi$.}
  \end{remark}

  In the following we denote by~$\Gamma_{i,\infty}$, $\Gamma_{i,2}$,
  $\Psi_{i,\infty}$, $\Psi_{i,2}$, $\Psi_{i}$, $\underline\Psi_{i}$,
  $\bmvarphi_{i,\infty}$, $\bmvarphi_{i,2}$, $\bmvarphi_{i}$ and
  $\underline\bmvarphi_{i}$ the supremum over~$u$ and~$v$ of the
  aforementioned bootstrap norms. We define, moreover
\begin{align*}
  &\bmGamma\equiv\sum_{i\leq6}(\Gamma_{i,\infty}+\Psi_{i,\infty}+\bmvarphi_{i,\infty})
    +\sum_{i\leq10}(\Gamma_{i,2}+\bmvarphi_{i,2})+\sum_{i\leq9}(\Psi_{i,2}+\tilde{\bmvarphi}_{i,2}),\\
&\bm{\Psi}\equiv\sum_{i\leq10}(\Psi_{i}+\underline\Psi_{i}), \quad
\bmvarphi\equiv\sum_{i\leq11}(\bmvarphi_{i}+\underline\bmvarphi_{i})
+\sum_{i\leq10}\underline{\tilde{\bmvarphi}}_{i}.
\end{align*}
and write~$\bmGamma_0$, $\bm\Psi_0$ and~$\bmvarphi_0$ for the norms on
the initial hypersurfaces.
 
We make use of $\Gamma(\Gamma,\varphi_i,\Psi)_k$ to 
denote the k-th derivative of $L^2$ norm of $\Gamma$, $\varphi_i$ and $\Psi$ 
which shows up on the right hand side in previous definition. 
Take $\sigma$ as an example:
\begin{align*}
\Gamma(\sigma)_k\equiv\frac{1}{a^{\frac{1}{2}}}
||(a^{\frac{1}{2}}\mathcal{D})^i\sigma||_{L^2_{sc}(\mathcal{S}_{u,v})}
\end{align*}
Follow the same spirit we make use of $\bm{\Gamma}[\Gamma]$ to denote
all the norms related to $\Gamma$.  For $\sigma$ we define
\begin{align*}
\bm{\Gamma}(\sigma)\equiv\sum_{i\leq10}\Gamma(\sigma)_i
\end{align*}

Moreover, for~$\varphi_i$, we use~$\bm{\varphi}[\varphi_i]_k$
and~$\underline{\bm{\varphi}}[\varphi_i]_k$ to denote the norm of
the~$k$-th derivative of~$\varphi_i$ along null
hypersurfaces~$\mathcal{N}_{u}$ and~$\mathcal{N}'_v$ respectively.
For example:
\begin{align*}
\bm{\varphi}[\varphi_0]_k&\equiv
\frac{1}{a^{\frac{1}{2}}}||(a^{\frac{1}{2}})^{i-1}\mathcal{D}^i\varphi_{0}||_{L^2_{sc}(\mathcal{N}_u(0,v))}, \\
\underline{\bm{\varphi}}[\varphi_2]_k&\equiv
||\frac{a^{\frac{1}{2}}}{|u'|}(a^{\frac{1}{2}})^{i-1}\mathcal{D}^i\varphi_{2}||_{L^2_{sc}(\mathcal{N}'_v(u_{\infty},u))}.
\end{align*}
We then write~$\bm{\varphi}[\varphi_i]$ to denote the sum
of~$\bm{\varphi}[\varphi_i]_k$ up to order
11. By~$\underline{\bm{\varphi}}[\varphi_i]$ we mean the sum
of~$\underline{\bm{\varphi}}[\varphi_i]_k$ up to order 11.

Similarly, we use~$\bm{\Psi}[\Psi_i]_k$
and~$\underline{\bm{\Psi}}[\Psi_i]_k$ to denote the norm of the~$k$-th
derivative of~$\Psi_i$ along null hypersurfaces~$\mathcal{N}_{u}$
and~$\mathcal{N}'_v$ respectively, and likewise
use~$\bm{\Psi}[\Psi_i]$ to denote the sum of~$\bm{\Psi}[\Psi_i]_k$ up
to 10th order. By~$\underline{\bm{\Psi}}[\Psi_i]$ we mean the sum
of~$\underline{\bm{\Psi}}[\Psi_i]_k$ up to 10th order.

We also make use of the \emph{initial data
  quantity}
\begin{align*}
\mathcal{I}_0\equiv\sup_{0\leq v\leq1}\mathcal{I}_0(v), 
\end{align*}
where
\begin{align*}
\mathcal{I}_0(v)\equiv\sum_{j=0}^1\sum_{i=0}^{15}\frac{1}{a^{\frac{1}{2}}}
||\mthorn^j(|u_{\infty}|\mathcal{D})^i(\sigma,\varphi_0)||_{L^{2}(\mathcal{S}_{u_{\infty},v})}.
\end{align*}

\smallskip
\noindent
\textbf{Notation.} In the rest of the article, for simplicity of the
presentation, we use~$A\lesssim B$ to denote that there exist a
constant~$C>0$ which is independent of~$a$, such that~$A\leq CB$. We
do not distinguish between~$\lesssim$ and~$\leq$ when there is no
source of ambiguity. In addition, we write~$\meth^{i_1}\Gamma^{i_2}$
to denote~$\meth^{j_1}\Gamma\meth^{j_2}\Gamma...\meth^{j_{i_2}}\Gamma$
where~$i_1\geq0$, $i_2\geq1$, $j_1$, $j_2$, ...,
$j_{i_2}\in\mathbb{N}$ and~$j_1+j_2+...+j_{i_2}=i_1$.

\subsection{Strategy of the bootstrap argument}
\label{BootstrapAssumption}

The main theorem in this article relies on an \emph{a priori} estimate
for the nonlinear Einstein-scalar system. In order to derive uniform
upper bounds on~$\bmGamma$, $\bm\Psi$ and~$\bmvarphi$, we make use of
a bootstrap argument. From the analysis of the initial data for the
CIVP, $\bmGamma_0$, $\bm\Psi_0$ and~$\bmvarphi_0$ we readily have the
bound
\begin{align}
\bmGamma_0+\bm\Psi_0+\bmvarphi_0\lesssim\mathcal{I}_0. \label{InitialData}
\end{align}
The proof of our main theorem will show that in the area 
\begin{align*}
\mathbb{D}=\big\{(u,v)|u_{\infty}\leq u\leq -a/4, \ \ 0\leq v\leq 1\big\}
\end{align*}
the inequality
\begin{align}
  \bmGamma(u,v)+\bm\Psi(u,v)+\bmvarphi(u,v)\lesssim(\mathcal{I}_0)^2
  +\mathcal{I}_0+1, \label{Conclusion}
\end{align}
holds. In turn, the above uniform bound leads to a last slice argument
from which it follows that a solution to Einstein-scalar system exists
within the area~$\mathbb{D}$.

In order to establish~\eqref{Conclusion}, we make the bootstrap
assumptions
\begin{align}
\bmGamma+\bm\Psi+\bmvarphi\leq\mathcal{O}. \label{Hypothesis}
\end{align}
for large~$\mathcal{O}$ such that
\begin{align*}
\mathcal{I}_{0}\ll\mathcal{O}
\end{align*}
but also
\begin{align*}
\mathcal{O}^{20}\leq a^{\frac{1}{16}}.
\end{align*}
Letting
\begin{align*}
  \bm{I}=\{u\,|\,u_{\infty}\leq u\leq -a/4,\ \eqref{Hypothesis} \
  \textrm{holds for every} \ 0\leq v\leq 1\},\nonumber
\end{align*}
our goal becomes to show that, in fact,
$\bm{I}=[u_{\infty}\leq u\leq -a/4]$. We will follow the
\emph{standard bootstrap principle} (or \emph{continuity method}) to
derive this conclusion ---i.e, showing that~$\bm{I}$ is both open and
closed.
 
Our argument is naturally divided into the following steps:

\smallskip
\noindent
\textbf{Step 1}: From the assumptions on the initial data, Condition
\eqref{InitialData}, and the local existence for the Einstein-scalar
system, the existence area can always be slightly extended
in~$u$. Then there exists a small~$\varepsilon$ such that for
$u_{\infty}\leq u\leq u_{\infty}+\varepsilon$ we have
\begin{align*}
\bmGamma+\bm\Psi+\bmvarphi\lesssim2\mathcal{I}_0\ll\mathcal{O}.
\end{align*}
This implies
that~$[u_{\infty}\leq u\leq u_{\infty}+\varepsilon]\subseteq\bm{I}$,
and it follows that~$\bm{I}$ is a closed set.

\smallskip
\noindent
\textbf{Step 2}: In Sections 5, 6 and 7 we show the estimates
 \begin{align*}
\bmGamma(u,v)&\lesssim\bm\Psi(u,v)^2+\bm\Psi(u,v)+\bm\varphi(u,v)^2+\bm\varphi(u,v)+\mathcal{I}_0+1, \\
\bm\varphi(u,v)&\lesssim\mathcal{I}_0+1, \quad
\bm\Psi(u,v)\lesssim\mathcal{I}_0+1.
\end{align*}
The above estimates improve the upper bounds in the bootstrap
assumption~\eqref{Hypothesis}. Accordingly, $\bm{I}$ can be extended a
bit forward along~$u$. This implies that~$\bm{I}$ is open.

\smallskip
\noindent
\textbf{Step 3}: Using local existence and basic topology, it follows
from Steps 1 and 2 that the only possibility for the set~$\bm{I}$ to
be both open and closed in~$[u_{\infty}\leq u\leq -a/4]$ is that it is
the whole set. We conclude
that~$\bm{I}\equiv[u_{\infty}\leq u\leq -a/4]$ and, moreover, that the
uniform bound in~\eqref{Conclusion} holds.

\subsection{Estimates for the components of the frame}

In order to kick-start the construction of our main estimates we need
first some basic estimates for the components of the frame in terms of
the basic bootstrap assumptions~\eqref{Hypothesis}.

\smallskip
\noindent
\textbf{Step 1.}  Integrating the definition of~$\ulchi$ in
the~$v$-direction and making use of the bootstrap
assumption~\eqref{Hypothesis} for~$\underline\chi$, one has that
\begin{align*}
   |Q-1|=\left|\int_0^1\ulchi\mathrm{d}v\right|\leq\int_0^1
   ||\ulchi||_{L^{\infty}(\mathcal{S}_{u,v'})}\mathrm{d}v'
   =\int_0^1\frac{1}{|u|}||\underline\chi||_{L^{\infty}_{sc}(\mathcal{S}_{u,v'})}\mathrm{d}v'
   \leq\frac{\mathcal{O}}{|u|}.
\end{align*}
In the above estimate it has been used that~$Q=1$
on~$\mathcal{N}'_{0}\equiv\mathcal{N}'_{\star}$ where Minkowski data
is given. From the assumption that~$a$ is suitably larger
than~$\mathcal{O}$, one has
that~$||Q,Q^{-1}||_{L^{\infty}(\mathcal{S}_{u,v'})}$ are close to~$1$.

\smallskip
\noindent
\textbf{Step 2.} Making use of the estimate
\eqref{DerivativeTwoSphere1}
\begin{align*}
\frac{\mathrm{d}}{\mathrm{d}v}\int_{\mathcal{S}_{u,v}}
\phi=\int_{\mathcal{S}_{u,v}}Q^{-1}\left(D\phi-(\rho+\bar\rho)\phi\right)
\end{align*}
and letting~$\phi=1$ one concludes that
\begin{align*}
\mbox{Area}(S_{u,v})-\mbox{Area}(S_{u,0})\lesssim
\int_0^v\int_{\mathcal{S}_{u,v}}2Q^{-1}||\rho||_{L^{\infty}(\mathcal{S}_{u,v'})}.
\end{align*}
Now, making use of the bootstrap assumption
of~$||\rho||_{L^{\infty}(\mathcal{S}_{u,v})}\leq\mathcal{O}/|u|$ and
that~$\mbox{Area}(S_{u,0})=2\pi u^2$ we obtain
\begin{align*}
\mbox{Area}(S_{u,v})-2\pi u^2\lesssim
\int_0^v\frac{1}{|u|^2}\mbox{Area}(S_{u,v'}).
\end{align*}
Finally, making use of the Gr\"onwall inequality we find that
\begin{align*}
|\mbox{Area}(S_{u,v})-2\pi u^2|\lesssim\frac{\mathcal{O}}{|u|}.
\end{align*}

\subsection{Gr\"onwall-type estimates}

One of the basic tools in our analysis will be a Gr\"onwall-type
estimate on transport equations. We begin by noticing the following:

\begin{proposition}
Under the bootstrap assumption~\eqref{Hypothesis}, we have that  
\begin{align*}
  &||\phi||_{L^2(\mathcal{S}_{u,v})}\lesssim 
  ||\phi||_{L^2(\mathcal{S}_{u,0})}+\int_0^v
  ||\mthorn\phi||_{L^2(\mathcal{S}_{u,v'})}\mathrm{d}v' , \\
  &||\phi||_{L^2(\mathcal{S}_{u,v})}\lesssim
  ||\phi||_{L^2(\mathcal{S}_{u_0,v})}+
  \int_{u_0}^u||\mthorn'\phi||_{L^2(\mathcal{S}_{u',v})}\mathrm{d}u'.
\end{align*}
\end{proposition}

\begin{proof}
  Let~$f\equiv |\phi|^2$. Making use of the discussion in
  Section~\ref{T-weightPDEanalysis} and the bootstrap assumption
  on~$\rho$, one has that
  \begin{align*}
    \frac{\mathrm{d}}{\mathrm{d}v}||\phi||_{L^2(\mathcal{S}_{u,v})}\leq ||\mthorn\phi||_{L^2(\mathcal{S}_{u,v})}
    +\frac{\mathcal{O}}{|u|}||\phi||_{L^2(\mathcal{S}_{u,v})}.
  \end{align*}
  Then, the Gr\"onwall inequality will give us that
  \begin{align*}
    ||\phi||_{L^2(\mathcal{S}_{u,v})}
    &
    \leq \exp\left(\int_0^1\frac{\mathcal{O}}{|u|}\mathrm{d}v'\right)
    \left(||\phi||_{L^2(\mathcal{S}_{u,0})}
    +\int_0^1||\mthorn\phi||_{L^2(\mathcal{S}_{u,v'})}\mathrm{d}v'\right) \\
    &\lesssim||\phi||_{L^2(\mathcal{S}_{u,0})}
    +\int_0^1||\mthorn\phi||_{L^2(\mathcal{S}_{u,v'})}\mathrm{d}v'.
\end{align*}

Along the~$u$ direction, making use of the bootstrap assumption
\begin{align*}
||\tilde\mu||_{L^{\infty}}\leq\frac{\mathcal{O}}{|u|^2}, 
\end{align*}
where~$\tilde\mu\equiv \mu+\displaystyle\frac{1}{|u|}$, we observe
that
\begin{align*}
-\frac{\mathcal{O}}{|u|^2}-\frac{1}{|u|}\leq\mu\leq\frac{\mathcal{O}}{|u|^2}-\frac{1}{|u|}\leq0.
\end{align*}
Hence, when~$f\equiv |\phi|^2$, using that~$Q\sim1$ we have
\begin{align*}
\frac{\mathrm{d}}{\mathrm{d}u}||\phi||_{L^2(\mathcal{S}_{u,v})}\leq ||\mthorn'\phi||_{L^2(\mathcal{S}_{u,v})},
\end{align*}
from which the desired result follows.
\end{proof}

We also have the following scale-invariant norm version of the last result: 
\begin{proposition}
  \begin{subequations}
\begin{align}
  &||\phi||_{L_{sc}^2(\mathcal{S}_{u,v})}\lesssim 
  ||\phi||_{L_{sc}^2(\mathcal{S}_{u,0})}+\int_0^v
  ||\mthorn\phi||_{L_{sc}^2(\mathcal{S}_{u,v'})}\mathrm{d}v' , \label{SCldirectiongronwall}\\
  &||\phi||_{L_{sc}^2(\mathcal{S}_{u,v})}\lesssim
  ||\phi||_{L_{sc}^2(\mathcal{S}_{0,v})}+
  \int_0^u\frac{a}{|u'|^2}||\mthorn'\phi||_{L_{sc}^2(\mathcal{S}_{u',v})}\mathrm{d}u'.
\end{align}
\end{subequations}
\end{proposition}

In the following we focus attention on the transport equations for the
components of the connection in the incoming direction. Schematically,
these can be written as
\begin{align*}
\mthorn'\Gamma= \lambda_0\mu\Gamma+F.
\end{align*}
In view of the important role to be played by the bootstrap assumption
on~$\mu$, we will require more precise estimates. In particular, we
have the following result:

\begin{proposition}
  For transport equations of the form
\begin{align*}
\mthorn'\Gamma=\lambda_0\mu\Gamma+F,
\end{align*}
one has the estimate
\begin{align*}
|u|^{\lambda_1}||\Gamma||_{L^2(\mathcal{S}_{u,v})}\lesssim 
|u_{\infty}|^{\lambda_1}||\Gamma||_{L^2(\mathcal{S}_{u_{\infty},v})}+
\int_{u_{\infty}}^u|u'|^{\lambda_1}||F||_{L^2(\mathcal{S}_{u',v})}\mathrm{d}u',
\end{align*}
where
\begin{align*}
\lambda_1\equiv -(\lambda_0+1).
\end{align*}
\end{proposition}

\begin{proof}
Let~$f\equiv |u|^{2\lambda_1}|\Gamma|^2$ and apply equation
\ref{DerivativeTwoSphere2}. The left-hand side of this equation is then given by 
\begin{align*}
\frac{\mathrm{d}}{\mathrm{d}u}\int_{\mathcal{S}_{u,v}}
f=\frac{\mathrm{d}}{\mathrm{d}u}|||u|^{\lambda_1}|\Gamma|||^2_{L^2(\mathcal{S}_{u,v})}=
2|||u|^{\lambda_1}|\Gamma|||_{L^2(\mathcal{S}_{u,v})}
\frac{\mathrm{d}}{\mathrm{d}u}|||u|^{\lambda_1}|\Gamma|||_{L^2(\mathcal{S}_{u,v})}.
\end{align*}
The integrands on the right-hand side satisfy
\begin{align*}
&\mthorn'(|u|^{2\lambda_1}|\Gamma|^2)
+2\mu|u|^{2\lambda_1}|\Gamma|^2 \\
&=2\lambda_1|u|^{2\lambda_1-1}\mthorn'(|u|)|\Gamma|^2
+|u|^{2\lambda_1}(\bar\Gamma\mthorn'\Gamma+\Gamma\mthorn'\bar\Gamma)+2\mu|u|^{2\lambda_1}|\Gamma|^2 \\
&\leq-2\lambda_1|u|^{2\lambda_1-1}|\Gamma|^2+
2\lambda_0\mu|u|^{2\lambda_1}|\Gamma|^2+2|u|^{2\lambda_1}|\Gamma||F|
+2\mu|u|^{2\lambda_1}|\Gamma|^2 \\
&=-2\lambda_1|u|^{2\lambda_1}|\Gamma|^2(\mu-\frac{1}{u})
+2|u|^{2\lambda_1}\Gamma F.
\end{align*}
Now, making use of the bootstrap assumption
\begin{align*}
  ||\mu-\frac{1}{u}||_{L^{\infty}(\mathcal{S}_{u,v})}\leq\frac{\mathcal{O}}{|u|^2},
\end{align*}
one can show that the right-hand side is less than
\begin{align*}
2\frac{\mathcal{O}}{|u|^2}|||u|^{\lambda_1}|\Gamma|||^2_{L^2(\mathcal{S}_{u,v})}+2|||u|^{\lambda_1}|\Gamma|||_{L^2(\mathcal{S}_{u,v})}|||u|^{\lambda_1}|F|||_{L^2(\mathcal{S}_{u,v})}.
\end{align*}
Putting everything together one gets
\begin{align*}
  \frac{\mathrm{d}}{\mathrm{d}u}|||u|^{\lambda_1}|\Gamma|||_{L^2(\mathcal{S}_{u,v})}
  \leq\frac{\mathcal{O}}{|u|^2}|||u|^{\lambda_1}|\Gamma|||_{L^2(\mathcal{S}_{u,v})}
  +|||u|^{\lambda_1}|F|||_{L^2(\mathcal{S}_{u,v})}.
\end{align*}
Integrating with respect to~$u$ and applying the Gr\"onwall
inequality, one concludes that
\begin{align*}
  |||u|^{\lambda_1}|\Gamma|||_{L^2(\mathcal{S}_{u,v})}
  &\leq\left(\exp{\int_{u_{\infty}}^u\frac{\mathcal{O}}{|u'|^2}}\right)
    \left(|||u|^{\lambda_1}|\Gamma|||_{L^2(\mathcal{S}_{u_{\infty},v})}
    +\int_{u_{\infty}}^u|||u'|^{\lambda_1}|F|||_{L^2(\mathcal{S}_{u',v})}\right)\\
  &\lesssim |||u|^{\lambda_1}|\Gamma|||_{L^2(\mathcal{S}_{u_{\infty},v})}
    +\int_{u_{\infty}}^u|||u'|^{\lambda_1}|F|||_{L^2(\mathcal{S}_{u',v})}.
\end{align*}

\end{proof}

In terms of scale-invariant norms the above result can be reformulated
as
\begin{proposition}
\label{SCweightedgronwall}
For  transport equations of the form 
\begin{align*}
\mthorn'\Gamma=\lambda_0\mu\Gamma+F,
\end{align*}
one has that 
\begin{align*}
  a^{s_2(\Gamma)}|u|^{\lambda_1-2s_2(\Gamma)}||\Gamma||_{L_{sc}^2(\mathcal{S}_{u,v})}
  &\lesssim 
    a^{s_2(\Gamma)}|u_{\infty}|^{\lambda_1-2s_2(\Gamma)}||\Gamma||_{L_{sc}^2(\mathcal{S}_{u_{\infty},v})}\\
  &+\int_{u_{\infty}}^ua^{s_2(F)}|u'|^{\lambda_1-2s_2(F)}||F||_{L_{sc}^2(\mathcal{S}_{u',v})}\mathrm{d}u',
\end{align*}
where~$s_2(F)=s_2(\mthorn'\Gamma)=s_2(\Gamma)+1$.
\end{proposition}

\subsection{Sobolev embedding}

Exploiting the fact that~$\mbox{Area}(\mathcal{S}_{u,v})\sim|u|^2$ and
working under the bootstrap assumptions, the Sobolev embedding
inequalities~\eqref{SobolevLp1} and~\eqref{SobolevLinfty1} imply the
following:

\begin{proposition}
  Under the bootstrap assumptions~\eqref{Hypothesis} one has
\begin{align}
||\phi||_{L^{\infty}(\mathcal{S}_{u,v})}\lesssim
\sum_{i\leq2}|||u|^{i-1}\mathcal{D}^i\phi||_{L^2(\mathcal{S}_{u,v})}.
\end{align}
In terms of scale-invariant norms one has that
\begin{align}
||\phi||_{L_{sc}^{\infty}(\mathcal{S}_{u,v})}\lesssim
  \sum_{i\leq2}||\left(a^{1/2}\mathcal{D}\right)^i
  \phi||_{L_{sc}^2(\mathcal{S}_{u,v})}.
\end{align}
\end{proposition}

\subsection{Commutators}

Suppose that the T-weighted quantity~$f$ satisfies the transport
equation~$\mthorn f=H_0$.  Then, using Stewart's gauge and the
commutator~\eqref{T-weightMasslessCommutator2} we have that
\begin{align*}
H_k=\sum_{i_1+i_2+i_3=k}\meth^{i_1}\Gamma(\pi,\tau)^{i_2}\meth^{i_3}H_0+
\sum_{i_1+i_2+i_3+i_4=k}\meth^{i_1}\Gamma(\tau,\pi)^{i_2}\meth^{i_3}\Gamma(\tau,\pi,\rho,\sigma)\meth^{i_4}f,
\end{align*}
where~$H_k\equiv\mthorn\meth^kf$. Similarly, suppose~$f$
satisfies~$\mthorn'f=G_0$, we have that
\begin{align*}
G_k=-k\mu\meth^kf+\meth^{k}G_0+
\sum_{i=1}^k\meth^i\mu\meth^{k-i}f+
\sum_{i=0}^k\meth^i\lambda\meth^{k-i}f,
\end{align*}
where~$G_k\equiv\mthorn'\meth^kf$. In the above expressions for~$H_k$
and~$G_k$ we need give neither the exact constants in each sum nor the
T-weight of~$f$. Furthermore, in the expressions for~$H_k$ and~$G_k$
we need not distinguish the~$\meth$ and~$\meth'$ operators or their
complex conjugates. The justification for this can be traced back to
the definition of the norms.

\section{$L^2(\mathcal{S})$ estimates}
\label{L2estimate}

In this section, under the bootstrap assumption
\begin{align*}
\bmGamma,\bm\Psi,\bm\varphi\leq\mathcal{O},
\end{align*}
we show~$L^2(\mathcal{S})$ estimates of the connection coefficients,
the auxiliary field~$\varphi_a$ and the Weyl curvature.  In particular
we demonstrate the inequality
 \begin{align*}
\bmGamma\lesssim\bm\Psi^2+\bm\Psi+\bm\varphi^2+\bm\varphi+\mathcal{I}_0+1.
\end{align*}
Note that~$\bmGamma$ contains up to 10 derivatives of the connection
coefficients, 10 derivatives of~$\varphi_{0,1,2}$ and 9 derivatives of
the Weyl scalars~$\Psi_{0,...,4}$. We use the 10th derivative of the
Weyl terms with the 11th derivative of~$\varphi_{0,1,2}$ on the
lightcone to control~$\bmGamma$.

The general strategy of the proof is to make use of the transport
equations satisfied by the various quantities involved in the field
equations. Then one applies the commutator relations into either
equation~\eqref{SCldirectiongronwall} or the inequalities in
Proposition~\ref{SCweightedgronwall}. Making use of scale-invariant
norms, we show in detail that the main contributions in the transport
equations come from the linear terms, and a handful of nonlinear terms
involving the scalar field. The T-weight of the various fields allows
us to obtain estimates for~$\mathcal{D}^kf$ once we have an estimate
on one of its strings of derivatives. For concreteness, in our
calculations we always choose $\meth^kf$ as an example. In the course
of calculation, we make frequent use of the bootstrap
assumption~\eqref{Hypothesis} and the fact
that~$\mathcal{O}\ll a\leq|u|$. In the case of the connection
coefficients we can always find structure equations which do not
involve angular derivatives. In contrast, for~$\varphi_{0,1,2}$ and
the components of the Weyl tensor we instead transform the estimates
of angular derivatives from~$\mathcal{S}$ to lightcone estimates.

\subsection{$L^2(\mathcal{S})$ estimates of connection coefficients}

We start by looking at the estimates for the connection
coefficients. Our first result is the following:

\begin{proposition}
\label{L2lambda}
For~$0\leq k\leq10$, one has
\begin{align*}
  \frac{a^{\frac{1}{2}}}{|u|}||
  (a^{\frac{1}{2}}\mathcal{D})^k\lambda||_{L^2_{sc}(\mathcal{S}_{u,v})}\lesssim1.
\end{align*}
\end{proposition}

\begin{proof}
  We illustrate the method of proof by looking in detail at the case
  of the spin connection coefficient~$\lambda$. In this case the
  relevant structure equation
  is~\eqref{T-weightMasslessStructureStEq8} ---namely
\begin{align*}
\mthorn'\lambda=-2\mu\lambda-\ulomega\lambda-\Psi_4.
\end{align*}
In order to estimate~$||(a^{\frac{1}{2}}\mathcal{D})^k\lambda||$ one
needs to commute each~$\mathcal{D}^{k_i}$ with~$\mthorn'$. We
take~$\meth^k\lambda$ as an example. One has that
\begin{align*}
\mthorn'\meth^k\lambda+(k+2)\mu\meth^k\lambda=-\meth^k\Psi_4
+\sum_{i=1}^k\meth^i\tilde\mu\meth^{k-i}\lambda 
+\sum_{i=0}^k\meth^i\Gamma(\ulomega,\lambda)\meth^{k-i}\lambda.
\end{align*}
Denote the right-hand side of the last equation by~$F$. One has that
\begin{align*}
&s_2(\lambda)=1, \ \ s_2(\meth^k\lambda)=\frac{k+2}{2}, \ \ s_2(F)=\frac{k+4}{2}, \\
&\lambda_0=-(k+2), \ \ \lambda_1=-\lambda_0-1=k+1, \\
& \lambda_1-2s_2(F)=-3, \ \ \lambda_1-2s_2(\meth^k\lambda)=-1.
\end{align*}

Now, making use of Proposition~\ref{SCweightedgronwall}, one has
\begin{align*}
\frac{a}{|u|}||(a^{\frac{1}{2}}\meth)^k\lambda||_{L^2_{sc}(\mathcal{S}_{u,v})}\lesssim
\frac{a}{|u_{\infty}|}||(a^{\frac{1}{2}}\meth)^k\lambda||_{L^2_{sc}(\mathcal{S}_{u_{\infty},v})}
+\int_{u_{\infty}}^u\frac{a^2}{|u'|^3}||a^{\frac{k}{2}}F||_{L^2_{sc}(\mathcal{S}_{u',v})}.
\end{align*}
Multiplying by~$a^{-\frac{1}{2}}$ one has 
\begin{align*}
\frac{a^{\frac{1}{2}}}{|u|}||(a^{\frac{1}{2}}\meth)^k\lambda||_{L^2_{sc}(\mathcal{S}_{u,v})}\lesssim
\frac{a^{\frac{1}{2}}}{|u_{\infty}|}||(a^{\frac{1}{2}}\mathcal{D})^k\lambda||_{L^2_{sc}(\mathcal{S}_{u_{\infty},v})}
+\int_{u_{\infty}}^u\frac{a^\frac{3}{2}}{|u'|^3}||a^{\frac{k}{2}}F||_{L^2_{sc}(\mathcal{S}_{u',v})}.
\end{align*}
Now, for convenience let 
\begin{align*}
H\equiv\int_{u_{\infty}}^u\frac{a^\frac{3}{2}}{|u|^3}||a^{\frac{k}{2}}F||_{L^2_{sc}(\mathcal{S}_{u,v})}.
\end{align*}
Substituting the expression of~$F$ one finds that 
\begin{align*}
  H&\leq
  \int_{u_{\infty}}^u\frac{a^\frac{3}{2}}{|u|^3}
    ||(a^{\frac{1}{2}}\mathcal{D})^k\Psi_4||_{L^2_{sc}(\mathcal{S}_{u',v})}
    +\sum_{i=1}^k\int_{u_{\infty}}^u\frac{a^\frac{3}{2}}{|u|^3}
    ||(a^{\frac{1}{2}}\mathcal{D})^i\tilde\mu(a^{\frac{1}{2}}
    \mathcal{D})^{k-1}\lambda||_{L^2_{sc}(\mathcal{S}_{u',v})} \\
  &\quad+\sum_{i=0}^k\int_{u_{\infty}}^u\frac{a^\frac{3}{2}}{|u|^3}||(a^{\frac{1}{2}}\mathcal{D})^i
    \Gamma(\ulomega,\lambda)(a^{\frac{1}{2}}\mathcal{D})^{k-1}
    \lambda||_{L^2_{sc}(\mathcal{S}_{u',v})} 
    (a^{\frac{1}{2}}\mathcal{D})^{i_3}\lambda||_{L^2_{sc}(\mathcal{S}_{u',v})}\\
&=I_1+I_2+I_3.
\end{align*}

For~$I_1$ it follows that
\begin{align*}
\int_{u_{\infty}}^u\frac{a^\frac{3}{2}}{|u|^3}
||(a^{\frac{1}{2}}\mathcal{D})^k\Psi_4||_{L^2_{sc}(\mathcal{S}_{u',v})}&\leq
\left(\int_{u_{\infty}}^u\frac{a}{|u|^2}
||(a^{\frac{1}{2}}\mathcal{D})^k\Psi_4||^2_{L^2_{sc}(\mathcal{S}_{u',v})}\right)^{\frac{1}{2}}
\left(\int_{u_{\infty}}^u\frac{a^2}{|u|^4}\right)^{\frac{1}{2}} \\
&\leq ||(a^{\frac{1}{2}}\mathcal{D})^i\Psi_4||_{L^2_{sc}(\mathcal{N}'_v(u_{\infty},u))}\frac{a}{|u|^{\frac{3}{2}}}
\leq\frac{\mathcal{O}}{a^{\frac{1}{2}}}\leq1.
\end{align*}

Now, making use of the bootstrap assumptions
\begin{align*}
||\tilde\mu||\leq\frac{|u|}{a}\mathcal{O}, \quad
||\lambda||\leq\frac{|u|}{a^{\frac{1}{2}}}\mathcal{O}, \quad
||\ulomega,\tau||\leq\mathcal{O}, \quad
||\mu||\leq\frac{|u|^2}{a}\mathcal{O},
\end{align*}
for~$I_2$ and~$I_3$, the above assumptions give
\begin{align*}
\int_{u_{\infty}}^u\frac{a^\frac{3}{2}}{|u'|^3}\frac{1}{|u'|}
\frac{|u'|}{a^{\frac{1}{2}}}\mathcal{O}
\frac{|u'|}{a^{\frac{1}{2}}}\mathcal{O}
=\int_{u_{\infty}}^u\frac{a^{\frac{1}{2}}\mathcal{O}^2}{|u'|^2}
\leq\frac{a^{\frac{1}{2}}\mathcal{O}^2}{|u|}
\leq\frac{\mathcal{O}^2}{a^{\frac{1}{2}}}
\leq1.
\end{align*}

Combining the results above with those on the initial data, one
concludes
\begin{align*}
  \frac{a^{\frac{1}{2}}}{|u|}||(a^{\frac{1}{2}}\meth)^k
  \lambda||_{L^2_{sc}(\mathcal{S}_{u,v})}\lesssim1.
\end{align*}
One can estimate the rest of the strings in~$\mathcal{D}^{k_i}\lambda$
in analogous way. Hence, from Definition~\eqref{T-weightL2Norm}, we
have
\begin{align*}
  \frac{a^{\frac{1}{2}}}{|u|}||(a^{\frac{1}{2}}\mathcal{D})^k
  \lambda||_{L^2_{sc}(\mathcal{S}_{u,v})}\lesssim1.
\end{align*}
\end{proof}

\begin{proposition}
\label{L2sigma}
For~$0\leq k\leq10$, one has that
\begin{align*}
  \frac{1}{a^{\frac{1}{2}}}
  ||(a^{\frac{1}{2}}\mathcal{D})^{k}\sigma||_{L^2_{sc}(\mathcal{S}_{u,v})}\lesssim
  \bm\Psi[\Psi_0]+1.
\end{align*}
\end{proposition}

\begin{proof}
  In this case we make use of the transport
  equation~\eqref{T-weightMasslessStructureStEq2}
\begin{align*}
\mthorn\sigma&=2\rho\sigma+\Psi_0
\end{align*}
to estimate~$\sigma$. For~$k\leq10$, commuting~$\meth^{k}$
with~$\mthorn$, one has that
\begin{align*}
\mthorn\meth^{k}\sigma&=\meth^k\Psi_0
+\sum_{i_1+i_2+i_3=k,i_3<k}\meth^{i_1}\Gamma(\tau,\pi)^{i_2}
\meth^{i_3}\Psi_0 
+\sum_{i_1+i_2+i_3+i_4=k}\meth^{i_1}\Gamma(\tau,\pi)^{i_2}
\meth^{i_3}\Gamma(\tau,\pi,\rho,\sigma)
\meth^{i_4}\sigma.
\end{align*}
Then, making use of the transport estimate in scale-invariant norm,
equation~\eqref{SCldirectiongronwall}, one has that
\begin{align*}
||(a^{\frac{1}{2}}\meth)^{k}\sigma||_{L^2_{sc}(\mathcal{S}_{u,v})}&\lesssim
||(a^{\frac{1}{2}}\meth)^{k}\sigma||_{L^2_{sc}(\mathcal{S}_{u,0})}+
\int_0^v||a^{\frac{k}{2}}\mthorn\meth^{k}\sigma||_{L^2_{sc}(\mathcal{S}_{u,v'})}.
\end{align*}
Note that the initial data for~$\sigma$ on the initial ingoing light
cone is Minkowskian ---i.e, it is zero. Now, multiplying
by~$a^{-\frac{1}{2}}$ both sides of the inequality we have 
\begin{align*}
\frac{1}{a^{\frac{1}{2}}}||(a^{\frac{1}{2}}\meth)^{k}\sigma||_{L^2_{sc}(\mathcal{S}_{u,v})}&\lesssim
\int_0^v\frac{1}{a^{\frac{1}{2}}}||a^{\frac{k}{2}}\mthorn\meth^{k}\sigma||_{L^2_{sc}(\mathcal{S}_{u,v'})} 
\leq\frac{1}{a^{\frac{1}{2}}}\int_0^v||a^{\frac{k}{2}}\mathcal{D}^{k}\Psi_0||_{L^2_{sc}(\mathcal{S}_{u,v'})}\\
&+\sum_{i_1+i_2+i_3=k,i_3<k}\int_0^v\frac{1}{a^{\frac{1}{2}}}
||a^{\frac{k}{2}}\mathcal{D}^{i_1}\Gamma(\tau,\pi)^{i_2}\mathcal{D}^{i_3}\Psi_0||_{L^2_{sc}(\mathcal{S}_{u,v'})} \\
&+\sum_{i_1+i_2+i_3+i_4=k}\int_0^v\frac{1}{a^{\frac{1}{2}}}
||a^{\frac{k}{2}}\mathcal{D}^{i_1}\Gamma(\tau,\pi)^{i_2}
\mathcal{D}^{i_3}\Gamma(\sigma,...)\mathcal{D}^{i_4}\sigma||_{L^2_{sc}(\mathcal{S}_{u,v'})} \\
&\leq\bm\Psi[\Psi_0]+\frac{a^{\frac{i_2}{2}}}{|u|^{i_2}}\mathcal{O}^{i_2+1}+
\frac{a^{\frac{i_2+1}{2}}}{|u|^{i_2+1}}\mathcal{O}^{i_2+2} \\
&\leq\bm\Psi[\Psi_0]+\frac{a^{\frac{1}{2}}}{|u|}\mathcal{O}
\leq\bm\Psi[\Psi_0]+1.
\end{align*}
The rest of the strings in~$\mathcal{D}^{k_i}\sigma$ can be estimated
in a similar way. Hence, from Definition~\eqref{T-weightL2Norm}, we
conclude that
\begin{align*}
  \frac{1}{a^{\frac{1}{2}}}||(a^{\frac{1}{2}}
  \mathcal{D})^{k}\sigma||_{L^2_{sc}(\mathcal{S}_{u,v})}\lesssim
\bm\Psi[\Psi_0]+1.
\end{align*}

\end{proof}

\begin{proposition}
\label{L2ulomega}
For~$0\leq k\leq10$, one has that 
\begin{align*}
||(a^{\frac{1}{2}}\mathcal{D})^{k}\ulomega||_{L^2_{sc}(\mathcal{S}_{u,v})}\lesssim
\bm\Psi[\TiPsi_2]+\underline{\bm\varphi}[\varphi_1]+1.
\end{align*}
\end{proposition}

\begin{proof}
  Here we make use of the structure
  equation~\eqref{T-weightMasslessStructureStEq11} ---namely,
\begin{align*}
\mthorn\ulomega&=2\tau\bar\tau+2\tau\pi+2\bar\tau\bar\pi
+\TiPsi_2+\bar\TiPsi_2+6\varphi_{0}\varphi_{2}.
\end{align*}
For~$k\leq10$, we commute~$\meth^{k}$ with~$\mthorn$, to obtain
\begin{align*}
\mthorn\meth^{k}\ulomega&=\meth^k\TiPsi_2
+\sum_{i_1+i_2+i_3=k,i_3<k}\meth^{i_1}\Gamma(\tau,\pi)^{i_2}
\meth^{i_3}\TiPsi_2
+\sum_{i_1+i_2+i_3+i_4=k}\meth^{i_1}\Gamma(\tau,\pi)^{i_2}
\meth^{i_3}\Gamma(\tau,\pi,\rho,\sigma)
\meth^{i_4}\Gamma(\tau,\pi,\ulomega) \\
&+\sum_{i_1+i_2+i_3+i_4=k}\meth^{i_1}\Gamma(\tau,\pi)^{i_2}
\meth^{i_3}\varphi_0
\meth^{i_4}\varphi_2.
\end{align*}
Then, the transport estimate in the scale-invariant norm,
inequality~\eqref{SCldirectiongronwall}, yields
\begin{align*}
||(a^{\frac{1}{2}}\meth)^{k}\ulomega||_{L^2_{sc}(\mathcal{S}_{u,v})}&\lesssim
||(a^{\frac{1}{2}}\meth)^{k}\ulomega||_{L^2_{sc}(\mathcal{S}_{u,0})}+
\int_0^v||a^{\frac{k}{2}}\mthorn\meth^{k}\ulomega||_{L^2_{sc}(\mathcal{S}_{u,v'})}. 
\end{align*}

As the initial data for~$\ulomega$ on the initial ingoing light cone
is Minkowskian (and thus vanishes), we have
\begin{align*}
||(a^{\frac{1}{2}}\meth)^{k}\ulomega||_{L^2_{sc}(\mathcal{S}_{u,v})}&\lesssim
\int_0^v||a^{\frac{k}{2}}\mthorn\meth^{k}\ulomega||_{L^2_{sc}(\mathcal{S}_{u,v'})} 
\leq\int_0^v||a^{\frac{k}{2}}\mathcal{D}^{k}\TiPsi_2||_{L^2_{sc}(\mathcal{S}_{u,v'})} \\
&+\sum_{i_1+i_2+i_3=k,i_3<k}\int_0^v
||a^{\frac{k}{2}}\mathcal{D}^{i_1}\Gamma(\tau,\pi)^{i_2}\mathcal{D}^{i_3}\TiPsi_2||_{L^2_{sc}(\mathcal{S}_{u,v'})} \\
&+\sum_{i_1+i_2+i_3+i_4=k}\int_0^v
||a^{\frac{k}{2}}\mathcal{D}^{i_1}\Gamma(\tau,\pi)^{i_2}
\mathcal{D}^{i_3}\Gamma(\sigma,...)\mathcal{D}^{i_4}\Gamma||_{L^2_{sc}(\mathcal{S}_{u,v'})} \\
&+\sum_{i_1+i_2+i_3+i_4=k}\int_0^v
||a^{\frac{k}{2}}\mathcal{D}^{i_1}\Gamma(\tau,\pi)^{i_2}
\mathcal{D}^{i_3}\varphi_0\mathcal{D}^{i_4}\varphi_2||_{L^2_{sc}(\mathcal{S}_{u,v'})} \\
&\leq\bm\Psi[\TiPsi_2]+\frac{a^{\frac{i_2}{2}}}{|u|^{i_2}}\mathcal{O}^{i_2+1}
+\frac{a^{\frac{i_2}{2}}}{|u|^{i_2+1}}\mathcal{O}^{i_2+1}a^{\frac{1}{2}}\mathcal{O}
+\frac{a^{\frac{i_2}{2}}}{|u|^{i_2+1}}\mathcal{O}^{i_2}a^{\frac{1}{2}}\mathcal{O}\frac{|u|}{a^{\frac{1}{2}}}\mathcal{O} \\
&\leq\bm\Psi[\TiPsi_2]+\frac{a^{\frac{1}{2}}}{|u|}\mathcal{O}^2+\bmGamma(\varphi_0)\bmGamma(\varphi_2)
\leq\bm\Psi[\TiPsi_2]+\bmGamma(\varphi_0)\bmGamma(\varphi_2)+1 \\
&\leq\bm\Psi[\TiPsi_2]+\underline{\bm\varphi}[\varphi_1]+1,
\end{align*}
where in the last step we have made use of the~$L^2(\mathcal{S})$
estimates for~$\varphi_0$ and~$\varphi_2$ which will be shown in
Proposition~\ref{L2varphi0} and Proposition~\ref{L2varphi12}.

The remaining strings of derivatives in~$\mathcal{D}^{k_i}\ulomega$
can be estimated in a similar way. Hence, from
Definition~\ref{T-weightL2Norm}, we conclude that
\begin{align*}
||(a^{\frac{1}{2}}\mathcal{D})^{k}\ulomega||_{L^2_{sc}(\mathcal{S}_{u,v})}\lesssim
\bm\Psi[\TiPsi_2]+\underline{\bm\varphi}[\varphi_1]+1.
\end{align*}
\end{proof}

\begin{proposition}
\label{L2ulchi}
For~$0\leq k\leq10$, one has that
\begin{align*}
||(a^{\frac{1}{2}}\mathcal{D})^k\ulchi||_{L^2_{sc}(\mathcal{S}_{u,v})}\lesssim
\underline{\bm\Psi}[\TiPsi_2]+\underline{\bm\varphi}[\varphi_1]+1.
\end{align*}
\end{proposition}

\begin{proof}
  We make use of the transport equation~\eqref{Tpulchi} ---namely
\begin{align*}
\mthorn'\ulchi=2\tau\pi+2\bar\tau\bar\pi+2\pi\bar\pi+\TiPsi_2+\bar\TiPsi_2+6\varphi_0\varphi_2+\ulomega\ulchi.
\end{align*}
Commuting~$\mthorn'$ with~$\meth^k$ we have that
\begin{align*}
\mthorn'\meth^k\ulchi+k\mu\meth^k\ulchi=\meth^k\TiPsi_2
+\sum_{i=0}^k\meth^i\Gamma(\tau,\pi)\meth^{k-i}\pi
+\sum_{i=0}^k\meth^i\ulomega\meth^{k-i}\ulchi
+\sum_{i=0}^k\meth^i\varphi_0\meth^{k-i}\varphi_2.
\end{align*}
Denote the right-hand side of the above equation by~$F$. One has that 
\begin{align*}
&s_2(\ulchi)=0, \ \ s_2(\meth^k\ulchi)=\frac{k}{2}, \ \ s_2(F)=\frac{k+2}{2}, \\
&\lambda_0=-k, \ \ \lambda_1=-\lambda_0-1=k-1, \\
& \lambda_1-2s_2(F)=-3, \ \ \lambda_1-2s_2(\meth^k\ulchi)=-1.
\end{align*}
Now, making use of Proposition~\ref{SCweightedgronwall}, one obtains
\begin{align*}
\frac{1}{|u|}||(a^{\frac{1}{2}}\meth)^k\ulchi||_{L^2_{sc}(\mathcal{S}_{u,v})}&\lesssim
\int_{u_{\infty}}^u\frac{a}{|u'|^3}||a^{\frac{k}{2}}F||_{L^2_{sc}(\mathcal{S}_{u',v})} \\
&\lesssim\int_{u_{\infty}}^u\frac{a}{|u'|^3}||a^{\frac{k}{2}}\meth^k\TiPsi_2||_{L^2_{sc}(\mathcal{S}_{u',v})} 
+\sum_{i=0}^k\int_{u_{\infty}}^u\frac{a}{|u'|^3}||a^{\frac{k}{2}}\meth^i\varphi_0
\meth^{k-i}\varphi_2||_{L^2_{sc}(\mathcal{S}_{u',v})} \\
&+\sum_{i=0}^k\int_{u_{\infty}}^u\frac{a}{|u'|^3}||a^{\frac{k}{2}}\meth^i\Gamma(\tau,\pi,\ulomega)
\meth^{k-i}\Gamma(\pi,\ulchi)||_{L^2_{sc}(\mathcal{S}_{u',v})} \\
&=I_1+I_2+I_3.
\end{align*}

The term~$I_1$ can be estimated as follows:
\begin{align*}
I_1&\leq\left(\int_{u_{\infty}}^u\frac{a}{|u'|^2}||a^{\frac{k}{2}}\meth^k\TiPsi_2||^2_{L^2_{sc}(\mathcal{S}_{u',v})}\right)^{\frac{1}{2}}
\left(\int_{u_{\infty}}^u\frac{a}{|u'|^4} \right)^{\frac{1}{2}} \\
&\leq\underline{\Psi}[\TiPsi_2]\frac{a^{\frac{1}{2}}}{|u|^{\frac{3}{2}}}
\leq\frac{\underline{\bm\Psi}[\TiPsi_2]}{|u|}.
\end{align*}

For~$I_2$ we have that
\begin{align*}
I_2&\leq\sum_{i=0}^k\int_{u_{\infty}}^u\frac{a}{|u'|^4}||(a^{\frac{1}{2}}\meth)^i\varphi_0||\times
||(a^{\frac{1}{2}}\meth)^{k-i}\varphi_2|| \\
&\leq\int_{u_{\infty}}^u\frac{a}{|u'|^4}a^{\frac{1}{2}}\Gamma(\varphi_0)\frac{|u|}{a^{\frac{1}{2}}}\Gamma(\varphi_2)
\leq\frac{\Gamma(\varphi_0)\Gamma(\varphi_2)}{|u|} \\
&\lesssim\frac{\underline{\bm\varphi}[\varphi_1]+1}{|u|}.
\end{align*}
The last step in the above chain of inequalities will be shown in the
Propositions~\ref{L2varphi0} and~\ref{L2varphi12}.

Finally, for~$I_3$ we have that 
\begin{align*}
I_3&\leq\sum_{i=0}^k\int_{u_{\infty}}^u\frac{a}{|u'|^4}||(a^{\frac{1}{2}}\meth)^i\Gamma||\times
||(a^{\frac{1}{2}}\meth)^{k-i}\Gamma|| 
\leq\int_{u_{\infty}}^u\frac{a}{|u'|^4}\mathcal{O}^2
\lesssim\frac{1}{|u|}.
\end{align*}

Collecting the previous estimates, we obtain that 
\begin{align*}
\frac{1}{|u|}||(a^{\frac{1}{2}}\meth)^k\ulchi||_{L^2_{sc}(\mathcal{S}_{u,v})}\lesssim
\frac{\underline{\bm\Psi}[\TiPsi_2]}{|u|}+\frac{\underline{\bm\varphi}[\varphi_1]+1}{|u|},
\end{align*}
from which it follows that 
\begin{align*}
||(a^{\frac{1}{2}}\mathcal{D})^k\ulchi||_{L^2_{sc}(\mathcal{S}_{u,v})}\lesssim
\underline{\Psi}[\bm\TiPsi_2]+\underline{\bm\varphi}[\varphi_1]+1.
\end{align*}
\end{proof}

\begin{proposition}
\label{L2tau}
For~$0\leq k\leq10$, one has that
\begin{align*}
||(a^{\frac{1}{2}}\mathcal{D})^{k}\tau||_{L^2_{sc}(\mathcal{S}_{u,v})}\lesssim
\bm\Psi[\TiPsi_1]+\underline{\bm\varphi}[\Tivarphi_2]\underline{\bm\varphi}[\varphi_1]
+\underline{\bm\varphi}[\Tivarphi_2]+\underline{\bm\varphi}[\varphi_1]+1.
\end{align*}
\end{proposition}

\begin{proof}

We make use of the structure equation~\eqref{T-weightMasslessStructureStEq9}
\begin{align*}
\mthorn\tau&=(\tau+\bar\pi)\rho+(\bar\tau+\pi)\sigma
+\TiPsi_1+6\varphi_0\varphi_1.
\end{align*}
For~$k\leq10$, commuting~$\meth^{k}$ with~$\mthorn$, we have that 
\begin{align*}
  \mthorn\meth^{k}\tau
  &=\meth^k\TiPsi_1
+\sum_{i_1+i_2+i_3=k,i_3<k}\meth^{i_1}\Gamma(\tau,\pi)^{i_2}
\meth^{i_3}\TiPsi_1
  \\
  &+\sum_{i_1+i_2+i_3+i_4=k}\meth^{i_1}\Gamma(\tau,\pi)^{i_2}
\meth^{i_3}\Gamma(\tau,\pi,\rho,\sigma)
\meth^{i_4}\Gamma(\tau,\pi)+\sum_{i=0}^k\meth^i\varphi_0\meth^{k-i}\varphi_1\\
  &+\sum_{i_1+i_2+i_3+i_4=k,i_3+i_4<k}\meth^{i_1}\Gamma(\tau,\pi)^{i_2}
    \meth^{i_3}\varphi_0\meth^{i_4}\varphi_1.
\end{align*}
Next, we make use of the transport estimate in scale-invariant norm
\eqref{SCldirectiongronwall}. One concludes that 
\begin{align*}
||(a^{\frac{1}{2}}\meth)^{k}\tau||_{L^2_{sc}(\mathcal{S}_{u,v})}&\lesssim
||(a^{\frac{1}{2}}\meth)^{k}\tau||_{L^2_{sc}(\mathcal{S}_{u,0})}+
\int_0^v||a^{\frac{k}{2}}\mthorn\meth^{k}\tau||_{L^2_{sc}(\mathcal{S}_{u,v'})}. 
\end{align*}
As the initial data of~$\tau$ on the initial ingoing light cone is
Minkowskian, we have that~$\tau=0$. Accordingly, we find that 
\begin{align*}
||(a^{\frac{1}{2}}\meth)^{k}\tau||_{L^2_{sc}(\mathcal{S}_{u,v})}&\lesssim
\int_0^v||a^{\frac{k}{2}}\mthorn\meth^{k}\tau||_{L^2_{sc}(\mathcal{S}_{u,v'})} 
\leq\int_0^v||a^{\frac{k}{2}}\mathcal{D}^{k}\TiPsi_1||_{L^2_{sc}(\mathcal{S}_{u,v'})} \\
&+\sum_{i_1+i_2+i_3=k,i_3<k}\int_0^v
||a^{\frac{k}{2}}\mathcal{D}^{i_1}\Gamma(\tau,\pi)^{i_2}\mathcal{D}^{i_3}\TiPsi_1||_{L^2_{sc}(\mathcal{S}_{u,v'})} \\
&+\sum_{i_1+i_2+i_3+i_4=k}\int_0^v
||a^{\frac{k}{2}}\mathcal{D}^{i_1}\Gamma(\tau,\pi)^{i_2}
\mathcal{D}^{i_3}\Gamma(\sigma,...)\mathcal{D}^{i_4}\Gamma||_{L^2_{sc}(\mathcal{S}_{u,v'})} \\
&+\sum_{i=0}^k\int_0^v||a^{\frac{k}{2}}\mathcal{D}^{i}\varphi_0
\mathcal{D}^{k-i}\varphi_1||_{L^2_{sc}(\mathcal{S}_{u,v'})} \\
&+\sum_{i_1+i_2+i_3+i_4=k,i_3+i_4<k}\int_0^v
||a^{\frac{k}{2}}\mathcal{D}^{i_1}\Gamma(\tau,\pi)^{i_2}
\mathcal{D}^{i_3}\varphi_0\mathcal{D}^{i_4}\varphi_1||_{L^2_{sc}(\mathcal{S}_{u,v'})} \\
&\leq\bm\Psi[\TiPsi_1]+\frac{a^{\frac{i_2}{2}}}{|u|^{i_2}}\mathcal{O}^{i_2+1}
+\frac{a^{\frac{i_2}{2}}}{|u|^{i_2+1}}\mathcal{O}^{i_2+1}a^{\frac{1}{2}}\mathcal{O}
+\frac{a}{|u|}\bmGamma(\varphi_0)\bmGamma(\varphi_1)
+\frac{a^{\frac{i_2}{2}}}{|u|^{i_2+1}}\mathcal{O}^{i_2}a^{\frac{1}{2}}\mathcal{O}\mathcal{O} \\
&\leq\bm\Psi[\TiPsi_1]+\frac{a}{|u|}\bmGamma(\varphi_0)\bmGamma(\varphi_1)+\frac{a^{\frac{1}{2}}}{|u|}\mathcal{O}
\leq\bm\Psi[\TiPsi_1]+\frac{a}{|u|}\bmGamma(\varphi_0)\bmGamma(\varphi_1)+1.
\end{align*}
The above leads to
\begin{align*}
||(a^{\frac{1}{2}}\mathcal{D})^{k}\tau||_{L^2_{sc}(\mathcal{S}_{u,v})}&\lesssim
\bm\Psi[\TiPsi_1]+\frac{a}{|u|}\bmGamma(\varphi_0)\bmGamma(\varphi_1)+1 \\
&\lesssim\bm\Psi[\TiPsi_1]+\underline{\bm\varphi}[\Tivarphi_2]\underline{\bm\varphi}[\varphi_1]
+\underline{\bm\varphi}[\Tivarphi_2]+\underline{\bm\varphi}[\varphi_1]+1.
\end{align*}

\end{proof}

\begin{proposition}
\label{L2rho}
For~$0\leq k\leq10$, one has that 
\begin{align*}
||(a^{\frac{1}{2}}\mathcal{D})^{k}\rho||_{L^2_{sc}(\mathcal{S}_{u,v})}\lesssim
(\bm\Psi[\Psi_0]+\underline{\bm\varphi}[\varphi_1]+1)^2.
\end{align*}
\end{proposition}

\begin{proof}
  We make use of the structure equation for~$\rho$ ---namely,
  \eqref{T-weightMasslessStructureStEq1}
\begin{align*}
\mthorn\rho&=\rho^2+\sigma\bar\sigma+3\varphi_{0}^2.
\end{align*}
For~$k\leq10$, commuting~$\meth^{k}$ with~$\mthorn$, we have that 
\begin{align*}
\mthorn\meth^{k}\rho&=\sum_{i=0}^k\meth^i\sigma\meth^{k-i}\bar\sigma+
\sum_{i=0}^k\meth^i\varphi_0\meth^{k-i}\varphi_0
+\sum_{i_1+i_2+i_3+i_4=k}\meth^{i_1}\Gamma(\tau,\pi)^{i_2}
\meth^{i_3}\rho\meth^{i_4}\rho \\
&+\sum_{i_1+i_2+i_3+i_4=k,i_3+i_4<k}\meth^{i_1}\Gamma(\tau,\pi)^{i_2}
\meth^{i_3}\sigma\meth^{i_4}\sigma
+\sum_{i_1+i_2+i_3+i_4=k,i_3+i_4<k}\meth^{i_1}\Gamma(\tau,\pi)^{i_2}
\meth^{i_3}\varphi_0\meth^{i_4}\varphi_0  \\
&+\sum_{i_1+i_2+i_3+i_4=k}\meth^{i_1}\Gamma(\tau,\pi)^{i_2}
\meth^{i_3}\Gamma(\tau,\pi,\rho,\sigma)
\meth^{i_4}\rho.
\end{align*}
Next, we make use of the transport estimate in terms of
scale-invariant norms, inequality~\eqref{SCldirectiongronwall}. One
finds that
\begin{align*}
&||(a^{\frac{1}{2}}\meth)^{k}\rho||_{L^2_{sc}(\mathcal{S}_{u,v})}\lesssim
||(a^{\frac{1}{2}}\meth)^{k}\rho||_{L^2_{sc}(\mathcal{S}_{u,0})}+
\int_0^v||a^{\frac{k}{2}}\mthorn\meth^{k}\rho||_{L^2_{sc}(\mathcal{S}_{u,v'})} \\
&\quad\quad\leq||(a^{\frac{1}{2}}\mathcal{D})^{k}\rho||_{L^2_{sc}(\mathcal{S}_{u,0})}
+\sum_{i=0}^k\int_0^v||a^{\frac{k}{2}}\mathcal{D}^{i}\sigma\mathcal{D}^{k-i}\sigma||_{L^2_{sc}(\mathcal{S}_{u,v'})} 
+\sum_{i=0}^k\int_0^v||a^{\frac{k}{2}}\mathcal{D}^{i}\varphi_0\mathcal{D}^{k-i}\varphi_0||_{L^2_{sc}(\mathcal{S}_{u,v'})} \\
&\quad\quad\quad+\sum_{i_1+i_2+i_3+i_4=k}\int_0^v
||a^{\frac{k}{2}}\mathcal{D}^{i_1}\Gamma(\tau,\pi)^{i_2}\mathcal{D}^{i_3}(\Gamma(\tau,\pi,\rho,\sigma),\varphi_0)
\mathcal{D}^{i_4}(\Gamma(\rho,\sigma),\varphi_0)||_{L^2_{sc}(\mathcal{S}_{u,v'})} \\
&\quad\quad\lesssim1+\frac{a}{|u|}\bmGamma(\sigma)_{\infty}\bmGamma(\sigma)_{2}
+\frac{a}{|u|}\bmGamma(\varphi_0)_{\infty}\bmGamma(\varphi_0)_{2}
+\frac{a^{\frac{i_2}{2}}}{|u|^{i_2+1}}\mathcal{O}^{i_2}a^{\frac{1}{2}}\mathcal{O}a^{\frac{1}{2}}\mathcal{O}
+\frac{a^{\frac{i_2}{2}}}{|u|^{i_2+1}}\mathcal{O}^{i_2+2} \\
&\quad\quad\leq1+\frac{a}{|u|}\left(\bmGamma(\sigma)_{\infty}\bmGamma(\sigma)_{2}
+\bmGamma(\varphi_0)_{\infty}\bmGamma(\varphi_0)_{2}\right)+\frac{\mathcal{O}}{|u|}+\frac{a\mathcal{O}}{|u|^2} \\
&\quad\quad\leq(\bm\Psi[\Psi_0]+\underline{\bm\varphi}[\varphi_1]+1)^2.
\end{align*}
Observe that in the third inequality one has
that~$i_2\neq0$. Similarly, one can estimate the rest of the strings
in~$\mathcal{D}^{k_i}\rho$ and obtain the same result. We conclude
then that
\begin{align*}
||(a^{\frac{1}{2}}\mathcal{D})^{k}\rho||_{L^2_{sc}(\mathcal{S}_{u,v})}\lesssim
(\bm\Psi[\Psi_0]+\underline{\bm\varphi}[\varphi_1]+1)^2.
\end{align*}
\end{proof}

\begin{remark}
\label{L2rhoAlt}
{\em The initial data of~$\rho$ on the ingoing light cone is
  Minkowskian ---that is, one has~$\rho=-\frac{1}{|u|}$. In terms of scale-invariant 
norms we have
\begin{align*}
||\rho||_{L^2_{sc}(\mathcal{S}_{u,0})}=1; \quad 
\sum_{1\leq k\leq10}||(a^{\frac{1}{2}}\mathcal{D})^{k}\rho||_{L^2_{sc}(\mathcal{S}_{u,0})}=0.
\end{align*}
Hence, when~$1\leq k\leq9$, we have that 
\begin{align*}
||(a^{\frac{1}{2}}\mathcal{D})^{k}\rho||_{L^2_{sc}(\mathcal{S}_{u,v})}\lesssim
\frac{a}{|u|}\left(\bmGamma(\sigma)_{\infty}\bmGamma(\sigma)_{2}
+\bmGamma(\varphi_0)_{\infty}\bmGamma(\varphi_0)_{2}\right)+
\frac{\mathcal{O}}{|u|}+\frac{a\mathcal{O}}{|u|^2}.
\end{align*}
}
\end{remark}

\begin{proposition}
\label{L2mu}
For~$0\leq k\leq10$, one has that 
\begin{align*}
\frac{a}{|u|}||(a^{\frac{1}{2}}\mathcal{D})^k\Timu||_{L^2_{sc}(\mathcal{S}_{u,v})}&\lesssim
\bm\Psi[\TiPsi_2]+\underline{\bm\varphi}[\varphi_1]+1,
\end{align*}
\begin{align*}
\frac{a}{|u|^2}||(a^{\frac{1}{2}}\mathcal{D})^k\mu||_{L^2_{sc}(\mathcal{S}_{u,v})}&\lesssim1.
\end{align*}
\end{proposition}

\begin{proof}
In this case we make use of the ingoing equation of~$\Timu$ ---namely,
\begin{align*}
\mthorn'\Timu+2\mu\Timu=\Timu^2-\ulomega\mu-\lambda\bar\lambda-3\varphi_2^2 .
\end{align*}
Commuting~$\meth^k$ with~$\mthorn'$, we find that 
\begin{align*}
\mthorn'\meth^k\Timu+(k+2)\mu\meth^k\Timu&=
\sum_{i=0}^k\meth^i\Gamma(\Timu,\lambda)\meth^{k-i}\Gamma(\Timu,\lambda)
+\sum_{i=0}^k\meth^i\ulomega\meth^{k-i}\mu
+\sum_{i=0}^k\meth^i\varphi_2\meth^{k-i}\varphi_2.
\end{align*}
Denoting the right-hand side of the above equation by~$F$ it follows
then that 
\begin{align*}
&s_2(\Timu)=1, \ \ s_2(\meth^k\Timu)=\frac{k+2}{2}, \ \ s_2(F)=\frac{k+4}{2}, \\
&\lambda_0=-(k+2), \ \ \lambda_1=-\lambda_0-1=k+1, \\
& \lambda_1-2s_2(F)=-3, \ \ \lambda_1-2s_2(\meth^k\lambda)=-1.
\end{align*}
Now, making use of Proposition~\ref{SCweightedgronwall}, one has that
\begin{align*}
\frac{a}{|u|}||(a^{\frac{1}{2}}\meth)^k\Timu||_{L^2_{sc}(\mathcal{S}_{u,v})}&\lesssim
\frac{a}{|u_{\infty}|}||(a^{\frac{1}{2}}\meth)^k\Timu||_{L^2_{sc}(\mathcal{S}_{u_{\infty},v})}
+\int_{u_{\infty}}^u\frac{a^2}{|u'|^3}||a^{\frac{k}{2}}F||_{L^2_{sc}(\mathcal{S}_{u',v})} \\
&\leq\frac{a}{|u_{\infty}|}||(a^{\frac{1}{2}}\mathcal{D})^k\Timu||_{L^2_{sc}(\mathcal{S}_{u_{\infty},v})}
+\int_{u_{\infty}}^u\frac{a^2}{|u'|^3}||a^{\frac{k}{2}}\mu
\mathcal{D}^{k}\ulomega||_{L^2_{sc}(\mathcal{S}_{u',v})} \\
&+\sum_{i=1}^k\int_{u_{\infty}}^u\frac{a^2}{|u'|^3}||a^{\frac{k}{2}}\mathcal{D}^i\Timu
\mathcal{D}^{k-i}\ulomega||_{L^2_{sc}(\mathcal{S}_{u',v})} \\
&+\sum_{i=0}^k\int_{u_{\infty}}^u\frac{a^2}{|u'|^3}||a^{\frac{k}{2}}\mathcal{D}^i\varphi_2
\mathcal{D}^{k-i}\varphi_2||_{L^2_{sc}(\mathcal{S}_{u',v})}  \\
&+\sum_{i=0}^k\int_{u_{\infty}}^u\frac{a^2}{|u'|^3}||a^{\frac{k}{2}}\mathcal{D}^i(\Timu,\lambda)
\mathcal{D}^{k-i}(\Timu,\lambda)||_{L^2_{sc}(\mathcal{S}_{u',v})}  \\
&\lesssim1+\int_{u_{\infty}}^u\frac{a^2}{|u'|^3}\frac{1}{|u'|}\frac{|u'|^2}{a}\mathcal{O}\bmGamma(\ulomega)
+\int_{u_{\infty}}^u\frac{a^2}{|u'|^3}\frac{1}{|u'|}\frac{|u'|}{a^{\frac{1}{2}}}\frac{|u'|}{a^{\frac{1}{2}}}\bmGamma(\varphi_2)^2 \\
&+\int_{u_{\infty}}^u\frac{a^2}{|u'|^3}\frac{1}{|u'|}
\left(\frac{|u'|}{a}\mathcal{O}\mathcal{O}
+\frac{|u'|}{a}\mathcal{O}\frac{|u'|}{a}\mathcal{O}
+\frac{|u'|}{a^{\frac{1}{2}}}\frac{|u'|}{a^{\frac{1}{2}}}
+\frac{|u'|}{a^{\frac{1}{2}}}\frac{|u'|}{a}\mathcal{O}\right) \\
&\leq1+\frac{a}{|u|}(\bmGamma(\ulomega)+\bmGamma(\varphi_2)^2)\\
&\leq1+\bm\Psi[\TiPsi_2]+\underline{\bm\varphi}[\varphi_1],
\end{align*}
so that, in the end, one has that 
\begin{align*}
\frac{a}{|u|}||(a^{\frac{1}{2}}\mathcal{D})^k\Timu||_{L^2_{sc}(\mathcal{S}_{u,v})}&\lesssim
\bm\Psi[\TiPsi_2]+\underline{\bm\varphi}[\varphi_1]+1,
\end{align*}
so that, consequently,
\begin{align*}
\frac{a}{|u|^2}||(a^{\frac{1}{2}}\mathcal{D})^k\mu||_{L^2_{sc}(\mathcal{S}_{u,v})}&\lesssim 1.
\end{align*}

\end{proof}

\begin{proposition}
\label{L2pi}
For~$0\leq k\leq10$, one has that 
\begin{align*}
||(a^{\frac{1}{2}}\mathcal{D})^k\pi||_{L^2_{sc}(\mathcal{S}_{u,v})}\lesssim
\underline{\bm\Psi}[\TiPsi_3]+\bm\Psi[\TiPsi_1]+\underline{\bm\varphi}[\Tivarphi_2]+1.
\end{align*}
\end{proposition}

\begin{proof}

  We make use of the structure
  equation~\eqref{T-weightMasslessStructureStEq10}, that is
\begin{align*}
\mthorn'\pi=-(\pi+\bar\tau)\mu-(\bar\pi+\tau)\lambda
-\TiPsi_3-6\varphi_{2}\bar\varphi_{1}.
\end{align*}
Commuting$\meth^k$ with~$\mthorn'$, we find that 
\begin{align*}
\mthorn'\meth^k\pi+(k+1)\mu\meth^k\pi&=-\meth^k\TiPsi_3-\mu\meth^k\bar\tau
+\sum_{i=0}^k\meth^i\Gamma(\tau,\pi)\meth^{k-i}\Gamma(\Timu,\lambda) 
+\sum_{i=0}^k\meth^i\bar\varphi_1\meth^{k-i}\varphi_2. 
\end{align*}
Denoting the right-hand side of the above expression by~$F$, it follows
that 
\begin{align*}
&s_2(\pi)=\frac{1}{2}, \ \ s_2(\meth^k\pi)=\frac{k+1}{2}, \ \ s_2(F)=\frac{k+3}{2}, \\
&\lambda_0=-(k+1), \ \ \lambda_1=-\lambda_0-1=k, \\
& \lambda_1-2s_2(F)=-3, \ \ \lambda_1-2s_2(\meth^k\pi)=-1.
\end{align*}
Now, making use of Proposition~\ref{SCweightedgronwall} one has
\begin{align*}
\frac{a^{\frac{1}{2}}}{|u|}||(a^{\frac{1}{2}}\meth)^k\pi||_{L^2_{sc}(\mathcal{S}_{u,v})}\lesssim
\frac{a^{\frac{1}{2}}}{|u_{\infty}|}||(a^{\frac{1}{2}}\meth)^k\pi||_{L^2_{sc}(\mathcal{S}_{u_{\infty},v})}
+\int_{u_{\infty}}^u\frac{a^\frac{3}{2}}{|u'|^3}||a^{\frac{k}{2}}F||_{L^2_{sc}(\mathcal{S}_{u',v})}.
\end{align*}
Multiplying by~$a^{-\frac{1}{2}}$ one has that
\begin{align*}
\frac{1}{|u|}||(a^{\frac{1}{2}}\meth)^k\pi||_{L^2_{sc}(\mathcal{S}_{u,v})}&\lesssim
\frac{1}{|u_{\infty}|}||(a^{\frac{1}{2}}\meth)^k\pi||_{L^2_{sc}(\mathcal{S}_{u_{\infty},v})}
+\int_{u_{\infty}}^u\frac{a}{|u'|^3}||a^{\frac{k}{2}}F||_{L^2_{sc}(\mathcal{S}_{u',v})} \\
&\lesssim\frac{1}{|u_{\infty}|}
+\int_{u_{\infty}}^u\frac{a}{|u'|^3}||a^{\frac{k}{2}}\mathcal{D}^k\TiPsi_3||_{L^2_{sc}(\mathcal{S}_{u',v})} 
+\int_{u_{\infty}}^u\frac{a}{|u'|^3}||a^{\frac{k}{2}}\mu\mathcal{D}^k\bar\tau||_{L^2_{sc}(\mathcal{S}_{u',v})} \\
&+\sum_{i=0}^k\int_{u_{\infty}}^u\frac{a}{|u'|^3}||a^{\frac{k}{2}}\mathcal{D}^i(\tau,\pi)
\mathcal{D}^{k-i}(\Timu,\lambda)||_{L^2_{sc}(\mathcal{S}_{u',v})} \\
&+\sum_{i=0}^k\int_{u_{\infty}}^u\frac{a}{|u'|^3}||a^{\frac{k}{2}}\mathcal{D}^i\varphi_1
\mathcal{D}^{k-i}\varphi_2||_{L^2_{sc}(\mathcal{S}_{u',v})} \\
&=\frac{1}{|u_{\infty}|}+I_1+...+I_4.
\end{align*}

Now, for the term~$I_1$ we have that
\begin{align*}
I_1&\leq\left(\int_{u_{\infty}}^u\frac{a}{|u'|^2}||a^{\frac{k}{2}}\mathcal{D}^k
\TiPsi_3||^2_{L^2_{sc}(\mathcal{S}_{u',v})} \right)^{\frac{1}{2}}
\left(\int_{u_{\infty}}^u\frac{a}{|u'|^4}\right)^{\frac{1}{2}} 
\leq\frac{1}{|u|}\underline{\bm\Psi}[\TiPsi_3].
\end{align*}
For~$I_2$ we have that
\begin{align*}
I_2\leq\int_{u_{\infty}}^u\frac{a}{|u'|^3}\frac{1}{|u'|}||\mu||_{L^{\infty}_{sc}(\mathcal{S}_{u',v})} 
||a^{\frac{k}{2}}\mathcal{D}^k\bar\tau||_{L^2_{sc}(\mathcal{S}_{u',v})} 
\leq\int_{u_{\infty}}^u\frac{a}{|u'|^4}\frac{|u'|^2}{a}\Gamma(\tau) 
\lesssim\frac{\bm{\Gamma}(\tau)}{|u|}.
\end{align*}
For~$I_3$ we have that
\begin{align*}
I_3\leq\int_{u_{\infty}}^u\frac{a}{|u'|^3}\frac{1}{|u'|}\mathcal{O}\frac{|u'|}{a^{\frac{1}{2}}}
\leq\frac{1}{|u|}.
\end{align*}
Finally, for~$I_4$ we have that
\begin{align*}
I_4\leq\int_{u_{\infty}}^u\frac{a}{|u'|^3}\frac{1}{|u'|}a^{\frac{1}{2}}\bmGamma(\varphi_1)\frac{|u'|}{a^{\frac{1}{2}}}\bmGamma(\varphi_2)
\leq\frac{\bmGamma(\varphi_1)\bmGamma(\varphi_2)}{|u|}.
\end{align*}

Now, combining~$I_1$ and~$I_4$ we see that 
\begin{align*}
\frac{1}{|u|}||(a^{\frac{1}{2}}\meth)^k\pi||_{L^2_{sc}(\mathcal{S}_{u,v})}\lesssim
\frac{\underline{\bm\Psi}[\TiPsi_3]+\bm\Gamma(\tau)+\bmGamma(\varphi_1)\bmGamma(\varphi_2)+1}{|u|} ,
\end{align*}
which leads to
\begin{align*}
||(a^{\frac{1}{2}}\mathcal{D})^k\pi||_{L^2_{sc}(\mathcal{S}_{u,v})}&\lesssim
\underline{\bm\Psi}[\TiPsi_3]+\bm\Gamma(\tau)+\bmGamma(\varphi_1)\bmGamma(\varphi_2)+1 \\
&\lesssim\underline{\bm\Psi}[\TiPsi_3]+\bm\Psi[\TiPsi_1]
+\underline{\bm\varphi}[\Tivarphi_2]\underline{\bm\varphi}[\varphi_1]
+\underline{\bm\varphi}[\Tivarphi_2]+\underline{\bm\varphi}[\varphi_1]+1.
\end{align*}
\end{proof}

\subsection{$L^2(\mathcal{S})$ estimates of the auxiliary field~$\varphi_a$}

Next, we proceed to discuss the construction of estimates for the
components of the auxiliary field.

\begin{proposition}
\label{L2varphi0}
For~$0\leq k\leq10$, one has that 
\begin{align*}
\frac{1}{a^{\frac{1}{2}}}||(a^{\frac{1}{2}}\mathcal{D})^k\varphi_0||_{L^2_{sc}(\mathcal{S}_{u,v})}\lesssim
\underline{\bm{\varphi}}[\varphi_1]+1.
\end{align*}
\end{proposition}

\begin{proof}
We make use of equation~\eqref{EOMMasslessScalarStT-weight1},
\begin{align*}
\mthorn'\varphi_{0}-\meth\bar\varphi_{1}&=(\ulomega-\mu)\varphi_{0}
+\rho\varphi_{2}-\bar\tau\varphi_{1}-\tau\bar\varphi_{1}.
\end{align*}
Commuting~$\meth^k$ with~$\mthorn'$, we find that 
\begin{align*}
\mthorn'\meth^k\varphi_0+(k+1)\mu\meth^k\varphi_0&=-\meth^{k+1}\bar\varphi_1
+\sum_{i=0}^k\meth^i\Gamma(\ulomega,\Timu,\lambda)\meth^{k-i}\varphi_0 
+\sum_{i=0}^k\meth^i\Gamma(\rho,\tau)\meth^{k-i}(\varphi_1,\varphi_2 ). 
\end{align*}
Denoting the right-hand side of the previous equation by~$F$, one has that
\begin{align*}
&s_2(\varphi_0)=0, \ \ s_2(\meth^k\pi)=\frac{k}{2}, \ \ s_2(F)=\frac{k+2}{2}, \\
&\lambda_0=-(k+1), \ \ \lambda_1=-\lambda_0-1=k, \\
& \lambda_1-2s_2(F)=-2, \ \ \lambda_1-2s_2(\meth^k\varphi_0)=0.
\end{align*}
Now, making use of Proposition~\ref{SCweightedgronwall}, one has
\begin{align*}
||(a^{\frac{1}{2}}\meth)^k\varphi_0||_{L^2_{sc}(\mathcal{S}_{u,v})}\lesssim
||(a^{\frac{1}{2}}\meth)^k\varphi_0||_{L^2_{sc}(\mathcal{S}_{u_{\infty},v})}
+\int_{u_{\infty}}^u\frac{a}{|u'|^2}||a^{\frac{k}{2}}F||_{L^2_{sc}(\mathcal{S}_{u',v})},
\end{align*}
so that multiplying by~$a^{-\frac{1}{2}}$ one sees
\begin{align*}
\frac{1}{a^{\frac{1}{2}}}||(a^{\frac{1}{2}}\meth)^k\varphi_0||_{L^2_{sc}(\mathcal{S}_{u,v})}&\lesssim
\frac{1}{a^{\frac{1}{2}}}||(a^{\frac{1}{2}}\meth)^k\varphi_0||_{L^2_{sc}(\mathcal{S}_{u_{\infty},v})}
+\frac{1}{a^{\frac{1}{2}}}\int_{u_{\infty}}^u\frac{a}{|u'|^2}||a^{\frac{k}{2}}F||_{L^2_{sc}(\mathcal{S}_{u',v})} \\
&\lesssim1+\frac{1}{a^{\frac{1}{2}}}\int_{u_{\infty}}^u\frac{a}{|u'|^2}
||a^{\frac{k}{2}}\mathcal{D}^{k+1}\varphi_1||_{L^2_{sc}(\mathcal{S}_{u',v})} \\
&+\sum_{i=0}^{k}\frac{1}{a^{\frac{1}{2}}}\int_{u_{\infty}}^u\frac{a}{|u'|^2}
||a^{\frac{k}{2}}\mathcal{D}^{i}\Gamma(\ulomega,\Timu,\lambda)\mathcal{D}^{k-i}\varphi_0||_{L^2_{sc}(\mathcal{S}_{u',v})} \\
&+\sum_{i=0}^{k}\frac{1}{a^{\frac{1}{2}}}\int_{u_{\infty}}^u\frac{a}{|u'|^2}
||a^{\frac{k}{2}}\mathcal{D}^{i}\Gamma(\rho,\tau)\mathcal{D}^{k-i}(\varphi_1,\varphi_2)||_{L^2_{sc}(\mathcal{S}_{u',v})} \\
&=1+I_1+I_2+I_3.
\end{align*}

For the term~$I_1$ we have 
\begin{align*}
I_1\leq\frac{1}{a^{\frac{1}{2}}}\left(\int_{u_{\infty}}^u\frac{a}{|u'|^2}
||a^{\frac{k}{2}}\mathcal{D}^{k+1}\varphi_1||^2_{L^2_{sc}(\mathcal{S}_{u',v})} \right)^{\frac{1}{2}}
\left( \int_{u_{\infty}}^u\frac{a}{|u'|^2}\right)^{\frac{1}{2}} 
\leq\underline{\bm{\varphi}}[\varphi_1]\frac{a^{\frac{1}{2}}}{|u|^{\frac{1}{2}}}.
\end{align*}
For the term~$I_2$ we have
\begin{align*}
I_2\leq\frac{1}{a^{\frac{1}{2}}}\int_{u_{\infty}}^u\frac{a}{|u'|^2}\frac{1}{|u'|}
(\frac{|u'|}{a}\mathcal{O}+\frac{|u'|}{a^{\frac{1}{2}}})a^{\frac{1}{2}}\mathcal{O} 
=\int_{u_{\infty}}^u(\frac{\mathcal{O}^2}{|u'|^2}+\frac{a^{\frac{1}{2}}\mathcal{O}}{|u'|^2})\leq1.
\end{align*}
And, finally, for the term~$I_3$ it can be verified that 
\begin{align*}
I_3\leq\frac{1}{a^{\frac{1}{2}}}\int_{u_{\infty}}^u\frac{a}{|u'|^2}\frac{1}{|u'|}
(\frac{|u'|}{a^{\frac{1}{2}}}+1)\mathcal{O}^2
=\frac{1}{a^{\frac{1}{2}}}\int_{u_{\infty}}^u
(\frac{a^{\frac{1}{2}}\mathcal{O}^2}{|u'|^2}+\frac{a\mathcal{O}^2}{|u'|^3})\leq1.
\end{align*}

Combining the above estimates we obtain
\begin{align*}
\frac{1}{a^{\frac{1}{2}}}||(a^{\frac{1}{2}}\meth)^k\varphi_0||_{L^2_{sc}(\mathcal{S}_{u,v})}\lesssim
\underline{\bm{\varphi}}[\varphi_1]\frac{a^{\frac{1}{2}}}{|u|^{\frac{1}{2}}}+1 
\leq\bm{\varphi}[\varphi_1]+1.
\end{align*}
A similar analysis of all the other strings of derivatives
in~$\mathcal{D}^{k_i}\varphi_0$ leads to the same result, so that we
can conclude that
\begin{align*}
  \frac{1}{a^{\frac{1}{2}}}||(a^{\frac{1}{2}}\mathcal{D})^k\varphi_0
  ||_{L^2_{sc}(\mathcal{S}_{u,v})}
  \lesssim\underline{\bm{\varphi}}[\varphi_1]+1.
\end{align*}

\end{proof}

The components~$\varphi_1$ and~$\varphi_2$ can be analysed
together. We have the following:

\begin{proposition}
\label{L2varphi12}
For~$0\leq k\leq10$, one has that
\begin{align*}
||(a^{\frac{1}{2}}\mathcal{D})^k\varphi_1||_{L^2_{sc}(\mathcal{S}_{u,v})}&\lesssim\underline{\bm\varphi}[\Tivarphi_2]+1, \\
\frac{a^{\frac{1}{2}}}{|u|}||(a^{\frac{1}{2}}\mathcal{D})^{k}\varphi_2||_{L^2_{sc}(\mathcal{S}_{u,v})}&\lesssim1.
\end{align*}
\end{proposition}

\begin{proof}
  In this proof we make use of the definition
\begin{align*}
\tilde\varphi_2\equiv\meth'\varphi_2+\mu\bar\varphi_1, 
\end{align*}
and the associated bootstrap assumption
\begin{align*}
||(a^{\frac{1}{2}}\mathcal{D})^{k}\Tivarphi_2||_{L^2_{sc}(\mathcal{S}_{u,v})}\leq\mathcal{O},
\end{align*}
for~$0\leq k\leq9$.

Starting from equation~\eqref{EOMMasslessScalarStT-weight3} 
\begin{align*}
\mthorn'\varphi_1+2\mu\varphi_1=\bar\Tivarphi_2-\bar\lambda\bar\varphi_1,
\end{align*}
we commute~$\meth^k$ with~$\mthorn'$ to obtain 
\begin{align*}
\mthorn'\meth^k\varphi_1+(k+2)\mu\meth^k\varphi_1=
\meth^k\bar\Tivarphi_2+\sum_{i=0}^k\meth^{i}\Gamma(\Timu,\lambda)\meth^{k-i}\varphi_1,
\end{align*}
where~$0\leq k\leq10$. Denote by~$F$ the right-hand side of the
previous equation. One has that  
\begin{align*}
&s_2(\varphi_1)=\frac{1}{2}, \ \ s_2(\meth^k\varphi_1)=\frac{k+1}{2}, \ \ s_2(F)=\frac{k+3}{2}, \\
&\lambda_0=-(k+2), \ \ \lambda_1=-\lambda_0-1=k+1, \\
& \lambda_1-2s_2(F)=-2, \ \ \lambda_1-2s_2(\meth^k\varphi_1)=0.
\end{align*}
Now, making use of Proposition~\ref{SCweightedgronwall}, one has that 
\begin{align*}
a^{\frac{1}{2}}||(a^{\frac{1}{2}}\meth)^k\varphi_1||_{L^2_{sc}(\mathcal{S}_{u,v})}\lesssim
a^{\frac{1}{2}}||(a^{\frac{1}{2}}\meth)^k\varphi_1||_{L^2_{sc}(\mathcal{S}_{u_{\infty},v})}
+\int_{u_{\infty}}^u\frac{a^{\frac{3}{2}}}{|u'|^2}||a^{\frac{k}{2}}F||_{L^2_{sc}(\mathcal{S}_{u',v})},
\end{align*}
so that~$a^{-\frac{1}{2}}$ to both hands side and substituting the
expression for~$F$ leads to 
\begin{align*}
||(a^{\frac{1}{2}}\meth)^k\varphi_1||_{L^2_{sc}(\mathcal{S}_{u,v})}&\lesssim
||(a^{\frac{1}{2}}\meth)^k\varphi_1||_{L^2_{sc}(\mathcal{S}_{u_{\infty},v})}
+\int_{u_{\infty}}^u\frac{a}{|u'|^2}||a^{\frac{k}{2}}F||_{L^2_{sc}(\mathcal{S}_{u',v})} \\
&\leq||(a^{\frac{1}{2}}\meth)^k\varphi_1||_{L^2_{sc}(\mathcal{S}_{u_{\infty},v})}
+\int_{u_{\infty}}^u\frac{a}{|u'|^2}||a^{\frac{k}{2}}\mathcal{D}^k\Tivarphi_2||_{L^2_{sc}(\mathcal{S}_{u',v})} \\
&+\sum_{i=0}^k\int_{u_{\infty}}^u\frac{a}{|u'|^2}||a^{\frac{k}{2}}\mathcal{D}^i\Gamma(\Timu,\lambda)
\mathcal{D}^{k-i}\varphi_1||_{L^2_{sc}(\mathcal{S}_{u',v})} \\
&\leq||(a^{\frac{1}{2}}\meth)^k\varphi_1||_{L^2_{sc}(\mathcal{S}_{u_{\infty},v})}
+\int_{u_{\infty}}^u\frac{a}{|u'|^2}||a^{\frac{k}{2}}\mathcal{D}^k\Tivarphi_2||_{L^2_{sc}(\mathcal{S}_{u',v})} \\
&+\int_{u_{\infty}}^u\frac{a}{|u'|^2}\frac{1}{|u'|}\frac{|u'|}{a^{\frac{1}{2}}}||(a^{\frac{1}{2}}\mathcal{D})^i\varphi_1|| \\
&\leq1+\underline{\bm\varphi}[\Tivarphi_2].
\end{align*}
Finally, observing that the rest of the strings
in~$\mathcal{D}^{k_i}\varphi_1$ give rise to the same estimate, one
can conclude that
\begin{align*}
||(a^{\frac{1}{2}}\mathcal{D})^k\varphi_1||_{L^2_{sc}(\mathcal{S}_{u,v})}\lesssim\underline{\bm\varphi}[\Tivarphi_2]+1.
\end{align*}

For the analysis of~$\Tivarphi_2$, we apply~$\meth'$
to~\eqref{EOMMasslessScalarStT-weight4} and make use of the commutator
relations to yield
\begin{align*}
\mthorn\tilde\varphi_2&=\meth'\meth'\varphi_1+\bar\varphi_1\meth\pi+\bar\pi\meth'\bar\varphi_1
+2\pi\meth'\varphi_1-\bar\tau\meth'\varphi_1
-\varphi_0\meth'\mu+\varphi_1\meth'\pi+\bar\varphi_1\meth'\bar\pi+\varphi_2\meth'\rho \\
&+\bar\varphi_1\tilde\varphi_2+2\rho\tilde\varphi_2+\bar\sigma\bar{\tilde{\varphi}}_2
+\pi^2\varphi_1+2\pi\bar\pi\bar\varphi_1+\pi\rho\varphi_2+\lambda\sigma\bar\varphi_1 
-\pi\bar\tau\varphi_1-\bar\pi\bar\tau\bar\varphi_1-\rho\bar\tau\varphi_2.
\end{align*}
Commuting~$\meth'^k$ with~$\mthorn$ we find that 
\begin{align*}
\mthorn\meth'^k\Tivarphi_2&=\meth'^{k+2}\varphi_1
+\sum_{i_1+i_2+i_3+i_4=k}\meth'^{i_1}\Gamma(\tau,\pi)^{i_2}\meth'^{i_3}\Gamma(\tau,\pi)\meth'^{i_4+1}\varphi_1 \\
&+\sum_{i_1+i_2+i_3+i_4=k}\meth'^{i_1}\Gamma(\tau,\pi)^{i_2}\meth'^{i_3}
\bm\varphi\meth'^{i_4+1}\Gamma(\pi,\mu,\rho) \\
&+\sum_{i_1+i_2+i_3+i_4=k}\meth'^{i_1}\Gamma(\tau,\pi)^{i_2}
\meth'^{i_3}(\tau,\pi,\rho,\sigma,\varphi_1)\meth'^{i_4}\Tivarphi_2 \\
&+\sum_{i_1+i_2+i_3+i_4+i_5=k}\meth'^{i_1}\Gamma(\tau,\pi)^{i_2}\meth'^{i_3}\Gamma(\lambda,\pi)
\meth'^{i_4}\Gamma(\sigma,\tau,\pi,\rho)\meth'^{i_5}\bm\varphi,
\end{align*}
where~$0\leq k\leq9$. Making use of the transport estimate in
scale-invariant norm~\eqref{SCldirectiongronwall} it follows that
\begin{align*}
||(a^{\frac{1}{2}}\meth')^{k}\Tivarphi_2||_{L^2_{sc}(\mathcal{S}_{u,v})}&\lesssim
||(a^{\frac{1}{2}}\meth')^{k}\Tivarphi_2||_{L^2_{sc}(\mathcal{S}_{u,0})}+
\int_0^v||a^{\frac{k}{2}}\mthorn\meth'^{k}\Tivarphi_2||_{L^2_{sc}(\mathcal{S}_{u,v'})} \\
&\leq||(a^{\frac{1}{2}}\mathcal{D})^{k}\Tivarphi_2||_{L^2_{sc}(\mathcal{S}_{u,0})}+
\int_0^v||a^{\frac{k}{2}}\mathcal{D}^{k+2}\varphi_1||_{L^2_{sc}(\mathcal{S}_{u,v'})} \\
&+\sum_{i_1+i_2+i_3+i_4=k}\int_0^v||a^{\frac{k}{2}}\mathcal{D}^{i_1}\Gamma(\tau,\pi)^{i_2}
\mathcal{D}^{i_3}\Gamma(\tau,\pi)\mathcal{D}^{i_4+1}\varphi_1||_{L^2_{sc}(\mathcal{S}_{u,v'})} \\
&+\sum_{i_1+i_2+i_3+i_4=k}\int_0^v||a^{\frac{k}{2}}\mathcal{D}^{i_1}\Gamma(\tau,\pi)^{i_2}\mathcal{D}^{i_3}
\bm\varphi\mathcal{D}^{i_4+1}\Gamma(\pi,\mu,\rho)||_{L^2_{sc}(\mathcal{S}_{u,v'})} \\
&+\sum_{i_1+i_2+i_3+i_4=k}\int_0^v||a^{\frac{k}{2}}\mathcal{D}^{i_1}\Gamma(\tau,\pi)^{i_2}\mathcal{D}^{i_3}
(\tau,\pi,\rho,\sigma,\varphi_1)\mathcal{D}^{i_4}\Tivarphi_2||_{L^2_{sc}(\mathcal{S}_{u,v'})} \\
&+\sum_{i_1+i_2+i_3+i_4+i_5=k}\int_0^v||a^{\frac{k}{2}}\mathcal{D}^{i_1}\Gamma(\tau,\pi)^{i_2}
\mathcal{D}^{i_3}\Gamma(\lambda,\pi)\mathcal{D}^{i_4}\Gamma(\sigma,\tau,\pi,\rho)
\mathcal{D}^{i_5}\bm\varphi||_{L^2_{sc}(\mathcal{S}_{u,v'})} \\
&=1+I_1+...+I_5.
\end{align*}

Now, for~$I_1$ we have that
\begin{align*}
I_1\leq||a^{\frac{k}{2}}\mathcal{D}^{k+2}\varphi_1||_{L^2_{sc}(\mathcal{N}_{u}(0,v))}
=\frac{1}{a^{\frac{1}{2}}}||a^{\frac{k+1}{2}}\mathcal{D}^{k+2}\varphi_1||_{L^2_{sc}(\mathcal{N}_{u}(0,v))}
\leq\frac{1}{a^{\frac{1}{2}}}\bm\varphi[\varphi_1]\leq1.
\end{align*}
For~$I_2$ we have that 
\begin{align*}
I_2\leq\frac{1}{a^{\frac{1}{2}}}\frac{a^{\frac{i_2}{2}}}{|u|^{i_2+1}}\mathcal{O}^{i_2}\mathcal{O}^2\leq1.
\end{align*}
For~$I_3$ we have that 
\begin{align*}
I_3\leq\frac{1}{a^{\frac{1}{2}}}\frac{a^{\frac{i_2}{2}}}{|u|^{i_2+1}}\mathcal{O}^{i_2}
\big(a^{\frac{1}{2}}\mathcal{O}\frac{|u|}{a}\mathcal{O}+\mathcal{O}^2+\frac{|u|}{a^{\frac{1}{2}}}\mathcal{O}^2\big)
\leq1.
\end{align*}
For~$I_4$ we have that
\begin{align*}
I_4\leq\frac{a^{\frac{i_2}{2}}}{|u|^{i_2+1}}\mathcal{O}^{i_2}a^{\frac{1}{2}}\mathcal{O}\mathcal{O}\leq1.
\end{align*}
Finally, for~$I_5$ we have that
\begin{align*}
I_5\leq\frac{a^{\frac{i_2}{2}}}{|u|^{i_2+2}}\mathcal{O}^{i_2}
\big(\frac{|u|}{a^{\frac{1}{2}}}a^{\frac{1}{2}}\mathcal{O}\mathcal{O}
+\mathcal{O}^2\frac{|u|}{a^{\frac{1}{2}}}\mathcal{O}\big)\leq1.
\end{align*}

Combining the above estimates for~$I_1,\ldots I_5$ we see that
\begin{align*}
||(a^{\frac{1}{2}}\meth')^{k}\Tivarphi_2||_{L^2_{sc}(\mathcal{S}_{u,v})}\lesssim1.
\end{align*}
Following the same strategy one can estimate the rest of the strings
in~$\mathcal{D}^{k_i}\Tivarphi_2$ to conclude that 
\begin{align*}
||(a^{\frac{1}{2}}\mathcal{D})^{k}\Tivarphi_2||_{L^2_{sc}(\mathcal{S}_{u,v})}\lesssim1.
\end{align*}

From the definition of~$\Tivarphi_2$ we have~$\meth'\varphi_2=\tilde\varphi_2-\mu\bar\varphi_1$, then
\begin{align*}
&\frac{a^{\frac{1}{2}}}{|u|}||(a^{\frac{1}{2}}\mathcal{D})^{k+1}\varphi_2||_{L^2_{sc}(\mathcal{S}_{u,v})}
\leq\frac{a}{|u|}||(a^{\frac{1}{2}}\mathcal{D})^{k}\Tivarphi_2||_{L^2_{sc}(\mathcal{S}_{u,v})}
+\frac{a}{|u|}||(a^{\frac{1}{2}}\mathcal{D})^{k}(\mu\varphi_1)||_{L^2_{sc}(\mathcal{S}_{u,v})} \\
&\quad\quad\leq1+\frac{a}{|u|}\frac{1}{|u|}||\mu||_{L^{\infty}_{sc}(\mathcal{S}_{u,v})}
||(a^{\frac{1}{2}}\mathcal{D})^{k}\varphi_1||_{L^2_{sc}(\mathcal{S}_{u,v})}+
\sum_{i=1}^k\frac{a}{|u|}\frac{1}{|u|}||(a^{\frac{1}{2}}\mathcal{D})^{i}\mu||\times
||(a^{\frac{1}{2}}\mathcal{D})^{k-i}\varphi_1|| \\
&\quad\quad\leq1+\frac{a}{|u|^2}\frac{|u|^2}{a}+\frac{a}{|u|^2}\frac{|u|}{a}\mathcal{O}
\leq1.
\end{align*}
when~$0\leq k\leq9$.

In the case of~$k=0$ we fall back into
equation~\eqref{EOMMasslessScalarStT-weight4} and make use of the
transport estimate in scale-invariant norm
\begin{align*}
||\varphi_2||_{L^2_{sc}(\mathcal{S}_{u,v})}&\lesssim
||\varphi_2||_{L^2_{sc}(\mathcal{S}_{u,0})}+
\int_0^v||\mthorn\varphi_2||_{L^2_{sc}(\mathcal{S}_{u,v'})},
\end{align*}
so that multiplying by$a^{\frac{1}{2}}/|u|$ we can conclude that
\begin{align*}
\frac{a^{\frac{1}{2}}}{|u|}||\varphi_2||_{L^2_{sc}(\mathcal{S}_{u,v})}&\lesssim
\frac{a^{\frac{1}{2}}}{|u|}||\varphi_2||_{L^2_{sc}(\mathcal{S}_{u,0})}+
\frac{a^{\frac{1}{2}}}{|u|}\int_0^v||\mthorn\varphi_2||_{L^2_{sc}(\mathcal{S}_{u,v'})}  \\
&\leq\frac{a^{\frac{1}{2}}}{|u|}||\varphi_2||_{L^2_{sc}(\mathcal{S}_{u,0})}+
\frac{a^{\frac{1}{2}}}{|u|}(||\mathcal{D}\varphi_1||_{L^2_{sc}(\mathcal{S}_{u,v'})}
+||\mu\varphi_0+\rho\varphi_2+\pi\varphi_1||_{L^2_{sc}(\mathcal{S}_{u,v'})}) \\
&\lesssim1+\frac{a^{\frac{1}{2}}}{|u|}(\frac{1}{a^{\frac{1}{2}}}\mathcal{O}+
\frac{1}{|u|}\frac{|u|^2}{a}a^{\frac{1}{2}}\underline{\bm\varphi}[\varphi_1]+
\frac{1}{|u|}\mathcal{O}^2) \\
&\leq1+\underline{\bm\varphi}[\varphi_1]
\leq1.
\end{align*}
In the last step of the previous chain of inequalities we have made
use of
\begin{align*}
\int_{u_{\infty}}^u\frac{a}{|u'|^2}
||\mathcal{D}\varphi_1||_{L^2_{sc}(\mathcal{S}_{u',v})}\leq
\frac{1}{a^{\frac{1}{2}}}\int_{u_{\infty}}^u\frac{a}{|u'|^2}
||a^{\frac{1}{2}}\mathcal{D}\varphi_1||_{L^2_{sc}(\mathcal{S}_{u',v})}
\leq\frac{1}{a^{\frac{1}{2}}}\int_{u_{\infty}}^u\frac{a}{|u'|^2}\mathcal{O}
\leq\frac{1}{a^{\frac{1}{2}}}\frac{a}{|u|}\mathcal{O}\leq1.
\end{align*}

Collecting the above estimates for~$\varphi_2$ we finally conclude
that
\begin{align*}
\frac{a^{\frac{1}{2}}}{|u|}||(a^{\frac{1}{2}}\mathcal{D})^{k}\varphi_2||_{L^2_{sc}(\mathcal{S}_{u,v})}\lesssim1
\end{align*}
for~$0\leq k\leq10$.
\end{proof}

\subsection{$L^2(\mathcal{S})$ estimates of Weyl curvature}

Next in our hierarchy, we construct sphere estimates for the
components of the Weyl tensor.

\begin{proposition}
\label{L2Psi0}
For~$0\leq k\leq9$, one has that 
\begin{align*}
\frac{1}{a^{\frac{1}{2}}}||(a^{\frac{1}{2}}\mathcal{D})^k\Psi_0||_{L^2_{sc}(\mathcal{S}_{u,v})}\lesssim1.
\end{align*}
\end{proposition}

\begin{proof}

We make use of the Bianchi identity~\eqref{T-weightMasslessStBianchi2}
\begin{align*}
\mthorn'\Psi_0+\mu\Psi_0-\meth\TiPsi_1=6\varphi_0\meth\varphi_1+2\ulomega\Psi_0
-5\tau\TiPsi_1+3\sigma\TiPsi_2 
+6\varphi_0\varphi_2\sigma-3\varphi_0^2\bar\lambda-3\varphi_1^2\rho
-6\varphi_1\bar\varphi_1\sigma-12\varphi_0\varphi_1\tau.
\end{align*}
Commuting~$\meth^k$ with~$\mthorn'$, we find that 
\begin{align*}
\mthorn'\meth^k\Psi_0+(k+1)\mu\meth^k\Psi_0&=\meth^{k+1}\TiPsi_1
+\sum_{i=0}^k\meth^i\varphi_0\meth^{k+1-i}\varphi_1
+\sum_{i=0}^k\meth^i\Gamma(\Timu,\lambda,\ulomega)\meth^{k-i}\Psi_0 \\
&+\sum_{i=0}^k\meth^i\Gamma(\tau,\sigma)\meth^{k-i}\TiPsi_{1,2}
+\sum_{i_1+i_2+i_3=k}\meth^{i_1}\Gamma\meth^{i_2}\bm\varphi\meth^{i_3}\bm\varphi .
\end{align*}
Denoting the right-hand side of the previous equality by~$F$, one has
\begin{align*}
&s_2(\Psi_0)=0, \ \ s_2(\meth^k\Psi_0)=\frac{k}{2}, \ \ s_2(F)=\frac{k+2}{2}, \\
&\lambda_0=-(k+1), \ \ \lambda_1=-\lambda_0-1=k, \\
& \lambda_1-2s_2(F)=-2, \ \ \lambda_1-2s_2(\meth^k\Psi_0)=0.
\end{align*}
Now, making use of Proposition~\ref{SCweightedgronwall}, one finds
\begin{align*}
||(a^{\frac{1}{2}}\meth)^k\Psi_0||_{L^2_{sc}(\mathcal{S}_{u,v})}\lesssim
||(a^{\frac{1}{2}}\meth)^k\Psi_0||_{L^2_{sc}(\mathcal{S}_{u_{\infty},v})}
+\int_{u_{\infty}}^u\frac{a}{|u'|^2}||a^{\frac{k}{2}}F||_{L^2_{sc}(\mathcal{S}_{u',v})},
\end{align*}
so that multiplying by~$a^{-\frac{1}{2}}$ one concludes that 
\begin{align*}
\frac{1}{a^{\frac{1}{2}}}||(a^{\frac{1}{2}}\meth)^k\Psi_0||_{L^2_{sc}(\mathcal{S}_{u,v})}&\lesssim
\frac{1}{a^{\frac{1}{2}}}||(a^{\frac{1}{2}}\meth)^k\Psi_0||_{L^2_{sc}(\mathcal{S}_{u_{\infty},v})}
+\int_{u_{\infty}}^u\frac{a^{\frac{1}{2}}}{|u'|^2}||a^{\frac{k}{2}}F||_{L^2_{sc}(\mathcal{S}_{u',v})} \\
&\leq\frac{1}{a^{\frac{1}{2}}}||(a^{\frac{1}{2}}\meth)^k\Psi_0||_{L^2_{sc}(\mathcal{S}_{u_{\infty},v})}
+\int_{u_{\infty}}^u\frac{a^{\frac{1}{2}}}{|u'|^2}||a^{\frac{k}{2}}\mathcal{D}^{k+1}\TiPsi_1||_{L^2_{sc}(\mathcal{S}_{u',v})} \\
&+\sum_{i=0}^k\int_{u_{\infty}}^u\frac{a^{\frac{1}{2}}}{|u'|^2}||a^{\frac{k}{2}}\mathcal{D}^{i}\varphi_0
\mathcal{D}^{k+1-i}\varphi_1||_{L^2_{sc}(\mathcal{S}_{u',v})} \\
&+\sum_{i=0}^k\int_{u_{\infty}}^u\frac{a^{\frac{1}{2}}}{|u'|^2}||a^{\frac{k}{2}}\mathcal{D}^{i}\Gamma(\Timu,\lambda,\ulomega)
\mathcal{D}^{k-i}\Psi_0||_{L^2_{sc}(\mathcal{S}_{u',v})} \\
&+\sum_{i=0}^k\int_{u_{\infty}}^u\frac{a^{\frac{1}{2}}}{|u'|^2}||a^{\frac{k}{2}}\mathcal{D}^{i}\Gamma(\tau,\sigma)
\mathcal{D}^{k-i}\TiPsi_{1,2}||_{L^2_{sc}(\mathcal{S}_{u',v})} \\
&+\sum_{i_1+i_2+i_3=k}\int_{u_{\infty}}^u\frac{a^{\frac{1}{2}}}{|u'|^2}||a^{\frac{k}{2}}\mathcal{D}^{i_1}\Gamma
\mathcal{D}^{i_2}\bm\varphi\mathcal{D}^{i_3}\bm\varphi||_{L^2_{sc}(\mathcal{S}_{u',v})} \\
&\lesssim 1+I_1+...+I_5.
\end{align*}

For~$I_1$ we have that
\begin{align*}
I_1\leq\frac{1}{a^{\frac{1}{2}}}\left(\int_{u_{\infty}}^u\frac{a}{|u'|^2}
||(a^{\frac{1}{2}}\mathcal{D})^{k+1}\TiPsi_1||^2_{L^2_{sc}(\mathcal{S}_{u',v})} \right)^{\frac{1}{2}}
\left(\int_{u_{\infty}}^u\frac{1}{|u'|^2}\right)^{\frac{1}{2}}
\leq\frac{1}{|u|^{\frac{1}{2}}}\underline{\bm\Psi}[\TiPsi_1].
\end{align*}
For~$I_2$ we have that
\begin{align*}
I_2\leq
\int_{u_{\infty}}^u\frac{1}{|u'|^2}\frac{1}{|u'|}a^{\frac{1}{2}}\mathcal{O}\mathcal{O} 
\leq
\int_{u_{\infty}}^u\frac{a^{\frac{1}{2}}}{|u'|^3}\mathcal{O}^2 
\leq
\frac{a^{\frac{1}{2}}}{|u'|^2}\mathcal{O}^2\leq1.
\end{align*}
For~$I_3$ we have that
\begin{align*}
I_3\leq\int_{u_{\infty}}^u\frac{a^{\frac{1}{2}}}{|u'|^2}\frac{1}{|u'|}
(\frac{|u'|}{a}\mathcal{O}+\frac{|u|}{a^{\frac{1}{2}}})a^{\frac{1}{2}}\mathcal{O}
\leq\int_{u_{\infty}}^u\frac{a^{\frac{1}{2}}}{|u'|^2}\mathcal{O}
\leq\frac{a^{\frac{1}{2}}}{|u|}\mathcal{O}\leq1.
\end{align*}
For~$I_4$ we have that
\begin{align*}
I_4\leq\int_{u_{\infty}}^u\frac{a^{\frac{1}{2}}}{|u'|^2}\frac{1}{|u'|}\mathcal{O}^2
=\int_{u_{\infty}}^u\frac{a^{\frac{1}{2}}}{|u'|^3}\mathcal{O}^2
\leq\frac{a^{\frac{1}{2}}}{|u|^2}\mathcal{O}^2\leq1,
\end{align*}
and, finally, for~$I_5$ we have that
\begin{align*}
I_5\leq\int_{u_{\infty}}^u\frac{a^{\frac{1}{2}}}{|u'|^2}\frac{1}{|u'|^2}
a^{\frac{1}{2}}\mathcal{O}a^{\frac{1}{2}}\mathcal{O}\frac{|u'|}{a^{\frac{1}{2}}}
=\int_{u_{\infty}}^u\frac{a}{|u'|^3}\mathcal{O}^2
\leq\frac{a}{|u|^2}\mathcal{O}^2\leq1.
\end{align*}

Combining the above estimates for~$I_1,\dots,I_5$ we conclude that 
\begin{align*}
\frac{1}{a^{\frac{1}{2}}}||(a^{\frac{1}{2}}\meth)^k\Psi_0||_{L^2_{sc}(\mathcal{S}_{u,v})}\lesssim
\frac{1}{|u|^{\frac{1}{2}}}\underline{\bm\Psi}[\TiPsi_1]+1\leq1.
\end{align*}
Observing that for the remaining strings
in~$\mathcal{D}^{k_i}\varphi_2$ one obtains the same result, we can
conclude that
\begin{align*}
\frac{1}{a^{\frac{1}{2}}}||(a^{\frac{1}{2}}\mathcal{D})^k\Psi_0||_{L^2_{sc}(\mathcal{S}_{u,v})}\lesssim1.
\end{align*}

\end{proof}

\begin{proposition}
\label{L2PsiII}
For~$0\leq k\leq9$, one has that 
\begin{align*}
||(a^{\frac{1}{2}}\mathcal{D})^{k}\{\TiPsi_1,\TiPsi_2,\TiPsi_3,\Psi_4\}||_{L^2_{sc}(\mathcal{S}_{u,v})}\lesssim
\bm\Psi[\Psi_0]+\underline{\bm\varphi}[\varphi_1]^2+\underline{\bm\varphi}[\varphi_1]+1.
\end{align*}
\end{proposition}

\begin{proof}
  For the~$L^2$ estimates of~$\TiPsi_1$, $\TiPsi_2$, $\TiPsi_3$
  and~$\Psi_4$, we make use of the Bianchi
  identities~\eqref{T-weightMasslessStBianchi1},
  \eqref{T-weightMasslessStBianchi3},
  \eqref{T-weightMasslessStBianchi5} and
  \eqref{T-weightMasslessStBianchi7}. These equations share the same
  structure ---namely, one has that
\begin{align*}
\mthorn\Psi_{II}=\meth(\Psi_0,\TiPsi)+\bm\varphi\meth\bm\varphi+\Gamma\Psi+\Gamma\bm\varphi\bm\varphi,
\end{align*}
where we have used~$\TiPsi$ denote~$\{\TiPsi_1,\TiPsi_2,\TiPsi_3\}$
and~$\Psi_{II}$ to denote~$\{\TiPsi_1,\TiPsi_2,\TiPsi_3,\Psi_4\}$. We
commute~$\meth^k$ with~$\mthorn'$ to obtain
\begin{align*}
\mthorn\meth^k\Psi_{II}&=\meth^{k+1}(\Psi_0,\TiPsi)
+\sum_{i_1+i_2+i_3=k,i_3<k}\meth^{i_1}\Gamma(\tau,\pi)^{i_2}\meth^{i_3+1}(\Psi_0,\TiPsi)
+\sum_{i_1+i_2+i_3+i_4=k}\meth^{i_1}\Gamma(\tau,\pi)^{i_2}\meth^{i_3}\bm\varphi\meth^{i_4+1}\bm\varphi \\
&+\sum_{i_1+i_2+i_3+i_4=k}\meth^{i_1}\Gamma(\tau,\pi)^{i_2}\meth^{i_3}\Gamma\meth^{i_4}\Psi
+\sum_{i_1+i_2+i_3+i_4=k}\meth^{i_1}\Gamma(\tau,\pi)^{i_2}\meth^{i_3}\Gamma\meth^{i_4}\bm\varphi\meth^{i_5}\bm\varphi .
\end{align*}
We then make use of the scale-invariant norm transport
estimate~\eqref{SCldirectiongronwall} so as to obtain
\begin{align*}
||(a^{\frac{1}{2}}\meth)^{k}\Psi_{II}||_{L^2_{sc}(\mathcal{S}_{u,v})}&\lesssim
||(a^{\frac{1}{2}}\meth)^{k}\Psi_{II}||_{L^2_{sc}(\mathcal{S}_{u,0})}+
\int_0^v||a^{\frac{k}{2}}\mthorn\meth^{k}\Psi_{II}||_{L^2_{sc}(\mathcal{S}_{u,v'})} \\
&\leq||(a^{\frac{1}{2}}\meth)^{k}\Psi_{II}||_{L^2_{sc}(\mathcal{S}_{u,0})}+
\int_0^v||a^{\frac{k}{2}}\mathcal{D}^{k+1}(\Psi_0,\TiPsi)||_{L^2_{sc}(\mathcal{S}_{u,v'})} \\
&+\sum_{i_1+i_2+i_3=k,i_3<k}\int_0^v
||a^{\frac{k}{2}}\mathcal{D}^{i_1}\Gamma^{i_2}\mathcal{D}^{i_3+1}(\Psi_0,\TiPsi)||_{L^2_{sc}(\mathcal{S}_{u,v'})} \\
&+\sum_{i_1+i_2+i_3+i_4=k}\int_0^v
||a^{\frac{k}{2}}\mathcal{D}^{i_1}\Gamma^{i_2}\mathcal{D}^{i_3}\bm\varphi
\mathcal{D}^{i_4+1}\bm\varphi||_{L^2_{sc}(\mathcal{S}_{u,v'})} \\
&+\sum_{i_1+i_2+i_3+i_4=k}\int_0^v
||a^{\frac{k}{2}}\mathcal{D}^{i_1}\Gamma^{i_2}\mathcal{D}^{i_3}\Gamma
\mathcal{D}^{i_4}\Psi||_{L^2_{sc}(\mathcal{S}_{u,v'})} \\
&+\sum_{i_1+i_2+i_3+i_4+i_5=k}\int_0^v
||a^{\frac{k}{2}}\mathcal{D}^{i_1}\Gamma^{i_2}\mathcal{D}^{i_3}\Gamma\mathcal{D}^{i_4}\bm\varphi
\mathcal{D}^{i_5}\bm\varphi||_{L^2_{sc}(\mathcal{S}_{u,v'})} \\
&\lesssim1+I_1+...+I_5.
\end{align*}

For the term~$I_1$ we have that 
\begin{align*}
I_1&\leq\int_0^v||a^{\frac{k}{2}}\mathcal{D}^{k+1}(\Psi_0,\TiPsi)||_{L^2_{sc}(\mathcal{S}_{u,v'})} 
\leq\frac{1}{a^{\frac{1}{2}}}||(a^{\frac{1}{2}}\mathcal{D})^{k+1}(\Psi_0,\TiPsi)||_{L^2_{sc}(\mathcal{N}_{u}(0,v))} \\
&\leq\bm\Psi[\Psi_0]+\frac{1}{a^{\frac{1}{2}}}\bm\Psi[\TiPsi]
\leq\bm\Psi[\Psi_0]+1.
\end{align*}
For~$I_2$ we find that
\begin{align*}
I_2\leq\frac{1}{a^{\frac{1}{2}}}\frac{a^{\frac{i_2}{2}}}{|u|^{i_2}}\mathcal{O}^{i_2}a^{\frac{1}{2}}\mathcal{O}
\leq\frac{a^{\frac{1}{2}}}{|u|}\mathcal{O}^2\leq1.
\end{align*}
For~$I_3$ we have that
\begin{align*}
I_3&\leq\sum_{i_1+i_2+i_3+i_4=k}\int_0^v(
||a^{\frac{k}{2}}\mathcal{D}^{i_1}\Gamma^{i_2}\mathcal{D}^{i_3}\varphi_{0,1}
\mathcal{D}^{i_4+1}\varphi_{0,1}||_{L^2_{sc}(\mathcal{S}_{u,v'})} \\
&+||a^{\frac{k}{2}}\mathcal{D}^{i_1}\Gamma^{i_2}\mathcal{D}^{i_3}\varphi_{2}
\mathcal{D}^{i_4+1}\varphi_{1}||_{L^2_{sc}(\mathcal{S}_{u,v'})} ) \\
&\leq\frac{1}{a^{\frac{1}{2}}}\frac{a^{\frac{i_2}{2}}}{|u|^{i_2+1}}\mathcal{O}^{i_2}
(a^{\frac{1}{2}}\mathcal{O}\mathcal{O}+\mathcal{O}\frac{|u|}{a^{\frac{1}{2}}}\mathcal{O}) \\
&\leq\frac{1}{|u|}\mathcal{O}^2+\frac{1}{a}\mathcal{O}^2
\leq1.
\end{align*}
For the term~$I_4$, the worst term is~$\lambda\Psi_0$. In this case we
have
\begin{align*}
I_4&\leq\sum_{i=0}^k\int_0^v
||a^{\frac{k}{2}}\mathcal{D}^{i}\lambda\mathcal{D}^{k-i}\Psi_0||_{L^2_{sc}(\mathcal{S}_{u,v'})}
+\sum_{i_1+i_2+i_3+i_4=k,i_3+i_4<k}\int_0^v
||a^{\frac{k}{2}}\mathcal{D}^{i_1}\Gamma^{i_2}\mathcal{D}^{i_3}\lambda
\mathcal{D}^{i_4}\Psi_0||_{L^2_{sc}(\mathcal{S}_{u,v'})} \\
&+\sum_{i_1+i_2+i_3+i_4=k}\int_0^v
||a^{\frac{k}{2}}\mathcal{D}^{i_1}\Gamma^{i_2}\mathcal{D}^{i_3}\Gamma(\lambda,...)
\mathcal{D}^{i_4}\Psi_{II}||_{L^2_{sc}(\mathcal{S}_{u,v'})} \\
&\leq\frac{1}{|u|}\frac{|u|}{a^{\frac{1}{2}}}a^{\frac{1}{2}}
+\frac{a^{\frac{i_2}{2}}}{|u|^{i_2+1}}\mathcal{O}^{i_2}\frac{|u|}{a^{\frac{1}{2}}}a^{\frac{1}{2}}
+\frac{a^{\frac{i_2}{2}}}{|u|^{i_2+1}}\mathcal{O}^{i_2}\frac{|u|}{a^{\frac{1}{2}}}\mathcal{O}
\leq1+\frac{a^{\frac{1}{2}}}{|u|}\mathcal{O}+\frac{\mathcal{O}}{a^{\frac{1}{2}}}
\leq1.
\end{align*}
Finally, for~$I_5$, we denote integral contain
terms~$\mu\varphi_{0,1}\varphi_{0,1}$ by~$I_{51}$. One has that
\begin{align*}
I_{51}&\leq\sum_{i=0}^k\int_0^v
||a^{\frac{k}{2}}\mu\mathcal{D}^{i}\varphi_{0,1}\mathcal{D}^{k-i}\varphi_{0,1}||_{L^2_{sc}(\mathcal{S}_{u,v'})} \\
&+\sum_{i_1+i_2+i_3+i_4+i_5=k, i_3+i_4+i_5<k}\int_0^v
||a^{\frac{k}{2}}\mathcal{D}^{i_1}\Gamma^{i_2}\mathcal{D}^{i_3}\mu
\mathcal{D}^{i_4}\varphi_{0,1}\mathcal{D}^{i_5}\varphi_{0,1}||_{L^2_{sc}(\mathcal{S}_{u,v'})} \\
&\leq\frac{1}{|u|^2}\frac{|u|^2}{a}a^{\frac{1}{2}}\Gamma(\varphi_{0,1})a^{\frac{1}{2}}\Gamma(\varphi_{0,1})
+\frac{a^{\frac{i_2}{2}}}{|u|^{i_2+2}}\mathcal{O}^{i_2}\frac{|u|^2}{a}a^{\frac{1}{2}}\mathcal{O}a^{\frac{1}{2}}\mathcal{O}\\
&\leq\bmGamma(\varphi_{0})^2+\frac{1}{a^{\frac{1}{2}}}\bmGamma(\varphi_{0})\bmGamma(\varphi_{1})
+\frac{1}{a}\bmGamma(\varphi_{1})^2+\frac{a^{\frac{1}{2}}}{|u|}\mathcal{O}^3 \\
&\leq\underline{\bm\varphi}[\varphi_1]^2
+\underline{\bm\varphi}[\varphi_1]+1.
\end{align*}
Similarly, denote the integral contain terms with~$\varphi_2$,
namely~$\sigma\varphi_2^2$, $\lambda\varphi_0\varphi_2$ and
$\pi\varphi_1\varphi_2$ by~$I_{52}$. It follows then that
\begin{align*}
I_{52}&\leq\frac{a^{\frac{i_2}{2}}}{|u|^{i_2+2}}\mathcal{O}^{i_2}\left(
a^{\frac{1}{2}}\mathcal{O}\frac{|u|}{a^{\frac{1}{2}}}\mathcal{O}\frac{|u|}{a^{\frac{1}{2}}}\mathcal{O}
+\frac{|u|}{a^{\frac{1}{2}}}a^{\frac{1}{2}}\mathcal{O}\frac{|u|}{a^{\frac{1}{2}}}\mathcal{O}
+\mathcal{O}\mathcal{O}\frac{|u|}{a^{\frac{1}{2}}}\mathcal{O}
\right) \\
&=\frac{1}{a^{\frac{1}{2}}}\frac{a^{\frac{i_2}{2}}}{|u|^{i_2}}\mathcal{O}^{i_2+3}
+\frac{1}{a^{\frac{1}{2}}}\frac{a^{\frac{i_2}{2}}}{|u|^{i_2}}\mathcal{O}^{i_2+2}
+\frac{1}{a^{\frac{1}{2}}}\frac{a^{\frac{i_2}{2}}}{|u|^{i_2+1}}\mathcal{O}^{i_2+3}
\leq\frac{\mathcal{O}^3}{a^{\frac{1}{2}}}
\leq1.
\end{align*}
Denoting the remaining terms by~$I_{53}$, one has
\begin{align*}
I_{53}&\leq\sum_{i_1+i_2+i_3+i_4+i_5=k}\int_0^v
||a^{\frac{k}{2}}\mathcal{D}^{i_1}\Gamma^{i_2}\mathcal{D}^{i_3}\Gamma(\lambda,\sigma,\tau,\pi,\rho)
\mathcal{D}^{i_4}\varphi_{0,1}\mathcal{D}^{i_5}\varphi_{0,1}||_{L^2_{sc}(\mathcal{S}_{u,v'})} \\
&\leq\frac{a^{\frac{i_2}{2}}}{|u|^{i_2+2}}\mathcal{O}^{i_2}\left(
\frac{|u|}{a^{\frac{1}{2}}}+a^{\frac{1}{2}}\mathcal{O}+\mathcal{O}
\right)a^{\frac{1}{2}}\mathcal{O}a^{\frac{1}{2}}\mathcal{O}
\leq\frac{1}{a^{\frac{1}{2}}}\frac{a^{\frac{i_2}{2}}}{|u|^{i_2+1}}\mathcal{O}^{i_2+2}
+\frac{a^{\frac{i_2+3}{2}}}{|u|^{i_2+2}}\mathcal{O}^{i_2+3} \\
&\leq\frac{\mathcal{O}^2}{a^{\frac{1}{2}}|u|}
+\frac{a^{\frac{3}{2}}}{|u|^2}\mathcal{O}^3
\leq1.
\end{align*}

To conclude, combining the above estimates for~$I_1,\ldots,I_5$ one
has that
\begin{align*}
||(a^{\frac{1}{2}}\meth)^{k}\Psi_{II}||_{L^2_{sc}(\mathcal{S}_{u,v})}\lesssim
\bm\Psi[\Psi_0]+\underline{\bm\varphi}[\varphi_1]^2+\underline{\bm\varphi}[\varphi_1]+1.
\end{align*}
An inspection then shows that the other strings of derivatives
in~$\mathcal{D}^{k_i}\Psi_{II}$ give rise to the same
estimate. Accordingly, we find that
\begin{align*}
||(a^{\frac{1}{2}}\mathcal{D})^{k}\Psi_{II}||_{L^2_{sc}(\mathcal{S}_{u,v})}\lesssim
\bm\Psi[\Psi_0]+\underline{\bm\varphi}[\varphi_1]^2+\underline{\bm\varphi}[\varphi_1]+1.
\end{align*}

\end{proof}

\subsection{Summary and discussion of~$L^2(\mathcal{S})$ estimate}

In this subsection we summarise and put together
the~$L^2(\mathcal{S})$ estimates we have obtained so far.

We have showed that under the bootstrap assumption
\begin{align*}
\bmGamma,\bm\varphi,\bm\Psi\leq\mathcal{O},
\end{align*}
we obtain the estimates
\begin{itemize}
\item[(i)] for~$0\leq k\leq10$, 
\begin{align*}
\frac{a^{\frac{1}{2}}}{|u|}||(a^{\frac{1}{2}}\mathcal{D})^k\lambda||_{L^2_{sc}(\mathcal{S}_{u,v})}&\lesssim1,\\
\frac{1}{a^{\frac{1}{2}}}||(a^{\frac{1}{2}}\mathcal{D})^{k}\sigma||_{L^2_{sc}(\mathcal{S}_{u,v})}&\lesssim\bm\Psi[\Psi_0]+1, \\
||(a^{\frac{1}{2}}\mathcal{D})^{k}\ulomega||_{L^2_{sc}(\mathcal{S}_{u,v})}&\lesssim
\bm\Psi[\TiPsi_2]+\underline{\bm\varphi}[\varphi_1]+1, \\
||(a^{\frac{1}{2}}\mathcal{D})^k\ulchi||_{L^2_{sc}(\mathcal{S}_{u,v})}&\lesssim
\underline{\bm\Psi}[\TiPsi_2]+\underline{\bm\varphi}[\varphi_1]+1,\\
||(a^{\frac{1}{2}}\mathcal{D})^{k}\rho||_{L^2_{sc}(\mathcal{S}_{u,v})}&\lesssim
(\bm\Psi[\Psi_0]+\underline{\bm\varphi}[\varphi_1]+1)^2, \\
||(a^{\frac{1}{2}}\mathcal{D})^{k}\tau||_{L^2_{sc}(\mathcal{S}_{u,v})}&\lesssim
\bm\Psi[\TiPsi_1]+\underline{\bm\varphi}[\Tivarphi_2]\underline{\bm\varphi}[\varphi_1]
+\underline{\bm\varphi}[\Tivarphi_2]+\underline{\bm\varphi}[\varphi_1]+1, \\
\frac{a}{|u|}||(a^{\frac{1}{2}}\mathcal{D})^k\Timu||_{L^2_{sc}(\mathcal{S}_{u,v})}&\lesssim
\bm\Psi[\TiPsi_2]+\underline{\bm\varphi}[\varphi_1]+1, \\
\frac{a}{|u|^2}||(a^{\frac{1}{2}}\mathcal{D})^k\mu||_{L^2_{sc}(\mathcal{S}_{u,v})}&\lesssim1, \\
||(a^{\frac{1}{2}}\mathcal{D})^k\pi||_{L^2_{sc}(\mathcal{S}_{u,v})}&\lesssim
\underline{\bm\Psi}[\TiPsi_3]+\bm\Psi[\TiPsi_1]
+\underline{\bm\varphi}[\Tivarphi_2]\underline{\bm\varphi}[\varphi_1]
+\underline{\bm\varphi}[\Tivarphi_2]+\underline{\bm\varphi}[\varphi_1]+1.
\end{align*}

\item[(ii)] for~$0\leq k\leq10$, 
\begin{align*}
\frac{1}{a^{\frac{1}{2}}}||(a^{\frac{1}{2}}\mathcal{D})^k\varphi_0||_{L^2_{sc}(\mathcal{S}_{u,v})}&\lesssim
\underline{\bm{\varphi}}[\varphi_1]+1, \quad
||(a^{\frac{1}{2}}\mathcal{D})^k\varphi_1||_{L^2_{sc}(\mathcal{S}_{u,v})}
\lesssim\underline{\bm\varphi}[\Tivarphi_2]+1, \\
\frac{a^{\frac{1}{2}}}{|u|}||(a^{\frac{1}{2}}\mathcal{D})^{k}\varphi_2||_{L^2_{sc}(\mathcal{S}_{u,v})}
&\lesssim1.
\end{align*}

\item[(iii)] for~$0\leq k\leq9$, 
\begin{align*}
||(a^{\frac{1}{2}}\mathcal{D})^{k}\Tivarphi_2||_{L^2_{sc}(\mathcal{S}_{u,v})}&\lesssim1,\\
  \frac{1}{a^{\frac{1}{2}}}||(a^{\frac{1}{2}}\mathcal{D})^k\Psi_0||_{L^2_{sc}(\mathcal{S}_{u,v})}&\lesssim1, \\
||(a^{\frac{1}{2}}\mathcal{D})^{k}\Psi_{II}||_{L^2_{sc}(\mathcal{S}_{u,v})}&\lesssim
\bm\Psi[\Psi_0]+\underline{\bm\varphi}[\varphi_1]^2+\underline{\bm\varphi}[\varphi_1]+1.
\end{align*}
where~$\Psi_{II}=\{\TiPsi_1,\TiPsi_2,\TiPsi_3,\Psi_4\}$.
\end{itemize}

Actually, from the bootstrap assumption~$\bmGamma\leq\mathcal{O}$, one
can obtain additional estimates for all derivatives of~$\bm\varphi$
and~$\bm\Psi$ except those at the highest level. We now give estimates
which will be used to control the Gaussian curvature.

When~$0\leq k\leq10$, from Propositions~\ref{L2varphi0}
and~\ref{L2varphi12} we have that
\begin{align*}
\bm\varphi[\varphi_0]_k&\equiv \frac{1}{a^{\frac{1}{2}}}\left(\int_0^v
||a^{\frac{k-1}{2}}\mathcal{D}^k\varphi_0||^2_{L^2_{sc}(\mathcal{S}_{u,v})}
\right)^{\frac{1}{2}}
=\frac{1}{a}\left(\int_0^v
||a^{\frac{k}{2}}\mathcal{D}^k\varphi_0||^2_{L^2_{sc}(\mathcal{S}_{u,v})}
\right)^{\frac{1}{2}} \\
&\leq\frac{1}{a}a^{\frac{1}{2}}\mathcal{O}
=\frac{\mathcal{O}}{a^{\frac{1}{2}}}\leq1,
\end{align*}
and furthermore that
\begin{align*}
\underline{\bm\varphi}[\varphi_1]_k&\equiv\frac{1}{a^{\frac{1}{2}}}\left(\int_{u_{\infty}}^u\frac{a}{|u'|^2}
||a^{\frac{k-1}{2}}\mathcal{D}^k\varphi_1||^2_{L^2_{sc}(\mathcal{S}_{u,v})}
\right)^{\frac{1}{2}}
=\frac{1}{a}\left(\int_{u_{\infty}}^u\frac{a}{|u'|^2}
||a^{\frac{k}{2}}\mathcal{D}^k\varphi_1||^2_{L^2_{sc}(\mathcal{S}_{u,v})}
\right)^{\frac{1}{2}} \\
&\leq\frac{1}{a}a^{\frac{1}{2}}\mathcal{O}
=\frac{\mathcal{O}}{a^{\frac{1}{2}}}\leq1.
\end{align*}

When~$0\leq k\leq9$, from Proposition~\ref{L2Psi0} we have that
\begin{align*}
\bm\Psi[\Psi_0]_k=\frac{1}{a^{\frac{1}{2}}}\left(\int_0^v
||(a^{\frac{1}{2}}\mathcal{D})^k\Psi_0||^2_{L^2_{sc}(\mathcal{S}_{u,v})}
\right)^{\frac{1}{2}}
\leq\frac{1}{a^{\frac{1}{2}}}a^{\frac{1}{2}}=1.
\end{align*}

With the aid of the above three results, one can make use of
the~$L^2(\mathcal{S})$ estimates from the previous subsection, so that
when~$0\leq k\leq9$ one has that
\begin{align*}
\bm\Gamma(\sigma,\varphi_0,\varphi_1)_k\lesssim1.
\end{align*}
Moreover, making use of the
estimate~\eqref{T-weightMasslessStBianchi3} in
Proposition~\ref{L2PsiII} to estimate~$\Psi_2$, we can conclude with
the inequality
\begin{align*}
\bm\Gamma(\TiPsi_2)_k\leq1,
\end{align*}
which leads to
\begin{align*}
\bm\Psi[\TiPsi_2]_k\leq1, \quad \bm\Gamma(\Timu)_k\leq1.
\end{align*}

From Lemma~\ref{L2rhoAlt} we have that for~$1\leq k\leq9$
\begin{align*}
\bm\Gamma(\rho)_k\lesssim\frac{a}{|u|}
\end{align*}
and 
\begin{align*}
\bm\Gamma(\rho)_0\lesssim1.
\end{align*}

Collecting the results above we have following corollary:

\begin{corollary}
\label{PreforGaussCurv}
Under bootstrap assumption
$\bmGamma,\bm\varphi,\bm\Psi\leq\mathcal{O}$, for~$0\leq k\leq9$, we
have that 
\begin{align*}
  \bm\Gamma(\sigma,\Timu,\varphi_0,\varphi_1,\TiPsi_2)_k,
  \bm\Psi[\Psi_0]_k, \bm\Psi[\TiPsi_2]_k
  \lesssim1.
\end{align*}
When~$0\leq k\leq10$, we have that
\begin{align*}
\bm\varphi[\varphi_0]_k,\,
\underline{\bm\varphi}[\varphi_1]_k\leq1,
\end{align*}
and for~$\rho$ we have that 
\begin{align*}
\bm\Gamma(\rho)_{1\leq k\leq9}\lesssim\frac{a}{|u|}, \quad
\bm\Gamma(\rho)_0\lesssim1.
\end{align*}
\end{corollary}

\section{Elliptic estimates}
\label{EllipticEstimate}

In this section we estimate the highest (11th order) derivative of
connection coefficients.  These results are needed in the energy
estimates of the auxiliary fields~$\varphi_{0,1,2}$. The reason one
cannot directly use the transport estimate is due to the presence of
angular derivatives and Weyl curvature terms leading to even higher
derivatives of said quantities. To overcome this problem, based on the
analysis on topological 2-spheres~$\mathcal{S}$, as introduced in
Chapter 2 of~\cite{ChrKla93}, which provided an estimate of up to nine
derivatives of the Gaussian curvature on~$\mathcal{S}$, we can obtain
an estimate for the highest derivative of the connection coefficients
from the information of the divergence and curl of its' lower
derivatives. The required Gaussian curvature estimates can be obtained
from the Codazzi equation. Further discussion on this strategy can be
found in~\cite{Luk12}, \cite{An202209}.

In this article we can simplify the discussion by making use of the
T-weight formalism and relying on~\ref{T-weightElliptic} ---which
shows that the divergence has the same information as the
curl. Additional simplifications are obtained by constructing some new
quantities via divergences and the Weyl curvature so that the
transport equation of these new quantities no longer involve the Weyl
tensor.

The connection coefficients can be divided into three categories
depending on the specific structure of their equations:

\begin{itemize}
\item[(i)] The first category contains~$(\rho,\sigma)$,
  $(\mu,\lambda)$.  For these coefficients we make use one transport
  equation for~$\rho$ ($\mu$) which does not contain spherical
  derivative and the Weyl tensor to obtain estimates.  The leading
  term in this estimate is~$\sigma$ ($\lambda$) which be controlled by
  its divergence via the Codazzi equation and
  Proposition~\ref{EllipticT-weightSC}. Details are given in
  Propositions~\ref{11Derrhosigma} and~\ref{11Dermulambda}.

\item[(ii)] The second category contains~$\tau$ and~$\pi$. Their
  transport equations do not contain angular derivative but have
  components of the Weyl tensor. To avoid the eleventh derivative of
  the Weyl, we construct new quantities with the divergence and Weyl
  tensor for which we estimate up to the tenth derivative.
  Proposition~\ref{EllipticT-weightSC} then allows us to control the
  top derivatives of~$\tau$ and~$\pi$. This discussion is presented in
  Propositions~\ref{11Dertau} and~\ref{11Derpi}.

\item[(iii)] The third category consists of~$\ulomega$ which is a pure
  scalar with zero T-weight.  This property leads to a relatively
  complex construction. Via an auxiliary field, we construct the new
  quantity~$\tilde\ulomega$ to avoid the analysis of the eleventh
  derivative of the Weyl tensor. More details are given in
  Proposition~\ref{11Derulomega}.

\end{itemize}

\subsection{Preliminary elliptic estimates}

Making use of the~$L^2(\mathcal{S})$ estimates, we have the estimate:

\begin{lemma}
\begin{align*}
\sum_{i=0}^7||(a^{\frac{1}{2}}\mathcal{D})^iK||_{L_{sc}^{\infty}(\mathcal{S}_{u,v})}
+\sum_{i=0}^9||(a^{\frac{1}{2}}\mathcal{D})^iK||_{L_{sc}^{2}(\mathcal{S}_{u,v})}\lesssim1.
\end{align*}
\end{lemma}

\begin{proof}
  The Gauss curvature satisfies the formula
\begin{align*}
K=2\mu\rho-\lambda\sigma-\bar\lambda\bar\sigma
+6\varphi_{1}\bar\varphi_{1}-\TiPsi_2-\bar\TiPsi_2.
\end{align*}
For~$1\leq k\leq9$ we have that 
\begin{align*}
||(a^{\frac{1}{2}}\mathcal{D})^kK||_{L_{sc}^{2}(\mathcal{S}_{u,v})}&\leq
||(a^{\frac{1}{2}}\mathcal{D})^k\Psi_2||_{L_{sc}^{2}(\mathcal{S}_{u,v})}
+\sum_{i=0}^k||(a^{\frac{1}{2}}\mathcal{D})^i\rho(a^{\frac{1}{2}}\mathcal{D})^{k-i}\mu||_{L_{sc}^{2}(\mathcal{S}_{u,v})} \\
&+\sum_{i=0}^k||(a^{\frac{1}{2}}\mathcal{D})^i\lambda(a^{\frac{1}{2}}\mathcal{D})^{k-i}\sigma||_{L_{sc}^{2}(\mathcal{S}_{u,v})}
+\sum_{i=0}^k||(a^{\frac{1}{2}}\mathcal{D})^{i}\varphi_{1}(a^{\frac{1}{2}}\mathcal{D})^{k-i}\varphi_{1}||_{L_{sc}^{2}(\mathcal{S}_{u,v})} \\
&\leq||(a^{\frac{1}{2}}\mathcal{D})^k\Psi_2||_{L_{sc}^{2}(\mathcal{S}_{u,v})}
+\sum_{i=0}^k\frac{1}{|u|}||(a^{\frac{1}{2}}\mathcal{D})^i\rho||\times||(a^{\frac{1}{2}}\mathcal{D})^{k-i}\mu||\\
&+\sum_{i=0}^k\frac{1}{|u|}||(a^{\frac{1}{2}}\mathcal{D})^i\lambda||\times||(a^{\frac{1}{2}}\mathcal{D})^{k-i}\sigma||
+\sum_{i=0}^k\frac{1}{|u|}||(a^{\frac{1}{2}}\mathcal{D})^{i}\varphi_{1}||\times||(a^{\frac{1}{2}}\mathcal{D})^{k-i}\varphi_{1}||.
\end{align*}
Then, applying the results in Corollary~\ref{PreforGaussCurv} we find
that
\begin{align*}
||(a^{\frac{1}{2}}\mathcal{D})^kK||_{L_{sc}^{2}(\mathcal{S}_{u,v})}\lesssim1.
\end{align*}

When~$k=0$, we have the troublesome term~$\rho\mu$. To estimate~$\mu$
we make use of the definition of~$\Timu$ to yield 
\begin{align*}
||\mu||_{L_{sc}^{2}(\mathcal{S}_{u,v})}\leq||\Timu||_{L_{sc}^{2}(\mathcal{S}_{u,v})}
+||\frac{1}{|u|}||_{L_{sc}^{2}(\mathcal{S}_{u,v})} 
\leq\frac{|u|}{a}+1.
\end{align*}
Hence, we have that 
\begin{align*}
||K||_{L_{sc}^{2}(\mathcal{S}_{u,v})}\leq\frac{1}{|u|}1\times\frac{|u|}{a}+1\lesssim1.
\end{align*}

Then we have when~$0\leq k\leq9$ that 
\begin{align*}
||(a^{\frac{1}{2}}\mathcal{D})^kK||_{L_{sc}^{2}(\mathcal{S}_{u,v})}\lesssim1.
\end{align*}
Finally, making use of the Sobolev embedding theorem we obtain the
desired result.
\end{proof}

\begin{remark}
{\em From the last lemma it follows that
\begin{align*}
&\sum_{i=0}^7|u|^2||(|u|^i\mathcal{D}^iK)||_{L^{\infty}(\mathcal{S}_{u,v})}\lesssim\frac{a}{|u|}\leq1, \\
&\sum_{i=0}^9|u|||(|u|^i\mathcal{D}^iK)||_{L^{2}(\mathcal{S}_{u,v})}\lesssim\frac{a}{|u|}\leq1.
\end{align*}
}
\end{remark}

With this result we have, in turn, that
\begin{align*}
|||u|^k\mathcal{D}^kf||^2_{L^2(\mathcal{S}_{u,v})}\lesssim\sum_{i\leq k-1}
\left(|||u|^{i+1}\mathcal{D}^j\mathscr{D}_f||^2_{L^2(\mathcal{S}_{u,v})}
+|||u|^j\mathcal{D}^if||^2_{L^2(\mathcal{S}_{u,v})}
\right) ,
\end{align*}
where~$\mathscr{D}_f$ is the divergence related scalar defined
in~\eqref{T-weightDiv}. In terms of scale-invariant norms we have:

\begin{proposition}
\label{EllipticT-weightSC}
Let~$f$ denote a quantity in Table~\ref{QuantityT-weight} with nonzero
T-weight. Then under the bootstrap assumption, for~$0\leq k\leq 11$,
one has that
\begin{align*}
||(a^{\frac{1}{2}}\mathcal{D})^kf||_{L^2_{sc}(\mathcal{S}_{u,v})}\lesssim\sum_{i\leq k-1}
\left(||(a^{\frac{1}{2}})^{i+1}\mathcal{D}^i\mathscr{D}_f||_{L_{sc}^2(\mathcal{S}_{u,v})}
+||(a^{\frac{1}{2}}\mathcal{D})^if||_{L_{sc}^2(\mathcal{S}_{u,v})}
\right).
\end{align*}
\end{proposition}
The proof can be found in Sec 6.1 in Paper~\cite{An202209}. When~$f$
is a pure scalar, we have the following:

\begin{proposition}
  \label{EllipticPureScalar}
  Let~$f$ denote a quantity in Table~\ref{QuantityT-weight} with zero
  T-weight. Then, under the bootstrap assumption,
  for~$0\leq k\leq 10$, one has that
\begin{align*}
||(a^{\frac{1}{2}}\mathcal{D})^{k+1}f||_{L^2_{sc}(\mathcal{S})}\lesssim
||(a^{\frac{1}{2}})^{k+1}\mathcal{D}^{k-1}(\bmDelta f)||_{L^2_{sc}(\mathcal{S})}
+\sum_{i=0}^k||(a^{\frac{1}{2}}\mathcal{D})^if||_{L^2_{sc}(\mathcal{S})},
\end{align*}
where~$\bmDelta f\equiv 2\meth\meth'f$.
\end{proposition}

\subsection{Elliptic estimates for the highest (11th) derivative of
  the connection}

The norm of the eleventh derivative quantities of the connection
coefficients is defined as
\begin{align*}
\bm{\Gamma}_{11}(u,v)&\equiv
||a^5\mathcal{D}^{11}\sigma||_{L^2_{sc}(\mathcal{N}_u(0,v))}
+||a^5\mathcal{D}^{11}\rho||_{L^2_{sc}(\mathcal{N}_u(0,v))} \\
&+||a^5\mathcal{D}^{11}\tau||_{L^2_{sc}(\mathcal{N}'_v(u_{\infty},u))}
+||a^5\mathcal{D}^{11}\tau||_{L^2_{sc}(\mathcal{N}_u(0,v))}\\
&+\frac{a}{|u|}||a^5\mathcal{D}^{11}\pi||_{L^2_{sc}(\mathcal{N}_u(0,v))}
+||a^5\mathcal{D}^{11}\ulomega||_{L^2_{sc}(\mathcal{N}'_v(u_{\infty},u))}\\
&+\int_{u_{\infty}}^u\frac{a^{\frac{1}{2}}}{|u|^3}||a^5\mathcal{D}^{11}\lambda||_{L_{sc}^2(\mathcal{S}_{u,v})}
+\int_{u_{\infty}}^u\frac{a^2}{|u|^3}||a^5\mathcal{D}^{11}\mu||_{L_{sc}^2(\mathcal{S}_{u,v})}.
\end{align*}
In the following estimates, we make use of the~$L^2(\mathcal{S})$
estimates and the bootstrap assumptions
\begin{align*}
 \bm{\varphi}, \bm{\Psi}, \bm{\Gamma}_{11}\leq\mathcal{O}.
\end{align*}
We will show that auxiliary bootstrap assumption on~$\bm{\Gamma}_{11}$
can be improved to
\begin{align*}
\bm{\Gamma}_{11}\lesssim\bm\Psi^2+\bm\Psi+\bm\varphi^2+\bm\varphi+1.
\end{align*}

\begin{proposition}
  \label{11Derrhosigma}
  One has that 
\begin{align*}
||a^5\mathcal{D}^{11}\rho||_{L^2_{sc}(\mathcal{S}_{u,v})}&\lesssim
\bm{\varphi}[\varphi_0]
+\bm\Psi[\Psi_0]+\bm\Psi[\TiPsi_1]+1,\\
||a^5\mathcal{D}^{11}\rho||_{L^2_{sc}(\mathcal{N}_u(0,v))}&\lesssim
\bm{\varphi}[\varphi_0]
+\bm\Psi[\Psi_0]+\bm\Psi[\TiPsi_1]+1
\end{align*}
and
\begin{align*}
||a^5\mathcal{D}^{11}\sigma||_{L^2_{sc}(\mathcal{N}_u(0,v))}
\lesssim\bm{\varphi}[\varphi_0]
+\bm\Psi[\Psi_0]+\bm\Psi[\TiPsi_1]+1.
\end{align*}
\end{proposition}

\begin{proof}
\smallskip
\noindent
\textbf{Step 1}: Estimate of~$\rho$.

We make use of structure equation for~$\rho$ ---namely, equation
\eqref{T-weightMasslessStructureStEq1}:
\begin{align*}
\mthorn\rho&=\rho^2+\sigma\bar\sigma+3\varphi_{0}^2.
\end{align*}
Commuting~$\meth^{11}$ with~$\mthorn$, one has that
\begin{align*}
\mthorn\meth^{11}\rho&=\Gamma(\rho,\sigma)\meth^{11}\rho+\sigma\meth^{11}\sigma+\varphi_0\meth^{11}\varphi_0 \\
&+\sum_{i_1+i_2+i_3+i_4=11,i_3+i_4<11}\meth^{i_1}\Gamma(\tau,\pi)^{i_2}
\meth^{i_3}[\Gamma(\tau,\pi,\rho,\sigma),\varphi_0]
\meth^{i_4}[\Gamma(\rho,\sigma),\varphi_0].
\end{align*}
Then, making use of the transport estimate in scale-invariant norm,
one has that
\begin{align*}
||a^5\meth^{11}\rho||_{L^2_{sc}(\mathcal{S}_{u,v})}&\lesssim
||a^5\meth^{11}\rho||_{L^2_{sc}(\mathcal{S}_{u,0})}+
\int_0^v||a^5\mthorn\meth^{11}\rho||_{L^2_{sc}(\mathcal{S}_{u,v'})}\\
&\leq\int_0^v||a^5\rho\meth^{11}\rho||_{L^2_{sc}(\mathcal{S}_{u,v'})} 
+\int_0^v||a^5\sigma\meth^{11}\sigma||_{L^2_{sc}(\mathcal{S}_{u,v'})} 
+\int_0^v||a^5\varphi_0\meth^{11}\varphi_0||_{L^2_{sc}(\mathcal{S}_{u,v'})} \\
&\sum_{i_1+i_2+i_3+i_4=11,i_3+i_4<11}\int_0^v||a^5\meth^{i_1}\Gamma(\tau,\pi)^{i_2}
\meth^{i_3}[\Gamma,\varphi_0]
\meth^{i_4}[\Gamma,\varphi_0]||_{L^2_{sc}(\mathcal{S}_{u,v'})} \\
&=I_1+...+I_4.
\end{align*}

For~$I_1$ we have that
\begin{align*}
I_1\leq\int_0^v\frac{1}{|u|}||\rho||_{L^{\infty}_{sc}(\mathcal{S}_{u,v'})}
||a^5\mathcal{D}^{11}\rho||_{L^2_{sc}(\mathcal{S}_{u,v})}
\leq\frac{\bmGamma(\rho)_{0,\infty}+a^{\frac{1}{2}}\bmGamma(\sigma)_{0,\infty}}{|u|}\int_0^v
||a^5\mathcal{D}^{11}\rho||_{L^2_{sc}(\mathcal{S}_{u,v})}.
\end{align*}
This term can be absorbed by the Gr\"onwall inequality.

For~$I_2$ we have that 
\begin{align*}
I_2&\leq\int_0^v\frac{1}{|u|}||\sigma||_{L^{\infty}_{sc}(\mathcal{S}_{u,v'})}
||a^5\mathcal{D}^{11}\sigma||_{L^2_{sc}(\mathcal{S}_{u,v})}
\leq\frac{a^{\frac{1}{2}}\bmGamma(\sigma)_{0,\infty}}{|u|}\int_0^v
||a^5\mathcal{D}^{11}\sigma||_{L^2_{sc}(\mathcal{S}_{u,v})} \\
&\leq\frac{a^{\frac{1}{2}}}{|u|}\int_0^v
||a^5\mathcal{D}^{11}\sigma||_{L^2_{sc}(\mathcal{S}_{u,v})}.
\end{align*}

For~$I_3$ we have
\begin{align*}
I_3&\leq\int_0^v\frac{1}{|u|}||\varphi_0||_{L^{\infty}_{sc}(\mathcal{S}_{u,v'})}
||a^5\mathcal{D}^{11}\varphi_0||_{L^2_{sc}(\mathcal{S}_{u,v})}
\leq\frac{a^{\frac{1}{2}}\bmGamma(\varphi_0)_{0,\infty}}{|u|}\int_0^v
||a^5\mathcal{D}^{11}\varphi_0||_{L^2_{sc}(\mathcal{S}_{u,v})} \\
&\leq\frac{a}{|u|}\bmGamma(\varphi_0)_{0,\infty}\bm{\varphi}[\varphi_0]
\leq\frac{a}{|u|}(\underline{\bm\varphi}[\varphi_1]+1)\bm{\varphi}[\varphi_0]
\leq\frac{a}{|u|}\bm{\varphi}[\varphi_0].
\end{align*}

For~$I_4$ we have
\begin{align*}
I_4&\leq\sum_{i_1+i_2+i_3+i_4=11}\int_0^v
\frac{(a^{\frac{1}{2}})^{i_2-1}}{|u|^{i_2+1}}||(a^{\frac{1}{2}}\mathcal{D})^{j_1}\Gamma||...
||(a^{\frac{1}{2}}\mathcal{D})^{j_{i_2}}\Gamma||\times
||(a^{\frac{1}{2}}\mathcal{D})^{i_3}[\Gamma,\varphi_0]||\times||(a^{\frac{1}{2}}\mathcal{D})^{i_4}[\Gamma,\varphi_0]|| \\
&\leq\frac{(a^{\frac{1}{2}})^{i_2-1}}{|u|^{i_2+1}}\mathcal{O}^{i_2}a^{\frac{1}{2}}\mathcal{O}a^{\frac{1}{2}}\mathcal{O}
=\frac{(a^{\frac{1}{2}})^{i_2+1}}{|u|^{i_2+1}}\mathcal{O}^{i_2+2}
\leq\frac{a}{|u|^2}\mathcal{O}^2\leq1.
\end{align*}

Collecting the results above we have that 
\begin{align*}
||a^5\meth^{11}\rho||_{L^2_{sc}(\mathcal{S}_{u,v})}\lesssim
\frac{a^{\frac{1}{2}}}{|u|}\int_0^v
||a^5\mathcal{D}^{11}\sigma||_{L^2_{sc}(\mathcal{S}_{u,v})}+
\frac{a}{|u|}\bm{\varphi}[\varphi_0].
\end{align*}
Following the same strategy one can estimate the rest of the strings
of derivatives in~$\mathcal{D}^{k_{i}}\Titau$ and obtain the same
result. That leads to
\begin{align}
\label{Ellipticrhomid}
||a^5\mathcal{D}^{11}\rho||_{L^2_{sc}(\mathcal{S}_{u,v})}\lesssim
\frac{a^{\frac{1}{2}}}{|u|}\int_0^v
||a^5\mathcal{D}^{11}\sigma||_{L^2_{sc}(\mathcal{S}_{u,v})}+
\frac{a}{|u|}\bm{\varphi}[\varphi_0].
\end{align}

\smallskip
\noindent
\textbf{Step 2}: Elliptic estimates for~$\sigma$. \\

For~$\sigma$ we apply the structure
equation~\eqref{T-weightMasslessStructureStEq14} ---namely
\begin{align*}
\meth\rho-\meth'\sigma&=\tau\rho-\bar\tau\sigma-\TiPsi_1.
\end{align*}
As the T-weight of~$\sigma$ is~$-2$, hence we have that
\begin{align*}
\mathscr{D}_{\sigma}=\meth'\sigma=\meth\rho-\tau\rho+\bar\tau\sigma+\TiPsi_1.
\end{align*}
Now, applying Proposition~\ref{EllipticT-weightSC}, we have that
\begin{align*}
||a^5\mathcal{D}^{11}\sigma||_{L^2_{sc}(\mathcal{S}_{u,v})}&\lesssim\sum_{i\leq10}
(\frac{1}{a^{\frac{1}{2}}}||(a^{\frac{1}{2}}\mathcal{D})^{i}\sigma||_{L^2_{sc}(\mathcal{S}_{u,v})}
+||(a^{\frac{1}{2}}\mathcal{D})^{i}\mathscr{D}_{\sigma}||_{L^2_{sc}(\mathcal{S}_{u,v})}) \\
&\lesssim||a^5\mathcal{D}^{11}\rho||_{L^2_{sc}(\mathcal{S}_{u,v})}+
\sum_{i=0}^{10}||(a^{\frac{1}{2}}\mathcal{D})^{i}\TiPsi_1||_{L^2_{sc}(\mathcal{S}_{u,v})} \\
&+\sum_{i=0}^{10}||(a^{\frac{1}{2}}\mathcal{D})^{i}\tau(a^{\frac{1}{2}}\mathcal{D})^{i}\rho||_{L^2_{sc}(\mathcal{S}_{u,v})} 
+\sum_{i=0}^{10}||(a^{\frac{1}{2}}\mathcal{D})^{i}\tau(a^{\frac{1}{2}}\mathcal{D})^{i}\sigma||_{L^2_{sc}(\mathcal{S}_{u,v})} \\
&+\bm\Psi[\Psi_0]+1 \\
&\lesssim||a^5\mathcal{D}^{11}\rho||_{L^2_{sc}(\mathcal{S}_{u,v})}+
\sum_{i=0}^{10}||(a^{\frac{1}{2}}\mathcal{D})^{i}\TiPsi_1||_{L^2_{sc}(\mathcal{S}_{u,v})}
+\bm\Psi[\Psi_0]+1 \\
&\lesssim\frac{a^{\frac{1}{2}}}{|u|}\int_0^v
||a^5\mathcal{D}^{11}\sigma||_{L^2_{sc}(\mathcal{S}_{u,v})}+
\frac{a}{|u|}\bm{\varphi}[\varphi_0]
+\sum_{i=0}^{10}||(a^{\frac{1}{2}}\mathcal{D})^{i}\TiPsi_1||_{L^2_{sc}(\mathcal{S}_{u,v})}
+\bm\Psi[\Psi_0]+1,
\end{align*}
where in the second step we have made use of 
\begin{align*}
\frac{1}{a^{\frac{1}{2}}}||(a^{\frac{1}{2}}\mathcal{D})^{i}\sigma||_{L^2_{sc}(\mathcal{S}_{u,v})}
\lesssim\bm\Psi[\Psi_0]+1 .
\end{align*}
for~$i\leq10$.

Then using Gr\"onwall's inequality we have that 
\begin{align}
\label{EllipticsigmaMid}
||a^5\mathcal{D}^{11}\sigma||_{L^2_{sc}(\mathcal{S}_{u,v})}&\lesssim
\frac{a}{|u|}\bm{\varphi}[\varphi_0]
+\sum_{i=0}^{10}||(a^{\frac{1}{2}}\mathcal{D})^{i}\TiPsi_1||_{L^2_{sc}(\mathcal{S}_{u,v})}
+\bm\Psi[\Psi_0]+1,
\end{align}
so that integrating along the outgoing lightcone we find that
\begin{align*}
||a^5\mathcal{D}^{11}\sigma||_{L^2_{sc}(\mathcal{N}_u(0,v))}
\lesssim\bm{\varphi}[\varphi_0]
+\bm\Psi[\Psi_0]+\bm\Psi[\TiPsi_1]+1.
\end{align*}
Substituting~\eqref{EllipticsigmaMid} back into~\eqref{Ellipticrhomid}
we see that
\begin{align*}
||a^5\mathcal{D}^{11}\rho||_{L^2_{sc}(\mathcal{S}_{u,v})}\lesssim
\bm{\varphi}[\varphi_0]
+\bm\Psi[\Psi_0]+\bm\Psi[\TiPsi_1]+1.
\end{align*}
Integrating along the outgoing lightcone we have that
\begin{align*}
||a^5\mathcal{D}^{11}\rho||_{L^2_{sc}(\mathcal{N}_u(0,v))}\lesssim
\bm{\varphi}[\varphi_0]
+\bm\Psi[\Psi_0]+\bm\Psi[\TiPsi_1]+1.
\end{align*}
\end{proof}

Next, we have that 
\begin{proposition}
\label{11Dertau}
\begin{align*}
||a^5\mathcal{D}^{11}\tau||_{L^2_{sc}(\mathcal{N}_u(0,v))}
\lesssim\bm\Psi[\Psi_0]+\bm\Psi[\TiPsi_2]+\underline{\bm\varphi}[\varphi_1]+1,
\end{align*}
\begin{align*}
||a^5\mathcal{D}^{11}\tau||_{L^2_{sc}(\mathcal{N}'_v(u_{\infty},u))}
\lesssim\bm\Psi[\Psi_0]+\underline{\bm\Psi}[\TiPsi_2]+\underline{\bm\varphi}[\varphi_1]+1.
\end{align*}
\end{proposition}

\begin{proof}
We define the field
\begin{align*}
\Titau\equiv\mathscr{D}_{\tau}-\TiPsi_2=\meth'\tau-\TiPsi_2.
\end{align*}
with zero T-weight and signature~$1$. The strategy of the proof is to
use of the estimates up to the tenth derivative of~$\Titau$ to obtain
the estimate of the eleventh derivative of~$\tau$ via
Proposition~\ref{EllipticT-weightSC}. Making use of the structure
equation of~$\tau$, equation~\eqref{T-weightMasslessStructureStEq9},
and the Bianchi identity for~$\mthorn\TiPsi_2$,
equation~\eqref{T-weightMasslessStBianchi3}, we find the following
equation for~$\Titau$:
\begin{align*}
\mthorn\Titau&=6\varphi_0\meth'\varphi_1+\tau\meth\bar\sigma+\bar\sigma\meth\tau+\sigma\meth'\pi+\pi\meth'\sigma
+\rho\meth'\bar\pi+\bar\pi\meth'\rho+\bar\tau\meth'\sigma+\sigma\meth'\bar\tau \\
&+2\rho\Titau+\lambda\Psi_0-\pi\TiPsi_1-\rho\TiPsi_2-3\mu\varphi_0^2-3\bar\sigma\varphi_1^2
+\pi\bar\pi\rho+\pi^2\sigma+\bar\pi\bar\sigma\tau-\bar\sigma\tau^2-\bar\pi\rho\bar\tau-\sigma\bar\tau^2.
\end{align*}
Remarkably, we see that this equation does not involve derivatives of
the Weyl tensor. Hence we can avoid the estimate of the eleventh
derivative of Weyl. We now first estimate the~$k=10$ derivative
of~$\Titau$. For this, we first
estimate~$\meth^k\Titau$. Commuting~$\meth^k$ with~$\mthorn$ one has
that
\begin{align*}
\mthorn\meth^k\Titau&=\varphi_{0}\meth^{11}\varphi_{1}+
\Gamma(\rho,\sigma)\meth^{11}\Gamma(\tau,\pi)
+\Gamma(\tau,\pi)\meth^{11}\Gamma(\rho,\sigma)
+\Gamma(\lambda,\pi,\rho)\meth^k\bm\Psi(\Psi_0,\TiPsi_1,\TiPsi_2) \\
&+\sum_{i_1+i_2+i_3+i_4=k,i_4\neq k}\meth^{i_1}\Gamma(\tau,\pi)^{i_2}\meth^{i_3}\varphi_{0}\meth^{i_4+1}\varphi_{1} 
+\sum_{i_1+i_2+i_3+i_4=k,i_4\neq k}\meth^{i_1}\Gamma(\tau,\pi)^{i_2}\meth^{i_3}\Gamma\meth^{i_4+1}\Gamma \\
&+\sum_{i_1+i_2+i_3+i_4=k,i_4\leq k-1}
\meth^{i_1}\Gamma(\tau,\pi)^{i_2}\meth^{i_3}\Gamma\meth^{i_4}\bm\Psi(\Psi_0,\TiPsi_1,\TiPsi_2) \\
&+\sum_{i_1+i_2+i_3+i_4+i_5=k}
\meth^{i_1}\Gamma(\tau,\pi)^{i_2}\meth^{i_3}\Gamma(\mu,\sigma)\meth^{i_4}\varphi\meth^{i_5}\varphi 
+\sum_{i_1+i_2+i_3+i_4+i_5=k}
\meth^{i_1}\Gamma(\tau,\pi)^{i_2}\meth^{i_3}\Gamma\meth^{i_4}\Gamma\meth^{i_5}\Gamma \\
&+\sum_{i_1+i_2+i_3+i_4=k}
\meth^{i_1}\Gamma(\tau,\pi)^{i_2}\meth^{i_3}\Gamma\meth^{i_4}\Titau ,
\end{align*}
where~$\Gamma$ contains~$\rho$, $\sigma$, $\pi$, $\tau$, $\varphi$
contain~$\varphi_{0}$ and~$\varphi_{1}$. Then, making use of the
transport estimate in scale-invariant norm, one finds that
\begin{align*}
&||(a^{\frac{1}{2}}\meth)^k\Titau||_{L^2_{sc}(\mathcal{S}_{u,v})}\lesssim
||(a^{\frac{1}{2}}\meth)^k\Titau||_{L^2_{sc}(\mathcal{S}_{u,0})}+
\int_0^v||\mthorn(a^{\frac{1}{2}}\meth)^k\Titau||_{L^2_{sc}(\mathcal{S}_{u,v'})} \\
&\quad\quad\quad\leq||(a^{\frac{1}{2}}\mathcal{D})^k\Titau||_{L^2_{sc}(\mathcal{S}_{u,0})}+
\int_0^va^{5}||\varphi_{0}\mathcal{D}^{11}\varphi_{1}||_{L^2_{sc}(\mathcal{S}_{u,v})}
 +\int_0^va^{5}||\Gamma\mathcal{D}^{11}\Gamma||_{L^2_{sc}(\mathcal{S}_{u,v})}\\
&\quad\quad\quad +\int_0^va^{5}||\Gamma\mathcal{D}^{10}\bm\Psi||_{L^2_{sc}(\mathcal{S}_{u,v})}
+\sum_{i_1+i_2+i_3+i_4=k,i_4\neq k}\int_0^v a^5
||\mathcal{D}^{i_1}\Gamma(\tau,\pi)^{i_2}\mathcal{D}^{i_3}\varphi_{0}
\mathcal{D}^{i_4+1}\varphi_{1}||_{L^2_{sc}(\mathcal{S}_{u,v})} \\
&\quad\quad\quad+\sum_{i_1+i_2+i_3+i_4=k,i_4\neq k}\int_0^v a^5
||\mathcal{D}^{i_1}\Gamma(\tau,\pi)^{i_2}\mathcal{D}^{i_3}\Gamma
\mathcal{D}^{i_4+1}\Gamma||_{L^2_{sc}(\mathcal{S}_{u,v})} \\
&\quad\quad\quad+\sum_{i_1+i_2+i_3+i_4=k,i_4\leq k-1}\int_0^v a^5
||\mathcal{D}^{i_1}\Gamma(\tau,\pi)^{i_2}\mathcal{D}^{i_3}\Gamma
\mathcal{D}^{i_4}\bm\Psi||_{L^2_{sc}(\mathcal{S}_{u,v})} \\
&\quad\quad\quad+\sum_{i_1+i_2+i_3+i_4+i_5=k}\int_0^v a^5
||\mathcal{D}^{i_1}\Gamma(\tau,\pi)^{i_2}\mathcal{D}^{i_3}\Gamma(\mu,\sigma)
\mathcal{D}^{i_4}\varphi\mathcal{D}^{i_5}\varphi||_{L^2_{sc}(\mathcal{S}_{u,v})} \\
&\quad\quad\quad+\sum_{i_1+i_2+i_3+i_4+i_5=k}\int_0^v a^5
||\mathcal{D}^{i_1}\Gamma(\tau,\pi)^{i_2}\mathcal{D}^{i_3}\Gamma
\mathcal{D}^{i_4}\Gamma\mathcal{D}^{i_5}\Gamma||_{L^2_{sc}(\mathcal{S}_{u,v})} \\
&\quad\quad\quad+\sum_{i_1+i_2+i_3+i_4=k}\int_0^v a^5
||\mathcal{D}^{i_1}\Gamma(\tau,\pi)^{i_2}\mathcal{D}^{i_3}\Gamma(\tau,\pi,\rho,\sigma)
\mathcal{D}^{i_4}\Titau||_{L^2_{sc}(\mathcal{S}_{u,v})} \\
&\quad\quad\quad\lesssim1+I_1+...+I_9.
\end{align*}

For~$I_1$ we have that 
\begin{align*}
I_1&\leq\int_0^v\frac{1}{|u|}||\varphi_{0}||\times||(a^{\frac{1}{2}})^{10}\mathcal{D}^{11}\varphi_{1}|| 
\leq\frac{1}{|u|}a^{\frac{1}{2}}\bmGamma(\varphi_0)_{0,\infty}
||(a^{\frac{1}{2}})^{10}\mathcal{D}^{11}\varphi_{1}||_{L^2_{sc}(\mathcal{N}_u(0,v))}\\
&\leq\frac{a^{\frac{1}{2}}}{|u|}\bmGamma(\varphi_0)_{0,\infty}\bm\varphi[\varphi_{1}]
\lesssim1.
\end{align*}
For~$I_2$ we have that 
\begin{align*}
I_2&\leq\int_0^v\frac{1}{|u|}||\Gamma||\times||(a^{\frac{1}{2}})^{10}\mathcal{D}^{11}\Gamma||
\leq\frac{a^{\frac{1}{2}}\mathcal{O}}{|u|}
\int_0^v||(a^{\frac{1}{2}})^{10}\mathcal{D}^{11}\tau||_{L^2_{sc}(\mathcal{S}_{u,v})} \\
&+\frac{a^{\frac{1}{2}}\mathcal{O}}{|u|}
\int_0^v||(a^{\frac{1}{2}})^{10}\mathcal{D}^{11}\pi||_{L^2_{sc}(\mathcal{S}_{u,v})} 
+\frac{\mathcal{O}}{|u|}\int_0^v||(a^{\frac{1}{2}})^{10}\mathcal{D}^{11}\Gamma(\rho,\sigma)||_{L^2_{sc}(\mathcal{S}_{u,v})} \\
&\lesssim\frac{a^{\frac{1}{2}}\mathcal{O}}{|u|}
\int_0^v||(a^{\frac{1}{2}})^{10}\mathcal{D}^{11}\tau||_{L^2_{sc}(\mathcal{S}_{u,v})}+1.
\end{align*}
This term can be handled using the Gr\"onwall inequality. For~$I_3$ we
have that 
\begin{align*}
I_3&\leq\int_0^v\frac{1}{|u|}||\Gamma(\lambda,...)||\times||(a^{\frac{1}{2}})^{10}\mathcal{D}^{10}\bm\Psi(\Psi_{0},...)|| \\
&\leq\frac{1}{|u|}\frac{|u|}{a^{\frac{1}{2}}}
||(a^{\frac{1}{2}}\mathcal{D})^{10}\Psi_{0}||_{L^2_{sc}(\mathcal{N}_u(0,v))} 
+\frac{1}{|u|}\mathcal{O}
||(a^{\frac{1}{2}}\mathcal{D})^{10}\{\TiPsi_{1},\TiPsi_{2}\}||_{L^2_{sc}(\mathcal{N}_u(0,v))} \\
&\leq\bm\Psi[\Psi_0]+1.
\end{align*}
For~$I_4$ we have
\begin{align*}
I_4&\leq\frac{(a^{\frac{1}{2}})^{i_2-1}}{|u|^{i_2+1}}||(a^{\frac{1}{2}}\mathcal{D})^{j_1}\Gamma||...
||(a^{\frac{1}{2}}\mathcal{D})^{j_{i_2}}\Gamma||\times
||(a^{\frac{1}{2}}\mathcal{D})^{i_3}\varphi_{0}||\times||(a^{\frac{1}{2}}\mathcal{D})^{i_4+1}\varphi_{1}||  \\
&\leq\frac{(a^{\frac{1}{2}})^{i_2-1}}{|u|^{i_2+1}}\mathcal{O}^{i_2}a^{\frac{1}{2}}\mathcal{O}\mathcal{O}
=\frac{a^{\frac{i_2}{2}}}{|u|^{i_2+1}}\mathcal{O}^{i_2+2}
\leq\frac{a^{\frac{1}{2}}}{|u|^2}\mathcal{O}^3\leq1.
\end{align*}
For~$I_5$ we have
\begin{align*}
I_5&\leq\frac{(a^{\frac{1}{2}})^{i_2-1}}{|u|^{i_2+1}}||(a^{\frac{1}{2}}\mathcal{D})^{j_1}\Gamma||...
||(a^{\frac{1}{2}}\mathcal{D})^{j_{i_2}}\Gamma||\times
||(a^{\frac{1}{2}}\mathcal{D})^{i_3}\Gamma(\sigma,...)||\times||(a^{\frac{1}{2}}\mathcal{D})^{i_4+1}\Gamma||  \\
&\leq\frac{(a^{\frac{1}{2}})^{i_2-1}}{|u|^{i_2+1}}\mathcal{O}^{i_2}a^{\frac{1}{2}}\mathcal{O}\mathcal{O}
=\frac{a^{\frac{i_2}{2}}}{|u|^{i_2+1}}\mathcal{O}^{i_2+2}
\leq\frac{a^{\frac{1}{2}}}{|u|^2}\mathcal{O}^3\leq1.
\end{align*}
For~$I_6$ we have
\begin{align*}
I_6&\leq\frac{(a^{\frac{1}{2}})^{i_2-1}}{|u|^{i_2+1}}||(a^{\frac{1}{2}}\mathcal{D})^{j_1}\Gamma||...
||(a^{\frac{1}{2}}\mathcal{D})^{j_{i_2}}\Gamma||\times
||(a^{\frac{1}{2}}\mathcal{D})^{i_3}\Gamma(\lambda,...)||\times||(a^{\frac{1}{2}}\mathcal{D})^{i_4}\bm\Psi(\Psi_0)||  \\
&\leq\frac{(a^{\frac{1}{2}})^{i_2-1}}{|u|^{i_2+1}}\mathcal{O}^{i_2}\frac{|u|}{a^{\frac{1}{2}}}\mathcal{O}a^{\frac{1}{2}}\mathcal{O}
=\frac{a^{\frac{i_2}{2}}}{a^{\frac{1}{2}}|u|^{i_2}}\mathcal{O}^{i_2+2}
\leq\frac{\mathcal{O}^2}{a^{\frac{1}{2}}}\leq1.
\end{align*}
For~$I_7$, the worst term
is~$\mu\meth^{i}\varphi_0\meth^{k-i}\varphi_0$. In this case we follow
the same strategy used for~$I_{51}$ in Proposition~\ref{L2PsiII} and
obtain that this term can be bounded by
\begin{align*}
I_{71}\leq\underline{\bm\varphi}[\varphi_1]+1.
\end{align*}
The rest of the terms in~$I_7$ are collectively denoted by
$I_{72}$. We find that
\begin{align*}
I_{72}&\leq a^5||\mathcal{D}^{j_1}\Gamma...\mathcal{D}^{j_{i_2}}\Gamma
\mathcal{D}^{i_3}\Gamma(\Timu,..)\mathcal{D}^{i_4}\varphi\mathcal{D}^{i_5}\varphi|| \\
&\leq\frac{(a^{\frac{1}{2}})^{i_2}}{|u|^{i_2+2}}||(a^{\frac{1}{2}}\mathcal{D})^{j_1}\Gamma||...
||(a^{\frac{1}{2}}\mathcal{D})^{j_{i_2}}\Gamma||\times
||(a^{\frac{1}{2}}\mathcal{D})^{i_3}\Gamma||\times||(a^{\frac{1}{2}}\mathcal{D})^{i_4}\varphi||
\times||(a^{\frac{1}{2}}\mathcal{D})^{i_5}\varphi||  \\
&\leq\frac{(a^{\frac{1}{2}})^{i_2}}{|u|^{i_2+2}}\mathcal{O}^{i_2}\frac{|u|}{a}\mathcal{O}
a^{\frac{1}{2}}\mathcal{O}a^{\frac{1}{2}}\mathcal{O}
=\frac{a^{\frac{i_2}{2}}}{|u|^{i_2+1}}\mathcal{O}^{i_2+3}
\leq\frac{1}{|u|}\mathcal{O}^3\leq1.
\end{align*}
For~$I_8$ we have that
\begin{align*}
I_8&\leq\frac{(a^{\frac{1}{2}})^{i_2}}{|u|^{i_2+2}}||(a^{\frac{1}{2}}\mathcal{D})^{j_1}\Gamma||...
||(a^{\frac{1}{2}}\mathcal{D})^{j_{i_2}}\Gamma||\times
||(a^{\frac{1}{2}}\mathcal{D})^{i_3}\Gamma||\times||(a^{\frac{1}{2}}\mathcal{D})^{i_4}\Gamma||
\times||(a^{\frac{1}{2}}\mathcal{D})^{i_5}\Gamma(\sigma,...)||  \\
&\leq\frac{(a^{\frac{1}{2}})^{i_2}}{|u|^{i_2+2}}\mathcal{O}^{i_2}\mathcal{O}\mathcal{O}a^{\frac{1}{2}}\mathcal{O}
=\frac{a^{\frac{i_2+1}{2}}}{|u|^{i_2+2}}\mathcal{O}^{i_2+3}
\leq\frac{a^{\frac{1}{2}}}{|u|^2}\mathcal{O}^3\leq1.
\end{align*}
Finally, for~$I_9$ we find that
\begin{align*}
I_9&\leq\frac{(a^{\frac{1}{2}})^{i_2}}{|u|^{i_2+1}}||(a^{\frac{1}{2}}\mathcal{D})^{j_1}\Gamma||...
||(a^{\frac{1}{2}}\mathcal{D})^{j_{i_2}}\Gamma||\times
||(a^{\frac{1}{2}}\mathcal{D})^{i_3}\Gamma(\sigma,...)||\times||(a^{\frac{1}{2}}\mathcal{D})^{i_4}\Titau||  \\
&\leq\frac{(a^{\frac{1}{2}})^{i_2}}{|u|^{i_2+1}}\mathcal{O}^{i_2}a^{\frac{1}{2}}\mathcal{O}\mathcal{O}
=\frac{a^{\frac{i_2+1}{2}}}{|u|^{i_2+1}}\mathcal{O}^{i_2+3}
\leq\frac{a^{\frac{1}{2}}}{|u|}\leq1.
\end{align*}
Collect the results above one concludes that
\begin{align*}
||(a^{\frac{1}{2}}\meth)^k\Titau||_{L^2_{sc}(\mathcal{S}_{u,v})}
\lesssim\bm\Psi[\Psi_0]+\underline{\bm\varphi}[\varphi_1]+1.
\end{align*}
Following the same strategy one can estimate the rest of the strings
in~$\mathcal{D}^{k_{i}}\Titau$ and obtain the same result. That leads
to
\begin{align*}
||(a^{\frac{1}{2}}\mathcal{D})^k\Titau||_{L^2_{sc}(\mathcal{S}_{u,v})}
\lesssim\bm\Psi[\Psi_0]+\underline{\bm\varphi}[\varphi_1]+1.
\end{align*}

Now, making use of the definition
$\Titau=\mathscr{D}_{\tau}-\TiPsi_2=\meth'\tau-\TiPsi_2$, we have that
\begin{align*}
||(a^{\frac{1}{2}}\mathcal{D})^k\mathscr{D}_\tau||_{L^2_{sc}(\mathcal{S}_{u,v})}\lesssim
||(a^{\frac{1}{2}}\mathcal{D})^k\TiPsi_2||_{L^2_{sc}(\mathcal{S}_{u,v})}
+\bm\Psi[\Psi_0]+\underline{\bm\varphi}[\varphi_1]+1.
\end{align*}
Applying Proposition~\ref{EllipticT-weightSC}, we have that  
\begin{align*}
||a^5\mathcal{D}^{11}\tau||_{L^2_{sc}(\mathcal{S}_{u,v})}&\lesssim\sum_{i\leq10}
(\frac{1}{a^{\frac{1}{2}}}||(a^{\frac{1}{2}}\mathcal{D})^{i}\tau||_{L^2_{sc}(\mathcal{S}_{u,v})}
+||(a^{\frac{1}{2}}\mathcal{D})^{i}\mathscr{D}_\tau||_{L^2_{sc}(\mathcal{S}_{u,v})}) \\
&\lesssim||(a^{\frac{1}{2}}\mathcal{D})^{10}\TiPsi_2||_{L^2_{sc}(\mathcal{S}_{u,v})}
+\bm\Psi[\Psi_0]+\underline{\bm\varphi}[\varphi_1]+1.
\end{align*}
We now transform the estimate on~$\mathcal{S}_{u,v}$ to one on the
light cone so that
\begin{align*}
||a^5\mathcal{D}^{11}\tau||_{L^2_{sc}(\mathcal{N}_u(0,v))}
\lesssim\bm\Psi[\Psi_0]+\bm\Psi[\TiPsi_2]+\underline{\bm\varphi}[\varphi_1]+1,
\end{align*}
and
\begin{align*}
||a^5\mathcal{D}^{11}\tau||_{L^2_{sc}(\mathcal{N}'_v(u_{\infty},u))}
\lesssim\bm\Psi[\Psi_0]+\underline{\bm\Psi}[\TiPsi_2]+\underline{\bm\varphi}[\varphi_1]+1.
\end{align*}
\end{proof}

A similar strategy is used to construct an estimate for the
coefficient~$\pi$.
\begin{proposition}
\label{11Derpi}
\begin{align*}
  \frac{a}{|u|}||a^5\mathcal{D}^{11}\pi||_{L^2_{sc}(\mathcal{N}_{u}(0,v))}
  \lesssim\bm\Psi[\TiPsi_2]+\bm\Psi[\Psi_0]+\underline{\bm\Psi}[\TiPsi_2]
  +\underline{\bm\varphi}[\varphi_1]+1.
\end{align*}
\end{proposition}

\begin{proof}
  In order to estimate the eleventh derivative of~$\pi$ we introduce
  the following quantity:
\begin{align*}
\Tipi\equiv\mathscr{D}_{\pi}+\TiPsi_2=\meth\pi+\TiPsi_2.
\end{align*}
The T-weight of~$\Tipi$ is zero and its signature of is~$1$. Next, we
make use of the structure
equation~\eqref{T-weightMasslessStructureStEq10} and the Bianchi
identities~\eqref{T-weightMasslessStBianchi6} to obtain the following
equation for~$\Tipi$:
\begin{align*}
  \mthorn'\Tipi+2\mu\Tipi&=
  -6\varphi_2\meth\bar\varphi_1
    -(\bar\pi+\tau)\meth\lambda-\bar\tau\meth\mu-\lambda(\meth\bar\pi+\meth\tau)
    -\mu\meth\bar\tau-\bar\lambda\meth'\pi-\pi\meth'\bar\lambda\\
  &-\tau\TiPsi_3-\mu\TiPsi_2+\sigma\Psi_4
    -3\bar\lambda\bar\varphi_1^2-3\rho\varphi_2^2 .
\end{align*}
Observe that this equation does not involve derivatives of the Weyl
tensor ---hence, we can avoid estimating its eleventh derivative. Now,
we first estimate the~$k=10$ derivative of~$\Tipi$ by first looking at
$\meth^k\Tipi$. Commuting~$\meth^k$ with~$\mthorn'$ one has that
\begin{align*}
&\mthorn'\meth^k\Tipi+(k+2)\mu\meth^k\Tipi=\varphi_{2}\meth^{11}\varphi_{1}+
\Gamma(\mu,\lambda)\meth^{11}\tau+\lambda\meth^{11}\pi
+\Gamma(\tau,\pi)\meth^{11}\Gamma(\Timu,\lambda)\\
  &\quad\quad
    +\Gamma(\tau,\mu,\sigma)\meth^k\bm\Psi(\TiPsi_2,\TiPsi_3,\Psi_4)
    +\sum_{i=1}^{k}\meth^i\varphi_2\meth^{k+1-i}\varphi_1
+\sum_{i_1+i_2+i_3=k}\meth^{i_1}\Gamma(\lambda,\rho)\meth^{i_2}\bm{\varphi}(1,2)\meth^{i_3}\bm{\varphi} \\
&\quad\quad+\sum_{i=1}^{k}\meth^i\Gamma(\tau,\pi,\lambda,\mu)\meth^{k+1-i}\Gamma(\tau,\pi,\lambda,\mu)
+\sum_{i=0}^k\meth^i\Gamma(\Timu,\lambda)\meth^{k-i}\Tipi \\
&\quad\quad+\sum_{i=1}^{k}\meth^i\Gamma(\tau,\Timu,\sigma)\meth^{k-1}\bm{\Psi}.
\end{align*}
Denote the right-hand side of the last equation by~$F$. One can verify
that
\begin{align*}
&s_2(\Tipi)=1, \ \ s_2(\meth^k\Tipi)=\frac{k+2}{2}, \ \ s_2(F)=\frac{k+4}{2}, \\
&\lambda_0=-(k+2), \ \ \lambda_1=-\lambda_0-1=k+1, \\
& \lambda_1-2s_2(F)=-3, \ \ \lambda-2s_2(\meth^k\Tipi)=-1.
\end{align*}
Now, making use of Proposition~\ref{SCweightedgronwall}, one has that
\begin{align*}
\frac{a}{|u|}||(a^{\frac{1}{2}}\meth)^k\Tipi||_{L^2_{sc}(\mathcal{S}_{u,v})}\lesssim
\frac{a}{|u_{\infty}|}||(a^{\frac{1}{2}}\meth)^k\Tipi||_{L^2_{sc}(\mathcal{S}_{u_{\infty},v})}
+\int_{u_{\infty}}^u\frac{a^2}{|u'|^3}||a^{\frac{k}{2}}F||_{L^2_{sc}(\mathcal{S}_{u',v})}.
\end{align*}
For convenience define 
\begin{align*}
H\equiv\int_{u_{\infty}}^u\frac{a^2}{|u|^3}||a^{\frac{k}{2}}F||_{L^2_{sc}(\mathcal{S}_{u,v})}.
\end{align*}
Substituting into the expression of~$F$ one has that
\begin{align*}
H&\leq\int_{u_{\infty}}^u\frac{a^2}{|u'|^3}||a^5\varphi_2\mathcal{D}^{11}\varphi_1||_{L^2_{sc}(\mathcal{S}_{u',v})}
+\int_{u_{\infty}}^u\frac{a^2}{|u'|^3}||a^5\Gamma\mathcal{D}^{11}\Gamma||_{L^2_{sc}(\mathcal{S}_{u',v})} 
+\int_{u_{\infty}}^u\frac{a^2}{|u'|^3}||a^5\Gamma\mathcal{D}^{10}\bm{\Psi}||_{L^2_{sc}(\mathcal{S}_{u',v})}\\
&+\sum_{i=1}^k\int_{u_{\infty}}^u\frac{a^2}{|u'|^3}
||a^5\mathcal{D}^i\varphi_2\mathcal{D}^{k+1-i}\varphi_1||_{L^2_{sc}(\mathcal{S}_{u',v})}
+\sum_{i_1+i_2+i_3=k}\int_{u_{\infty}}^u\frac{a^2}{|u'|^3}
||a^5\mathcal{D}^{i_1}\Gamma(\lambda,\rho)\mathcal{D}^{i_2}\bm\varphi
\mathcal{D}^{i_3}\bm\varphi||_{L^2_{sc}(\mathcal{S}_{u',v})} \\
&+\sum_{i=1}^k\int_{u_{\infty}}^u\frac{a^2}{|u'|^3}
||a^5\mathcal{D}^i\Gamma(\mu,..)\mathcal{D}^{k+1-i}\Gamma||_{L^2_{sc}(\mathcal{S}_{u',v})}
+\sum_{i=0}^k\int_{u_{\infty}}^u\frac{a^2}{|u'|^3}
||a^5\mathcal{D}^i\Gamma(\Timu,..)\mathcal{D}^{k-i}\Tipi||_{L^2_{sc}(\mathcal{S}_{u',v})} \\
&+\sum_{i=1}^k\int_{u_{\infty}}^u\frac{a^2}{|u'|^3}
||a^5\mathcal{D}^i\Gamma(\Timu,..)\mathcal{D}^{k-i}\bm\Psi||_{L^2_{sc}(\mathcal{S}_{u',v})} 
=I_1+...+I_8.
\end{align*}

For the term~$I_1$ we have that 
\begin{align*}
I_1&\leq\int_{u_{\infty}}^u\frac{a^2}{|u'|^3}\frac{1}{|u'|}||\varphi_2||_{L^{\infty}_{sc}(\mathcal{S}_{u',v})}
||a^5\mathcal{D}^{11}\varphi_1||_{L^2_{sc}(\mathcal{S}_{u',v})} 
\leq\int_{u_{\infty}}^u\frac{a^2}{|u'|^4}\frac{|u'|}{a^{\frac{1}{2}}}\Gamma(\varphi_2)_{0,\infty}
||a^5\mathcal{D}^{11}\varphi_1||_{L^2_{sc}(\mathcal{S}_{u',v})} \\
&=\Gamma(\varphi_2)_{0,\infty}\int_{u_{\infty}}^u\frac{a^{\frac{3}{2}}}{|u|^3}
||a^5\mathcal{D}^{11}\varphi_1||_{L^2_{sc}(\mathcal{S}_{u',v})} 
\leq\Gamma(\varphi_2)_{0,\infty}\left(
\int_{u_{\infty}}^u\frac{a}{|u|^2}
||a^5\mathcal{D}^{11}\varphi_1||^2_{L^2_{sc}(\mathcal{S}_{u',v})}
\right)^{\frac{1}{2}}
\left(\int_{u_{\infty}}^u\frac{a^2}{|u|^4} \right)^{\frac{1}{2}}\\
&\leq\frac{a\Gamma(\varphi_2)_{0,\infty}}{|u|^{\frac{3}{2}}}
||a^5\mathcal{D}^{11}\varphi_1||_{L^2_{sc}(\mathcal{N}'_{v}(u_{\infty},u))} 
\leq\frac{a^{\frac{3}{2}}}{|u|^{\frac{3}{2}}}\Gamma(\varphi_2)_{0,\infty}\underline{\bm\varphi}[\varphi_1]
\leq1.
\end{align*}
For~$I_2$ we have that 
\begin{align*}
I_2&\leq\int_{u_{\infty}}^u\frac{a^2}{|u'|^3}\frac{1}{|u'|}||\Gamma||_{L^{\infty}_{sc}(\mathcal{S}_{u',v})}
||a^5\mathcal{D}^{11}\Gamma||_{L^2_{sc}(\mathcal{S}_{u',v})}\leq
\int_{u_{\infty}}^u\frac{a^2}{|u'|^4}||\mu||_{L^{\infty}_{sc}(\mathcal{S}_{u',v})}
||a^5\mathcal{D}^{11}\tau||_{L^2_{sc}(\mathcal{S}_{u',v})} \\
&+\int_{u_{\infty}}^u\frac{a^2}{|u'|^4}||\lambda||_{L^{\infty}_{sc}(\mathcal{S}_{u',v})}
||a^5\mathcal{D}^{11}\tau||_{L^2_{sc}(\mathcal{S}_{u',v})}
+\int_{u_{\infty}}^u\frac{a^2}{|u'|^4}||\lambda||_{L^{\infty}_{sc}(\mathcal{S}_{u',v})}
||a^5\mathcal{D}^{11}\pi||_{L^2_{sc}(\mathcal{S}_{u',v})}  \\
&+\int_{u_{\infty}}^u\frac{a^2}{|u'|^4}||\Gamma(\tau,\pi)||_{L^{\infty}_{sc}(\mathcal{S}_{u',v})}
||a^5\mathcal{D}^{11}\mu||_{L^2_{sc}(\mathcal{S}_{u',v})} 
+\int_{u_{\infty}}^u\frac{a^2}{|u'|^4}||\Gamma(\tau,\pi)||_{L^{\infty}_{sc}(\mathcal{S}_{u',v})}
||a^5\mathcal{D}^{11}\lambda||_{L^2_{sc}(\mathcal{S}_{u',v})} \\
&\leq\int_{u_{\infty}}^u\frac{a^2}{|u'|^4}\frac{|u'|^2}{a}\bmGamma(\mu)_{0,\infty}
||a^5\mathcal{D}^{11}\tau||_{L^2_{sc}(\mathcal{S}_{u',v})} 
+\int_{u_{\infty}}^u\frac{a^2}{|u'|^4}\frac{|u'|}{a^{\frac{1}{2}}}\bmGamma(\lambda)_{0,\infty}
||a^5\mathcal{D}^{11}\tau||_{L^2_{sc}(\mathcal{S}_{u',v})} \\
&+\int_{u_{\infty}}^u\frac{a^2}{|u'|^4}\frac{|u'|}{a^{\frac{1}{2}}}\bmGamma(\lambda)_{0,\infty}
||a^5\mathcal{D}^{11}\pi||_{L^2_{sc}(\mathcal{S}_{u',v})} 
+\int_{u_{\infty}}^u\frac{a^2}{|u'|^4}\bmGamma(\tau,\pi)_{0,\infty}
||a^5\mathcal{D}^{11}\mu||_{L^2_{sc}(\mathcal{S}_{u',v})} \\
&+\int_{u_{\infty}}^u\frac{a^2}{|u'|^4}\bmGamma(\tau,\pi)_{0,\infty}
||a^5\mathcal{D}^{11}\lambda||_{L^2_{sc}(\mathcal{S}_{u',v})} \\
&\leq\frac{a^{\frac{1}{2}}}{|u|^{\frac{1}{2}}}\bmGamma(\mu)_{0,\infty}\bmGamma_{11}(\tau)
+\frac{1}{|u|^{\frac{3}{2}}}\bmGamma(\lambda)_{0,\infty}\bmGamma_{11}(\tau)
+\int_{u_{\infty}}^u\frac{1}{|u'|^2}
||\frac{a}{|u'|}a^5\mathcal{D}^{11}\pi||_{L^2_{sc}(\mathcal{S}_{u',v})} \\
&+\frac{1}{a}\bmGamma(\tau,\pi)_{0,\infty}\bmGamma_{11}(\mu)
+\frac{1}{a^{\frac{1}{2}}}\bmGamma(\tau,\pi)_{0,\infty}\bmGamma_{11}(\lambda) \\
&\leq\int_{u_{\infty}}^u\frac{1}{|u'|^2}
||\frac{a}{|u'|}a^5\mathcal{D}^{11}\pi||_{L^2_{sc}(\mathcal{S}_{u',v})}
+\bmGamma_{11}(\tau)+1.
\end{align*}
The integral in the last line can be absorbed using the Gr\"onwall
inequality. For~$I_3$, we have that
\begin{align*}
I_3&\leq\int_{u_{\infty}}^u\frac{a^2}{|u'|^3}\frac{1}{|u'|}||\Gamma||_{L^{\infty}_{sc}(\mathcal{S}_{u',v})}
||(a^{\frac{1}{2}}\mathcal{D})^{10}\TiPsi||_{L^2_{sc}(\mathcal{S}_{u',v})}
\leq\int_{u_{\infty}}^u\frac{a^2}{|u'|^4}||\mu||_{L^{\infty}_{sc}(\mathcal{S}_{u',v})}
||(a^{\frac{1}{2}}\mathcal{D})^{10}\TiPsi_2||_{L^2_{sc}(\mathcal{S}_{u',v})} \\
&+\int_{u_{\infty}}^u\frac{a^2}{|u'|^4}||\tau||_{L^{\infty}_{sc}(\mathcal{S}_{u',v})}
||(a^{\frac{1}{2}}\mathcal{D})^{10}\TiPsi_3||_{L^2_{sc}(\mathcal{S}_{u',v})} 
+\int_{u_{\infty}}^u\frac{a^2}{|u'|^4}||\sigma||_{L^{\infty}_{sc}(\mathcal{S}_{u',v})}
||(a^{\frac{1}{2}}\mathcal{D})^{10}\Psi_4||_{L^2_{sc}(\mathcal{S}_{u',v})} \\
&\leq\int_{u_{\infty}}^u\frac{a}{|u'|^2}\bmGamma(\mu)_{0,\infty}
||(a^{\frac{1}{2}}\mathcal{D})^{10}\TiPsi_2||_{L^2_{sc}(\mathcal{S}_{u',v})} 
+\int_{u_{\infty}}^u\frac{a^2}{|u'|^4}\bmGamma(\tau)_{0,\infty}
||(a^{\frac{1}{2}}\mathcal{D})^{10}\TiPsi_3||_{L^2_{sc}(\mathcal{S}_{u',v})} \\
&+\int_{u_{\infty}}^u\frac{a^2}{|u'|^4}\bmGamma(\sigma)_{0,\infty}
||(a^{\frac{1}{2}}\mathcal{D})^{10}\TiPsi_4||_{L^2_{sc}(\mathcal{S}_{u',v})} \\
\leq&\frac{a^{\frac{1}{2}}}{|u|^{\frac{1}{2}}}\bmGamma(\mu)_{0,\infty}\underline{\bm\Psi}[\TiPsi_2]
+\frac{a^{\frac{3}{2}}}{|u|^{\frac{5}{2}}}\bmGamma(\tau)_{0,\infty}\underline{\bm\Psi}[\TiPsi_3]
+\frac{a}{|u|^{\frac{5}{2}}}\bmGamma(\sigma)_{0,\infty}\underline{\bm\Psi}[\Psi_4] 
\leq\underline{\bm\Psi}[\TiPsi_2]+1.
\end{align*}
For~$I_4$ we have that
\begin{align*}
I_4&\leq\sum_{i=1}^k\int_{u_{\infty}}^u\frac{a^2}{|u'|^3}\frac{1}{|u'|}\frac{1}{a^{\frac{1}{2}}}
||(a^{\frac{1}{2}}\mathcal{D})^{i}\varphi_2||\times||(a^{\frac{1}{2}}\mathcal{D})^{k+1-i}\varphi_1|| \\
&\leq\int_{u_{\infty}}^u\frac{a^{\frac{3}{2}}}{|u'|^4}\frac{|u'|}{a^{\frac{1}{2}}}\mathcal{O}\mathcal{O}
\leq\frac{a}{|u|^2}\mathcal{O}^2\leq1.
\end{align*}
For~$I_5$ we find that 
\begin{align*}
I_5&\leq\sum_{i_1+i_2+i_3=k}\int_{u_{\infty}}^u\frac{a^2}{|u'|^3}\frac{1}{|u'|^2}
||\Gamma(\lambda,\rho)||\times||\bm\varphi(1,2)||\times||\bm\varphi|| \\
&\leq\int_{u_{\infty}}^u\frac{a^2}{|u'|^5}(\frac{|u'|}{a^{\frac{1}{2}}}a^{\frac{1}{2}}\mathcal{O}a^{\frac{1}{2}}\mathcal{O}
+\frac{|u'|^2}{a}\mathcal{O}^3)\leq1.
\end{align*}
For~$I_6$ we have that 
\begin{align*}
I_6&\leq\sum_{i=1}^k\int_{u_{\infty}}^u\frac{a^2}{|u'|^3}\frac{1}{|u'|}\frac{1}{a^{\frac{1}{2}}}
||(a^{\frac{1}{2}}\mathcal{D})^{i}\Gamma(\Timu,\lambda,..)||\times
||(a^{\frac{1}{2}}\mathcal{D})^{k+1-i}\Gamma|| \\
&\leq\int_{u_{\infty}}^u\frac{a^{\frac{3}{2}}}{|u'|^4}\frac{|u'|}{a^{\frac{1}{2}}}\mathcal{O}^2
\leq\frac{a}{|u|^2}\mathcal{O}^2\leq1.
\end{align*}
For~$I_7$ we have that 
\begin{align*}
I_7&\leq\sum_{i=1}^k\int_{u_{\infty}}^u\frac{a^2}{|u'|^3}\frac{1}{|u'|}
||(a^{\frac{1}{2}}\mathcal{D})^{i}\Gamma(\Timu,\lambda)||\times
||(a^{\frac{1}{2}}\mathcal{D})^{k-i}\Tipi|| \\
&\leq\int_{u_{\infty}}^u\frac{a^2}{|u'|^4}\frac{|u'|}{a^{\frac{1}{2}}}\mathcal{O}\frac{|u'|}{a}\mathcal{O}
\leq\frac{a^{\frac{1}{2}}}{|u|}\mathcal{O}^2\leq1.
\end{align*}
For~$I_8$ we have that 
\begin{align*}
I_8\leq\sum_{i=1}^k\int_{u_{\infty}}^u\frac{a^2}{|u'|^3}\frac{1}{|u'|}||\Gamma||\times||\Psi||
\leq\int_{u_{\infty}}^u\frac{a^2}{|u'|^4}\frac{|u'|}{a}\mathcal{O}^2\leq1.
\end{align*}

Collecting the results above we can conclude that 
\begin{align*}
\frac{a}{|u|}||(a^{\frac{1}{2}}\meth)^k\Tipi||_{L^2_{sc}(\mathcal{S}_{u,v})}&\lesssim
\bmGamma_{11}(\tau)+\underline{\bm\Psi}[\TiPsi_2]+1 \\
&\lesssim\bm\Psi[\Psi_0]+\underline{\bm\Psi}[\TiPsi_2]+\underline{\bm\varphi}[\varphi_1]+1.
\end{align*}
This inequality leads to
\begin{align*}
\frac{a}{|u|}||(a^{\frac{1}{2}}\mathcal{D})^k\Tipi||_{L^2_{sc}(\mathcal{S}_{u,v})}
\lesssim\bm\Psi[\Psi_0]+\underline{\bm\Psi}[\TiPsi_2]+\underline{\bm\varphi}[\varphi_1]+1.
\end{align*}
By the definition of
$\Tipi=\mathscr{D}_{\pi}+\TiPsi_2=\meth\pi+\TiPsi_2$, we have that 
\begin{align*}
\frac{a}{|u|}||(a^{\frac{1}{2}}\mathcal{D})^k\mathscr{D}_{\pi}||_{L^2_{sc}(\mathcal{S}_{u,v})}
\lesssim\frac{a}{|u|}||a^5\mathcal{D}^{10}\TiPsi_2||_{L^2_{sc}(\mathcal{S}_{u,v})}
+\bm\Psi[\Psi_0]+\underline{\bm\Psi}[\TiPsi_2]+\underline{\bm\varphi}[\varphi_1]+1.
\end{align*}
Now applying Proposition~\ref{EllipticT-weightSC}, we see that 
\begin{align*}
\frac{a}{|u|}||a^5\mathcal{D}^{11}\pi||_{L^2_{sc}(\mathcal{S}_{u,v})}&\lesssim\sum_{i\leq10}
(\frac{1}{a^{\frac{1}{2}}}\frac{a}{|u|}||(a^{\frac{1}{2}}\mathcal{D})^{i}\pi||_{L^2_{sc}(\mathcal{S}_{u,v})}
+\frac{a}{|u|}||(a^{\frac{1}{2}}\mathcal{D})^{i}\mathscr{D}_\pi||_{L^2_{sc}(\mathcal{S}_{u,v})}) \\
&\lesssim\frac{a}{|u|}||a^5\mathcal{D}^{10}\TiPsi_2||_{L^2_{sc}(\mathcal{S}_{u,v})}
+\bm\Psi[\Psi_0]+\underline{\bm\Psi}[\TiPsi_2]+\underline{\bm\varphi}[\varphi_1]+1.
\end{align*}
Then, integrating along the outgoing lightcone we finally conclude
that
\begin{align*}
\frac{a}{|u|}||a^5\mathcal{D}^{11}\pi||_{L^2_{sc}(\mathcal{N}_{u}(0,v))}\lesssim
\bm\Psi[\TiPsi_2]+\bm\Psi[\Psi_0]+\underline{\bm\Psi}[\TiPsi_2]+\underline{\bm\varphi}[\varphi_1]+1.
\end{align*}

\end{proof}

\begin{proposition}
\label{11Derulomega}
\begin{align*}
||a^5\mathcal{D}^{11}\ulomega||_{L^2_{sc}(\mathcal{N}'_v(u_{\infty},u))}\lesssim
\underline{\bm\varphi}[\varphi_2]+\bm\varphi[\varphi_0]
+\underline{\bm\Psi}[\TiPsi_3]+1.
\end{align*}
\end{proposition}

\begin{proof}
  We estimate~$\ulomega$ in three steps.

\smallskip
\noindent
\textbf{Step 1}: Construct a new scalar~$\tilde\ulomega$.

First we construct an auxiliary function~$\ulomega^{\dag}$ with zero
T-weight through the relation
\begin{align}
\label{EQulomegastar}
\mthorn\ulomega^{\dag}=-i(\bar\TiPsi_2-\TiPsi_2)
\end{align}
with trivial initial data on~$\mathcal{N}'_{\star}$. Here~$\ulomega^{\dag}$ is real.
First we apply~$\meth$ and use the commutator relation to find that 
\begin{align}
\label{PethUlomegadag}
\mthorn(i\meth\ulomega^{\dag})=\meth\bar\TiPsi_2-\meth\TiPsi_2+i\rho\meth\ulomega^{\dag}+i\sigma\meth'\ulomega^{\dag}
+(\bar\pi-\tau)(\bar\TiPsi_2-\TiPsi_2).
\end{align}
Then, making use of the structure equation for~$\ulomega$,
\eqref{T-weightMasslessStructureStEq11}, we have that 
\begin{align*}
\mthorn\ulomega=2\tau\bar\tau+2\tau\pi+2\bar\tau\bar\pi
+\TiPsi_2+\bar\TiPsi_2+6\varphi_{0}\varphi_{2}, 
\end{align*}
and then we have
\begin{align}
\label{PethUlomega}
\mthorn\meth\ulomega=&\meth\TiPsi_2+\meth\bar\TiPsi_2+\rho\meth\ulomega+\sigma\meth'\ulomega \nonumber\\
&+[\meth+(\bar\pi-\tau)](2\tau\bar\tau+2\tau\pi+2\bar\tau\bar\pi
+6\varphi_{0}\varphi_{2} ).
\end{align}
Together with~\eqref{PethUlomegadag} we see that 
\begin{align}
\label{PUlomega1}
\mthorn[\meth\ulomega+i\meth\ulomega^{\dag}]=&2\meth\bar\TiPsi_2+\rho(\meth\ulomega+i\meth\ulomega^{\dag})
+\sigma(\meth'\ulomega+i\meth'\ulomega^{\dag}) \nonumber \\
&+\Gamma\meth\Gamma+\Gamma^3+\bm{\varphi}\meth\bm{\varphi}+\Gamma\bm{\varphi}^2
+\Gamma\TiPsi_2,
\end{align}
where in the previous equation~$\Gamma$ contains~$\tau$ and~$\pi$,
while~$\bm{\varphi}$ contains~$\varphi_0$ and~$\varphi_2$. Then,
making use of the Bianchi identity~\eqref{T-weightMasslessStBianchi5} 
\begin{align*}
\mthorn\TiPsi_3-\meth'\TiPsi_2=&-6\bar\varphi_1\meth\bar\varphi_1-2\lambda\TiPsi_1
+3\pi\TiPsi_2+2\rho\TiPsi_3 \nonumber \\
&-3\bar\varphi_1^2\bar\pi+6\varphi_0\bar\varphi_1\mu
-6\varphi_0\varphi_1\lambda-6\varphi_1\bar\varphi_1\pi ,
\end{align*}
we have that 
\begin{align*}
\mthorn [\meth(\ulomega+i\ulomega^{\dag})-2\bar\TiPsi_3 ]=&\rho[\meth(\ulomega+i\ulomega^{\dag})-2\bar\TiPsi_3]
+\sigma\meth'(\ulomega+i\ulomega^{\dag}) \nonumber \\
&+\Gamma\meth\Gamma+\Gamma^3+\bm{\varphi}\meth\bm{\varphi}+\Gamma\bm{\varphi}^2
+\Gamma\TiPsi.
\end{align*}
In the previous equation~$\Gamma$ contains~$\tau$, $\pi$, $\mu$
and~$\lambda$; $\bm{\varphi}$ contains~$\varphi_0$, $\varphi_1$ and
$\varphi_2$; and~$\TiPsi$ contains~$\TiPsi_1$, $\TiPsi_2$
and~$\TiPsi_3$. Crucially, there is no derivative of curvature on the
right-hand side.

Now, we define 
\begin{align*}
\tilde\ulomega\equiv\meth(\ulomega+i\ulomega^{\dag})-2\bar\TiPsi_3,
\end{align*}
and find that 
\begin{align*}
\mthorn\tilde\ulomega&=\rho\tilde\ulomega-\sigma\tilde\ulomega^*+2\sigma\meth'\ulomega
+2\tau(\meth\pi+\meth\bar\tau)+2\bar\tau(\meth\bar\pi+\meth\tau)+2\pi\meth\tau+2\bar\pi\meth\bar\tau
+6\varphi_2\meth\varphi_0+6\varphi_0\meth\varphi_2 \\
&+12\varphi_1\meth'\varphi_1+12\bar\lambda\varphi_0\bar\varphi_1-12\mu\varphi_0\varphi_1+6\pi\varphi_1^2+6\bar\pi\varphi_0\varphi_2+
12\bar\pi\varphi_1\bar\varphi_1-6\tau\varphi_0\varphi_2
+4\bar\lambda\bar\TiPsi_1-4\bar\pi\bar\TiPsi_2 \\
&-2\rho\bar\TiPsi_3-2\sigma\TiPsi_3-2\tau\bar\TiPsi_2
+2\pi\bar\pi\tau-2\pi\tau^2+2\bar\pi^2\bar\tau-2\tau^2\bar\tau.
\end{align*}
The key observation is that this equation does not contain derivatives
of the Weyl curvature. Hence we can avoid estimating the eleventh
derivative of the Weyl tensor. The T-weight of~$\tilde\ulomega$
is~$-1$ and its signature is~$\frac{3}{2}$.

\smallskip
\noindent
\textbf{Step 2}: Estimate of the~$k=10$ derivative
of~$\tilde\ulomega$.

Commuting~$\meth^k$ with~$\mthorn$ one has that
\begin{align*}
\mthorn\meth^k\tilde\ulomega&=\bm{\varphi}\meth^{11}\bm{\varphi}+
\Gamma(\tau,\pi,\lambda,\rho,\sigma)\meth^{10}\TiPsi
+\Gamma(\tau,\pi)\meth^{11}\Gamma(\tau,\pi)
+\Gamma(\rho,\sigma)\meth^{10}\tilde\ulomega
+\sigma\meth^{11}\ulomega \\
&+\sum_{i_1+i_2+i_3+i_4=k,i_4<k}\meth^{i_1}\Gamma(\tau,\pi)^{i_2}\meth^{i_3}\bm{\varphi}\meth^{i_4+1}\bm{\varphi} \\
&+\sum_{i_1+i_2+i_3+i_4=k,i_4<k}\meth^{i_1}\Gamma(\tau,\pi)^{i_2}
\meth^{i_3}\Gamma(\tau,\pi,\sigma)\meth^{i_4+1}\Gamma(\tau,\pi,\sigma,\ulomega) \\
&+\sum_{i_1+i_2+i_3+i_4=k,i_4<k}
\meth^{i_1}\Gamma(\tau,\pi)^{i_2}\meth^{i_3}\Gamma(\tau,\pi,\lambda,\rho)\meth^{i_4}\TiPsi \\
&+\sum_{i_1+i_2+i_3+i_4+i_5=k}
\meth^{i_1}\Gamma(\tau,\pi)^{i_2}\meth^{i_3}\Gamma(\tau,\pi,\mu,\lambda)\meth^{i_4}\bm{\varphi}\meth^{i_5}\bm{\varphi}\\
&+\sum_{i_1+i_2+i_3+i_4+i_5=k}
\meth^{i_1}\Gamma(\tau,\pi)^{i_2}\meth^{i_3}\Gamma(\tau,\pi)\meth^{i_4}\Gamma(\tau,\pi)\meth^{i_5}\Gamma(\tau,\pi) \\
&+\sum_{i_1+i_2+i_3+i_4=k}
\meth^{i_1}\Gamma(\tau,\pi)^{i_2}\meth^{i_3}\Gamma(\tau,\pi,\rho,\sigma)\meth^{i_4}\tilde\ulomega.
\end{align*}
Now, we make use of the transport estimate in scale-invariant norm, to
find that
\begin{align*}
||(a^{\frac{1}{2}}\meth)^k\tilde\ulomega||_{L^2_{sc}(\mathcal{S}_{u,v})}&\lesssim
||(a^{\frac{1}{2}}\meth)^k\tilde\ulomega||_{L^2_{sc}(\mathcal{S}_{u,0})}+
\int_0^v||\mthorn(a^{\frac{1}{2}}\meth)^k\tilde\ulomega||_{L^2_{sc}(\mathcal{S}_{u,v'})} \\
&\leq||(a^{\frac{1}{2}}\mathcal{D})^k\tilde\ulomega||_{L^2_{sc}(\mathcal{S}_{u,0})}+
\int_0^va^{5}||\bm{\varphi}\mathcal{D}^{11}\bm{\varphi}||_{L^2_{sc}(\mathcal{S}_{u,v})}
 +\int_0^va^{5}||\Gamma\mathcal{D}^{11}\Gamma||_{L^2_{sc}(\mathcal{S}_{u,v})}\\
& +\int_0^va^{5}||\Gamma\mathcal{D}^{10}\TiPsi||_{L^2_{sc}(\mathcal{S}_{u,v})}
+\int_0^va^{5}||\Gamma\mathcal{D}^{10}\tilde\ulomega||_{L^2_{sc}(\mathcal{S}_{u,v})}
+\int_0^va^{5}||\sigma\mathcal{D}^{11}\ulomega||_{L^2_{sc}(\mathcal{S}_{u,v})} \\
&+\sum_{i_1+i_2+i_3+i_4=k,i_4<k}\int_0^v a^5
||\mathcal{D}^{i_1}\Gamma(\tau,\pi)^{i_2}\mathcal{D}^{i_3}\bm{\varphi}
\mathcal{D}^{i_4+1}\bm{\varphi}||_{L^2_{sc}(\mathcal{S}_{u,v})} \\
&+\sum_{i_1+i_2+i_3+i_4=k,i_4<k}\int_0^v a^5
||\mathcal{D}^{i_1}\Gamma(\tau,\pi)^{i_2}\mathcal{D}^{i_3}\Gamma
\mathcal{D}^{i_4+1}\Gamma||_{L^2_{sc}(\mathcal{S}_{u,v})} \\
&+\sum_{i_1+i_2+i_3+i_4=k,i_4<k}\int_0^v a^5
||\mathcal{D}^{i_1}\Gamma(\tau,\pi)^{i_2}\mathcal{D}^{i_3}\Gamma
\mathcal{D}^{i_4}\TiPsi||_{L^2_{sc}(\mathcal{S}_{u,v})} \\
&+\sum_{i_1+i_2+i_3+i_4+i_5=k}\int_0^v a^5
||\mathcal{D}^{i_1}\Gamma(\tau,\pi)^{i_2}\mathcal{D}^{i_3}\Gamma
\mathcal{D}^{i_4}\bm{\varphi}\mathcal{D}^{i_5}\bm{\varphi}||_{L^2_{sc}(\mathcal{S}_{u,v})} \\
&+\sum_{i_1+i_2+i_3+i_4+i_5=k}\int_0^v a^5
||\mathcal{D}^{i_1}\Gamma(\tau,\pi)^{i_2}\mathcal{D}^{i_3}\Gamma
\mathcal{D}^{i_4}\Gamma\mathcal{D}^{i_5}\Gamma||_{L^2_{sc}(\mathcal{S}_{u,v})} \\
&+\sum_{i_1+i_2+i_3+i_4=k}\int_0^v a^5
||\mathcal{D}^{i_1}\Gamma(\tau,\pi)^{i_2}\mathcal{D}^{i_3}\Gamma
\mathcal{D}^{i_4}\tilde\ulomega||_{L^2_{sc}(\mathcal{S}_{u,v})} \\
&\lesssim1+I_1+...+I_{11}.
\end{align*}

For the term~$I_1$ we have that 
\begin{align*}
I_1&\leq\int_0^va^{5}(||\varphi_0\mathcal{D}^{11}\varphi_2||_{L^2_{sc}(\mathcal{S}_{u,v})}
+||\varphi_2\mathcal{D}^{11}\varphi_0||_{L^2_{sc}(\mathcal{S}_{u,v})}
+||\varphi_1\mathcal{D}^{11}\varphi_1||_{L^2_{sc}(\mathcal{S}_{u,v})}) \\
\leq&\int_0^v\frac{1}{|u|}a^{5}||\varphi_0||_{L^{\infty}_{sc}(\mathcal{S}_{u,v})}
||\mathcal{D}^{11}\varphi_2||_{L^2_{sc}(\mathcal{S}_{u,v})}
+\int_0^v\frac{1}{|u|}a^{5}||\varphi_2||_{L^{\infty}_{sc}(\mathcal{S}_{u,v})}
||\mathcal{D}^{11}\varphi_0||_{L^2_{sc}(\mathcal{S}_{u,v})} \\
&+\int_0^v\frac{1}{|u|}a^{5}||\varphi_1||_{L^{\infty}_{sc}(\mathcal{S}_{u,v})}
||\mathcal{D}^{11}\varphi_1||_{L^2_{sc}(\mathcal{S}_{u,v})} \\
\leq&\frac{a^{\frac{1}{2}}}{|u|}\bmGamma(\varphi_0)
||a^{5}\mathcal{D}^{11}\varphi_2||_{L^2_{sc}(\mathcal{S}_{u,v})}
+\frac{\bmGamma(\varphi_2)}{a^{\frac{1}{2}}}\int_0^v
||a^5\mathcal{D}^{11}\varphi_0||_{L^2_{sc}(\mathcal{S}_{u,v})}\\
&+\frac{\bmGamma(\varphi_1)}{|u|}\int_0^v
||a^5\mathcal{D}^{11}\varphi_1||_{L^2_{sc}(\mathcal{S}_{u,v})} \\
&\leq\bmGamma(\varphi_0)\frac{a^{\frac{1}{2}}}{|u|}
||a^{5}\mathcal{D}^{11}\varphi_2||_{L^2_{sc}(\mathcal{S}_{u,v})}
+\frac{\bmGamma(\varphi_2)}{a^{\frac{1}{2}}}
||a^5\mathcal{D}^{11}\varphi_0||_{L^2_{sc}(\mathcal{N}_u(0,v))}
+\frac{\bmGamma(\varphi_1)}{|u|}
||a^5\mathcal{D}^{11}\varphi_1||_{L^2_{sc}(\mathcal{N}_u(0,v))} \\
&\lesssim
\bmGamma(\varphi_0)\frac{a^{\frac{1}{2}}}{|u|}||a^5\mathcal{D}^{11}\varphi_2||_{L^2_{sc}(\mathcal{S}_{u,v})}
+\bmGamma(\varphi_2)\bm\varphi[\varphi_0]
+\frac{\bmGamma(\varphi_1)}{|u|}\bm\varphi[\varphi_1] \\
&\leq
\frac{a^{\frac{1}{2}}}{|u|}||a^5\mathcal{D}^{11}\varphi_2||_{L^2_{sc}(\mathcal{S}_{u,v})}
+\bm\varphi[\varphi_0]
+\frac{\mathcal{O}^2}{a}.
\end{align*}
For~$I_2$ we have that
\begin{align*}
I_2 &\leq
\int_0^v\frac{a^{5}}{|u|}||\Gamma||_{L^{\infty}_{sc}(\mathcal{S}_{u,v})}
||\mathcal{D}^{11}\Gamma||_{L^2_{sc}(\mathcal{S}_{u,v})} 
\leq\frac{\mathcal{O}}{|u|}\int_0^v||a^5\mathcal{D}^{11}\Gamma||_{L^2_{sc}(\mathcal{S}_{u,v})} 
\leq\frac{\mathcal{O}}{|u|}||a^5\mathcal{D}^{11}\Gamma||_{L_{sc}^2(\mathcal{N}_u(0,v))} \\
&\leq\frac{\mathcal{O}}{a}\left(\frac{a}{|u|}||a^5\mathcal{D}^{11}\Gamma||_{L_{sc}^2(\mathcal{N}_u(0,v))}\right)
\leq\frac{\mathcal{O}^2}{a}\leq1,
\end{align*}
where in the previous expression~$\Gamma$ contains~$\tau$
and~$\pi$. For~$I_3$ we have that
\begin{align*}
I_3&\leq\int_0^v\frac{1}{|u|}||\Gamma||_{L^{\infty}_{sc}(\mathcal{S}_{u,v})}
||(a^{\frac{1}{2}}\mathcal{D})^{10}\TiPsi||_{L^2_{sc}(\mathcal{S}_{u,v})} 
\leq\int_0^v\frac{1}{|u|}\frac{|u|\mathcal{O}}{a^{\frac{1}{2}}}
||(a^{\frac{1}{2}}\mathcal{D})^{10}\TiPsi||_{L^2_{sc}(\mathcal{S}_{u,v})} \\
&\leq\frac{\mathcal{O}}{a^{\frac{1}{2}}}||(a^{\frac{1}{2}}\mathcal{D})^{10}\TiPsi||_{L_{sc}^2(\mathcal{N}_u(0,v))}
\leq1,
\end{align*}
where~$\Gamma$ contains~$\tau$, $\pi$, $\rho$, $\sigma$
and~$\lambda$. For~$I_4$ we have that
\begin{align*}
I_4\leq\int_0^v\frac{1}{|u|}||\Gamma||_{L^{\infty}_{sc}(\mathcal{S}_{u,v})}
||(a^{\frac{1}{2}}\mathcal{D})^{10}\tilde\ulomega||_{L^2_{sc}(\mathcal{S}_{u,v})} 
\leq\int_0^v\frac{a^{\frac{1}{2}}\mathcal{O}}{|u|}||(a^{\frac{1}{2}}\mathcal{D})^{10}\tilde\ulomega||_{L^2_{sc}(\mathcal{S}_{u,v})} ,
\end{align*}
which can be absorbed making use of the Gr\"onwall inequality. In the
last expression~$\bm{\Gamma}$ is meant to contain~$\rho$
and~$\sigma$. For~$I_5$ we have that
\begin{align*}
I_5\leq
\int_0^v\frac{1}{|u|}||\sigma||_{L^{\infty}_{sc}(\mathcal{S}_{u,v})}
||a^{5}\mathcal{D}^{11}\ulomega||_{L^2_{sc}(\mathcal{S}_{u,v})} 
\leq\frac{a^{\frac{1}{2}}\mathcal{O}}{|u|}||a^{5}\mathcal{D}^{11}\ulomega||_{L^2_{sc}(\mathcal{S}_{u,v})}.
\end{align*}
For~$I_6$ we have that 
\begin{align*}
I_6&\leq\frac{(a^{\frac{1}{2}})^{i_2-1}}{|u|^{i_2+1}}||(a^{\frac{1}{2}}\mathcal{D})^{j_1}\Gamma||...
||(a^{\frac{1}{2}}\mathcal{D})^{j_{i_2}}\Gamma||\times
||(a^{\frac{1}{2}}\mathcal{D})^{i_3}\bm\varphi||\times||(a^{\frac{1}{2}}\mathcal{D})^{i_4+1}\bm\varphi||  \\
&\leq\frac{(a^{\frac{1}{2}})^{i_2-1}}{|u|^{i_2+1}}\mathcal{O}^{i_2}a^{\frac{1}{2}}\mathcal{O}\frac{|u|}{a^{\frac{1}{2}}}\mathcal{O}
=\frac{a^{\frac{i_2}{2}}}{a^{\frac{1}{2}}|u|^{i_2}}\mathcal{O}^{i_2+2}
\leq\frac{1}{|u|}\mathcal{O}^{3}\leq1.
\end{align*}
For~$I_7$ we have that 
\begin{align*}
I_7&\leq\frac{(a^{\frac{1}{2}})^{i_2-1}}{|u|^{i_2+1}}||(a^{\frac{1}{2}}\mathcal{D})^{j_1}\Gamma||...
||(a^{\frac{1}{2}}\mathcal{D})^{j_{i_2}}\Gamma||\times
||(a^{\frac{1}{2}}\mathcal{D})^{i_3}\Gamma(\sigma,...)||\times||(a^{\frac{1}{2}}\mathcal{D})^{i_4+1}\Gamma||  \\
&\leq\frac{(a^{\frac{1}{2}})^{i_2-1}}{|u|^{i_2+1}}\mathcal{O}^{i_2}a^{\frac{1}{2}}\mathcal{O}\mathcal{O}
=\frac{a^{\frac{i_2}{2}}}{|u|^{i_2+1}}\mathcal{O}^{i_2+2}
\leq\frac{a}{|u|^2}\mathcal{O}^3\leq1.
\end{align*}
For~$I_8$ we have that 
\begin{align*}
I_8&\leq\frac{(a^{\frac{1}{2}})^{i_2-1}}{|u|^{i_2+1}}||(a^{\frac{1}{2}}\mathcal{D})^{j_1}\Gamma||...
||(a^{\frac{1}{2}}\mathcal{D})^{j_{i_2}}\Gamma||\times
||(a^{\frac{1}{2}}\mathcal{D})^{i_3}\Gamma(\lambda,...)||\times||(a^{\frac{1}{2}}\mathcal{D})^{i_4}\bm\Psi||  \\
&\leq\frac{(a^{\frac{1}{2}})^{i_2-1}}{|u|^{i_2+1}}\mathcal{O}^{i_2}\frac{|u|}{a^{\frac{1}{2}}}\mathcal{O}\mathcal{O}
=\frac{a^{\frac{i_2-2}{2}}}{|u|^{i_2}}\mathcal{O}^{i_2+2}
\leq\frac{1}{a|u|}\mathcal{O}^2\leq\frac{\mathcal{O}^2}{a^{\frac{3}{2}}}\leq1.
\end{align*}
For~$I_9$ the worst behaved terms
contain~$\mu\meth^i\varphi_0\meth^j\varphi_1$. It follows then that
\begin{align*}
I_9&\leq a^5||\mathcal{D}^{j_1}\Gamma...\mathcal{D}^{j_{i_2}}\Gamma
\mathcal{D}^{i_3}\Gamma(\mu,..)\mathcal{D}^{i_4}\varphi\mathcal{D}^{i_5}\varphi|| \\
&\leq\frac{(a^{\frac{1}{2}})^{i_2}}{|u|^{i_2+2}}||(a^{\frac{1}{2}}\mathcal{D})^{j_1}\Gamma||...
||(a^{\frac{1}{2}}\mathcal{D})^{j_{i_2}}\Gamma||\times
||(a^{\frac{1}{2}}\mathcal{D})^{i_3}\Gamma||\times||(a^{\frac{1}{2}}\mathcal{D})^{i_4}\varphi||
\times||(a^{\frac{1}{2}}\mathcal{D})^{i_5}\varphi||  \\
&\leq\frac{1}{a^{\frac{1}{2}}}\bmGamma(\varphi_0)\bmGamma(\varphi_1)+
\frac{(a^{\frac{1}{2}})^{i_2}}{|u|^{i_2+2}}\mathcal{O}^{i_2}
\frac{|u|}{a}a^{\frac{1}{2}}\mathcal{O}\frac{|u|}{a}\mathcal{O}
=\frac{1}{a^{\frac{1}{2}}}\bmGamma(\varphi_0)\bmGamma(\varphi_1)+\frac{a^{\frac{i_2}{2}}}{a^{\frac{3}{2}}|u|^{i_2}}\mathcal{O}^{i_2+3}
\leq1.
\end{align*}
For~$I_{10}$ we have that 
\begin{align*}
I_{10}&\leq\frac{(a^{\frac{1}{2}})^{i_2}}{|u|^{i_2+2}}||(a^{\frac{1}{2}}\mathcal{D})^{j_1}\Gamma||...
||(a^{\frac{1}{2}}\mathcal{D})^{j_{i_2}}\Gamma||\times
||(a^{\frac{1}{2}}\mathcal{D})^{i_3}\Gamma||\times||(a^{\frac{1}{2}}\mathcal{D})^{i_4}\Gamma||
\times||(a^{\frac{1}{2}}\mathcal{D})^{i_5}\Gamma||  \\
&\leq\frac{(a^{\frac{1}{2}})^{i_2}}{|u|^{i_2+2}}\mathcal{O}^{i_2}\mathcal{O}\mathcal{O}\mathcal{O}
=\frac{a^{\frac{i_2}{2}}}{|u|^{i_2+2}}\mathcal{O}^{i_2+3}
\leq\frac{1}{|u|^2}\mathcal{O}^3\leq1.
\end{align*}
Finally, for~$I_{11}$ we see that 
\begin{align*}
I_{11}&\leq\frac{(a^{\frac{1}{2}})^{i_2}}{|u|^{i_2+1}}||(a^{\frac{1}{2}}\mathcal{D})^{j_1}\Gamma||...
||(a^{\frac{1}{2}}\mathcal{D})^{j_{i_2}}\Gamma||\times
||(a^{\frac{1}{2}}\mathcal{D})^{i_3}\Gamma(\sigma,...)||\times||(a^{\frac{1}{2}}\mathcal{D})^{i_4}\tilde\ulomega||  \\
&\leq\frac{(a^{\frac{1}{2}})^{i_2}}{|u|^{i_2+1}}\mathcal{O}^{i_2}\mathcal{O}a^{\frac{1}{2}}\mathcal{O}
=\frac{a^{\frac{i_2+1}{2}}}{|u|^{i_2+1}}\mathcal{O}^{i_2+2}
\leq\frac{a^{\frac{1}{2}}}{|u|}\mathcal{O}^2\leq1.
\end{align*}

Combining the above estimates for~$I_1,\ldots, I_{11}$ we have that 
\begin{align*}
||(a^{\frac{1}{2}}\meth)^{10}\tilde\ulomega||_{L^2_{sc}(\mathcal{S}_{u,v})}
\lesssim\frac{a^{\frac{1}{2}}}{|u|}||a^5\mathcal{D}^{11}\varphi_2||_{L^2_{sc}(\mathcal{S}_{u,v})}
+\bm\varphi[\varphi_0]+1.
\end{align*}
Following the same strategy one can estimate the remaining strings of
derivatives in~$\mathcal{D}^{k_{i}}\tilde\ulomega$. This yields the
same result. Hence, one concludes that
\begin{align*}
||(a^{\frac{1}{2}}\mathcal{D})^{10}\tilde\ulomega||_{L^2_{sc}(\mathcal{S}_{u,v})}
\lesssim\frac{a^{\frac{1}{2}}}{|u|}||a^5\mathcal{D}^{11}\varphi_2||_{L^2_{sc}(\mathcal{S}_{u,v})}
+\bm\varphi[\varphi_0]+1.
\end{align*}

\smallskip
\noindent
\textbf{Step 3}: Make use of the elliptic estimate.

The estimate of the tenth derivative of~$\tilde\ulomega$ implies an
estimate for the ninth derivative of~$\bm\Delta\ulomega$, where
\begin{align*}
\bm\Delta\ulomega\equiv \nablasl_a\nablasl^a\ulomega=2\meth'\meth\ulomega.
\end{align*}
Making use of Proposition~\ref{EllipticPureScalar} we have that 
\begin{align*}
||a^5\mathcal{D}^{11}\ulomega||_{L^2_{sc}(\mathcal{S}_{u,v})}&\lesssim
||a^5\mathcal{D}^9(\meth'\meth\ulomega)||_{L^2_{sc}(\mathcal{S}_{u,v})}
+\sum_{i=0}^{10}\frac{1}{a^{\frac{1}{2}}}||(a^{\frac{1}{2}}\mathcal{D})^i\ulomega||_{L^2_{sc}(\mathcal{S}_{u,v})} \\
&\lesssim||a^5\mathcal{D}^9(\meth'\meth\ulomega+i\meth'\meth\ulomega^{\dag})||_{L^2_{sc}(\mathcal{S}_{u,v})} 
+1 \\
&\leq||(a^{\frac{1}{2}}\mathcal{D})^{10}\tilde\ulomega||_{L^2_{sc}(\mathcal{S}_{u,v})}+||a^5\mathcal{D}^{10}\TiPsi_3||_{L^2_{sc}(\mathcal{S}_{u,v})}+1 \\
&\lesssim\frac{a^{\frac{1}{2}}}{|u|}||a^5\mathcal{D}^{11}\varphi_2||_{L^2_{sc}(\mathcal{S}_{u,v})}
+\bm\varphi[\varphi_0]
+||(a^{\frac{1}{2}}\mathcal{D})^{10}\TiPsi_3||_{L^2_{sc}(\mathcal{S}_{u,v})}+1.
\end{align*}
Integrating along the ingoing lightcone we have that 
\begin{align*}
||a^5\mathcal{D}^{11}\ulomega||_{L^2_{sc}(\mathcal{N}'_v(u_{\infty},u))}&\lesssim
\underline{\bm\varphi}[\varphi_2]+\bm\varphi[\varphi_0]
+\underline{\bm\Psi}[\TiPsi_3]+1. 
\end{align*}
\end{proof}

\begin{proposition}
\label{11Dermulambda}
We have that 
\begin{align*}
\int_{u_{\infty}}^u\frac{a^2}{|u'|^3}||a^5\mathcal{D}^{11}\mu||_{L^2_{sc}(\mathcal{S}_{u,v})}
&\lesssim\bm\varphi[\varphi_0]+\underline{\bm\varphi}[\varphi_2]
+\underline{\bm\Psi}[\TiPsi_3]+1, \\
\int_{u_{\infty}}^u\frac{a^{\frac{3}{2}}}{|u'|^3}||a^5\mathcal{D}^{11}\lambda||_{L^2_{sc}(\mathcal{S}_{u,v})}
&\lesssim1,
\end{align*}
and that 
\begin{align*}
\int_{u_{\infty}}^u\frac{a^3}{|u'|^4}||a^5\mathcal{D}^{11}\Timu||^2_{L^2_{sc}(\mathcal{S}_{u,v})}
&\lesssim(\bm\varphi[\varphi_0]+\underline{\bm\varphi}[\varphi_2]
+\underline{\bm\Psi}[\TiPsi_3]+1)^2, \\
\int_{u_{\infty}}^u\frac{a^2}{|u'|^4}||a^5\mathcal{D}^{11}\lambda||^2_{L^2_{sc}(\mathcal{S}_{u,v})}&\lesssim1.
\end{align*}
\end{proposition}

\begin{proof}

\smallskip
\noindent
\textbf{Step 1}: Estimate of~$\Timu$, $\mu$.

We make use of 
\begin{align*}
\mthorn'\Timu+2\mu\Timu=\Timu^2-\ulomega\mu-\lambda\bar\lambda-2\varphi_2^2
\end{align*}
and then commute~$\meth^k$ with~$\mthorn'$ for~$k=11$. We find that 
\begin{align*}
\mthorn'\meth^k\Timu+(k+2)\mu\meth^k\Timu=&\Gamma(\Timu,\lambda)\meth^{11}\Timu
+\Gamma(\Timu,\lambda)\meth^{11}\lambda+\mu\meth^{11}\ulomega
+\ulomega\meth^{11}\Timu+\varphi_2\meth^{11}\varphi_2\\
&+\sum_{i=1}^{k-1}\meth^i\Gamma(\Timu,\lambda,\ulomega)\meth^{k-i}\Gamma(\Timu,\lambda,\ulomega)
+\sum_{i=1}^{k-1}\meth^i\varphi_2\meth^{k-i}\varphi_2 .
\end{align*}
Denote the right-hand side of the previous expression by~$F$. It
follows then that
\begin{align*}
&s_2(\Timu)=1, \ \ s_2(\meth^k\Timu)=\frac{k+2}{2}, \ \ s_2(F)=\frac{k+4}{2}, \\
&\lambda_0=-(k+2), \ \ \lambda_1=-\lambda_0-1=k+1, \\
& \lambda_1-2s_2(F)=-3, \ \ \lambda_1-2s_2(\meth^k\Timu)=-1.
\end{align*}
Now, we make use of Proposition~\ref{SCweightedgronwall}, so as to
obtain 
\begin{align*}
\frac{a}{|u|}||(a^{\frac{1}{2}}\meth)^k\Timu||_{L^2_{sc}(\mathcal{S}_{u,v})}\lesssim
\frac{a}{|u_{\infty}|}||(a^{\frac{1}{2}}\meth)^k\Timu||_{L^2_{sc}(\mathcal{S}_{u_{\infty},v})}
+\int_{u_{\infty}}^u\frac{a^2}{|u'|^3}||a^{\frac{k}{2}}F||_{L^2_{sc}(\mathcal{S}_{u',v})}.
\end{align*}
Multiplying by~$a^{-\frac{1}{2}}$ one concludes that  
\begin{align*}
\frac{a}{|u|}||a^5\meth^{11}\Timu||_{L^2_{sc}(\mathcal{S}_{u,v})}&\lesssim
\frac{a}{|u_{\infty}|}||a^5\mathcal{D}^{11}\Timu||_{L^2_{sc}(\mathcal{S}_{u_{\infty},v})}
+\int_{u_{\infty}}^u\frac{a^2}{|u'|^3}||a^{5}F||_{L^2_{sc}(\mathcal{S}_{u',v})} \\
&\leq\frac{a}{|u_{\infty}|}||a^5\mathcal{D}^{11}\Timu||_{L^2_{sc}(\mathcal{S}_{u_{\infty},v})}
+\int_{u_{\infty}}^u\frac{a^2}{|u'|^3}||a^{5}\Gamma(\Timu,\lambda,\ulomega)
\mathcal{D}^{11}\Timu||_{L^2_{sc}(\mathcal{S}_{u',v})} \\
&+\int_{u_{\infty}}^u\frac{a^2}{|u'|^3}||a^{5}\Gamma(\Timu,\lambda)\mathcal{D}^{11}\lambda||_{L^2_{sc}(\mathcal{S}_{u',v})} \\
&+\int_{u_{\infty}}^u\frac{a^2}{|u'|^3}||a^{5}\mu\mathcal{D}^{11}\ulomega||_{L^2_{sc}(\mathcal{S}_{u',v})} 
+\int_{u_{\infty}}^u\frac{a^2}{|u'|^3}||a^{5}\varphi_2\mathcal{D}^{11}\varphi_2||_{L^2_{sc}(\mathcal{S}_{u',v})} \\
&+\sum_{i=1}^{k-1}\int_{u_{\infty}}^u\frac{a^2}{|u'|^3}
||a^{5}\mathcal{D}^i\Gamma(\Timu,\lambda,\ulomega)
\mathcal{D}^{k-i}\Gamma(\Timu,\lambda,\ulomega)||_{L^2_{sc}(\mathcal{S}_{u',v})} \\
&+\sum_{i=1}^{k-1}\int_{u_{\infty}}^u\frac{a^2}{|u'|^3}
||a^{5}\mathcal{D}^i\varphi_2
\mathcal{D}^{k-i}\varphi_2||_{L^2_{sc}(\mathcal{S}_{u',v})} 
\lesssim1+I_1+...+I_6.
\end{align*}

For the term~$I_1$ we have that 
\begin{align*}
I_1&\leq\int_{u_{\infty}}^u\frac{a^2}{|u'|^3}\frac{1}{|u'|}||\Gamma||_{L^{\infty}_{sc}(\mathcal{S}_{u',v})}
||a^5\mathcal{D}^{11}\Timu||_{L^2_{sc}(\mathcal{S}_{u',v})} \\
&\leq\int_{u_{\infty}}^u\frac{a^2}{|u'|^3}\frac{1}{|u'|}\frac{|u'|}{a^{\frac{1}{2}}}\mathcal{O}
||a^5\mathcal{D}^{11}\Timu||_{L^2_{sc}(\mathcal{S}_{u',v})} 
=\int_{u_{\infty}}^u\frac{a^{\frac{1}{2}}\mathcal{O}}{|u'|^2}\frac{a}{|u'|}||a^5\mathcal{D}^{11}\Timu||_{L^2_{sc}(\mathcal{S}_{u',v})} .
\end{align*}
This term can be absorbed using the Gr\"onwall inequality. For~$I_2$
we have that
\begin{align*}
I_2&\leq\int_{u_{\infty}}^u\frac{a^2}{|u'|^3}\frac{1}{|u'|}||\Gamma||_{L^{\infty}_{sc}(\mathcal{S}_{u',v})}
||a^5\mathcal{D}^{11}\lambda||_{L^2_{sc}(\mathcal{S}_{u',v})} \\
&\leq\int_{u_{\infty}}^u\frac{a^2}{|u'|^3}\frac{1}{|u'|}\frac{|u'|}{a^{\frac{1}{2}}}
||a^5\mathcal{D}^{11}\lambda||_{L^2_{sc}(\mathcal{S}_{u',v})} 
=\int_{u_{\infty}}^u\frac{a^{\frac{3}{2}}}{|u'|^3}
||a^5\mathcal{D}^{11}\lambda||_{L^2_{sc}(\mathcal{S}_{u',v})} .
\end{align*}
For~$I_3$ we have that 
\begin{align*}
I_3&\leq\int_{u_{\infty}}^u\frac{a^2}{|u'|^3}\frac{1}{|u'|}||\mu||_{L^{\infty}_{sc}(\mathcal{S}_{u',v})}
||a^5\mathcal{D}^{11}\ulomega||_{L^2_{sc}(\mathcal{S}_{u',v})}  \\
&\leq\int_{u_{\infty}}^u\frac{a^2}{|u'|^3}\frac{1}{|u'|}\frac{|u'|^2}{a}
||a^5\mathcal{D}^{11}\ulomega||_{L^2_{sc}(\mathcal{S}_{u',v})} 
=\int_{u_{\infty}}^u\frac{a}{|u'|^2}||a^5\mathcal{D}^{11}\ulomega||_{L^2_{sc}(\mathcal{S}_{u',v})} \\
&\leq||a^5\mathcal{D}^{11}\ulomega||_{L^2_{sc}(\mathcal{N}'_v(u_{\infty},u))}.
\end{align*}
For~$I_4$ we have that 
\begin{align*}
I_4&\leq\int_{u_{\infty}}^u\frac{a^2}{|u'|^3}\frac{1}{|u'|}||\varphi_2||_{L^{\infty}_{sc}(\mathcal{S}_{u',v})}
||a^5\mathcal{D}^{11}\varphi_2||_{L^2_{sc}(\mathcal{S}_{u',v})}  \\
&\leq\int_{u_{\infty}}^u\frac{a^2}{|u'|^4}\frac{|u'|}{a^{\frac{1}{2}}}\bmGamma(\varphi_2)
||a^5\mathcal{D}^{11}\varphi_2||_{L^2_{sc}(\mathcal{S}_{u',v})} 
=\bmGamma(\varphi_2)\int_{u_{\infty}}^u\frac{a}{|u'|^2}
\frac{a^{\frac{1}{2}}}{|u'|}||a^5\mathcal{D}^{11}\varphi_2||_{L^2_{sc}(\mathcal{S}_{u',v})} \\
&\leq\bmGamma(\varphi_2)\left(\int_{u_{\infty}}^u\frac{a^2}{|u'|^4}
||a^5\mathcal{D}^{11}\varphi_2||^2_{L^2_{sc}(\mathcal{S}_{u',v})} \right)^{\frac{1}{2}}
\left(\int_{u_{\infty}}^u\frac{a}{|u'|^2}\right)^{\frac{1}{2}} \\
&\leq\frac{\bmGamma(\varphi_2)a^{\frac{1}{2}}}{|u|^{\frac{1}{2}}}
||\frac{a^{\frac{1}{2}}}{|u'|}a^5\mathcal{D}^{11}\varphi_2||_{L^2_{sc}(\mathcal{N}'_v(u_{\infty},u))}
\lesssim\bmGamma(\varphi_2)\underline{\bm\varphi}[\varphi_2]
\lesssim\underline{\bm\varphi}[\varphi_2].
\end{align*}
For~$I_5$ we see that 
\begin{align*}
I_5&\leq\int_{u_{\infty}}^u\frac{a^{\frac{3}{2}}}{|u'|^3}\frac{1}{|u'|}||(a^{\frac{1}{2}}\mathcal{D})^i\Gamma||\times
||(a^{\frac{1}{2}}\mathcal{D})^{k-i}\Gamma|| \\
&\leq\int_{u_{\infty}}^u\frac{a^{\frac{3}{2}}}{|u'|^3}\frac{1}{|u'|}\frac{|u'|\mathcal{O}}{a^{\frac{1}{2}}}
\frac{|u'|\mathcal{O}}{a^{\frac{1}{2}}}
=\mathcal{O}^2\int_{u_{\infty}}^u\frac{a^{\frac{1}{2}}}{|u'|^2}\leq\frac{\mathcal{O}^2}{a^{\frac{1}{2}}}\leq1.
\end{align*}
For~$I_6$ we have that 
\begin{align*}
I_6&\leq\int_{u_{\infty}}^u\frac{a^{\frac{3}{2}}}{|u'|^3}\frac{1}{|u'|}||(a^{\frac{1}{2}}\mathcal{D})^i\varphi_2||\times
||(a^{\frac{1}{2}}\mathcal{D})^{k-i}\varphi_2|| \\
&\leq\int_{u_{\infty}}^u\frac{a^{\frac{3}{2}}}{|u'|^3}\frac{1}{|u'|}\frac{|u'|\mathcal{O}}{a^{\frac{1}{2}}}
\frac{|u'|\mathcal{O}}{a^{\frac{1}{2}}}
=\mathcal{O}^2\int_{u_{\infty}}^u\frac{a^{\frac{1}{2}}}{|u'|^2}\leq\frac{\mathcal{O}^2}{a^{\frac{1}{2}}}\leq1.
\end{align*}

Combining the estimates for~$I_1,\ldots, I_6$ we have that 
\begin{align*}
\frac{a}{|u|}||a^5\meth^{11}\Timu||_{L^2_{sc}(\mathcal{S}_{u,v})}\lesssim
\int_{u_{\infty}}^u\frac{a^{\frac{3}{2}}}{|u'|^3}
||a^5\mathcal{D}^{11}\lambda||_{L^2_{sc}(\mathcal{S}_{u',v})}
+||a^5\mathcal{D}^{11}\ulomega||_{L^2_{sc}(\mathcal{N}'_v(u_{\infty},u))}+\underline{\bm\varphi}[\varphi_2]+1.
\end{align*}
Following the same strategy one can estimate the remaining strings of
derivatives in~$\mathcal{D}^{k_{i}}\tilde\ulomega$ and obtain the same
result. This leads us to conclude that
\begin{align}
\label{11DerivmuMid}
\frac{a}{|u|}||a^5\mathcal{D}^{11}\Timu||_{L^2_{sc}(\mathcal{S}_{u,v})}\lesssim
\int_{u_{\infty}}^u\frac{a^{\frac{3}{2}}}{|u'|^3}
||a^5\mathcal{D}^{11}\lambda||_{L^2_{sc}(\mathcal{S}_{u',v})}+
||a^5\mathcal{D}^{11}\ulomega||_{L^2_{sc}(\mathcal{N}'_v(u_{\infty},u))}+\underline{\bm\varphi}[\varphi_2]+1.
\end{align}

\smallskip
\noindent
\textbf{Step 2}: Elliptic estimates of~$\lambda$.

To estimate~$\lambda$ we make use of the structure equation
\eqref{T-weightMasslessStructureStEq15} ---namely
\begin{align*}
\meth\lambda-\meth'\mu=\bar\tau\mu-\tau\lambda-\TiPsi_3.
\end{align*}
As the T-weight of~$\lambda$ is~$2$, it follows
that~$\mathscr{D}_{\lambda}=\meth\lambda$. Accordingly, the above
structure equation can be rewritten as
\begin{align*}
\mathscr{D}_{\lambda}=\meth'\mu+\bar\tau\mu-\tau\lambda-\TiPsi_3.
\end{align*}

Now, applying Proposition~\ref{EllipticT-weightSC}, we have that
\begin{align*}
||a^5\mathcal{D}^{11}\lambda||_{L^2_{sc}(\mathcal{S}_{u,v})}&\lesssim\sum_{i\leq10}
(\frac{1}{a^{\frac{1}{2}}}||(a^{\frac{1}{2}}\mathcal{D})^{i}\lambda||_{L^2_{sc}(\mathcal{S}_{u,v})}
+||(a^{\frac{1}{2}}\mathcal{D})^{i}\mathscr{D}_{\lambda}||_{L^2_{sc}(\mathcal{S}_{u,v})}) \\
&\lesssim||a^5\mathcal{D}^{11}\mu||_{L^2_{sc}(\mathcal{S}_{u,v})}+
\sum_{i=0}^{10}||(a^{\frac{1}{2}}\mathcal{D})^{i}\TiPsi_3||_{L^2_{sc}(\mathcal{S}_{u,v})} \\
&+\sum_{i=0}^{10}||(a^{\frac{1}{2}}\mathcal{D})^{i}\mu(a^{\frac{1}{2}}\mathcal{D})^{i}\tau||_{L^2_{sc}(\mathcal{S}_{u,v})} 
+\sum_{i=0}^{10}||(a^{\frac{1}{2}}\mathcal{D})^{i}\lambda(a^{\frac{1}{2}}\mathcal{D})^{i}\tau||_{L^2_{sc}(\mathcal{S}_{u,v})}+1.
\end{align*}
Hence, integrating along the ingoing lightcone we conclude that 
\begin{align*}
\int_{u_{\infty}}^u\frac{a^{\frac{3}{2}}}{|u'|^3}||a^5\mathcal{D}^{11}\lambda||_{L^2_{sc}(\mathcal{S}_{u,v})}&\lesssim
\int_{u_{\infty}}^u\frac{a^{\frac{3}{2}}}{|u'|^3}||a^5\mathcal{D}^{11}\mu||_{L^2_{sc}(\mathcal{S}_{u,v})}
+\sum_{i=0}^{10}\int_{u_{\infty}}^u\frac{a^{\frac{3}{2}}}{|u'|^3}
||(a^{\frac{1}{2}}\mathcal{D})^{i}\TiPsi_3||_{L^2_{sc}(\mathcal{S}_{u,v})} \\
&+\sum_{i=0}^{10}\int_{u_{\infty}}^u\frac{a^{\frac{3}{2}}}{|u'|^3}
||(a^{\frac{1}{2}}\mathcal{D})^{i}\mu(a^{\frac{1}{2}}\mathcal{D})^{i}\tau||_{L^2_{sc}(\mathcal{S}_{u,v})}  \\
&+\sum_{i=0}^{10}\int_{u_{\infty}}^u\frac{a^{\frac{3}{2}}}{|u'|^3}
||(a^{\frac{1}{2}}\mathcal{D})^{i}\lambda(a^{\frac{1}{2}}\mathcal{D})^{i}\tau||_{L^2_{sc}(\mathcal{S}_{u,v})}+1 \\
&=\int_{u_{\infty}}^u\frac{a^{\frac{3}{2}}}{|u'|^3}||a^5\mathcal{D}^{11}\mu||_{L^2_{sc}(\mathcal{S}_{u,v})}
+I_1+I_2+I_3+1.
\end{align*}

Now, for the second integral in the last inequality we have that 
\begin{align*}
I_1&\leq\sum_{i=0}^{10}\left(\int_{u_{\infty}}^u\frac{a}{|u'|^2}
||(a^{\frac{1}{2}}\mathcal{D})^{i}\TiPsi_3||^2_{L^2_{sc}(\mathcal{S}_{u,v})}  \right)^{\frac{1}{2}}
\left(\int_{u_{\infty}}^u\frac{a^2}{|u'|^4} \right)^{\frac{1}{2}} \\
&\leq\sum_{i=0}^{10}\frac{a}{|u|^{\frac{3}{2}}}||(a^{\frac{1}{2}}\mathcal{D})^{i}\TiPsi_3||_{L^2_{sc}(\mathcal{N}_v(u_{\infty},u))}
\leq1.
\end{align*}
For the third and fourth integrals we see that 
\begin{align*}
I_2+I_3\leq\int_{u_{\infty}}^u\frac{a^{\frac{3}{2}}}{|u'|^3}\frac{1}{|u'|}
\frac{|u'|^2}{a}\mathcal{O}\mathcal{O}
\leq\frac{a^{\frac{1}{2}}}{|u|}\mathcal{O}\leq1.
\end{align*}
Hence, we conclude that 
\begin{align}
\label{11DerivlambdaMid}
\int_{u_{\infty}}^u\frac{a^{\frac{3}{2}}}{|u'|^3}||a^5\mathcal{D}^{11}\lambda||_{L^2_{sc}(\mathcal{S}_{u,v})}&\lesssim
\int_{u_{\infty}}^u\frac{a^{\frac{3}{2}}}{|u'|^3}||a^5\mathcal{D}^{11}\mu||_{L^2_{sc}(\mathcal{S}_{u,v})}+1 .
\end{align}
Substituting~\eqref{11DerivmuMid} and apply Gr\"onwall's inequality we
have that
\begin{align}
\label{11DerivmuMid01}
\frac{a}{|u|}||a^5\mathcal{D}^{11}\Timu||_{L^2_{sc}(\mathcal{S}_{u,v})}\lesssim&
||a^5\mathcal{D}^{11}\ulomega||_{L^2_{sc}(\mathcal{N}'_v(u_{\infty},u))}+\underline{\bm\varphi}[\varphi_2]+1 \nonumber\\
\lesssim&\bm\varphi+\underline{\bm\varphi}
+\underline{\bm\Psi}[\TiPsi_3]+1,
\end{align}
so that integrating along the ingoing lightcone we end up with
\begin{align*}
\int_{u_{\infty}}^u\frac{a^2}{|u'|^3}||a^5\mathcal{D}^{11}\Timu||_{L^2_{sc}(\mathcal{S}_{u,v})}
\lesssim\bm\varphi+\underline{\bm\varphi}
+\underline{\bm\Psi}[\TiPsi_3]+1.
\end{align*}
Now, we observe that for spherical derivatives, the norm of~$\mu$
coincides with that of~$\Timu$. Hence, we have that  
\begin{align*}
\int_{u_{\infty}}^u\frac{a^2}{|u'|^3}||a^5\mathcal{D}^{11}\mu||_{L^2_{sc}(\mathcal{S}_{u,v})}
\lesssim\bm\varphi+\underline{\bm\varphi}
+\underline{\bm\Psi}[\TiPsi_3]+1.
\end{align*}
Accordingly, observing~\eqref{11DerivlambdaMid}, we see that 
\begin{align*}
\int_{u_{\infty}}^u\frac{a^{\frac{3}{2}}}{|u'|^3}||a^5\mathcal{D}^{11}\lambda||_{L^2_{sc}(\mathcal{S}_{u,v})}
\lesssim1.
\end{align*}
Moreover, from~\eqref{11DerivmuMid01} we have that 
\begin{align*}
\int_{u_{\infty}}^u\frac{a^3}{|u'|^4}||a^5\mathcal{D}^{11}\Timu||^2_{L^2_{sc}(\mathcal{S}_{u,v})}
\lesssim(\bm\varphi+\underline{\bm\varphi}
+\underline{\bm\Psi}[\TiPsi_3]+1)^2.
\end{align*}

Following the same strategy as in the previous paragraphs we finally
conclude that
\begin{align*}
\int_{u_{\infty}}^u\frac{a^2}{|u'|^4}||a^5\mathcal{D}^{11}\lambda||^2_{L^2_{sc}(\mathcal{S}_{u,v})}&\lesssim
\int_{u_{\infty}}^u\frac{a^2}{|u'|^4}||a^5\mathcal{D}^{11}\Timu||^2_{L^2_{sc}(\mathcal{S}_{u,v})}
+\int_{u_{\infty}}^u\frac{a^2}{|u'|^4}||(a^{\frac{1}{2}}\mathcal{D})^{10}\TiPsi_3||^2_{L^2_{sc}(\mathcal{S}_{u,v})}+1\\
&\lesssim\frac{1}{a}(\bm\varphi+\underline{\bm\varphi}
+\underline{\bm\Psi}[\TiPsi_3]+1)^2+\frac{a}{|u|^2}(\underline{\bm\Psi}[\TiPsi_3]^2+1) \\
&\leq1.
\end{align*}

\end{proof}

\subsection{Summary of the analysis of the top derivative of the
  connection}

In this section we have made use of the results in Section \ref{L2estimate} and the bootstrap assumption

\begin{align*}
\bm\Gamma_{11},\bm\varphi,\bm\Psi\leq\mathcal{O}, 
\end{align*}
to conclude with the estimates
\begin{align*}
||a^5\mathcal{D}^{11}\rho||_{L^2_{sc}(\mathcal{S}_{u,v})}&\lesssim\bm{\varphi}[\varphi_0]+\bm\Psi[\Psi_0]+\bm\Psi[\TiPsi_1]+1,\\
||a^5\mathcal{D}^{11}\rho||_{L^2_{sc}(\mathcal{N}_u(0,v))}&\lesssim\bm{\varphi}[\varphi_0]+\bm\Psi[\Psi_0]+\bm\Psi[\TiPsi_1]+1,\\
||a^5\mathcal{D}^{11}\sigma||_{L^2_{sc}(\mathcal{N}_u(0,v))}
&\lesssim\bm{\varphi}[\varphi_0]+\bm\Psi[\Psi_0]+\bm\Psi[\TiPsi_1]+1, \\
||a^5\mathcal{D}^{11}\tau||_{L^2_{sc}(\mathcal{N}_u(0,v))}
&\lesssim\bm\Psi[\Psi_0]+\bm\Psi[\TiPsi_2]+\underline{\bm\varphi}[\varphi_1]+1,\\
||a^5\mathcal{D}^{11}\tau||_{L^2_{sc}(\mathcal{N}'_v(u_{\infty},u))}
&\lesssim\bm\Psi[\Psi_0]+\underline{\bm\Psi}[\TiPsi_2]+\underline{\bm\varphi}[\varphi_1]+1, \\
\frac{a}{|u|}||a^5\mathcal{D}^{11}\pi||_{L^2_{sc}(\mathcal{N}_{u}(0,v))}&\lesssim
\bm\Psi[\TiPsi_2]+\bm\Psi[\Psi_0]+\underline{\bm\Psi}[\TiPsi_2]+\underline{\bm\varphi}[\varphi_1]+1,\\
||a^5\mathcal{D}^{11}\ulomega||_{L^2_{sc}(\mathcal{N}'_v(u_{\infty},u))}&\lesssim
\underline{\bm\varphi}[\varphi_2]+\bm\varphi[\varphi_0]
+\underline{\bm\Psi}[\TiPsi_3]+1, \\
\int_{u_{\infty}}^u\frac{a^2}{|u'|^3}||a^5\mathcal{D}^{11}\mu||_{L^2_{sc}(\mathcal{S}_{u,v})}
&\lesssim\bm\varphi[\varphi_0]+\underline{\bm\varphi}[\varphi_2]
+\underline{\bm\Psi}[\TiPsi_3]+1, \\
\int_{u_{\infty}}^u\frac{a^3}{|u'|^4}||a^5\mathcal{D}^{11}\Timu||^2_{L^2_{sc}(\mathcal{S}_{u,v})}
&\lesssim(\bm\varphi[\varphi_0]+\underline{\bm\varphi}[\varphi_2]
+\underline{\bm\Psi}[\TiPsi_3]+1)^2, \\
\int_{u_{\infty}}^u\frac{a^{\frac{3}{2}}}{|u'|^3}||a^5\mathcal{D}^{11}\lambda||_{L^2_{sc}(\mathcal{S}_{u,v})}
&\lesssim1, \\
\int_{u_{\infty}}^u\frac{a^2}{|u'|^4}||a^5\mathcal{D}^{11}\lambda||^2_{L^2_{sc}(\mathcal{S}_{u,v})}
&\lesssim1.
\end{align*}

So, in general we have that
\begin{align*}
  \bmGamma_{11}\lesssim\bm\Psi^2+\bm\Psi+\bm\varphi^2+\bm\varphi+1.
\end{align*}

\section{Energy estimates}
\label{EnergyEstimate}

In this section we discuss the construction of energy estimates for
the components of the Weyl tensor and for the auxiliary
fields~$\varphi_{0,1,2}$. Making use of the bootstrap
assumption~\eqref{Hypothesis} and the the results in previous sections
we will show that
\begin{align*}
\bm\Psi+\bm\varphi\lesssim\mathcal{I}_{0}+1.
\end{align*}

In Section~\ref{subsection:preliminary_energy} we introduce our main
tool in the analysis ---namely, the energy inequality. Due to the
central role played by the connection coefficient~$\mu$, which are
large, we follow the strategy in An's work~\cite{An2022} to manage the
terms appearing in the principal part with an explicit~$\mu$, and thus
obtain an improved energy inequality for the auxiliary fields and the
Weyl curvature. The latter will also be a key tool in the
analysis. For the pairs~$(\Psi_{i},\Psi_{i+1})$ where~$i=0,...,3$
and~$(\varphi_0,\varphi_1)$, this inequality can be used to manage the
terms involving~$\mu$. Now, for the pair~$(\varphi_1,\varphi_2)$, the
term~$\mu\varphi_1$ in~\eqref{EOMMasslessScalarStT-weight3} leads to a
significant difference in the energy inequality:
\begin{align*}
||\frac{a^{\frac{1}{2}}}{|u|}(a^{\frac{1}{2}})^{k-1}\meth^k\varphi_1||^2_{L^2_{sc}(\mathcal{N}_u(0,v))}
+||\frac{a^{\frac{1}{2}}}{|u'|}(a^{\frac{1}{2}})^{k-1}\meth^k\varphi_2||^2_{L^2_{sc}(\mathcal{N}'_v(u_{\infty},u))} 
\lesssim... .
\end{align*}
and~$\mu\varphi_0$ in~\eqref{EOMMasslessScalarStT-weight4} is a large
term. As a result, the estimate for~$(\varphi_1,\varphi_2)$ is
comparatively coarse, and so not good enough for the purposes of
studying the formation of trapped surfaces.

Another problem in the energy estimates comes from the Bianchi
equation~\eqref{T-weightMasslessStBianchi8}. According to the
bootstrap assumption, the terms~$\varphi_2\meth’\varphi_2$
and~$\varphi_2\bar\varphi_1\mu$ are large, and lead to a logarithmic
contribution in the integral over~$D_{u,v}$. This behaviour can also
obstruct the proof.

To address these problems we introduce an auxiliary field
\begin{align*}
\tilde\varphi_2\equiv\meth'\varphi_2+\mu\bar\varphi_1.
\end{align*}
The first advantage of this definition is that the structure of
equations of pair~$(\meth'\varphi_1, \tilde\varphi_2)$ have the same
\emph{good structure} of the other pairs and is also consistent with
the signature~$s_2$. For this auxiliary variable the energy inequality
leads to an estimate of exactly the same type as the other pairs. The
second advantage is that the leading terms of~$\meth'\varphi_2$
and~$\mu\bar\varphi_1$, which are large terms, but cancel exactly to
leading order because of our choice of initial data combined with the
fact that they appear in the field equations with the desired
coefficients (see remark~\ref{remark:etheth'decay} for more
detail). This implies that~$\tilde\varphi_2$ is in fact not a leading
order term and, hence, the estimate of~$(\Psi_3,\Psi_4)$ goes through
undamaged. The details of this heuristic discussion are provided in
Propositions~\ref{EnergyEstimatevarphi12} and~\ref{EnergyEstrestWeyl}.

\subsection{Preliminary energy estimates}
\label{subsection:preliminary_energy}

We begin by recalling integration by parts on causal diamonds. Making
use of~\eqref{DerivativeTwoSphere1} and~\eqref{DerivativeTwoSphere2}
one has the following:

\begin{lemma}
[\textbf{\em integral over causal diamonds of derivatives of
      a scalar}] Let~$f$ be a scalar with zero T-weight
 in the causal diamond~$\mathcal{D}_{u,v}$. One has then that
  \begin{align*}
  &\int_{\mathcal{D}_{u,v}}\Delta f=\int_{\mathcal{N}_u(0,v)}Q^{-1}f
  -\int_{\mathcal{N}_{u_{\infty}}(0,v)}Q^{-1}f
 +\int_{\mathcal{D}_{u,v}}(\ulomega-2\mu)f ,\\
 &\int_{\mathcal{D}_{u,v}}Df=\int_{\mathcal{N}'_v(u_{\infty},u)}f
 -\int_{\mathcal{N}'_0(u_{\infty},u)}f+\int_{\mathcal{D}_{u,v}}2\rho f .
\end{align*}
\end{lemma}

\begin{corollary}
\label{corollaryIntegration}
Let~$f=f_1f_2$ and assume~$s(f)=0$, then
\begin{align*}
  & \int_{\mathcal{D}_{u,v}}f_1\mthorn' f_2+\int_{\mathcal{D}_{u,v}}f_2\mthorn' f_1
  =\int_{\mathcal{N}_u(0,v)}Q^{-1}f_1f_2-\int_{\mathcal{N}_{u_{\infty}}(0,v)}Q^{-1}f_1f_2
  +\int_{\mathcal{D}_{u,v}}(\ulomega-2\mu)f_1f_2 ,\\
  & \int_{\mathcal{D}_{u,v}}f_1\mthorn f_2+\int_{\mathcal{D}_{u,v}}f_2
  \mthorn f_1=\int_{\mathcal{N}'_v(u_{\infty},u)}f_1f_2-\int_{\mathcal{N}'_0(u_{\infty},u)}f_1f_2
  +\int_{\mathcal{D}_{u,v}}2\rho f_1f_2 .
\end{align*}
\end{corollary}

\begin{proposition}
\label{PropEnergyEst}
Let~$\phi$ be either a T-weighted quantity or a derivative of
one. Moreover, let~$\lambda_1=\lambda_0-1$. Then
  \begin{align*}
&2\int_{\mathcal{D}_{u,v}}|u|^{2\lambda_1}\langle\phi,\mthorn'\phi+\lambda_0\mu\phi\rangle\\
=&\int_{\mathcal{N}_u(0,v)}Q^{-1}|u|^{2\lambda_1}|\phi|^2
-\int_{\mathcal{N}_{u_{\infty}}(0,v)}Q^{-1}|u_{\infty}|^{2\lambda_1}|\phi|^2
+\int_{\mathcal{D}_{u,v}}h|u'|^{2\lambda_1}|\phi|^2 ,
\end{align*}
where~$h$ obeys the bound 
\begin{align*}
|h|\lesssim\frac{\mathcal{O}}{|u|^2},
\end{align*}
where 
\begin{align*}
2\langle\phi,\mthorn'\phi+2\lambda_0\mu\phi\rangle=\bar\phi(\mthorn'\phi+\lambda_0\mu\phi)+
\phi(\mthorn'\bar\phi+\lambda_0\mu\bar\phi).
\end{align*}
\end{proposition}

\begin{proof}
  Let~$f\equiv Q^{-1}|u|^{2\lambda_1}|\phi|^2$. Making use
  of~\ref{DerivativeTwoSphere2} one then has that
\begin{align*}
&\frac{\mathrm{d}}{\mathrm{d}u}\int_{\mathcal{S}_{u,v}}
Q^{-1}|u|^{2\lambda_1}|\phi|^2 \\
&\quad=\int_{\mathcal{S}_{u,v}}2Q^{-1}|u|^{2\lambda_1}\langle\phi,\mthorn'\phi+\lambda_0\mu\phi\rangle+
\int_{\mathcal{S}_{u,v}}Q^{-1}|u|^{2\lambda_1}\bigg(2\mu(1-\lambda_0)-\ulomega-\frac{2\lambda_1}{|u|}\bigg),
\end{align*}
where we have made use of the fact that~$\Delta Q=\ulomega Q$. Now, let 
\begin{align*}
h\equiv-2\mu(1-\lambda_0)+\ulomega+\frac{2\lambda_1}{|u|}.
\end{align*}
It then follows that 
\begin{align*}
h=2\lambda_1(\mu-\frac{1}{u})+\ulomega.
\end{align*}
Making use of the bootstrap assumption one concludes
that~$|h|\lesssim\frac{\mathcal{O}}{|u|^2}$. Now, integrating with
respect to~$\mathrm{d}u\mathrm{d}v$ and applying the \emph{Fundamental
  Theorem of Calculus} with respect to~$u$ we obtain the desired
result.
\end{proof}

In the following, let
\begin{align*}
  (\Psi_I,\Psi_{II})\in
  \{(\Psi_0,\tilde\Psi_1),(\tilde\Psi_1,\tilde\Psi_2),
  (\tilde\Psi_2,\tilde\Psi_3),(\tilde\Psi_3,\Psi_4),
  (\varphi_{0},\varphi_{1}),(\varphi_{1},\varphi_{2})
  \}.
\end{align*}
Using this notation, the Bianchi identities and the equations for the
scalar field take the schematic form
\begin{align*}
&\mthorn'\Psi_I+\lambda_0\mu\Psi_I-\meth\Psi_{II}=P_0, \\
&\mthorn\Psi_{II}-\meth'\Psi_I=Q_0.
\end{align*}
Applying the derivative~$\meth$ and the commutator relations it
follows that
\begin{subequations}
\begin{align}
\mthorn'\meth^k\Psi_{I}+(\lambda_0+k)\mu\meth^k\Psi_{I}-\meth^{k+1}\Psi_{II}&=P_k, \label{kthBianchi1} \\
\mthorn\meth^k\Psi_{II}-\meth'\meth^k\Psi_{I}&=Q_k, \label{kthBianchi2} 
\end{align}
\end{subequations}
where
\begin{align*}
P_k=\meth^{k}P_0+
\sum_{i=0}^k\meth^i\Gamma(\Timu,\lambda)\meth^{k-i}\Psi_{I}
\end{align*}
and 
\begin{align*}
Q_k&=\sum_{i_1+i_2+i_3=k}\meth^{i_1}\Gamma(\pi,\tau)^{i_2}\meth^{i_3}\Psi_{I}
+\sum_{i_1+i_2+i_3=k}\meth^{i_1}\Gamma(\pi,\tau)^{i_2}\meth^{i_3}Q_0 \\
&+\sum_{i_1+i_2+i_3+i_4=k}\meth^{i_1}\Gamma(\tau,\pi)^{i_2}\meth^{i_3}\Gamma(\tau,\pi,\rho,\sigma)\meth^{i_4}\Psi_{II}
+\sum_{i_1+i_2=k-1}\meth^{i_1}K\meth^{i_2}\Psi_{I}.
\end{align*}

The signature of the various terms involved satisfy
\begin{align*}
s_2(\meth^k\Psi_I)=\frac{k}{2}+s_2(\Psi_I), \quad 
&s_2(\meth^k\Psi_2)=\frac{k+1}{2}+s_2(\Psi_I),\\
s_2(P_k)=\frac{k+2}{2}+s_2(\Psi_I), \quad 
&s_2(Q_k)=\frac{k+1}{2}+s_2(\Psi_I).
\end{align*}

Now applying the expressions in Proposition
\ref{PropEnergyEst} to~$\Psi_{I}$ one has that
  \begin{align*}
&2\int_{\mathcal{D}_{u,v}}|u|^{2\lambda_1}
\langle\meth^k\Psi_{I},\mthorn'\meth^k\Psi_{I}+(\lambda_0+k)\mu\meth^k\Psi_{I}\rangle\\
&\quad=\int_{\mathcal{N}_u(0,v)}Q^{-1}|u|^{2\lambda_1}|\meth^k\Psi_{I}|^2
-\int_{\mathcal{N}_{u_{\infty}}(0,v)}Q^{-1}|u_{\infty}|^{2\lambda_1}|\meth^k\Psi_{I}|^2
+\int_{\mathcal{D}_{u,v}}h|u'|^{2\lambda_1}|\meth^k\Psi_{I}|^2,
\end{align*}
where~$\lambda_1=\lambda_0+k-1$. Applying
Corollary~\ref{corollaryIntegration}
to~$f_1=\bar{f}_2=|u|^{2\lambda_1}\meth^k\Psi_{II}$ it follows that
\begin{align*}
&2\int_{\mathcal{D}_{u,v}}|u|^{2\lambda_1}\langle\meth^k\Psi_{II},\mthorn\meth^k\Psi_{II}\rangle \\
&\quad=\int_{\mathcal{N}'_v(u_{\infty},u)}|u'|^{2\lambda_1}|\meth^k\Psi_{II}|^2
-\int_{\mathcal{N}'_0(u_{\infty},u)}|u'|^{2\lambda_1}|\meth^k\Psi_{II}|^2
 +\int_{\mathcal{D}_{u,v}}2\rho|u'|^{2\lambda_1}|\meth^k\Psi_{II}|^2 ,
\end{align*}
so that combine the above two equalities one ends up with
\begin{align*}
&2\int_{\mathcal{D}_{u,v}}|u|^{2\lambda_1}
\langle\meth^k\Psi_{I},\mthorn'\meth^k\Psi_{I}+(\lambda_0+k)\mu\meth^k\Psi_{I}\rangle
+2\int_{\mathcal{D}_{u,v}}|u|^{2\lambda_1}\langle\meth^k\Psi_{II},\mthorn\meth^k\Psi_{II}\rangle \\
&\quad=\int_{\mathcal{N}_u(0,v)}Q^{-1}|u|^{2\lambda_1}|\meth^k\Psi_{I}|^2
-\int_{\mathcal{N}_{u_{\infty}}(0,v)}Q^{-1}|u_{\infty}|^{2\lambda_1}|\meth^k\Psi_{I}|^2 \\
&\quad+\int_{\mathcal{N}'_v(u_{\infty},u)}|u'|^{2\lambda_1}|\meth^k\Psi_{II}|^2
-\int_{\mathcal{N}'_0(u_{\infty},u)}|u'|^{2\lambda_1}|\meth^k\Psi_{II}|^2 \\
&\quad+\int_{\mathcal{D}_{u,v}}h|u'|^{2\lambda_1}|\meth^k\Psi_{I}|^2
 +\int_{\mathcal{D}_{u,v}}2\rho|u'|^{2\lambda_1}|\meth^k\Psi_{II}|^2 .
\end{align*}
Now, making use of equations~\eqref{kthBianchi1}, \eqref{kthBianchi2}
in the left-hand side of the previous expression and making use of
integration by parts (Lemma~\ref{IntegralbyPartT-weight}), we have
that
\begin{align*}
&\int_{\mathcal{N}_u(0,v)}Q^{-1}|u|^{2\lambda_1}|\meth^k\Psi_{I}|^2
+\int_{\mathcal{N}'_v(u_{\infty},u)}|u'|^{2\lambda_1}|\meth^k\Psi_{II}|^2 \\
&\quad=\int_{\mathcal{N}_{u_{\infty}}(0,v)}Q^{-1}|u_{\infty}|^{2\lambda_1}|\meth^k\Psi_{I}|^2
+\int_{\mathcal{N}'_0(u_{\infty},u)}|u'|^{2\lambda_1}|\meth^k\Psi_{II}|^2 \\
&\quad+2\int_{\mathcal{D}_{u,v}}|u|^{2\lambda_1}
\langle\meth^k\Psi_{I},P_k\rangle
+2\int_{\mathcal{D}_{u,v}}|u|^{2\lambda_1}\langle\meth^k\Psi_{II},Q_k\rangle \\
&\quad-\int_{\mathcal{D}_{u,v}}h|u'|^{2\lambda_1}|\meth^k\Psi_{I}|^2
 -\int_{\mathcal{D}_{u,v}}2\rho|u'|^{2\lambda_1}|\meth^k\Psi_{II}|^2 .
\end{align*}
From the bootstrap assumption we have that
\begin{align*}
|h|\leq\frac{\mathcal{O}}{|u|^2},\quad |\rho|\leq\frac{\mathcal{O}}{|u|}. 
\end{align*}
Letting 
\begin{align*}
f(u)&\equiv\int_{\mathcal{N}_u(0,v)}Q^{-1}|u|^{2\lambda_1}|\meth^k\Psi_{I}|^2, \quad
g(v)\equiv\int_{\mathcal{N}'_v(u_{\infty},u)}|u'|^{2\lambda_1}|\meth^k\Psi_{II}|^2, \\
X(u,v)&\equiv\int_{\mathcal{N}_{u_{\infty}}(0,v)}Q^{-1}|u_{\infty}|^{2\lambda_1}|\meth^k\Psi_{I}|^2
+\int_{\mathcal{N}'_0(u_{\infty},u)}|u'|^{2\lambda_1}|\meth^k\Psi_{II}|^2 \\
&+2\int_{\mathcal{D}_{u,v}}|u|^{2\lambda_1}
\langle\meth^k\Psi_{I},P_k\rangle
+2\int_{\mathcal{D}_{u,v}}|u|^{2\lambda_1}\langle\meth^k\Psi_{II},Q_k\rangle,
\end{align*}
one has that
\begin{align*}
  f(u)+g(v)\leq X(u,v)+\int_{u_\infty}^u
  \frac{1}{|u'|^2}f(u')\mathrm{d}u+\int_0^v\frac{1}{|u|}g(v')\mathrm{d}v.
\end{align*}
Hence, we find that 
\begin{align*}
  g(v)\leq X(u,v)+\int_{u_\infty}^u
  \frac{1}{|u'|^2}f(u')\mathrm{d}u+\int_0^v\frac{1}{|u|}g(v')\mathrm{d}v.
\end{align*}
Applying the Gr\"onwall-type estimate, we thus obtain that
\begin{align*}
g(v)\lesssim X(u,v)+\int_{u_\infty}^u\frac{1}{|u'|^2}f(u')\mathrm{d}u,
\end{align*}
and, moreover, that 
\begin{align*}
f(u)\lesssim X(u,v)+\int_{u_\infty}^u\frac{1}{|u'|^2}f(u')\mathrm{d}u.
\end{align*}
Applying the Gr\"onwall-type estimate again, one has
\begin{align*}
f(u)\lesssim X(u,v).
\end{align*}
In other words, one has that
\begin{align*}
&\int_{\mathcal{N}_u(0,v)}Q^{-1}|u|^{2\lambda_1}|\meth^k\Psi_{I}|^2
+\int_{\mathcal{N}'_v(u_{\infty},u)}|u'|^{2\lambda_1}|\meth^k\Psi_{II}|^2 \\
&\lesssim\int_{\mathcal{N}_{u_{\infty}}(0,v)}Q^{-1}|u_{\infty}|^{2\lambda_1}|\meth^k\Psi_{I}|^2
+\int_{\mathcal{N}'_0(u_{\infty},u)}|u'|^{2\lambda_1}|\meth^k\Psi_{II}|^2 \\
&\quad+2\int_{\mathcal{D}_{u,v}}|u|^{2\lambda_1}
\langle\meth^k\Psi_{I},P_k\rangle
+2\int_{\mathcal{D}_{u,v}}|u|^{2\lambda_1}\langle\meth^k\Psi_{II},Q_k\rangle \\
&\lesssim\int_{\mathcal{N}_{u_{\infty}}(0,v)}Q^{-1}|u_{\infty}|^{2\lambda_1}|\meth^k\Psi_{I}|^2
+\int_{\mathcal{N}'_0(u_{\infty},u)}|u'|^{2\lambda_1}|\meth^k\Psi_{II}|^2 \\
&\quad+2\int_{\mathcal{D}_{u,v}}|u|^{2\lambda_1}
|\meth^k\Psi_{I}||P_k|
+2\int_{\mathcal{D}_{u,v}}|u|^{2\lambda_1}|\meth^k\Psi_{II}||Q_k|.
\end{align*}
Now, recalling the estimate stating
that~$||Q,Q^{-1}||_{L^{\infty}(\mathcal{S}_{u,v})}$ are not far
from~$1$, we can neglect term~$Q^{-1}$ in the above inequality. Then,
multiplying by~$a^{-k-2s_2(\Psi_I)}$ on both sides and making use of
the definitions
 \begin{align*}
   ||\phi||_{L^2_{sc}(\mathcal{S}_{u,v})}&\equiv a^{-s_2(\phi)}|u|^{2s_2(\phi)}||\phi||_{L^2(\mathcal{S}_{u,v})},\\
   ||\phi||_{L^1_{sc}(\mathcal{S}_{u,v})}&\equiv a^{-s_2(\phi)}|u|^{2s_2(\phi)-1}||\phi||_{L^1(\mathcal{S}_{u,v})},
\end{align*}
we can rewrite the inequality as
\begin{align*}
&\int_0^v|u|^{2\lambda_1-4s_2(\Psi_{I})-2k}||\meth^k\Psi_{I}||^2_{L^2_{sc}(\mathcal{S}_{u,v'})}
+\int_{u_{\infty}}^ua|u'|^{2\lambda_1-4s_2(\Psi_{I})-2k-2}||\meth^k\Psi_{II}||^2_{L^2_{sc}(\mathcal{S}_{u',v})} \\
&\lesssim\int_0^v|u_{\infty}|^{2\lambda_1-4s_2(\Psi_{I})-2k}||\meth^k\Psi_{I}||^2_{L^2_{sc}(\mathcal{S}_{u_{\infty},v'})}
+\int_{u_{\infty}}^ua|u'|^{2\lambda_1-4s_2(\Psi_{I})-2k-2}||\meth^k\Psi_{II}||^2_{L^2_{sc}(\mathcal{S}_{u',0})} \\
&\quad+2\int_0^v\int_{u_{\infty}}^ua|u'|^{2\lambda_1-4s_2(\Psi_{I})-2k-1}
|||\meth^k\Psi_{I}||P_k|||_{L^1_{sc}(\mathcal{S}_{u',v'})} \\
&\quad+2\int_0^v\int_{u_{\infty}}^ua|u'|^{2\lambda_1-4s_2(\Psi_{I})-2k-1}
|||\meth^k\Psi_{II}||Q_k|||_{L^2_{sc}(\mathcal{S}_{u',v'})}.
\end{align*}
Now, applying the H\"older inequality in scale-invariant norm and
observing that~$\lambda_1=\lambda_0+k-1$, we have that
\begin{align*}
&\int_0^v|u|^{2\lambda_0-4s_2(\Psi_{I})-2}||\meth^k\Psi_{I}||^2_{L^2_{sc}(\mathcal{S}_{u,v'})}
+\int_{u_{\infty}}^ua|u'|^{2\lambda_0-4s_2(\Psi_{I})-4}||\meth^k\Psi_{II}||^2_{L^2_{sc}(\mathcal{S}_{u',v})} \\
\lesssim&\int_0^v|u_{\infty}|^{2\lambda_0-4s_2(\Psi_{I})-2}||\meth^k\Psi_{I}||^2_{L^2_{sc}(\mathcal{S}_{u_{\infty},v'})}
+\int_{u_{\infty}}^ua|u'|^{2\lambda_0-4s_2(\Psi_{I})-4}||\meth^k\Psi_{II}||^2_{L^2_{sc}(\mathcal{S}_{u',0})} \\
&+2\int_0^v\int_{u_{\infty}}^ua|u'|^{2\lambda_0-4s_2(\Psi_{I})-4}
||\meth^k\Psi_{I}||_{L^2_{sc}(\mathcal{S}_{u',v'})}
||P_k||_{L^2_{sc}(\mathcal{S}_{u',v'})} \\
&+2\int_0^v\int_{u_{\infty}}^ua|u'|^{2\lambda_0-4s_2(\Psi_{I})-4}
||\meth^k\Psi_{II}||_{L^2_{sc}(\mathcal{S}_{u',v'})}
||Q_k||_{L^2_{sc}(\mathcal{S}_{u',v'})}.
\end{align*}

For~$(\Psi_I,\Psi_{II})\in
\{(\Psi_0,\tilde\Psi_1),(\tilde\Psi_1,\tilde\Psi_2),(\tilde\Psi_2,\tilde\Psi_3),(\tilde\Psi_3,\Psi_4),(\varphi_{0},\varphi_{1})
\}$, the Bianchi identities and
the equations for the scalar field satisfy that~$\lambda_0=2s_2(\Psi_{I})+1$. Consequently,
the above inequality takes the form
\begin{align*}
&\int_0^v||\meth^k\Psi_{I}||^2_{L^2_{sc}(\mathcal{S}_{u,v'})}
+\int_{u_{\infty}}^u\frac{a}{|u'|^2}||\meth^k\Psi_{II}||^2_{L^2_{sc}(\mathcal{S}_{u',v})} \\
\lesssim&\int_0^v||\meth^k\Psi_{I}||^2_{L^2_{sc}(\mathcal{S}_{u_{\infty},v'})}
+\int_{u_{\infty}}^u\frac{a}{|u'|^2}||\meth^k\Psi_{II}||^2_{L^2_{sc}(\mathcal{S}_{u',0})} \\
&+2\int_0^v\int_{u_{\infty}}^u\frac{a}{|u'|^2}
||\meth^k\Psi_{I}||_{L^2_{sc}(\mathcal{S}_{u',v'})}
||P_k||_{L^2_{sc}(\mathcal{S}_{u',v'})} \\
&+2\int_0^v\int_{u_{\infty}}^u\frac{a}{|u'|^2}
||\meth^k\Psi_{II}||_{L^2_{sc}(\mathcal{S}_{u',v'})}
||Q_k||_{L^2_{sc}(\mathcal{S}_{u',v'})}.
\end{align*}
Finally, recalling the definitions of the norms on the lightcone
---namely,
\begin{align*}
||\phi||^2_{L^2_{sc}(\mathcal{N}_u(0,v))}&
\equiv\int_0^v||\phi||^2_{L^2_{sc}(\mathcal{S}_{u,v'})}\mathrm{d}v',\\
||\phi||^2_{L^2_{sc}(\mathcal{N}'_v(u_{\infty},u))}&
\equiv\int_{u_{\infty}}^u\frac{a}{|u'|^2}||\phi||^2_{L^2_{sc}(\mathcal{S}_{u',v})}\mathrm{d}u',
\end{align*}
we conclude that 
\begin{align*}
&||\meth^k\Psi_{I}||^2_{L^2_{sc}(\mathcal{N}_u(0,v))}
+||\meth^k\Psi_{II}||^2_{L^2_{sc}(\mathcal{N}'_v(u_{\infty},u))} \\
&\lesssim||\meth^k\Psi_{I}||^2_{L^2_{sc}(\mathcal{N}_{u_{\infty}}(0,v))}
+||\meth^k\Psi_{II}||^2_{L^2_{sc}(\mathcal{N}'_0(u_{\infty},u))} \\
&\quad+2\int_0^v\int_{u_{\infty}}^u\frac{a}{|u'|^2}
||\meth^k\Psi_{I}||_{L^2_{sc}(\mathcal{S}_{u',v'})}
||P_k||_{L^2_{sc}(\mathcal{S}_{u',v'})} \\
&\quad+2\int_0^v\int_{u_{\infty}}^u\frac{a}{|u'|^2}
||\meth^k\Psi_{II}||_{L^2_{sc}(\mathcal{S}_{u',v'})}
||Q_k||_{L^2_{sc}(\mathcal{S}_{u',v'})}.
\end{align*}

\subsection{Energy estimates for the auxiliary field}

\begin{proposition}
\label{EnergyEstimatevarphi01}
For~$0\leq k\leq11$,  one has that 
\begin{align*}
\frac{1}{a}||(a^{\frac{1}{2}})^{k-1}\mathcal{D}^k\varphi_0||^2_{L^2_{sc}(\mathcal{N}_u(0,v))}
+\frac{1}{a}||(a^{\frac{1}{2}})^{k-1}\mathcal{D}^k\varphi_1||^2_{L^2_{sc}(\mathcal{N}'_v(u_{\infty},u))} 
\leq\mathcal{I}^2_0+\frac{1}{a^{\frac{1}{4}}}.
\end{align*}
\end{proposition}

\begin{proof}
  When~$(\Psi_I,\Psi_{II})=(\varphi_{0},\varphi_{1})$, we make use of
  the equations~\eqref{EOMMasslessScalarStT-weight1}
  and~\eqref{EOMMasslessScalarStT-weight1}
\begin{align*}
\mthorn'\varphi_{0}-\meth\bar\varphi_{1}&=(\ulomega-\mu)\varphi_{0}
+\rho\varphi_{2}-\bar\tau\varphi_{1}-\tau\bar\varphi_{1},  \\
\mthorn\bar\varphi_{1}-\meth'\varphi_{0}&=(\pi-\bar\tau)\varphi_{0}
+\bar\sigma\varphi_{1}+\rho\bar\varphi_{1}.
\end{align*}
Now, let 
\begin{align*}
P_k&\equiv\Gamma(\ulomega,\lambda)\meth^k\varphi_0+\tau\meth^k\varphi_1+\rho\meth^k\varphi_2
+\varphi_0\meth^{k}\Gamma(\ulomega,\Timu)+\varphi_1\meth^k\tau+\varphi_2\meth^k\rho \\
&+\sum_{i=1}^{k-1}\meth^i\Gamma(\ulomega,\tau,\rho,\Timu,\lambda)\meth^{k-1}\bm\varphi(\varphi_{0,1,2}).
\end{align*}
\begin{align*}
Q_k&\equiv\Gamma(\tau,\pi)\meth^k\varphi_0+\varphi_0\meth^k\Gamma(\tau,\pi)
+\Gamma(\rho,\sigma)\meth^k\varphi_1+\varphi_1\meth^k\Gamma(\rho,\sigma) \\
&+\sum_{i_1+i_2+i_3+i_4=k,i_3+i_4<k}\meth^{i_1}\Gamma(\tau,\pi)^{i_2}\meth^{i_3}
\Gamma(\tau,\pi,\rho,\sigma)\meth^{i_4}(\varphi_{0},\varphi_{1})
+\sum_{i_1+i_2=k-1}\meth^{i_1}K\meth^{i_2}\varphi_0.
\end{align*}

Making use of the energy inequality we have that
\begin{align*}
&||\meth^k\varphi_0||^2_{L^2_{sc}(\mathcal{N}_u(0,v))}
+||\meth^k\varphi_1||^2_{L^2_{sc}(\mathcal{N}'_v(u_{\infty},u))} \\
&\lesssim||\meth^k\varphi_0||^2_{L^2_{sc}(\mathcal{N}_{u_{\infty}}(0,v))}
+||\meth^k\varphi_1||^2_{L^2_{sc}(\mathcal{N}'_0(u_{\infty},u))} \\
&\quad+2\int_0^v\int_{u_{\infty}}^u\frac{a}{|u'|^2}
||\meth^k\varphi_0||_{L^2_{sc}(\mathcal{S}_{u',v'})}
||P_k||_{L^2_{sc}(\mathcal{S}_{u',v'})} \\
&\quad+2\int_0^v\int_{u_{\infty}}^u\frac{a}{|u'|^2}
||\meth^k\varphi_1||_{L^2_{sc}(\mathcal{S}_{u',v'})}
||Q_k||_{L^2_{sc}(\mathcal{S}_{u',v'})}.
\end{align*}
so that multiplying both sides by~$a^{\frac{2k-2}{2}}$ on both hands
side we find that
\begin{align*}
&||(a^{\frac{1}{2}})^{k-1}\meth^k\varphi_0||^2_{L^2_{sc}(\mathcal{N}_u(0,v))}
+||(a^{\frac{1}{2}})^{k-1}\meth^k\varphi_1||^2_{L^2_{sc}(\mathcal{N}'_v(u_{\infty},u))} \\
&\lesssim||(a^{\frac{1}{2}})^{k-1}\meth^k\varphi_0||^2_{L^2_{sc}(\mathcal{N}_{u_{\infty}}(0,v))}
+||(a^{\frac{1}{2}})^{k-1}\meth^k\varphi_1||^2_{L^2_{sc}(\mathcal{N}'_0(u_{\infty},u))} \\
&\quad+2\int_0^v\int_{u_{\infty}}^u\frac{a}{|u'|^2}
||(a^{\frac{1}{2}})^{k-1}\meth^k\varphi_0||_{L^2_{sc}(\mathcal{S}_{u',v'})}
||(a^{\frac{1}{2}})^{k-1}P_k||_{L^2_{sc}(\mathcal{S}_{u',v'})} \\
&\quad+2\int_0^v\int_{u_{\infty}}^u\frac{a}{|u'|^2}
||(a^{\frac{1}{2}})^{k-1}\meth^k\varphi_1||_{L^2_{sc}(\mathcal{S}_{u',v'})}
||(a^{\frac{1}{2}})^{k-1}Q_k||_{L^2_{sc}(\mathcal{S}_{u',v'})} \\
&\leq||(a^{\frac{1}{2}})^{k-1}\meth^k\varphi_0||^2_{L^2_{sc}(\mathcal{N}_{u_{\infty}}(0,v))}
+||(a^{\frac{1}{2}})^{k-1}\meth^k\varphi_1||^2_{L^2_{sc}(\mathcal{N}'_0(u_{\infty},u))} 
+M+N,
\end{align*}
where
\begin{align*}
M&\equiv2\int_0^v\int_{u_{\infty}}^u\frac{a}{|u'|^2}
||(a^{\frac{1}{2}})^{k-1}\meth^k\varphi_0||_{L^2_{sc}(\mathcal{S}_{u',v'})}
||(a^{\frac{1}{2}})^{k-1}P_k||_{L^2_{sc}(\mathcal{S}_{u',v'})}, \\
N&\equiv2\int_0^v\int_{u_{\infty}}^u\frac{a}{|u'|^2}
||(a^{\frac{1}{2}})^{k-1}\meth^k\varphi_1||_{L^2_{sc}(\mathcal{S}_{u',v'})}
||(a^{\frac{1}{2}})^{k-1}Q_k||_{L^2_{sc}(\mathcal{S}_{u',v'})}. \\
\end{align*}
Finally, multiplying by~$a^{-1}$ we are led to
\begin{align}
\label{EstimateScalar01Mid}
&\frac{1}{a}||(a^{\frac{1}{2}})^{k-1}\meth^k\varphi_0||^2_{L^2_{sc}(\mathcal{N}_u(0,v))}
+\frac{1}{a}||(a^{\frac{1}{2}})^{k-1}\meth^k\varphi_1||^2_{L^2_{sc}(\mathcal{N}'_v(u_{\infty},u))} \nonumber \\
&\lesssim\frac{1}{a}||(a^{\frac{1}{2}})^{k-1}\mathcal{D}^k\varphi_0||^2_{L^2_{sc}(\mathcal{N}_{u_{\infty}}(0,v))}
+\frac{1}{a}||(a^{\frac{1}{2}})^{k-1}\mathcal{D}^k\varphi_1||^2_{L^2_{sc}(\mathcal{N}'_0(u_{\infty},u))} 
+\frac{1}{a}(M+N).
\end{align}
From here one the proof requires a number of intermediate steps.

\smallskip
\noindent
\textbf{Step 1}: Estimate~$M$.

We have that
\begin{align*}
M&\leq2\left(\int_0^v\int_{u_{\infty}}^u\frac{a}{|u'|^2}
||(a^{\frac{1}{2}})^{k-1}\meth^k\varphi_0||^2_{L^2_{sc}(\mathcal{S}_{u',v'})}\right)^{\frac{1}{2}}
\left(\int_0^v\int_{u_{\infty}}^u\frac{a}{|u'|^2}||(a^{\frac{1}{2}})^{k-1}P_k||^2_{L^2_{sc}(\mathcal{S}_{u',v'})}\right)^{\frac{1}{2}} \\
&=2\left(\int_{u_{\infty}}^u\frac{a}{|u'|^2}
||(a^{\frac{1}{2}})^{k-1}\meth^k\varphi_0||^2_{L^2_{sc}(\mathcal{N}_{u'}(0,v))}\right)^{\frac{1}{2}}J^{\frac{1}{2}}
\leq\frac{a^{\frac{1}{2}}}{|u|^{\frac{1}{2}}}a^{\frac{1}{2}}\bm{\varphi}[\varphi_0]J^{\frac{1}{2}},
\end{align*}
where 
\begin{align*}
  J\equiv\int_0^v\int_{u_{\infty}}^u
  \frac{a}{|u'|^2}||(a^{\frac{1}{2}})^{k-1}P_k||^2_{L^2_{sc}(\mathcal{S}_{u',v'})}.
\end{align*}
Now, substituting the definition of~$P_k$ into~$J$ we have that
\begin{align*}
J&\leq\int_0^v\int_{u_{\infty}}^u\frac{a}{|u'|^2}
||(a^{\frac{1}{2}})^{k-1}\Gamma(\ulomega,\lambda)\mathcal{D}^k\varphi_0||^2_{L^2_{sc}(\mathcal{S}_{u',v'})} 
+\int_0^v\int_{u_{\infty}}^u\frac{a}{|u'|^2}
||(a^{\frac{1}{2}})^{k-1}\tau\mathcal{D}^k\varphi_1||^2_{L^2_{sc}(\mathcal{S}_{u',v'})} \\
&+\int_0^v\int_{u_{\infty}}^u\frac{a}{|u'|^2}
||(a^{\frac{1}{2}})^{k-1}\rho\mathcal{D}^k\varphi_2||^2_{L^2_{sc}(\mathcal{S}_{u',v'})} 
+\int_0^v\int_{u_{\infty}}^u\frac{a}{|u'|^2}
||(a^{\frac{1}{2}})^{k-1}\varphi_0\mathcal{D}^k\Gamma(\ulomega,\Timu)||^2_{L^2_{sc}(\mathcal{S}_{u',v'})} \\
&+\int_0^v\int_{u_{\infty}}^u\frac{a}{|u'|^2}
||(a^{\frac{1}{2}})^{k-1}\varphi_1\mathcal{D}^k\tau||^2_{L^2_{sc}(\mathcal{S}_{u',v'})} 
+\int_0^v\int_{u_{\infty}}^u\frac{a}{|u'|^2}
||(a^{\frac{1}{2}})^{k-1}\varphi_2\mathcal{D}^k\rho||^2_{L^2_{sc}(\mathcal{S}_{u',v'})} \\
&+\sum_{i=1}^{k-1}\int_0^v\int_{u_{\infty}}^u\frac{a}{|u'|^2}
||(a^{\frac{1}{2}})^{k-1}\mathcal{D}^{i}\Gamma(\lambda,...)\mathcal{D}^{k-i}\bm{\varphi}||^2_{L^2_{sc}(\mathcal{S}_{u',v'})}
=I_1+...+I_7.
\end{align*}

For the term~$I_1$ we have that
\begin{align*}
I_1&\leq\int_0^v\int_{u_{\infty}}^u\frac{a}{|u'|^2}\frac{1}{|u'|^2}||\Gamma(\ulomega,\lambda)||^2_{L^{\infty}_{sc}(\mathcal{S}_{u',v'})} 
||(a^{\frac{1}{2}})^{k-1}\mathcal{D}^k\varphi_0||^2_{L^2_{sc}(\mathcal{S}_{u',v'})} \\
&\leq\int_{u_{\infty}}^u\frac{a}{|u'|^2}\frac{1}{|u'|^2}(\bmGamma(\ulomega)_{0,\infty}^2+\frac{|u'|^2}{a}\bmGamma(\lambda)_{0,\infty}^2)
||(a^{\frac{1}{2}})^{k-1}\mathcal{D}^k\varphi_0||^2_{L^2_{sc}(\mathcal{N}_{u'}(0,v))} \\
&\leq\int_{u_{\infty}}^u\frac{a}{|u'|^2}\bm{\varphi}[\varphi_0]^2
\leq\frac{a}{|u|}\bm{\varphi}[\varphi_0]^2.
\end{align*}
For~$I_2$ we have that 
\begin{align*}
I_2&\leq\int_0^v\int_{u_{\infty}}^u\frac{a}{|u'|^2}\frac{1}{|u'|^2}||\tau||^2_{L^{\infty}_{sc}(\mathcal{S}_{u',v'})} 
||(a^{\frac{1}{2}})^{k-1}\mathcal{D}^k\varphi_1||^2_{L^2_{sc}(\mathcal{S}_{u',v'})} \\
&\leq\int_0^v\frac{1}{|u|^2}\bmGamma(\tau)_{0,\infty}^2
||(a^{\frac{1}{2}})^{k-1}\mathcal{D}^k\varphi_1||^2_{L^2_{sc}(\mathcal{N}'_{v'}(u_{\infty},u))} \\
&\leq\int_0^v\frac{a}{|u|^2}\bmGamma(\tau)_{0,\infty}^2\underline{\bm{\varphi}}[\varphi_1]^2
\leq\frac{a}{|u|^2}\mathcal{O}^2(\bm{\Psi}[\TiPsi_1]+1)^2
\leq\frac{a\mathcal{O}^4}{|u|^2}.
\end{align*}
For~$I_3$ we have that
\begin{align*}
I_3&\leq\int_0^v\int_{u_{\infty}}^u\frac{a}{|u'|^2}\frac{1}{|u'|^2}||\rho||^2_{L^{\infty}_{sc}(\mathcal{S}_{u',v'})} 
||(a^{\frac{1}{2}})^{k-1}\mathcal{D}^k\varphi_2||^2_{L^2_{sc}(\mathcal{S}_{u',v'})} \\
&\leq\int_0^v\frac{1}{a}\bmGamma(\rho)_{0,\infty}^2
||\frac{a^{\frac{1}{2}}}{|u'|}(a^{\frac{1}{2}})^{k-1}\mathcal{D}^k\varphi_2||^2_{L^2_{sc}(\mathcal{N}'_{v'}(u_{\infty},u))} 
\leq\frac{1}{a}\underline{\bm{\varphi}}[\varphi_2]^2.
\end{align*}
For~$I_4$ we have
\begin{align*}
I_4&\leq\int_0^v\int_{u_{\infty}}^u\frac{a}{|u'|^2}\frac{1}{|u'|^2}||\varphi_0||^2_{L^{\infty}_{sc}(\mathcal{S}_{u',v'})} 
||(a^{\frac{1}{2}})^{k-1}\mathcal{D}^k\ulomega||^2_{L^2_{sc}(\mathcal{S}_{u',v'})} \\
&\quad+\int_0^v\int_{u_{\infty}}^u\frac{a}{|u'|^2}\frac{1}{|u'|^2}||\varphi_0||^2_{L^{\infty}_{sc}(\mathcal{S}_{u',v'})} 
||(a^{\frac{1}{2}})^{k-1}\mathcal{D}^k\Timu||^2_{L^2_{sc}(\mathcal{S}_{u',v'})} \\
&\leq\int_0^v\frac{a}{|u|^2}\bmGamma(\varphi_0)_{0,\infty}^2
||(a^{\frac{1}{2}})^{k-1}\mathcal{D}^k\ulomega||^2_{L^2_{sc}(\mathcal{N}_{v'}(u_{\infty},u))} \\
&\quad+\int_0^v\frac{1}{a}\bmGamma(\varphi_0)_{0,\infty}^2
||(a^{\frac{1}{2}})^{k-1}\mathcal{D}^k\Timu||^2_{L^2_{sc}(\mathcal{N}_{v'}(u_{\infty},u))} \\
&\leq\frac{1}{a}(\bm\varphi+\underline{\bm\varphi}
+\underline{\bm\Psi}[\TiPsi_3]+1)^2.
\end{align*}
For~$I_5$ one sees that 
\begin{align*}
I_5&\leq\int_0^v\int_{u_{\infty}}^u\frac{a}{|u'|^2}\frac{1}{|u'|^2}||\varphi_1||^2_{L^{\infty}_{sc}(\mathcal{S}_{u',v'})} 
||(a^{\frac{1}{2}})^{k-1}\mathcal{D}^k\tau||^2_{L^2_{sc}(\mathcal{S}_{u',v'})} \\
&\leq\int_0^v\frac{a}{|u|^2}\bmGamma(\varphi_1)_{0,\infty}^2
||(a^{\frac{1}{2}})^{k-1}\mathcal{D}^k\tau||^2_{L^2_{sc}(\mathcal{N}'_{v'}(u_{\infty},u))} \\
&\lesssim\frac{a}{|u|^2}(\bm\Psi[\Psi_0]+\underline{\bm\Psi}[\TiPsi_2]+\underline{\bm\varphi}[\varphi_1]+1)^2.
\end{align*}
For~$I_6$ we have that 
\begin{align*}
I_6&\leq\int_0^v\int_{u_{\infty}}^u\frac{a}{|u'|^2}\frac{1}{|u'|^2}||\varphi_2||^2_{L^{\infty}_{sc}(\mathcal{S}_{u',v'})} 
||(a^{\frac{1}{2}})^{k-1}\mathcal{D}^k\rho||^2_{L^2_{sc}(\mathcal{S}_{u',v'})} \\
&\leq\int_{u_{\infty}}^u\frac{a}{|u'|^2}\frac{1}{|u'|^2}\frac{|u'|^2}{a}\mathcal{O}^2
||(a^{\frac{1}{2}})^{k-1}\mathcal{D}^k\rho||^2_{L^2_{sc}(\mathcal{N}_{u'}(0,v))} 
\leq\frac{1}{|u|}\mathcal{O}^4.
\end{align*}
Finally, for~$I_7$ we have that
\begin{align*}
I_7&\leq\sum_{i=1}^{k-1}\int_0^v\int_{u_{\infty}}^u\frac{1}{|u'|^2}\frac{1}{|u'|^2}
||(a^{\frac{1}{2}}\mathcal{D})^{i}\Gamma(\lambda,...)||^2\times
||a^{\frac{1}{2}}\mathcal{D})^{k-i}\bm{\varphi}||^2 \\
&\leq\int_{u_{\infty}}^u\frac{1}{|u'|^4}\frac{|u'|^2}{a}\mathcal{O}^2a\mathcal{O}^2
\leq\frac{1}{|u|}\mathcal{O}^4.
\end{align*}

Combining the estimates for~$I_1,\ldots,I_7$ we conclude that
\begin{align*}
J&\lesssim\frac{a}{|u|}\bm{\varphi}[\varphi_0]^2+\frac{\mathcal{O}^4}{a}
\end{align*}
and then 
\begin{align*}
M\lesssim\frac{a^{\frac{3}{2}}}{|u|}\bm{\varphi}[\varphi_0]^2
+\mathcal{O}^2\bm{\varphi}[\varphi_0].
\end{align*}

\smallskip
\noindent
\textbf{Step 2}: Estimate~$N$.

We have
\begin{align*}
N&\leq2\left(\int_0^v\int_{u_{\infty}}^u\frac{a}{|u'|^2}
||(a^{\frac{1}{2}})^{k-1}\meth^k\varphi_1||^2_{L^2_{sc}(\mathcal{S}_{u',v'})}\right)^{\frac{1}{2}}
\left(\int_0^v\int_{u_{\infty}}^u\frac{a}{|u'|^2}||(a^{\frac{1}{2}})^{k-1}Q_k||^2_{L^2_{sc}(\mathcal{S}_{u',v'})}\right)^{\frac{1}{2}} \\
&=2\left(\int_0^v
||(a^{\frac{1}{2}})^{k-1}\meth^k\varphi_1||^2_{L^2_{sc}(\mathcal{N}'_{v'}(u_{\infty},u))}\right)^{\frac{1}{2}}H^{\frac{1}{2}}
\lesssim a^{\frac{1}{2}}\bm{\varphi}[\varphi_1]H^{\frac{1}{2}},
\end{align*}
where 
\begin{align*}
H\equiv\int_0^v\int_{u_{\infty}}^u\frac{a}{|u'|^2}||(a^{\frac{1}{2}})^{k-1}Q_k||^2_{L^2_{sc}(\mathcal{S}_{u',v'})}.
\end{align*}
Substituting the definition of~$Q_k$ into~$H$ we obtain
\begin{align*}
&H\leq\int_0^v\int_{u_{\infty}}^u\frac{a}{|u'|^2}
||(a^{\frac{1}{2}})^{k-1}\Gamma(\tau,\pi)\mathcal{D}^k\varphi_0||^2_{L^2_{sc}(\mathcal{S}_{u',v'})} 
+\int_0^v\int_{u_{\infty}}^u\frac{a}{|u'|^2}
||(a^{\frac{1}{2}})^{k-1}\Gamma(\rho,\sigma)\mathcal{D}^k\varphi_1||^2_{L^2_{sc}(\mathcal{S}_{u',v'})} \\
&\quad+\int_0^v\int_{u_{\infty}}^u\frac{a}{|u'|^2}
||(a^{\frac{1}{2}})^{k-1}\varphi_0\mathcal{D}^k\Gamma(\tau,\pi)||^2_{L^2_{sc}(\mathcal{S}_{u',v'})} 
+\int_0^v\int_{u_{\infty}}^u\frac{a}{|u'|^2}
||(a^{\frac{1}{2}})^{k-1}\varphi_1\mathcal{D}^k\Gamma(\rho,\sigma)||^2_{L^2_{sc}(\mathcal{S}_{u',v'})} \\
&\quad+\sum_{i_1+i_2+i_3+i_4=k,i_3+i_4<k}\int_0^v\int_{u_{\infty}}^u\frac{a}{|u'|^2}
||(a^{\frac{1}{2}})^{k-1}\mathcal{D}^{i_1}\Gamma(\pi,\tau)^{i_2}
\mathcal{D}^{i_3}\Gamma(\sigma...)\mathcal{D}^{i_4}\bm{\varphi}||^2_{L^2_{sc}(\mathcal{S}_{u',v'})} \\
&\quad+\sum_{i_1+i_2=k-1}\int_0^v\int_{u_{\infty}}^u\frac{a}{|u'|^2}
||(a^{\frac{1}{2}})^{k-1}\mathcal{D}^{i_1}K\mathcal{D}^{i_2}\varphi_0||^2_{L^2_{sc}(\mathcal{S}_{u',v'})} 
=I_1+...+I_6.
\end{align*}

For the term~$I_1$ we have that
\begin{align*}
I_1&\leq\int_0^v\int_{u_{\infty}}^u\frac{a}{|u'|^2}\frac{1}{|u'|^2}||\Gamma(\tau,\pi)||^2_{L^{\infty}_{sc}(\mathcal{S}_{u',v'})} 
||(a^{\frac{1}{2}})^{k-1}\mathcal{D}^k\varphi_0||^2_{L^2_{sc}(\mathcal{S}_{u',v'})} \\
&\leq\int_{u_{\infty}}^u\frac{a}{|u'|^2}\frac{1}{|u'|^2}\bmGamma(\tau,\pi)_{0,\infty}^2
||(a^{\frac{1}{2}})^{k-1}\mathcal{D}^k\varphi_0||^2_{L^2_{sc}(\mathcal{N}_{u'}(0,v))} \\
&\leq\int_{u_{\infty}}^u\frac{a^2}{|u'|^4}\bmGamma(\tau,\pi)_{0,\infty}^2\bm{\varphi}[\varphi_0]^2
\leq\frac{a^2}{|u|^3}(\bm{\Psi}[\TiPsi_1]+\underline{\bm{\Psi}}[\TiPsi_3]+1)^2\bm{\varphi}[\varphi_0]^2.
\end{align*}
For~$I_2$ we have that 
\begin{align*}
I_2&\leq\int_0^v\int_{u_{\infty}}^u\frac{a}{|u'|^2}\frac{1}{|u'|^2}||\Gamma(\rho,\sigma)||^2_{L^{\infty}_{sc}(\mathcal{S}_{u',v'})} 
||(a^{\frac{1}{2}})^{k-1}\mathcal{D}^k\varphi_1||^2_{L^2_{sc}(\mathcal{S}_{u',v'})} \\
&\leq\int_0^v\frac{1}{|u'|^2}(\bmGamma(\rho)_{0,\infty}^2+a\bmGamma(\sigma)_{0,\infty}^2)
||(a^{\frac{1}{2}})^{k-1}\mathcal{D}^k\varphi_1||^2_{L^2_{sc}(\mathcal{N}'_{v'}(u_{\infty},u))} \\
&\leq\int_0^v\frac{a}{|u'|^2}\bmGamma(\sigma)_{0,\infty}^2\underline{\bm{\varphi}}[\varphi_1]^2
\leq\frac{a}{|u|^2}\mathcal{O}^2(\bm{\Psi}[\Psi_0]+1)^2.
\end{align*}
For~$I_3$ and~$I_4$ we see that 
\begin{align*}
I_3+I_4&\leq\int_0^v\int_{u_{\infty}}^u\frac{a}{|u'|^2}\frac{1}{|u'|^2}||\varphi_{0,1}||^2_{L^{\infty}_{sc}(\mathcal{S}_{u',v'})} 
||(a^{\frac{1}{2}})^{k-1}\mathcal{D}^k\Gamma(\rho,\sigma)||^2_{L^2_{sc}(\mathcal{S}_{u',v'})} \\
&\leq\int_{u_{\infty}}^u\frac{a}{|u'|^4}a\bmGamma(\varphi_0,\varphi_1)_{0,\infty}^2
||(a^{\frac{1}{2}})^{k-1}\mathcal{D}^k\Gamma(\rho,\sigma)||^2_{L^2_{sc}(\mathcal{N}_{u'}(0,v))} \\
&\leq\frac{a^2}{|u|^3}(\bm{\varphi}[\varphi_0]+\bm\Psi[\Psi_0]+\bm\Psi[\TiPsi_1]+1)^2.
\end{align*}
For~$I_5$ we have that 
\begin{align*}
I_5&=\sum_{i_1+...+i_4=k,i_3+i_4<k}\int_{u_{\infty}}^u\frac{a}{|u'|^2}
||a^{\frac{k-1}{2}}\mathcal{D}^{j_1}\Gamma...\mathcal{D}^{j_{i_2}}\Gamma
\mathcal{D}^{i_3}\Gamma(\sigma,...)\mathcal{D}^{i_4}\bm{\varphi}||^2 \\
&\leq\int_{u_{\infty}}^u\frac{1}{|u'|^2}\left(\frac{(a^{\frac{1}{2}})^{i_2}}{|u|^{i_2+1}}||(a^{\frac{1}{2}}\mathcal{D})^{j_1}\Gamma||...
||(a^{\frac{1}{2}}\mathcal{D})^{j_{i_2}}\Gamma||\times
||(a^{\frac{1}{2}}\mathcal{D})^{i_3}\Gamma||\times||(a^{\frac{1}{2}}\mathcal{D})^{i_4}\bm{\varphi}|| \right)^2 \\
&\leq\int_{u_{\infty}}^u\frac{1}{|u'|^2}
(\frac{(a^{\frac{1}{2}})^{i_2}}{|u|^{i_2+1}}\mathcal{O}^{i_2}a^{\frac{1}{2}}\mathcal{O}a^{\frac{1}{2}}\mathcal{O})^2
=\int_{u_{\infty}}^u\frac{1}{|u'|^2}\frac{a^{i_2+2}}{|u'|^{2i_2+2}}\mathcal{O}^{2(i_2+2)}
\leq\int_{u_{\infty}}^u\frac{1}{|u'|^2}\leq\frac{1}{|u|}.
\end{align*}
For~$I_6$ we have that 
\begin{align*}
I_6&\leq\sum_{i_1+i_2=k-1}\int_{u_{\infty}}^u\frac{1}{|u'|^2}\frac{1}{|u'|^2}||(a^{\frac{1}{2}}\mathcal{D})^{i_2}K||^2\times
||(a^{\frac{1}{2}}\mathcal{D})^{i_2}\varphi_0||^2\\
&\leq\int_{u_{\infty}}^u\frac{1}{|u'|^4}a\mathcal{O}^2 
\leq\frac{1}{|u|}.
\end{align*}

Finally, combining the estimates for~$I_1,\ldots,I_6$ we conclude that 
\begin{align*}
H\lesssim\frac{a^2}{|u|^3}(\bm{\Psi}[\TiPsi_1]+\underline{\bm{\Psi}}[\TiPsi_3]+1)^2\bm{\varphi}[\varphi_0]^2
+\frac{1}{|u|}\mathcal{O}^2(\underline{\bm\Psi}+\bm\Psi+\bm\varphi+1)^2+\frac{1}{|u|},
\end{align*}
and that 
\begin{align*}
N\lesssim\frac{a^{\frac{1}{2}}}{|u|^{\frac{1}{2}}}\bm{\varphi}[\varphi_0]\bm{\varphi}[\varphi_1](\underline{\bm\Psi}+\bm\Psi+\bm\varphi+1).
\end{align*}

\smallskip
\noindent
\textbf{Step 3}: Summary.

From the analysis in Steps 1 and 2 we conclude that 
\begin{align*}
&\frac{1}{a}||(a^{\frac{1}{2}})^{k-1}\meth^k\varphi_0||^2_{L^2_{sc}(\mathcal{N}_u(0,v))}
+\frac{1}{a}||(a^{\frac{1}{2}})^{k-1}\meth^k\varphi_1||^2_{L^2_{sc}(\mathcal{N}'_v(u_{\infty},u))} \\
&\lesssim\frac{1}{a}||(a^{\frac{1}{2}})^{k-1}\mathcal{D}^k\varphi_0||^2_{L^2_{sc}(\mathcal{N}_{u_{\infty}}(0,v))}
+\frac{1}{a}||(a^{\frac{1}{2}})^{k-1}\mathcal{D}^k\varphi_1||^2_{L^2_{sc}(\mathcal{N}'_0(u_{\infty},u))} 
+\frac{1}{a}(M+N) \\
&\lesssim\mathcal{I}^2_0+
\frac{1}{a}(\bm{\varphi}[\varphi_0]\bm{\varphi}[\varphi_1](\underline{\bm\Psi}+\bm\Psi+\bm\varphi+1)
+\frac{a^{\frac{3}{2}}}{|u|}\bm{\varphi}[\varphi_0]^2
+\mathcal{O}^2\bm{\varphi}[\varphi_0]) \\
&\leq\mathcal{I}^2_0+\frac{1}{a^{\frac{1}{4}}}.
\end{align*}

Following the same strategy one can estimate the remaining strings
in~$\mathcal{D}^{k_{i}}\varphi_0$ and~$\mathcal{D}^{k_{i}}\varphi_1$
and obtain the same result. We thus obtain the estimate
\begin{align*}
\frac{1}{a}||(a^{\frac{1}{2}})^{k-1}\mathcal{D}^k\varphi_0||^2_{L^2_{sc}(\mathcal{N}_u(0,v))}
+\frac{1}{a}||(a^{\frac{1}{2}})^{k-1}\mathcal{D}^k\varphi_1||^2_{L^2_{sc}(\mathcal{N}'_v(u_{\infty},u))} 
\leq\mathcal{I}^2_0+\frac{1}{a^{\frac{1}{4}}}.
\end{align*}
\end{proof}

\begin{proposition}
\label{EnergyEstimatevarphi12}
For~$0\leq k\leq10$, one has that 
\begin{align*}
&\int_0^v||a^{\frac{k}{2}}\mathcal{D}^{k+1}\varphi_1||^2_{L^2_{sc}(\mathcal{S}_{u,v'})}
+\int_{u_{\infty}}^u\frac{a}{|u'|^2}||(a^{\frac{1}{2}}\mathcal{D})^k\Tivarphi_2||^2_{L^2_{sc}(\mathcal{S}_{u',v})} \\
&+\int_{u_{\infty}}^u\frac{a^2}{|u'|^4}||a^{\frac{k}{2}}\mathcal{D}^{k+1}\varphi_2||^2_{L^2_{sc}(\mathcal{S}_{u',v})} 
\lesssim\mathcal{I}_0^2+1.
\end{align*}
\end{proposition}

\begin{proof}
  In this proof we make use of the fields~$\tilde\varphi_1$
  and~$\tilde\varphi_2$ where
\begin{align*}
  \tilde\varphi_1\equiv\meth'\varphi_1, \qquad \tilde\varphi_2\equiv\meth'\varphi_2+\mu\bar\varphi_1.
\end{align*}
Starting from the equations~\eqref{EOMMasslessScalarStT-weight3}
and~\eqref{EOMMasslessScalarStT-weight4} ---i.e.
\begin{align*}
\mthorn'\varphi_{1}-\meth\varphi_{2}&=
-\mu\varphi_{1}-\bar\lambda\bar\varphi_{1}, \\
\mthorn\varphi_{2}-\meth'\varphi_{1}&=-\mu\varphi_{0}
+\rho\varphi_{2}+\pi\varphi_{1}+\bar\pi\bar\varphi_{1},  \\
\end{align*}
we apply~$\meth'$ and make use of the commutator relations so as to
obtain
\begin{align}
\mthorn'\tilde{\varphi}_1+3\mu\tilde{\varphi}_1-\meth\tilde\varphi_2&=
-\lambda\meth\varphi_1-\varphi_1\meth\lambda-3\bar\varphi_1\meth\mu
-\bar\lambda\meth'\bar\varphi_1+\bar\varphi_1\meth'\bar\lambda, \label{tildevarphi1}\\
\mthorn\tilde\varphi_2-\meth'\tilde\varphi_1&=\bar\varphi_1\meth\pi+\bar\pi\meth'\bar\varphi_1
+2\pi\meth'\varphi_1-\bar\tau\meth'\varphi_1
-\varphi_0\meth'\mu+\varphi_1\meth'\pi+\bar\varphi_1\meth'\bar\pi+\varphi_2\meth'\rho \nonumber\\
&+\bar\varphi_1\tilde\varphi_2+2\rho\tilde\varphi_2+\bar\sigma\bar{\tilde{\varphi}}_2
+\pi^2\varphi_1+2\pi\bar\pi\bar\varphi_1+\pi\rho\varphi_2+\lambda\sigma\bar\varphi_1 \nonumber \\
&-\pi\bar\tau\varphi_1-\bar\pi\bar\tau\bar\varphi_1-\rho\bar\tau\varphi_2, \label{tildevarphi2}
\end{align}
where we have made use
of~$\bar{\tilde\varphi}_1=\meth\bar\varphi_1=\meth'\varphi_1=\tilde\varphi_1$
---i.e.~$\tilde\varphi_1$ is real. We can also use the same strategy
for the pair~$(\meth\varphi_2,\meth\varphi_1)$ as long as we
apply~$\meth$ to~\eqref{EOMMasslessScalarStT-weight3}
and~\eqref{EOMMasslessScalarStT-weight4}. By definition,
$s_2(\Tivarphi_2)=3/2$, $s_2(\Tivarphi_1)=1$. We now apply the energy
estimate to the pair~$(\Tivarphi_1,\Tivarphi_2)$.
Choosing~$\lambda_0=3=2s_2(\Tivarphi_1)+1$, we have that
\begin{align*}
&\int_0^v||\meth^k\Tivarphi_1||^2_{L^2_{sc}(\mathcal{S}_{u,v'})}
+\int_{u_{\infty}}^u\frac{a}{|u'|^2}||\meth^k\Tivarphi_2||^2_{L^2_{sc}(\mathcal{S}_{u',v})} \\
&\lesssim\int_0^v||\meth^k\Tivarphi_1||^2_{L^2_{sc}(\mathcal{S}_{u_{\infty},v'})}
+\int_{u_{\infty}}^u\frac{a}{|u'|^2}||\meth^k\Tivarphi_2||^2_{L^2_{sc}(\mathcal{S}_{u',0})} \\
&\quad+2\int_0^v\int_{u_{\infty}}^u\frac{a}{|u'|^2}
||\meth^k\Tivarphi_1||_{L^2_{sc}(\mathcal{S}_{u',v'})}
||P_k||_{L^2_{sc}(\mathcal{S}_{u',v'})} \\
&\quad+2\int_0^v\int_{u_{\infty}}^u\frac{a}{|u'|^2}
||\meth^k\Tivarphi_2||_{L^2_{sc}(\mathcal{S}_{u',v'})}
||Q_k||_{L^2_{sc}(\mathcal{S}_{u',v'})}.
\end{align*}
where~$0\leq k\leq10$. Multiplying by~$a^{k}$ we arrive to
\begin{align*}
&\int_0^v||(a^{\frac{1}{2}}\meth)^k\Tivarphi_1||^2_{L^2_{sc}(\mathcal{S}_{u,v'})}
+\int_{u_{\infty}}^u\frac{a}{|u'|^2}||(a^{\frac{1}{2}}\meth)^k\Tivarphi_2||^2_{L^2_{sc}(\mathcal{S}_{u',v})} \\
&\lesssim\int_0^v||(a^{\frac{1}{2}}\meth)^k\Tivarphi_1||^2_{L^2_{sc}(\mathcal{S}_{u_{\infty},v'})}
+\int_{u_{\infty}}^u\frac{a}{|u'|^2}||(a^{\frac{1}{2}}\meth)^k\Tivarphi_2||^2_{L^2_{sc}(\mathcal{S}_{u',0})} \\
&\quad+2\int_0^v\int_{u_{\infty}}^u\frac{a}{|u'|^2}
||(a^{\frac{1}{2}}\meth)^k\Tivarphi_1||_{L^2_{sc}(\mathcal{S}_{u',v'})}
||a^{\frac{k}{2}}P_k||_{L^2_{sc}(\mathcal{S}_{u',v'})} \\
&\quad+2\int_0^v\int_{u_{\infty}}^u\frac{a}{|u'|^2}
||(a^{\frac{1}{2}}\meth)^k\Tivarphi_2||_{L^2_{sc}(\mathcal{S}_{u',v'})}
||a^{\frac{k}{2}}Q_k||_{L^2_{sc}(\mathcal{S}_{u',v'})} \\
&=\int_0^v||(a^{\frac{1}{2}}\meth)^k\Tivarphi_1||^2_{L^2_{sc}(\mathcal{S}_{u_{\infty},v'})}
+\int_{u_{\infty}}^u\frac{a}{|u'|^2}||(a^{\frac{1}{2}}\meth)^k\Tivarphi_2||^2_{L^2_{sc}(\mathcal{S}_{u',0})} \\
&\quad+M+N,
\end{align*}
where
\begin{align*}
M&\equiv2\int_0^v\int_{u_{\infty}}^u\frac{a}{|u'|^2}
||(a^{\frac{1}{2}}\meth)^k\Tivarphi_1||_{L^2_{sc}(\mathcal{S}_{u',v'})}
||a^{\frac{k}{2}}P_k||_{L^2_{sc}(\mathcal{S}_{u',v'})}, \\
N&\equiv2\int_0^v\int_{u_{\infty}}^u\frac{a}{|u'|^2}
||(a^{\frac{1}{2}}\meth)^k\Tivarphi_2||_{L^2_{sc}(\mathcal{S}_{u',v'})}
||a^{\frac{k}{2}}Q_k||_{L^2_{sc}(\mathcal{S}_{u',v'})} ,
\end{align*}
and we define
\begin{align*}
P_k=\lambda\meth^{k+1}\varphi_1+\varphi_1\meth^{k+1}\Gamma(\Timu,\lambda)
+\sum_{i=1}^k\meth^{i}\Gamma(\Timu,\lambda)\meth^{k+1-i}\varphi_1
+\sum_{i=1}^k\meth^{i}\varphi_1\meth^{k+1-i}\Gamma(\Timu,\lambda)
\end{align*}
\begin{align*}
Q_k&=\sum_{i_1+i_2+i_3+i_4=k}\meth^{i_1}\Gamma(\tau,\pi)^{i_2}\meth^{i_3}\Gamma(\tau,\pi)\meth^{i_4+1}\varphi_1 
+\sum_{i_1+i_2+i_3+i_4=k}\meth^{i_1}\Gamma(\tau,\pi)^{i_2}\meth^{i_3}
\bm\varphi\meth^{i_4+1}\Gamma(\pi,\mu,\rho) \\
&+\sum_{i_1+i_2+i_3+i_4=k}\meth^{i_1}\Gamma(\tau,\pi)^{i_2}
\meth^{i_3}(\tau,\pi,\rho,\sigma,\varphi_1)\meth^{i_4}\Tivarphi_2 \\
&+\sum_{i_1+i_2+i_3+i_4+i_5=k}\meth^{i_1}\Gamma(\tau,\pi)^{i_2}\meth^{i_3}\Gamma(\lambda,\pi)
\meth^{i_4}\Gamma(\sigma,\tau,\pi,\rho)\meth^{i_5}\bm\varphi
+\sum_{i_1+i_2=k-1}\meth^{i_1}K\meth^{i_2+1}\varphi_1.
\end{align*}
The rest of the proof consists of three steps.

\smallskip
\noindent
\textbf{Step 1}: Estimate~$M$.
We have that 
\begin{align*}
M\leq&\left(\int_0^v\int_{u_{\infty}}^u\frac{a}{|u'|^2}
||(a^{\frac{1}{2}}\meth)^k\Tivarphi_1||^2_{L^2_{sc}(\mathcal{S}_{u',v'})}\right)^{\frac{1}{2}}
\left(\int_0^v\int_{u_{\infty}}^u\frac{a}{|u'|^2}||a^{\frac{k}{2}}P_k||^2_{L^2_{sc}(\mathcal{S}_{u',v'})}\right)^{\frac{1}{2}} \\
\leq&\int_0^v\int_{u_{\infty}}^u\frac{a}{|u'|^2}
||(a^{\frac{1}{2}}\meth)^k\Tivarphi_1||^2_{L^2_{sc}(\mathcal{S}_{u',v'})}
+\int_0^v\int_{u_{\infty}}^u\frac{a}{|u'|^2}||a^{\frac{k}{2}}P_k||^2_{L^2_{sc}(\mathcal{S}_{u',v'})}.
\end{align*}
Denoting the second integral by~$J$ and substituting~$P_k$ we have
that
\begin{align*}
J&\leq\int_0^v\int_{u_{\infty}}^u\frac{a}{|u'|^2}
||a^{\frac{k}{2}}\lambda\mathcal{D}^{k+1}\varphi_1||^2_{L^2_{sc}(\mathcal{S}_{u',v'})}
+\int_0^v\int_{u_{\infty}}^u\frac{a}{|u'|^2}
||a^{\frac{k}{2}}\varphi_1\mathcal{D}^{k+1}\Gamma(\Timu,\lambda)||^2_{L^2_{sc}(\mathcal{S}_{u',v'})} \\
&+\sum_{i=1}^k\int_0^v\int_{u_{\infty}}^u\frac{a}{|u'|^2}
||a^{\frac{k}{2}}\mathcal{D}^i\Gamma(\Timu,\lambda)\mathcal{D}^{k+1-i}\varphi_1||^2_{L^2_{sc}(\mathcal{S}_{u',v'})}\\
&+\sum_{i=1}^k\int_0^v\int_{u_{\infty}}^u\frac{a}{|u'|^2}
||a^{\frac{k}{2}}\mathcal{D}^{i}\varphi_1\mathcal{D}^{k+1-i}\Gamma(\Timu,\lambda)||^2_{L^2_{sc}(\mathcal{S}_{u',v'})}
=I_1+...+I_4.
\end{align*}

For the term~$I_1$ we have that
\begin{align*}
I_1&\leq\int_0^v\int_{u_{\infty}}^u\frac{a}{|u'|^2}\frac{1}{|u'|^2}||\lambda||^2_{L^{\infty}_{sc}(\mathcal{S}_{u',v'})}
||a^{\frac{k}{2}}\mathcal{D}^{k+1}\varphi_1||^2_{L^2_{sc}(\mathcal{S}_{u',v'})} \\
&\lesssim\int_0^v\int_{u_{\infty}}^u\frac{a}{|u'|^2}\frac{1}{|u'|^2}\frac{|u'|^2}{a}
||a^{\frac{k}{2}}\mathcal{D}^{k+1}\varphi_1||^2_{L^2_{sc}(\mathcal{S}_{u',v'})} \\
&=\int_{u_{\infty}}^u\frac{1}{|u'|^2}||a^{\frac{k}{2}}\mathcal{D}^{k+1}\varphi_1||^2_{L^2_{sc}(\mathcal{N}_{u'}(0,v))}
\leq\frac{1}{|u|}\bm{\varphi}[\varphi_1]\leq1.
\end{align*}
For~$I_2$ we have that
\begin{align*}
I_2&\leq\int_0^v\int_{u_{\infty}}^u\frac{a}{|u'|^2}\frac{1}{|u'|^2}||\varphi_1||^2_{L^{\infty}_{sc}(\mathcal{S}_{u',v'})}
||a^{\frac{k}{2}}\mathcal{D}^{k+1}\Gamma(\Timu,\lambda)||^2_{L^2_{sc}(\mathcal{S}_{u',v'})} \\
&\lesssim\frac{\bmGamma(\varphi_1)^2}{a}\int_0^v\int_{u_{\infty}}^u\frac{a^2}{|u'|^4}
||a^{\frac{k}{2}}\mathcal{D}^{k+1}\Gamma(\Timu,\lambda)||^2_{L^2_{sc}(\mathcal{S}_{u',v'})}
\lesssim\frac{1}{a}\mathcal{O}^4\leq1.
\end{align*}
For~$I_3$ we see that
\begin{align*}
I_3&\leq\sum_{i=1}^k\int_0^v\int_{u_{\infty}}^u\frac{1}{|u'|^2}\frac{1}{|u'|^2}
||(a^{\frac{1}{2}}\mathcal{D})^i\Gamma(\Timu,\lambda)||^2\times
||(a^{\frac{1}{2}}\mathcal{D})^{k+1-i}\varphi_1||^2 \\
&\lesssim\int_0^v\int_{u_{\infty}}^u\frac{1}{|u'|^4}\frac{|u'|^2}{a}\mathcal{O}^2\leq1.
\end{align*}
Finally, $I_4$ leads to the same estimate as for~$I_3$. Combining the
estimates for~$I_1,\dots,I_4$ we have~$J\lesssim1$.

\smallskip
\noindent
\textbf{Step 2}: Estimating~$N$.

We have that 
\begin{align*}
M&\leq\left(\int_0^v\int_{u_{\infty}}^u\frac{a}{|u'|^2}
||(a^{\frac{1}{2}}\meth)^k\Tivarphi_2||^2_{L^2_{sc}(\mathcal{S}_{u',v'})}\right)^{\frac{1}{2}}
\left(\int_0^v\int_{u_{\infty}}^u\frac{a}{|u'|^2}||a^{\frac{k}{2}}Q_k||^2_{L^2_{sc}(\mathcal{S}_{u',v'})}\right)^{\frac{1}{2}} \\
&\quad\leq\int_0^v\int_{u_{\infty}}^u\frac{a}{|u'|^2}
||(a^{\frac{1}{2}}\meth)^k\Tivarphi_2||^2_{L^2_{sc}(\mathcal{S}_{u',v'})}
+\int_0^v\int_{u_{\infty}}^u\frac{a}{|u'|^2}||a^{\frac{k}{2}}Q_k||^2_{L^2_{sc}(\mathcal{S}_{u',v'})}.
\end{align*}
Denoting the second integral by~$H$ and substituting~$Q_k$ we have that
\begin{align*}
H&\leq\sum_{i_1+i_2+i_3+i_4=k}\int_0^v\int_{u_{\infty}}^u\frac{a}{|u'|^2}
||a^{\frac{k}{2}}\mathcal{D}^{i_1}\Gamma(\tau,\pi)^{i_2}\mathcal{D}^{i_3}\Gamma(\tau,\pi)
\mathcal{D}^{i_4+1}\varphi_1||^2_{L^2_{sc}(\mathcal{S}_{u',v'})} \\
&\quad+\sum_{i_1+i_2+i_3+i_4=k}\int_0^v\int_{u_{\infty}}^u\frac{a}{|u'|^2}
||a^{\frac{k}{2}}\mathcal{D}^{i_1}\Gamma(\tau,\pi)^{i_2}\mathcal{D}^{i_3}\bm\varphi
\mathcal{D}^{i_4+1}\bmGamma(\pi,\mu,\rho)||^2_{L^2_{sc}(\mathcal{S}_{u',v'})} \\
&\quad+\sum_{i_1+i_2+i_3+i_4=k}\int_0^v\int_{u_{\infty}}^u\frac{a}{|u'|^2}
||a^{\frac{k}{2}}\mathcal{D}^{i_1}\Gamma(\tau,\pi)^{i_2}\mathcal{D}^{i_3}(\sigma,...,\varphi_1)
\mathcal{D}^{i_4}\Tivarphi_2||^2_{L^2_{sc}(\mathcal{S}_{u',v'})} \\
&\quad+\sum_{i_1+i_2+i_3+i_4+i_5=k}\int_0^v\int_{u_{\infty}}^u\frac{a}{|u'|^2}
||a^{\frac{k}{2}}\mathcal{D}^{i_1}\Gamma(\tau,\pi)^{i_2}\mathcal{D}^{i_3}\Gamma(\lambda,\pi)
\mathcal{D}^{i_4}\Gamma(\sigma,...)\mathcal{D}^{i_5}\bm\varphi||^2_{L^2_{sc}(\mathcal{S}_{u',v'})} \\
&\quad+\sum_{i_1+i_2=k-1}\int_0^v\int_{u_{\infty}}^u\frac{a}{|u'|^2}
||a^{\frac{k}{2}}\mathcal{D}^{i_1}K\mathcal{D}^{i_2+1}\varphi_1||^2_{L^2_{sc}(\mathcal{S}_{u',v'})} 
=I_1+...+I_5.
\end{align*}

For~$I_1$ we have that 
\begin{align*}
I_1&\leq\int_0^v\int_{u_{\infty}}^u\frac{a}{|u'|^2}\frac{1}{|u'|^2}||\Gamma(\tau,\pi)||^2_{L^{\infty}_{sc}(\mathcal{S}_{u',v'})} 
||a^{\frac{k}{2}}\mathcal{D}^{k+1}\varphi_1||^2_{L^2_{sc}(\mathcal{S}_{u',v'})} 
+\int_0^v\int_{u_{\infty}}^u\frac{1}{|u'|^2}(\frac{a^{\frac{i_2}{2}}}{|u'|^{i_2+1}}\mathcal{O}^{i_2+2})^2 \\
&\leq\int_{u_{\infty}}^u\frac{a}{|u'|^4}\mathcal{O}^2
||a^{\frac{k}{2}}\mathcal{D}^{k+1}\varphi_1||^2_{L^2_{sc}(\mathcal{N}_{u'}(0,v))}+\int_0^v\int_{u_{\infty}}^u\frac{\mathcal{O}^4}{|u'|^4}
\leq1.
\end{align*}
For~$I_2$ we find that 
\begin{align*}
I_2&\leq\int_0^v\int_{u_{\infty}}^u\frac{a}{|u'|^2}\frac{1}{|u'|^2}||\varphi_0||^2_{L^{\infty}_{sc}(\mathcal{S}_{u',v'})} 
||a^{\frac{k}{2}}\mathcal{D}^{k+1}\mu||^2_{L^2_{sc}(\mathcal{S}_{u',v'})} \\
&\quad+\int_0^v\int_{u_{\infty}}^u\frac{a}{|u'|^2}\frac{1}{|u'|^2}||\varphi_1||^2_{L^{\infty}_{sc}(\mathcal{S}_{u',v'})} 
||a^{\frac{k}{2}}\mathcal{D}^{k+1}\pi||^2_{L^2_{sc}(\mathcal{S}_{u',v'})} \\
&\quad+\int_0^v\int_{u_{\infty}}^u\frac{a}{|u'|^2}\frac{1}{|u'|^2}||\varphi_2||^2_{L^{\infty}_{sc}(\mathcal{S}_{u',v'})} 
||a^{\frac{k}{2}}\mathcal{D}^{k+1}\rho||^2_{L^2_{sc}(\mathcal{S}_{u',v'})} \\
&\quad+\int_0^v\int_{u_{\infty}}^u\frac{1}{|u'|^2}\left(\frac{a^{\frac{i_2}{2}}}{|u'|^{i_2+1}}\mathcal{O}^{i_2}(
a^{\frac{1}{2}}\mathcal{O}\frac{|u|}{a}\mathcal{O}+\mathcal{O}^2+\frac{|u|}{a^{\frac{1}{2}}}\mathcal{O}^2) \right)^2 \\
&\lesssim\int_0^v\int_{u_{\infty}}^u\frac{a^2}{|u'|^4}\bmGamma(\varphi_0)^2
||a^{\frac{k}{2}}\mathcal{D}^{k+1}\mu||^2_{L^2_{sc}(\mathcal{S}_{u',v'})}
+\int_0^v\int_{u_{\infty}}^u\frac{a}{|u'|^4}\bmGamma(\varphi_1)^2
||a^{\frac{k}{2}}\mathcal{D}^{k+1}\pi||^2_{L^2_{sc}(\mathcal{S}_{u',v'})} \\
&\quad+\int_0^v\int_{u_{\infty}}^u\frac{a}{|u'|^4}\frac{|u'|^2}{a}\bmGamma(\varphi_2)^2
||a^{\frac{k}{2}}\mathcal{D}^{k+1}\rho||^2_{L^2_{sc}(\mathcal{S}_{u',v'})}+\frac{\mathcal{O}^4}{|u|} \\
&\leq\frac{1}{a}\mathcal{O}^4+\frac{1}{a|u|}\mathcal{O}^4+\frac{\mathcal{O}^4}{|u|}
\leq1.
\end{align*}
For~$I_3$ we have that 
\begin{align*}
I_3\leq\int_0^v\int_{u_{\infty}}^u\frac{a}{|u'|^2}\left( 
\frac{a^{\frac{i_2}{2}}}{|u'|^{i_2+1}}\mathcal{O}^{i_2}a^{\frac{1}{2}}\mathcal{O}^2
\right)^2
\leq\frac{a^2}{|u|^3}\mathcal{O}^4\leq1.
\end{align*}
For~$I_4$ we have that 
\begin{align*}
I_4\leq\int_0^v\int_{u_{\infty}}^u\frac{a}{|u'|^2}\left( 
\frac{a^{\frac{i_2}{2}}}{|u'|^{i_2+2}}\mathcal{O}^{i_2}a^{\frac{1}{2}}\mathcal{O}\frac{|u'|}{a^{\frac{1}{2}}}\mathcal{O}
\right)^2
\leq\frac{a}{|u|^3}\mathcal{O}^4\leq1.
\end{align*}
Finally, for~$I_5$ we find that
\begin{align*}
I_5\leq\int_0^v\int_{u_{\infty}}^u\frac{1}{|u'|^2}\frac{1}{|u'|^2}\mathcal{O}^2\leq1.
\end{align*}
Combine the above estimates for~$I_1,\ldots,I_5$ we have~$H\lesssim1$.
 
\smallskip
\noindent
\textbf{Step 3}: Summary for~$0\leq k\leq1$.

 Collecting the estimates for~$M$ and~$N$ we have that 
\begin{align*}
\int_0^v||(a^{\frac{1}{2}}\meth)^k\Tivarphi_1||^2_{L^2_{sc}(\mathcal{S}_{u,v'})}
+\int_{u_{\infty}}^u\frac{a}{|u'|^2}||(a^{\frac{1}{2}}\meth)^k\Tivarphi_2||^2_{L^2_{sc}(\mathcal{S}_{u',v})} 
\lesssim\mathcal{I}_0^2+1.
\end{align*}
Observing that the rest of the strings in~$\mathcal{D}^{k_1}$ give
rise to the same results, we conclude that 
\begin{align*}
\int_0^v||(a^{\frac{1}{2}}\mathcal{D})^k\Tivarphi_1||^2_{L^2_{sc}(\mathcal{S}_{u,v'})}
+\int_{u_{\infty}}^u\frac{a}{|u'|^2}||(a^{\frac{1}{2}}\mathcal{D})^k\Tivarphi_2||^2_{L^2_{sc}(\mathcal{S}_{u',v})} 
\lesssim\mathcal{I}_0^2+1.
\end{align*}
We can also consider the pair~$(\meth\varphi_1, \meth\varphi_2+\mu\varphi_1)$ and obtain that 
\begin{align*}
\int_0^v||(a^{\frac{1}{2}}\mathcal{D})^k\meth\varphi_1||^2_{L^2_{sc}(\mathcal{S}_{u,v'})}
+\int_{u_{\infty}}^u\frac{a}{|u'|^2}||(a^{\frac{1}{2}}\mathcal{D})^k(\meth\varphi_2+\mu\varphi_1)||^2_{L^2_{sc}(\mathcal{S}_{u',v})} 
\lesssim\mathcal{I}_0^2+1.
\end{align*}
Thus, for~$\varphi_1$, we have that
\begin{align*}
\int_0^v||a^{\frac{k}{2}}\mathcal{D}^{k+1}\varphi_1||^2_{L^2_{sc}(\mathcal{S}_{u,v'})}
\lesssim\mathcal{I}_0^2+1.
\end{align*}
Similarly, for~$\varphi_2$ we have that 
\begin{align*}
&\int_{u_{\infty}}^u\frac{a^2}{|u'|^4}||a^{\frac{k}{2}}\mathcal{D}^{k+1}\varphi_2||^2_{L^2_{sc}(\mathcal{S}_{u',v})} 
\leq\int_{u_{\infty}}^u\frac{a^2}{|u'|^4}||(a^{\frac{1}{2}}\mathcal{D})^k(\meth\varphi_2+\mu\varphi_1)||^2_{L^2_{sc}(\mathcal{S}_{u',v})} \\
&\quad+\int_{u_{\infty}}^u\frac{a^2}{|u'|^4}||(a^{\frac{1}{2}}\mathcal{D})^k(\meth'\varphi_2+\mu\bar\varphi_1)||^2_{L^2_{sc}(\mathcal{S}_{u',v})} 
+\int_{u_{\infty}}^u\frac{a^2}{|u'|^4}||(a^{\frac{1}{2}}\mathcal{D})^k(\mu\varphi_1)||^2_{L^2_{sc}(\mathcal{S}_{u',v})} \\
&\lesssim1+\int_{u_{\infty}}^u\frac{a^2}{|u'|^4}\frac{1}{|u'|^2}||\mu||^2_{L^{\infty}_{sc}(\mathcal{S}_{u',v})}
||(a^{\frac{1}{2}}\mathcal{D})^k\varphi_1||^2_{L^2_{sc}(\mathcal{S}_{u',v})} \\
&\quad+\sum_{i=1}^k\int_{u_{\infty}}^u\frac{a^2}{|u'|^4}\frac{1}{|u'|^2}
||(a^{\frac{1}{2}}\mathcal{D})^i\Timu||^2_{L^{\infty}_{sc}(\mathcal{S}_{u',v})}
||(a^{\frac{1}{2}}\mathcal{D})^{k-i}\varphi_1||^2_{L^2_{sc}(\mathcal{S}_{u',v})} \\
&\lesssim1+\int_{u_{\infty}}^u\frac{a^2}{|u'|^6}\frac{|u'|^4}{a^2}\mathcal{O}^2
+\int_{u_{\infty}}^u\frac{a^2}{|u'|^6}\frac{|u'|^2}{a^2}\mathcal{O}^2
\lesssim1.
\end{align*}
Hence, one can conclude that
\begin{align*}
&\int_0^v||a^{\frac{k}{2}}\mathcal{D}^{k+1}\varphi_1||^2_{L^2_{sc}(\mathcal{S}_{u,v'})}
+\int_{u_{\infty}}^u\frac{a}{|u'|^2}||(a^{\frac{1}{2}}\mathcal{D})^k\Tivarphi_2||^2_{L^2_{sc}(\mathcal{S}_{u',v})} \\
&\quad+\int_{u_{\infty}}^u\frac{a^2}{|u'|^4}||a^{\frac{k}{2}}\mathcal{D}^{k+1}\varphi_2||^2_{L^2_{sc}(\mathcal{S}_{u',v})} 
\lesssim\mathcal{I}_0^2+1.
\end{align*}
for~$0\leq k\leq10$.

The estimates of~$\varphi_1$ and~$\varphi_2$ on the light cone give
the same result. For the sake of conciseness we do not discuss the
details.
\end{proof}

\subsection{Energy estimate for the Weyl curvature}

Finally, in this section we obtain the energy estimates for the
components of the Weyl tensor.

\begin{proposition}
\label{EnergyEstimatePsi01}
For~$0\leq k\leq10$, one has that 
\begin{align*}
\frac{1}{a}||(a^{\frac{1}{2}}\mathcal{D})^k\Psi_0||^2_{L^2_{sc}(\mathcal{N}_u(0,v))}
+\frac{1}{a}||(a^{\frac{1}{2}}\mathcal{D})^k\TiPsi_1||^2_{L^2_{sc}(\mathcal{N}'_v(u_{\infty},u))} 
\lesssim\mathcal{I}_0^2+\frac{1}{a^{\frac{1}{4}}}.
\end{align*}
\end{proposition}

\begin{proof}
  For the pair~$(\Psi_0,\TiPsi_1)$, we make use of the Bianchi
  equations~\eqref{T-weightMasslessStBianchi2}
  and~\eqref{T-weightMasslessStBianchi1}
\begin{align*}
\mthorn'\Psi_0-\meth\TiPsi_1&=6\varphi_0\meth\varphi_1+(2\ulomega-\mu)\Psi_0
-5\tau\TiPsi_1+3\sigma\TiPsi_2
\nonumber \\
&+6\varphi_0\varphi_2\sigma-3\varphi_0^2\bar\lambda-3\varphi_1^2\rho
-6\varphi_1\bar\varphi_1\sigma-12\varphi_0\varphi_1\tau,  \\
\mthorn\TiPsi_1-\meth'\Psi_0=&-6\varphi_0\meth\varphi_0+(\pi-2\bar\tau)\Psi_0
+4\rho\TiPsi_1
 \nonumber\\
&+3\varphi_0^2(2\tau-\bar\pi)+6\varphi_0\varphi_1\rho-6\varphi_0\bar\varphi_1\sigma.
\end{align*}
Applying~$\meth^k$ and commuting with~$\mthorn$ and~$\mthorn'$ we have that
\begin{align*}
\mthorn'\meth^k\Psi_0+(k+1)\mu\meth^k\TiPsi_1-\meth^{k+1}\TiPsi_1&=P_k, \\
\mthorn\meth^k\TiPsi_1-\meth'\meth^k\Psi_0&=Q_k,
\end{align*}
where
\begin{align*}
P_k&=\varphi_0\meth^{k+1}\varphi_1+\ulomega\meth^k\Psi_0+\tau\meth^k\TiPsi_1+\sigma\meth^k\TiPsi_2
+\sum_{i=1}^k\meth^i\varphi_0\meth^{k+1-i}\varphi_1
+\sum_{i=1}^k\meth^i\Gamma(\ulomega,\Timu)\meth^{k-i}\Psi_0 \\
&\quad+\sum_{i=0}^k\meth^i\lambda\meth^{k-i}\Psi_0
+\sum_{i=1}^k\meth^i\Gamma(\tau,\sigma)\meth^{k-i}(\TiPsi_1,\TiPsi_2) 
+\sum_{i_1+i_2+i_3=k}\meth^{i_1}\Gamma(\rho,\sigma,\lambda,\tau)\meth^{i_2}\bm\varphi\meth^{i_3}\bm\varphi,
\end{align*}

\begin{align*}
Q_k&=\varphi_0\meth^{k+1}\varphi_0+\Gamma(\tau,\pi)\meth^k\Psi_0+\rho\meth^k\TiPsi_1
+\sum_{i_1+i_2+i_3+i_4=k,i_4<k}\meth^{i_1}\Gamma(\tau,\pi)^{i_2}\meth^{i_3}\varphi_0\meth^{i_4+1}\varphi_0 \\
&\quad+\sum_{i_1+i_2+i_3+i_4=k,i_3+i_4<k}\meth^{i_1}\Gamma(\tau,\pi)^{i_2}
\meth^{i_3}\Gamma(\tau,\pi,\rho,\sigma)\meth^{i_4}(\Psi_0,\TiPsi_1) \\
&\quad+\sum_{i_1+i_2+i_3+i_4+i_5=k}\meth^{i_1}\Gamma(\tau,\pi)^{i_2}
\meth^{i_3}\Gamma(\tau,\pi,\rho,\sigma)\meth^{i_4}\bm\varphi\meth^{i_5}\bm\varphi .
+\sum_{i_1+i_2=k-1}\meth^{i_1}K\meth^{i_2}\Psi_{0}.
\end{align*}
It follows then that
\begin{align*}
&||(a^{\frac{1}{2}}\meth)^k\Psi_0||^2_{L^2_{sc}(\mathcal{N}_u(0,v))}
+||(a^{\frac{1}{2}}\meth)^k\TiPsi_1||^2_{L^2_{sc}(\mathcal{N}'_v(u_{\infty},u))} \\
&\lesssim||(a^{\frac{1}{2}}\meth)^k\Psi_0||^2_{L^2_{sc}(\mathcal{N}_{u_{\infty}}(0,v))}
+||(a^{\frac{1}{2}}\meth)^k\TiPsi_1||^2_{L^2_{sc}(\mathcal{N}'_0(u_{\infty},u))} \\
&\quad+2\int_0^v\int_{u_{\infty}}^u\frac{a}{|u'|^2}
||(a^{\frac{1}{2}}\meth)^k\Psi_0||_{L^2_{sc}(\mathcal{S}_{u',v'})}
||(a^{\frac{1}{2}})^{k}P_k||_{L^2_{sc}(\mathcal{S}_{u',v'})} \\
&\quad+2\int_0^v\int_{u_{\infty}}^u\frac{a}{|u'|^2}
||(a^{\frac{1}{2}}\meth)^k\TiPsi_1||_{L^2_{sc}(\mathcal{S}_{u',v'})}
||(a^{\frac{1}{2}})^{k}Q_k||_{L^2_{sc}(\mathcal{S}_{u',v'})} \\
&\leq||(a^{\frac{1}{2}}\meth)^k\Psi_0||^2_{L^2_{sc}(\mathcal{N}_{u_{\infty}}(0,v))}
+||(a^{\frac{1}{2}}\meth)^k\TiPsi_1||^2_{L^2_{sc}(\mathcal{N}'_0(u_{\infty},u))} 
+M+N,
\end{align*}
where 
\begin{align*}
M&\equiv2\int_0^v\int_{u_{\infty}}^u\frac{a}{|u'|^2}
||(a^{\frac{1}{2}}\meth)^k\Psi_0||_{L^2_{sc}(\mathcal{S}_{u',v'})}
||(a^{\frac{1}{2}})^{k}P_k||_{L^2_{sc}(\mathcal{S}_{u',v'})} \\
&\leq2\left(\int_0^v\int_{u_{\infty}}^u\frac{a}{|u'|^2}
||(a^{\frac{1}{2}}\meth)^k\Psi_0||^2_{L^2_{sc}(\mathcal{S}_{u',v'})}\right)^{\frac{1}{2}}
\left(\int_0^v\int_{u_{\infty}}^u\frac{a}{|u'|^2}
||(a^{\frac{1}{2}})^{k}P_k||^2_{L^2_{sc}(\mathcal{S}_{u',v'})}\right)^{\frac{1}{2}} \\
&\leq\left(\int_{u_{\infty}}^u\frac{a}{|u'|^2}
||(a^{\frac{1}{2}}\meth)^k\Psi_0||^2_{L^2_{sc}(\mathcal{N}_{u'}(0,v))}\right)^{\frac{1}{2}}J^{\frac{1}{2}}
\left(\int_0^v\int_{u_{\infty}}^u\frac{a}{|u'|^2}
||(a^{\frac{1}{2}})^{k}P_k||^2_{L^2_{sc}(\mathcal{S}_{u',v'})}\right)^{\frac{1}{2}} \\
&\leq\frac{a}{|u|^{\frac{1}{2}}}\bm\Psi[\Psi_0]\left(\int_0^v\int_{u_{\infty}}^u\frac{a}{|u'|^2}
||(a^{\frac{1}{2}})^{k}P_k||^2_{L^2_{sc}(\mathcal{S}_{u',v'})}\right)^{\frac{1}{2}} ,
\end{align*}

\begin{align*}
N&\equiv2\int_0^v\int_{u_{\infty}}^u\frac{a}{|u'|^2}
||(a^{\frac{1}{2}}\meth)^k\TiPsi_1||_{L^2_{sc}(\mathcal{S}_{u',v'})}
||(a^{\frac{1}{2}})^{k}Q_k||_{L^2_{sc}(\mathcal{S}_{u',v'})} \\
&\leq2\left(\int_0^v\int_{u_{\infty}}^u\frac{a}{|u'|^2}
||(a^{\frac{1}{2}}\meth)^k\TiPsi_1||^2_{L^2_{sc}(\mathcal{S}_{u',v'})}\right)^{\frac{1}{2}}
\left(\int_0^v\int_{u_{\infty}}^u\frac{a}{|u'|^2}
||(a^{\frac{1}{2}})^{k}Q_k||^2_{L^2_{sc}(\mathcal{S}_{u',v'})}\right)^{\frac{1}{2}}\\
&\leq a^{\frac{1}{2}}\underline{\bm\Psi}[\TiPsi_1]\left(\int_0^v\int_{u_{\infty}}^u\frac{a}{|u'|^2}
||(a^{\frac{1}{2}})^{k}Q_k||^2_{L^2_{sc}(\mathcal{S}_{u',v'})}\right)^{\frac{1}{2}}.
\end{align*}
It is convenient to define
\begin{align*}
J\equiv\int_0^v\int_{u_{\infty}}^u\frac{a}{|u'|^2}
||(a^{\frac{1}{2}})^{k}P_k||^2_{L^2_{sc}(\mathcal{S}_{u',v'})}, \quad
H\equiv\int_0^v\int_{u_{\infty}}^u\frac{a}{|u'|^2}
||(a^{\frac{1}{2}})^{k}Q_k||^2_{L^2_{sc}(\mathcal{S}_{u',v'})}.
\end{align*}

The rest of the proof proceeds in three steps.

\smallskip
\noindent
\textbf{Step 1}: Estimate~$J$.

Substituting~$P_k$ we have that
\begin{align*}
J&\leq\int_0^v\int_{u_{\infty}}^u\frac{a}{|u'|^2}
||(a^{\frac{1}{2}})^{k}\varphi_0\mathcal{D}^{k+1}\varphi_1||^2_{L^2_{sc}(\mathcal{S}_{u',v'})}
+\int_0^v\int_{u_{\infty}}^u\frac{a}{|u'|^2}
||(a^{\frac{1}{2}})^{k}\ulomega\mathcal{D}^{k}\Psi_0||^2_{L^2_{sc}(\mathcal{S}_{u',v'})}\\
&\quad+\int_0^v\int_{u_{\infty}}^u\frac{a}{|u'|^2}
||(a^{\frac{1}{2}})^{k}\tau\mathcal{D}^{k}\TiPsi_1||^2_{L^2_{sc}(\mathcal{S}_{u',v'})}
+\int_0^v\int_{u_{\infty}}^u\frac{a}{|u'|^2}
||(a^{\frac{1}{2}})^{k}\sigma\mathcal{D}^{k}\TiPsi_2||^2_{L^2_{sc}(\mathcal{S}_{u',v'})} \\
&\quad+\sum_{i=1}^k\int_0^v\int_{u_{\infty}}^u\frac{a}{|u'|^2}
||(a^{\frac{1}{2}})^{k}\mathcal{D}^i\varphi_0\mathcal{D}^{k+1-i}\varphi_1||^2_{L^2_{sc}(\mathcal{S}_{u',v'})} \\
&\quad+\sum_{i=0}^k\int_0^v\int_{u_{\infty}}^u\frac{a}{|u'|^2}
||(a^{\frac{1}{2}})^{k}\mathcal{D}^i\Gamma(\lambda,...)\mathcal{D}^{k-i}\Psi||^2_{L^2_{sc}(\mathcal{S}_{u',v'})}\\
&\quad+\sum_{i_1+i_2+i_3=k}\int_0^v\int_{u_{\infty}}^u\frac{a}{|u'|^2}
||(a^{\frac{1}{2}})^{k}\mathcal{D}^{i_1}\Gamma(\lambda,...)\mathcal{D}^{i_2}\bm\varphi
\mathcal{D}^{i_3}\bm\varphi||^2_{L^2_{sc}(\mathcal{S}_{u',v'})}
=I_1+...+I_7.
\end{align*}
For the term~$I_1$ we have that
\begin{align*}
I_1&\leq\int_0^v\int_{u_{\infty}}^u\frac{a}{|u'|^2}\frac{1}{|u'|^2}||\varphi_0||^2_{L^{\infty}_{sc}(\mathcal{S}_{u',v'})}
||(a^{\frac{1}{2}})^{k}\mathcal{D}^{k+1}\varphi_1||^2_{L^2_{sc}(\mathcal{S}_{u',v'})} \\
&\leq\int_{u_{\infty}}^u\frac{a}{|u'|^4}\bmGamma(\varphi_0)_{0,\infty}^2
\int_0^v||(a^{\frac{1}{2}})^{k}\mathcal{D}^{k+1}\varphi_1||^2_{L^2_{sc}(\mathcal{S}_{u',v'})} \\
&\leq\int_{u_{\infty}}^u\frac{a}{|u'|^4}\bmGamma(\varphi_0)_{0,\infty}^2\bm\varphi[\varphi_1]^2
\leq\frac{a}{|u|^3}\bmGamma(\varphi_0)_{0,\infty}^2\bm\varphi[\varphi_1]^2.
\end{align*}
For~$I_2$ we have that
\begin{align*}
I_2&\leq\int_0^v\int_{u_{\infty}}^u\frac{a}{|u'|^2}\frac{1}{|u'|^2}||\ulomega||^2_{L^{\infty}_{sc}(\mathcal{S}_{u',v'})}
||(a^{\frac{1}{2}})^{k}\mathcal{D}^{k}\Psi_0||^2_{L^2_{sc}(\mathcal{S}_{u',v'})} \\
&\leq\int_{u_{\infty}}^u\frac{a}{|u'|^4}\bmGamma(\ulomega)_{0,\infty}^2
\int_0^v||(a^{\frac{1}{2}})^{k}\mathcal{D}^{k}\Psi_0||^2_{L^2_{sc}(\mathcal{S}_{u',v'})} \\
&\leq\int_{u_{\infty}}^u\frac{a^2}{|u'|^4}\bmGamma(\ulomega)_{0,\infty}^2\bm\Psi[\Psi_0]^2
\leq\frac{a^2}{|u|^3}\bmGamma(\ulomega)_{0,\infty}^2\bm\Psi[\Psi_0]^2.
\end{align*}

For~$I_3$ we see that 
\begin{align*}
I_3&\leq\int_0^v\int_{u_{\infty}}^u\frac{a}{|u'|^2}\frac{1}{|u'|^2}||\tau||^2_{L^{\infty}_{sc}(\mathcal{S}_{u',v'})}
||(a^{\frac{1}{2}})^{k}\mathcal{D}^{k}\TiPsi_1||^2_{L^2_{sc}(\mathcal{S}_{u',v'})} \\
&\leq\int_{u_{\infty}}^u\frac{a}{|u'|^4}\bmGamma(\tau)_{0,\infty}^2
\int_0^v||(a^{\frac{1}{2}})^{k}\mathcal{D}^{k}\TiPsi_1||^2_{L^2_{sc}(\mathcal{S}_{u',v'})} \\
&\leq\int_{u_{\infty}}^u\frac{a}{|u'|^4}\bmGamma(\tau)_{0,\infty}^2\bm\Psi[\TiPsi_1]^2
\leq\frac{a}{|u|^3}\bmGamma(\tau)_{0,\infty}^2\bm\Psi[\TiPsi_1]^2.
\end{align*}
For~$I_4$ we observe that 
\begin{align*}
I_4&\leq\int_0^v\int_{u_{\infty}}^u\frac{a}{|u'|^2}\frac{1}{|u'|^2}||\sigma||^2_{L^{\infty}_{sc}(\mathcal{S}_{u',v'})}
||(a^{\frac{1}{2}})^{k}\mathcal{D}^{k}\TiPsi_2||^2_{L^2_{sc}(\mathcal{S}_{u',v'})} \\
&\leq\int_{u_{\infty}}^u\frac{a^2}{|u'|^4}\bmGamma(\sigma)_{0,\infty}^2
\int_0^v||(a^{\frac{1}{2}})^{k}\mathcal{D}^{k}\TiPsi_2||^2_{L^2_{sc}(\mathcal{S}_{u',v'})} \\
&\leq\int_{u_{\infty}}^u\frac{a^2}{|u'|^4}\bmGamma(\sigma)_{0,\infty}^2\bm\Psi[\TiPsi_2]^2
\leq\frac{a^2}{|u|^3}\bmGamma(\sigma)_{0,\infty}^2\bm\Psi[\TiPsi_2]^2.
\end{align*}
For~$I_5$ we have that
\begin{align*}
I_5&\leq\int_0^v\int_{u_{\infty}}^u\frac{1}{|u'|^2}\frac{1}{|u'|^2}||(a^{\frac{1}{2}}\mathcal{D})^i\varphi_0||^2\times
||(a^{\frac{1}{2}}\mathcal{D})^{k+1-i}\varphi_1||^2 \\
&\leq\int_0^v\int_{u_{\infty}}^u\frac{1}{|u'|^4}a\bmGamma(\varphi_0)^2\bmGamma(\varphi_1)^2
\leq\frac{a}{|u|^3}\mathcal{O}^4.
\end{align*}
For~$I_6$ we have that
\begin{align*}
I_6&\leq\int_0^v\int_{u_{\infty}}^u\frac{a}{|u'|^2}\frac{1}{|u'|^2}||(a^{\frac{1}{2}}\mathcal{D})^i\Gamma(\lambda,...)||^2\times
||(a^{\frac{1}{2}}\mathcal{D})^{k-i}\Psi||^2 \\
&\leq\int_0^v\int_{u_{\infty}}^u\frac{a}{|u'|^4}\frac{|u'|^2}{a}\bmGamma(\lambda)^2a\bmGamma(\Psi_0)^2
\leq\int_0^v\int_{u_{\infty}}^u\frac{a}{|u'|^2}\leq\frac{a}{|u|}\leq1.
\end{align*}
Finally, for this term~$I_7$ we find that
\begin{align*}
I_7&\leq\int_0^v\int_{u_{\infty}}^u\frac{a}{|u'|^2}\frac{1}{|u'|^4}||(a^{\frac{1}{2}}\mathcal{D})^{i_1}\Gamma(\lambda,...)||^2\times
||(a^{\frac{1}{2}}\mathcal{D})^{i_2}\bm\varphi||^2\times||(a^{\frac{1}{2}}\mathcal{D})^{i_2}\bm\varphi||^2 \\
&\leq\int_0^v\int_{u_{\infty}}^u\frac{a}{|u'|^6}\frac{|u'|^2}{a}\mathcal{O}^2a\mathcal{O}^2a\mathcal{O}^2
\leq\frac{a^2}{|u|^3}\mathcal{O}^6.
\end{align*}
Collecting the estimates for~$I_1,\ldots, I_7$ we conclude that
\begin{align*}
J\leq\frac{1}{|u|}\mathcal{O}^4,
\end{align*}
so that, in fact, one has that 
\begin{align*}
M\lesssim\frac{a}{|u|}\bm\Psi[\Psi_0]\mathcal{O}^2.
\end{align*}

\smallskip
\noindent
\textbf{Step 2}: Estimate~$H$.

Substituting~$Q_k$ in the definition for~$H$ we have that
\begin{align*}
H&\leq\int_0^v\int_{u_{\infty}}^u\frac{a}{|u'|^2}
||(a^{\frac{1}{2}})^{k}\varphi_0\mathcal{D}^{k+1}\varphi_0||^2_{L^2_{sc}(\mathcal{S}_{u',v'})} 
+\int_0^v\int_{u_{\infty}}^u\frac{a}{|u'|^2}
||(a^{\frac{1}{2}})^{k}\Gamma(\tau,\pi)\mathcal{D}^{k}\Psi_0||^2_{L^2_{sc}(\mathcal{S}_{u',v'})} \\
&\quad+\int_0^v\int_{u_{\infty}}^u\frac{a}{|u'|^2}
||(a^{\frac{1}{2}})^{k}\rho\mathcal{D}^{k}\TiPsi_1||^2_{L^2_{sc}(\mathcal{S}_{u',v'})} \\
&\quad+\sum_{i_1+i_2+i_3+i_4=k,i_4<k}\int_0^v\int_{u_{\infty}}^u\frac{a}{|u'|^2}
||(a^{\frac{1}{2}})^{k}\mathcal{D}^{i_1}\Gamma^{i_2}\mathcal{D}^{i_3}\varphi_0
\mathcal{D}^{i_4+1}\varphi_0||^2_{L^2_{sc}(\mathcal{S}_{u',v'})} \\
&\quad+\sum_{i_1+i_2+i_3+i_4=k,i_3+i_4<k}\int_0^v\int_{u_{\infty}}^u\frac{a}{|u'|^2}
||(a^{\frac{1}{2}})^{k}\mathcal{D}^{i_1}\Gamma^{i_2}\mathcal{D}^{i_3}\Gamma(\sigma,...)
\mathcal{D}^{i_4}\Psi||^2_{L^2_{sc}(\mathcal{S}_{u',v'})} \\
&\quad+\sum_{i_1+i_2+i_3+i_4+i_5=k}\int_0^v\int_{u_{\infty}}^u\frac{a}{|u'|^2}
||(a^{\frac{1}{2}})^{k}\mathcal{D}^{i_1}\Gamma^{i_2}\mathcal{D}^{i_3}\Gamma(\sigma,...)
\mathcal{D}^{i_4}\varphi\mathcal{D}^{i_5}\varphi||^2_{L^2_{sc}(\mathcal{S}_{u',v'})} \\
&\quad+\sum_{i_1+i_2=k-1}\int_0^v\int_{u_{\infty}}^u\frac{a}{|u'|^2}
||(a^{\frac{1}{2}})^{k}\mathcal{D}^{i_1}K\mathcal{D}^{i_2}\Psi_0||^2_{L^2_{sc}(\mathcal{S}_{u',v'})} 
=I_1+...+I_7.
\end{align*}

For the term~$I_1$ we have that
\begin{align*}
I_1&\leq\int_0^v\int_{u_{\infty}}^u\frac{a}{|u'|^2}\frac{1}{|u'|^2}||\varphi_0||^2_{L^{\infty}_{sc}(\mathcal{S}_{u',v'})}
||(a^{\frac{1}{2}})^{k}\mathcal{D}^{k+1}\varphi_0||^2_{L^2_{sc}(\mathcal{S}_{u',v'})} \\
&\leq\int_{u_{\infty}}^u\frac{a^2}{|u'|^4}\bmGamma(\varphi_0)_{0,\infty}^2
\int_0^v||(a^{\frac{1}{2}})^{k}\mathcal{D}^{k+1}\varphi_0||^2_{L^2_{sc}(\mathcal{S}_{u',v'})} \\
&\leq\int_{u_{\infty}}^u\frac{a^3}{|u'|^4}\bmGamma(\varphi_0)_{0,\infty}^2\bm\varphi[\varphi_0]^2
\leq\frac{a^3}{|u|^3}\bmGamma(\varphi_0)_{0,\infty}^2\bm\varphi[\varphi_0]^2
\leq\bm\varphi[\varphi_0]^2.
\end{align*}
For~$I_2$ we see that
\begin{align*}
I_2&\leq\int_0^v\int_{u_{\infty}}^u\frac{a}{|u'|^2}\frac{1}{|u'|^2}||\Gamma(\tau,\pi)||^2_{L^{\infty}_{sc}(\mathcal{S}_{u',v'})}
||(a^{\frac{1}{2}})^{k}\mathcal{D}^{k}\Psi_0||^2_{L^2_{sc}(\mathcal{S}_{u',v'})} \\
&\leq\int_{u_{\infty}}^u\frac{a}{|u'|^4}\bmGamma(\tau,\pi)_{0,\infty}^2
\int_0^v||(a^{\frac{1}{2}})^{k}\mathcal{D}^{k}\Psi_0||^2_{L^2_{sc}(\mathcal{S}_{u',v'})} \\
&\leq\int_{u_{\infty}}^u\frac{a^2}{|u'|^4}\bmGamma(\tau,\pi)_{0,\infty}^2\bm\Psi[\Psi_0]^2
\leq\frac{a^2}{|u|^3}\bmGamma(\tau,\pi)_{0,\infty}^2\bm\Psi[\Psi_0]^2.
\end{align*}
For~$I_3$ we have that
\begin{align*}
I_3&\leq\int_0^v\int_{u_{\infty}}^u\frac{a}{|u'|^2}\frac{1}{|u'|^2}||\rho||^2_{L^{\infty}_{sc}(\mathcal{S}_{u',v'})}
||(a^{\frac{1}{2}})^{k}\mathcal{D}^{k}\Psi_1||^2_{L^2_{sc}(\mathcal{S}_{u',v'})} \\
&\leq\int_{u_{\infty}}^u\frac{a}{|u'|^4}\bmGamma(\rho)_{0,\infty}^2
\int_0^v||(a^{\frac{1}{2}})^{k}\mathcal{D}^{k}\Psi_1||^2_{L^2_{sc}(\mathcal{S}_{u',v'})} \\
&\leq\int_{u_{\infty}}^u\frac{a}{|u'|^4}\bmGamma(\rho)_{0,\infty}^2\bm\Psi[\Psi_1]^2
\leq\frac{a}{|u|^3}\bmGamma(\rho)_{0,\infty}^2\bm\Psi[\Psi_1]^2.
\end{align*}
For~$I_4$ we have
\begin{align*}
I_4&=\sum_{i_1+...+i_4=k,i_4<k}\int_{u_{\infty}}^u\frac{1}{|u'|^2}
||a^{\frac{k+1}{2}}\mathcal{D}^{j_1}\Gamma...\mathcal{D}^{j_{i_2}}\Gamma
\mathcal{D}^{i_3}\varphi_0\mathcal{D}^{i_4+1}\varphi_0||^2 \\
&\leq\int_{u_{\infty}}^u\frac{1}{|u'|^2}\left(\frac{(a^{\frac{1}{2}})^{i_2}}{|u|^{i_2+1}}||(a^{\frac{1}{2}}\mathcal{D})^{j_1}\Gamma||...
||(a^{\frac{1}{2}}\mathcal{D})^{j_{i_2}}\Gamma||\times
||(a^{\frac{1}{2}}\mathcal{D})^{i_3}\varphi_0||\times||(a^{\frac{1}{2}}\mathcal{D})^{i_4+1}\varphi_0|| \right)^2 \\
&\leq\int_{u_{\infty}}^u\frac{1}{|u'|^2}
(\frac{(a^{\frac{1}{2}})^{i_2}}{|u|^{i_2+1}}\mathcal{O}^{i_2}a^{\frac{1}{2}}\mathcal{O}a^{\frac{1}{2}}\mathcal{O})^2
=\int_{u_{\infty}}^u\frac{a^{i_2+2}}{|u'|^{2i_2+4}}\mathcal{O}^{2(i_2+2)}
\leq\int_{u_{\infty}}^u\frac{a^2}{|u'|^4}\mathcal{O}^4\leq\frac{a^2}{|u|^3}\mathcal{O}^4.
\end{align*}
For~$I_5$ we have that
\begin{align*}
I_5&\leq\int_{u_{\infty}}^u\frac{a}{|u'|^2}\left(\frac{(a^{\frac{1}{2}})^{i_2}}{|u|^{i_2+1}}||(a^{\frac{1}{2}}\mathcal{D})^{j_1}\Gamma||...
||(a^{\frac{1}{2}}\mathcal{D})^{j_{i_2}}\Gamma||\times
||(a^{\frac{1}{2}}\mathcal{D})^{i_3}\Gamma(\sigma,...)||\times||(a^{\frac{1}{2}}\mathcal{D})^{i_4}\Psi|| \right)^2 \\
&\leq\int_{u_{\infty}}^u\frac{a}{|u'|^2}\left(\frac{(a^{\frac{1}{2}})^{i_2}}{|u|^{i_2+1}}
\mathcal{O}^{i_2}a^{\frac{1}{2}}\mathcal{O}\mathcal{O} \right)^2
\leq\int_{u_{\infty}}^u\frac{a}{|u'|^2}\frac{a}{|u'|^2}\mathcal{O}^4
\leq\frac{a^2}{|u|^3}\mathcal{O}^4.
\end{align*}
For~$I_6$ we see that
\begin{align*}
I_6&\leq\int_{u_{\infty}}^u\frac{a}{|u'|^2}\left(\frac{(a^{\frac{1}{2}})^{i_2}}{|u|^{i_2+2}}||(a^{\frac{1}{2}}\mathcal{D})^{j_1}\Gamma||...
||(a^{\frac{1}{2}}\mathcal{D})^{j_{i_2}}\Gamma||\times
||(a^{\frac{1}{2}}\mathcal{D})^{i_3}\Gamma(\sigma,...)||\times||(a^{\frac{1}{2}}\mathcal{D})^{i_4}\bm\varphi||
\times||(a^{\frac{1}{2}}\mathcal{D})^{i_5}\bm\varphi|| \right)^2 \\
&\leq\int_{u_{\infty}}^u\frac{a}{|u'|^2}\left(\frac{(a^{\frac{1}{2}})^{i_2}}{|u|^{i_2+2}}\mathcal{O}^{i_2}a^{\frac{1}{2}}
\mathcal{O}a^{\frac{1}{2}}\mathcal{O}a^{\frac{1}{2}}\mathcal{O}\right)^2
=\int_{u_{\infty}}^u\frac{a}{|u'|^2}\frac{a^{i_2+3}}{|u'|^{i_2+4}}\mathcal{O}^{2i_2+6} \\
&\leq\int_{u_{\infty}}^u\frac{a^4}{|u'|^6}\mathcal{O}^6
\leq\frac{a^4}{|u|^5}\mathcal{O}^6.
\end{align*}
Finally, for~$I_7$ we have that
\begin{align*}
I_7&\leq\sum_{i_1+i_2=k-1}\int_{u_{\infty}}^u\frac{a}{|u'|^2}\frac{1}{|u'|^2}||(a^{\frac{1}{2}}\mathcal{D})^{i_2}K||^2
||(a^{\frac{1}{2}}\mathcal{D})^{i_2}\Psi_0||^2\\
&\leq\int_{u_{\infty}}^u\frac{a}{|u'|^4}a\mathcal{O}^2 
\leq\frac{a^2}{|u|^3}\mathcal{O}^2.
\end{align*}

Collecting the estimates for~$I_1,\ldots, I_7$ we obtain
\begin{align*}
H\lesssim\bm\varphi[\varphi_0]^2+\frac{1}{|u|}\mathcal{O}^6,
\end{align*}
so that 
\begin{align*}
N\leq a^{\frac{1}{2}}\underline{\bm\Psi}[\TiPsi_1](\bm\varphi[\varphi_0]+\frac{\mathcal{O}^3}{|u|^{\frac{1}{2}}}).
\end{align*}

\smallskip
\noindent
\textbf{Step 3}: Summary.

From the estimate
\begin{align*}
&||(a^{\frac{1}{2}}\meth)^k\Psi_0||^2_{L^2_{sc}(\mathcal{N}_u(0,v))}
+||(a^{\frac{1}{2}}\meth)^k\TiPsi_1||^2_{L^2_{sc}(\mathcal{N}'_v(u_{\infty},u))} \\
&\lesssim||(a^{\frac{1}{2}}\meth)^k\Psi_0||^2_{L^2_{sc}(\mathcal{N}_{u_{\infty}}(0,v))}
+||(a^{\frac{1}{2}}\meth)^k\TiPsi_1||^2_{L^2_{sc}(\mathcal{N}'_0(u_{\infty},u))} 
+M+N,
\end{align*}
multiplying by~$\frac{1}{a}$ we find that 
\begin{align*}
&\frac{1}{a}||(a^{\frac{1}{2}}\meth)^k\Psi_0||^2_{L^2_{sc}(\mathcal{N}_u(0,v))}
+\frac{1}{a}||(a^{\frac{1}{2}}\meth)^k\TiPsi_1||^2_{L^2_{sc}(\mathcal{N}'_v(u_{\infty},u))} 
\lesssim\mathcal{I}_0^2+\frac{1}{a}(M+N) \\
&\lesssim\mathcal{I}_0^2+\frac{1}{|u|}\bm\Psi[\Psi_0]\mathcal{O}^2+
\frac{1}{a^{\frac{1}{2}}}\underline{\bm\Psi}[\TiPsi_1](\bm\varphi[\varphi_0]+\frac{\mathcal{O}^3}{|u|^{\frac{1}{2}}}) \\
&\leq\mathcal{I}_0^2+\frac{1}{a^{\frac{1}{4}}}.
\end{align*}
Now, we observe that the rest of the strings
in~$\mathcal{D}^{k_i}\Psi_0$ and~$\mathcal{D}^{k_i}\TiPsi_1$ give rise
to the same estimate. Hence, one concludes that
\begin{align*}
\frac{1}{a}||(a^{\frac{1}{2}}\mathcal{D})^k\Psi_0||^2_{L^2_{sc}(\mathcal{N}_u(0,v))}
+\frac{1}{a}||(a^{\frac{1}{2}}\mathcal{D})^k\TiPsi_1||^2_{L^2_{sc}(\mathcal{N}'_v(u_{\infty},u))} 
\lesssim\mathcal{I}_0^2+\frac{1}{a^{\frac{1}{4}}}.
\end{align*}

\end{proof}

\begin{proposition}
\label{EnergyEstrestWeyl}
For~$0\leq k\leq10$, one has that 
\begin{align*}
||(a^{\frac{1}{2}}\mathcal{D})^k\{\TiPsi_1,\TiPsi_2,\TiPsi_3\}||^2_{L^2_{sc}(\mathcal{N}_u(0,v))}
+||(a^{\frac{1}{2}}\mathcal{D})^k\{\TiPsi_2,\TiPsi_3,\Psi_4\}||^2_{L^2_{sc}(\mathcal{N}'_v(u_{\infty},u))} 
\lesssim\mathcal{I}_0^2+1.
\end{align*} 
\end{proposition}

\begin{proof}
  We use the notation~$(\Psi_I,\Psi_{II})$ to denote any of the
  pairs~$\{(\tilde\Psi_1,\tilde\Psi_2),(\tilde\Psi_2,\tilde\Psi_3),(\tilde\Psi_3,\Psi_4)
  \}$. The Bianchi identities of~$(\Psi_I,\Psi_{II})$ ---i.e
  \eqref{T-weightMasslessStBianchi3},
  \eqref{T-weightMasslessStBianchi4},
  \eqref{T-weightMasslessStBianchi5},
  \eqref{T-weightMasslessStBianchi6},
  \eqref{T-weightMasslessStBianchi7},
  \eqref{T-weightMasslessStBianchi8}--- have the structure
\begin{align*}
\mthorn'\Psi_{I}+(2s_2(\Psi_{I})+1)\mu\Psi_I-\meth\Psi_{II}&=
\bm\varphi\meth\bm\varphi+\Gamma\Psi+\Gamma\bm\varphi\bm\varphi, \\
\mthorn\Psi_{II}-\meth'\Psi_{I}&=
\bm\varphi\meth\bm\varphi+\Gamma\Psi+\Gamma\bm\varphi\bm\varphi.
\end{align*}
Commuting~$\mthorn$ and~$\mthorn'$ with~$\meth^k$ we have that
\begin{align*}
\mthorn'\meth^k\Psi_{I}+(2s_2(\Psi_{I})+k+1)\mu\meth^k\Psi_I-\meth^{k+1}\Psi_{II}&=P_k,\\
\mthorn\meth^k\Psi_{II}-\meth'\meth^k\Psi_{I}&=Q_k,
\end{align*}
where
\begin{align*}
P_k\equiv \sum_{i=0}^k\meth^{i}\bm\varphi\meth^{k+1-i}\bm\varphi
+\sum_{i=0}^k\meth^{i}\Gamma\meth^{k-i}\Psi
+\sum_{i_1+i_2+i_3=k}\meth^{i_1}\Gamma\meth^{i_2}\bm\varphi\meth^{i_3}\bm\varphi
\end{align*}

\begin{align*}
Q_k&\equiv\sum_{i_1+i_2+i_3+i_4=k}\meth^{i_1}\Gamma(\tau,\pi)^{i_2}\meth^{i_3}\bm\varphi\meth^{i_4+1}\bm\varphi+
\sum_{i_1+i_2+i_3+i_4=k}\meth^{i_1}\Gamma(\tau,\pi)^{i_2}\meth^{i_3}\Gamma\meth^{i_4}\Psi \\
&+\sum_{i_1+i_2+i_3+i_4+i_5=k}\meth^{i_1}\Gamma(\tau,\pi)^{i_2}\meth^{i_3}\Gamma\meth^{i_4}\bm\varphi\meth^{i_5}\bm\varphi
+\sum_{i_1+i_2=k-1}\meth^{i_1}K\meth^{i_2}\Psi_{I}.
\end{align*}

It follows then that 
\begin{align*}
&||(a^{\frac{1}{2}}\meth)^k\Psi_{I}||^2_{L^2_{sc}(\mathcal{N}_u(0,v))}
+||(a^{\frac{1}{2}}\meth)^k\Psi_{II}||^2_{L^2_{sc}(\mathcal{N}'_v(u_{\infty},u))} \\
&\lesssim||(a^{\frac{1}{2}}\meth)^k\Psi_{I}||^2_{L^2_{sc}(\mathcal{N}_{u_{\infty}}(0,v))}
+||(a^{\frac{1}{2}}\meth)^k\Psi_{II}||^2_{L^2_{sc}(\mathcal{N}'_0(u_{\infty},u))} \\
&\quad+2\int_0^v\int_{u_{\infty}}^u\frac{a}{|u'|^2}
||(a^{\frac{1}{2}}\meth)^k\Psi_{I}||_{L^2_{sc}(\mathcal{S}_{u',v'})}
||a^{\frac{k}{2}}P_k||_{L^2_{sc}(\mathcal{S}_{u',v'})} \\
&\quad+2\int_0^v\int_{u_{\infty}}^u\frac{a}{|u'|^2}
||(a^{\frac{1}{2}}\meth)^k\Psi_{II}||_{L^2_{sc}(\mathcal{S}_{u',v'})}
||a^{\frac{k}{2}}Q_k||_{L^2_{sc}(\mathcal{S}_{u',v'})}\\
&\leq\mathcal{I}_0^2+M+N,
\end{align*}
where we have defined
\begin{align*}
M\equiv 2\int_0^v\int_{u_{\infty}}^u\frac{a}{|u'|^2}
||(a^{\frac{1}{2}}\meth)^k\Psi_{I}||_{L^2_{sc}(\mathcal{S}_{u',v'})}
||a^{\frac{k}{2}}P_k||_{L^2_{sc}(\mathcal{S}_{u',v'})}
\end{align*}
and
\begin{align*}
N\equiv 2\int_0^v\int_{u_{\infty}}^u\frac{a}{|u'|^2}
||(a^{\frac{1}{2}}\meth)^k\Psi_{I}||_{L^2_{sc}(\mathcal{S}_{u',v'})}
||a^{\frac{k}{2}}Q_k||_{L^2_{sc}(\mathcal{S}_{u',v'})}.
\end{align*}

As in previous proofs, the argument proceeds in three steps.

\smallskip
\noindent
\textbf{Step 1}: Estimate for~$M$.

One readily has that 
\begin{align*}
M&\leq\left(\int_0^v\int_{u_{\infty}}^u\frac{a}{|u'|^2}
||(a^{\frac{1}{2}}\meth)^k\Psi_{I}||^2_{L^2_{sc}(\mathcal{S}_{u',v'})} \right)^{\frac{1}{2}}
\left(\int_0^v\int_{u_{\infty}}^u\frac{a}{|u'|^2}
||a^{\frac{k}{2}}P_k||^2_{L^2_{sc}(\mathcal{S}_{u',v'})} \right)^{\frac{1}{2}} \\
&\leq\int_0^v\int_{u_{\infty}}^u\frac{a}{|u'|^2}
||(a^{\frac{1}{2}}\meth)^k\Psi_{I}||^2_{L^2_{sc}(\mathcal{S}_{u',v'})} 
+\int_0^v\int_{u_{\infty}}^u\frac{a}{|u'|^2}
||a^{\frac{k}{2}}P_k||^2_{L^2_{sc}(\mathcal{S}_{u',v'})} \\
&=\int_{u_{\infty}}^u\frac{a}{|u'|^2}
||(a^{\frac{1}{2}}\meth)^k\Psi_{I}||^2_{L^2_{sc}(\mathcal{N}_{u'}(0,v))}+H ,
\end{align*}
where 
\begin{align*}
H\equiv\int_0^v\int_{u_{\infty}}^u\frac{a}{|u'|^2}
||a^{\frac{k}{2}}P_k||^2_{L^2_{sc}(\mathcal{S}_{u',v'})}.
\end{align*}
Substituting the definition of~$P_k$ we have that 
\begin{align*}
H&\leq\sum_{i=0}^k\int_0^v\int_{u_{\infty}}^u\frac{a}{|u'|^2}
||a^{\frac{k}{2}}\mathcal{D}^i\bm\varphi\mathcal{D}^{k+1-i}\bm\varphi||^2_{L^2_{sc}(\mathcal{S}_{u',v'})}\\
&\quad+\sum_{i=0}^k\int_0^v\int_{u_{\infty}}^u\frac{a}{|u'|^2}
||a^{\frac{k}{2}}\mathcal{D}^i\Gamma\mathcal{D}^{k-i}\Psi||^2_{L^2_{sc}(\mathcal{S}_{u',v'})} \\
&\quad+\sum_{i_1+i_2+i_3=k}\int_0^v\int_{u_{\infty}}^u\frac{a}{|u'|^2}
||a^{\frac{k}{2}}\mathcal{D}^{i_1}\Gamma\mathcal{D}^{i_2}\bm\varphi
\mathcal{D}^{i_3}\bm\varphi||^2_{L^2_{sc}(\mathcal{S}_{u',v'})} \\
&\quad+\sum_{i_1+i_2=k}\int_0^v\int_{u_{\infty}}^u\frac{a}{|u'|^2}
||a^{\frac{k}{2}}\mathcal{D}^{i_1}\varphi_2\mathcal{D}^{i_2}\Tivarphi_2||^2_{L^2_{sc}(\mathcal{S}_{u',v'})} \\
&=I_1+...+I_4.
\end{align*}

\smallskip
\noindent
For the term~$I_1$ we have that 
\begin{align*}
I_1&\leq\int_0^v\int_{u_{\infty}}^u\frac{a}{|u'|^2}
||a^{\frac{k}{2}}\varphi_1\mathcal{D}^{k+1}\varphi_1||^2_{L^2_{sc}(\mathcal{S}_{u',v'})}
+\int_0^v\int_{u_{\infty}}^u\frac{a}{|u'|^2}
||a^{\frac{k}{2}}\varphi_1\mathcal{D}^{k+1}\varphi_2||^2_{L^2_{sc}(\mathcal{S}_{u',v'})} \\
&\quad+\sum_{i=1}^{k-1}\int_0^v\int_{u_{\infty}}^u\frac{a}{|u'|^2}
||a^{\frac{k}{2}}\mathcal{D}^i\bm\varphi\mathcal{D}^{k+1-i}\bm\varphi||^2_{L^2_{sc}(\mathcal{S}_{u',v'})}
=I_{11}+I_{12}+I_{13},
\end{align*}
so that for~$I_{11}$ we estimate
\begin{align*}
I_{11}&\leq\int_0^v\int_{u_{\infty}}^u\frac{a}{|u'|^2}\frac{1}{|u'|^2}||\varphi_1||^2_{L^{\infty}_{sc}(\mathcal{S}_{u',v'})}
||a^{\frac{k}{2}}\mathcal{D}^{k+1}\varphi_1||^2_{L^2_{sc}(\mathcal{S}_{u',v'})} \\
&\leq\int_{u_{\infty}}^u\frac{a}{|u'|^4}\bmGamma(\varphi_1)^2_{0,\infty}
\int_0^v||a^{\frac{k}{2}}\mathcal{D}^{k+1}\varphi_1||^2_{L^2_{sc}(\mathcal{S}_{u',v'})} \\
&\leq\frac{a}{|u|^3}\bmGamma(\varphi_1)^2_{0,\infty}\bm\varphi[\varphi_1]^2
\leq1,
\end{align*}
while for~$I_{12}$ one has that 
\begin{align*}
I_{12}&\leq\int_0^v\int_{u_{\infty}}^u\frac{a}{|u'|^2}\frac{1}{|u'|^2}||\varphi_1||^2_{L^{\infty}_{sc}(\mathcal{S}_{u',v'})}
||a^{\frac{k}{2}}\mathcal{D}^{k+1}\varphi_2||^2_{L^2_{sc}(\mathcal{S}_{u',v'})} \\
&\leq\bmGamma(\varphi_1)^2_{0,\infty}
\int_{u_{\infty}}^u\frac{a}{|u|^4}
||a^{\frac{k}{2}}\mathcal{D}^{k+1}\varphi_2||^2_{L^2_{sc}(\mathcal{S}_{u',v'})} \\
&\leq\bmGamma(\varphi_1)^2_{0,\infty}\frac{1}{a}\underline{\bm\varphi}[\varphi_2]^2
\leq1,
\end{align*}
and for~$I_{13}$ we have
\begin{align*}
I_{13}&\leq\sum_{i=1}^{k-1}\int_0^v\int_{u_{\infty}}^u\frac{1}{|u'|^2}\frac{1}{|u'|^2}
||(a^{\frac{1}{2}}\mathcal{D})^i\bm\varphi||^2_{L^{\infty}_{sc}(\mathcal{S}_{u',v'})}
||(a^{\frac{1}{2}}\mathcal{D})^{k+1-i}\bm\varphi||^2_{L^2_{sc}(\mathcal{S}_{u',v'})} \\
&\leq\int_0^v\int_{u_{\infty}}^u\frac{1}{|u'|^4}\frac{|u'|^2}{a}\mathcal{O}^2\mathcal{O}^2
\leq\frac{\mathcal{O}^4}{|u|}\leq1.
\end{align*}
It follows from the above that~$I_1\lesssim1$.

\smallskip
\noindent
Next, for the term~$I_2$ we have that 
\begin{align*}
I_2&\leq\int_0^v\int_{u_{\infty}}^u\frac{a}{|u'|^2}
||a^{\frac{k}{2}}\Gamma\mathcal{D}^{k}\Psi||^2_{L^2_{sc}(\mathcal{S}_{u',v'})}
+\sum_{i=1}^k\int_0^v\int_{u_{\infty}}^u\frac{a}{|u'|^2}
||a^{\frac{k}{2}}\mathcal{D}^i\Gamma\mathcal{D}^{k-i}\Psi||^2_{L^2_{sc}(\mathcal{S}_{u',v'})}\\
&=I_{21}+I_{22},
\end{align*}
so that for~$I_{21}$ we have the estimate
\begin{align*}
I_{21}&\leq\int_0^v\int_{u_{\infty}}^u\frac{a}{|u'|^2}\frac{1}{|u'|^2}
||\Gamma(\Timu,\ulomega)||^2_{L^{\infty}_{sc}(\mathcal{S}_{u',v'})}
||a^{\frac{k}{2}}\mathcal{D}^{k}\TiPsi_1||^2_{L^2_{sc}(\mathcal{S}_{u',v'})} \\
&\quad+\int_0^v\int_{u_{\infty}}^u\frac{a}{|u'|^2}\frac{1}{|u'|^2}
||\Gamma(\Timu,\sigma,\ulomega,\tau)||^2_{L^{\infty}_{sc}(\mathcal{S}_{u',v'})}
||a^{\frac{k}{2}}\mathcal{D}^{k}\Psi_{2,3,4}||^2_{L^2_{sc}(\mathcal{S}_{u',v'})} \\
&\leq\int_0^v\int_{u_{\infty}}^u\frac{a}{|u'|^4}\frac{|u'|^2}{a^2}\Gamma(\Timu)
\left(||a^{\frac{k}{2}}\mathcal{D}^{k}\TiPsi_1||^2_{L^2_{sc}(\mathcal{S}_{u',v'})} 
+||a^{\frac{k}{2}}\mathcal{D}^{k}\Psi_{2,3,4}||^2_{L^2_{sc}(\mathcal{S}_{u',v'})}\right) \\
&=\frac{1}{a^2}\Gamma(\Timu)\int_0^v\int_{u_{\infty}}^u\frac{a}{|u'|^2}
\left(||a^{\frac{k}{2}}\mathcal{D}^{k}\Psi_1||^2_{L^2_{sc}(\mathcal{S}_{u',v'})} 
+||a^{\frac{k}{2}}\mathcal{D}^{k}\Psi_{2,3,4}||^2_{L^2_{sc}(\mathcal{S}_{u',v'})}\right) \\
&\leq\frac{\Gamma(\Timu)}{a}(\underline{\bm\Psi}[\TiPsi_1]^2
+\frac{1}{a}\underline{\bm\Psi}[\TiPsi_{2,3},\Psi_4]^2)
\lesssim1.
\end{align*}
Moreover, for~$I_{22}$ we have that
\begin{align*}
I_{22}&\leq\sum_{i=1}^k\int_0^v\int_{u_{\infty}}^u\frac{a}{|u'|^2}\frac{1}{|u'|^2}
||(a^{\frac{1}{2}}\mathcal{D})^i\Gamma||^2\times
||(a^{\frac{1}{2}}\mathcal{D})^{k-i}\Psi||^2 \\
&\lesssim\int_0^v\int_{u_{\infty}}^u\frac{a}{|u'|^4}\frac{|u'|^2}{a^2}\mathcal{O}^2a\mathcal{O}^2
\leq\frac{\mathcal{O}^4}{|u|}\leq1.
\end{align*}
Accordingly, we have that~$I_2\lesssim1$.

\smallskip
\noindent
For the term~$I_3$ we have that 
\begin{align*}
I_3&\leq\sum_{i_1+i_2+i_3=k}\int_0^v\int_{u_{\infty}}^u\frac{a}{|u'|^2}\frac{1}{|u'|^4}
||(a^{\frac{1}{2}}\mathcal{D})^{i_1}\Gamma||^2\times||(a^{\frac{1}{2}}\mathcal{D})^{i_2}\bm\varphi||^2\times
||(a^{\frac{1}{2}}\mathcal{D})^{i_3}\bm\varphi||^2_{L^2_{sc}(\mathcal{S}_{u',v'})} \\
&\lesssim\int_0^v\int_{u_{\infty}}^u\frac{a}{|u'|^6}\frac{|u'|^2}{a}\mathcal{O}^2\frac{|u'|^2}{a}\mathcal{O}^2\mathcal{O}^2
=\int_0^v\int_{u_{\infty}}^u\frac{\mathcal{O}^6}{a|u'|^2}\leq\frac{1}{|u|}\leq1.
\end{align*}

\smallskip
\noindent
For~$I_4$ we find that 
\begin{align*}
I_4&\leq\sum_{i_1+i_2=k}\int_0^v\int_{u_{\infty}}^u\frac{a}{|u'|^2}\frac{1}{|u'|^2}
||(a^{\frac{1}{2}}\mathcal{D})^{i_1}\varphi_2||^2\times
||(a^{\frac{1}{2}}\mathcal{D})^{i_2}\Tivarphi_2||^2_{L^2_{sc}(\mathcal{S}_{u',v'})} \\
&\lesssim\int_0^v\int_{u_{\infty}}^u\frac{a}{|u'|^4}\frac{|u'|^2}{a}\mathcal{O}^2\mathcal{O}^2
\leq\frac{\mathcal{O}^4}{|u|}\leq1.
\end{align*}

Collecting the estimates for~$I_1,\ldots, I_4$ we have
that~$H\lesssim1$.

\smallskip
\noindent
\textbf{Step 2}: Estimate for~$N$

In this case one readily sees that 
\begin{align*}
N&\leq\left(\int_0^v\int_{u_{\infty}}^u\frac{a}{|u'|^2}
||(a^{\frac{1}{2}}\meth)^k\Psi_{II}||^2_{L^2_{sc}(\mathcal{S}_{u',v'})} \right)^{\frac{1}{2}}
\left(\int_0^v\int_{u_{\infty}}^u\frac{a}{|u'|^2}
||a^{\frac{k}{2}}Q_k||^2_{L^2_{sc}(\mathcal{S}_{u',v'})} \right)^{\frac{1}{2}} \\
&\leq\int_0^v\int_{u_{\infty}}^u\frac{a}{|u'|^2}
||(a^{\frac{1}{2}}\meth)^k\Psi_{II}||^2_{L^2_{sc}(\mathcal{S}_{u',v'})} 
+\int_0^v\int_{u_{\infty}}^u\frac{a}{|u'|^2}
||a^{\frac{k}{2}}Q_k||^2_{L^2_{sc}(\mathcal{S}_{u',v'})} \\
&=\int_0^v||(a^{\frac{1}{2}}\meth)^k\Psi_{I}||^2_{L^2_{sc}(\mathcal{N}_{v'}(u_{\infty},u))}+J,
\end{align*}
where
\begin{align*}
J\equiv\int_0^v\int_{u_{\infty}}^u\frac{a}{|u'|^2}
||a^{\frac{k}{2}}Q_k||^2_{L^2_{sc}(\mathcal{S}_{u',v'})}.
\end{align*}
Substituting the definition of~$Q_k$ we have that 
\begin{align*}
J&\equiv\sum_{i_1+i_2+i_3+i_4=k}\int_0^v\int_{u_{\infty}}^u\frac{a}{|u'|^2}
||a^{\frac{k}{2}}\mathcal{D}^{i_1}\Gamma(\tau,\pi)^{i_2}\mathcal{D}^{i_3}\bm\varphi
\mathcal{D}^{i_4+1}\bm\varphi||^2_{L^2_{sc}(\mathcal{S}_{u',v'})} \\
&+\sum_{i_1+i_2+i_3+i_4=k}\int_0^v\int_{u_{\infty}}^u\frac{a}{|u'|^2}
||a^{\frac{k}{2}}\mathcal{D}^{i_1}\Gamma(\tau,\pi)^{i_2}\mathcal{D}^{i_3}\Gamma
\mathcal{D}^{i_4}\Psi||^2_{L^2_{sc}(\mathcal{S}_{u',v'})} \\
&+\sum_{i_1+i_2+i_3+i_4+i_5=k}\int_0^v\int_{u_{\infty}}^u\frac{a}{|u'|^2}
||a^{\frac{k}{2}}\mathcal{D}^{i_1}\Gamma(\tau,\pi)^{i_2}\mathcal{D}^{i_3}\Gamma
\mathcal{D}^{i_4}\bm\varphi\mathcal{D}^{i_5}\bm\varphi||^2_{L^2_{sc}(\mathcal{S}_{u',v'})} \\
&+\sum_{i_1+i_2=k-1}\int_0^v\int_{u_{\infty}}^u\frac{a}{|u'|^2}
||a^{\frac{k}{2}}\mathcal{D}^{i_1}K\mathcal{D}^{i_2}\Psi||^2_{L^2_{sc}(\mathcal{S}_{u',v'})} 
=I_1+...+I_4.
\end{align*}
We now proceed to estimate each of the terms~$I_1,\ldots,I_4$. 

\smallskip
\noindent
For the term~$I_1$ we have that 
\begin{align*}
I_1&\leq\int_0^v\int_{u_{\infty}}^u\frac{a}{|u'|^2}
||a^{\frac{k}{2}}\varphi_1\mathcal{D}^{k+1}\varphi_0||^2_{L^2_{sc}(\mathcal{S}_{u',v'})}
+\int_0^v\int_{u_{\infty}}^u\frac{a}{|u'|^2}
||a^{\frac{k}{2}}\varphi_1\mathcal{D}^{k+1}\varphi_1||^2_{L^2_{sc}(\mathcal{S}_{u',v'})} \\
&+\int_0^v\int_{u_{\infty}}^u\frac{a}{|u'|^2}
||a^{\frac{k}{2}}\varphi_2\mathcal{D}^{k+1}\varphi_1||^2_{L^2_{sc}(\mathcal{S}_{u',v'})} \\
&+\sum_{i_1+i_2+i_3+i_4=k,i_4<k}\int_0^v\int_{u_{\infty}}^u\frac{a}{|u'|^2}
||a^{\frac{k}{2}}\mathcal{D}^{i_1}\Gamma(\tau,\pi)^{i_2}\mathcal{D}^{i_3}\bm\varphi
\mathcal{D}^{i_4+1}\bm\varphi||^2_{L^2_{sc}(\mathcal{S}_{u',v'})} \\
&=I_{11}+I_{12}+I_{13}+I_{14},
\end{align*}
where for~$I_{11}$ one sees that
\begin{align*}
I_{11}&\leq\int_0^v\int_{u_{\infty}}^u\frac{a}{|u'|^2}\frac{1}{|u'|^2}||\varphi_1||^2_{L^{\infty}_{sc}(\mathcal{S}_{u',v'})}
||a^{\frac{k}{2}}\mathcal{D}^{k+1}\varphi_0||^2_{L^2_{sc}(\mathcal{S}_{u',v'})} \\
&\leq\int_{u_{\infty}}^u\frac{a^2}{|u'|^4}\bmGamma(\varphi_1)^2_{0,\infty}
\frac{1}{a}\int_0^v||a^{\frac{k}{2}}\mathcal{D}^{k+1}\varphi_0||^2_{L^2_{sc}(\mathcal{S}_{u',v'})} \\
&\leq\frac{a^2}{|u|^3}\bmGamma(\varphi_1)^2_{0,\infty}\bm\varphi[\varphi_0]^2
\leq1.
\end{align*}
For~$I_{12}$ we have that 
\begin{align*}
I_{12}&\leq\int_0^v\int_{u_{\infty}}^u\frac{a}{|u'|^2}\frac{1}{|u'|^2}||\varphi_1||^2_{L^{\infty}_{sc}(\mathcal{S}_{u',v'})}
||a^{\frac{k}{2}}\mathcal{D}^{k+1}\varphi_1||^2_{L^2_{sc}(\mathcal{S}_{u',v'})} \\
&\leq\int_{u_{\infty}}^u\frac{a}{|u'|^4}\bmGamma(\varphi_1)^2_{0,\infty}
\int_0^v||a^{\frac{k}{2}}\mathcal{D}^{k+1}\varphi_0||^2_{L^2_{sc}(\mathcal{S}_{u',v'})} \\
&\leq\frac{a}{|u|^3}\bmGamma(\varphi_1)^2_{0,\infty}\bm\varphi[\varphi_1]^2
\leq1.
\end{align*}
For~$I_{13}$ we see that 
\begin{align*}
I_{13}&\leq\int_0^v\int_{u_{\infty}}^u\frac{a}{|u'|^2}\frac{1}{|u'|^2}||\varphi_2||^2_{L^{\infty}_{sc}(\mathcal{S}_{u',v'})}
||a^{\frac{k}{2}}\mathcal{D}^{k+1}\varphi_1||^2_{L^2_{sc}(\mathcal{S}_{u',v'})} \\
&\leq\int_{u_{\infty}}^u\frac{a}{|u'|^4}\frac{|u'|^2}{a}\bmGamma(\varphi_2)^2_{0,\infty}
\int_0^v||a^{\frac{k}{2}}\mathcal{D}^{k+1}\varphi_0||^2_{L^2_{sc}(\mathcal{S}_{u',v'})} \\
&\leq\frac{1}{|u|}\bmGamma(\varphi_2)^2_{0,\infty}\bm\varphi[\varphi_1]^2
\leq1,
\end{align*}
while for~$I_{14}$ we have that 
\begin{align*}
I_{14}\leq\int_0^v\int_{u_{\infty}}^u\frac{1}{|u'|^2}\left(\frac{a^{\frac{i_2}{2}}}{|u'|^{i_2+1}}\mathcal{O}^{i_2}
(a^{\frac{1}{2}}\mathcal{O}^2+\mathcal{O}^2+\frac{|u|}{a^{\frac{1}{2}}}\mathcal{O}^2) \right)^2 
\leq\frac{1}{|u|}\leq1.
\end{align*}

\smallskip
\noindent
For the term~$I_2$ we have that 
\begin{align*}
I_2&\leq\int_0^v\int_{u_{\infty}}^u\frac{a}{|u'|^2}
||a^{\frac{k}{2}}\lambda\mathcal{D}^{k}\Psi_0||^2_{L^2_{sc}(\mathcal{S}_{u',v'})}
+\int_0^v\int_{u_{\infty}}^u\frac{a}{|u'|^2}
||a^{\frac{k}{2}}\Gamma(\lambda,...)\mathcal{D}^{k}(\TiPsi_{1,2,3},\Psi_4)||^2_{L^2_{sc}(\mathcal{S}_{u',v'})}  \\
&\quad+\sum_{i_1+i_2+i_3+i_4=k,i_4<k}\int_0^v\int_{u_{\infty}}^u\frac{a}{|u'|^2}
||a^{\frac{k}{2}}\mathcal{D}^{i_1}\Gamma(\tau,\pi)^{i_2}\mathcal{D}^{i_3}\Gamma
\mathcal{D}^{i_4}\Psi||^2_{L^2_{sc}(\mathcal{S}_{u',v'})} \\
&=I_{21}+I_{22}+I_{23},
\end{align*}
where for~$I_{21}$ one sees  that 
\begin{align*}
I_{21}&\leq\int_0^v\int_{u_{\infty}}^u\frac{a}{|u'|^2}\frac{1}{|u'|^2}||\lambda||^2_{L^{\infty}_{sc}(\mathcal{S}_{u',v'})}
||a^{\frac{k}{2}}\mathcal{D}^{k}\Psi_0||^2_{L^2_{sc}(\mathcal{S}_{u',v'})} \\
&\lesssim\int_{u_{\infty}}^u\frac{a^2}{|u'|^4}\frac{|u'|^2}{a}\bmGamma(\lambda)_{0,\infty}
\frac{1}{a}\int_0^v||a^{\frac{k}{2}}\mathcal{D}^{k}\Psi_0||^2_{L^2_{sc}(\mathcal{S}_{u',v'})} \\
&\leq\int_{u_{\infty}}^u\frac{a}{|u'|^2}\bmGamma(\lambda)_{0,\infty}\bm\Psi[\Psi_0]
\lesssim\frac{a}{|u|}\leq1.
\end{align*}
Moreover, for~$I_{22}$ we have that 
\begin{align*}
  I_{22}\leq
  &\int_0^v\int_{u_{\infty}}^u\frac{a}{|u'|^2}\frac{1}{|u'|^2}
    \left(||\Gamma(\lambda...)||^2_{L^{\infty}_{sc}(\mathcal{S}_{u',v'})}
    ||a^{\frac{k}{2}}\mathcal{D}^{k}\TiPsi_{1,2,3}||^2_{L^2_{sc}(\mathcal{S}_{u',v'})} 
    +||\rho||^2_{L^{\infty}_{sc}(\mathcal{S}_{u',v'})}
    ||a^{\frac{k}{2}}\mathcal{D}^{k}\Psi_4||^2_{L^2_{sc}(\mathcal{S}_{u',v'})} \right) \\
  &\leq\int_{u_{\infty}}^u\frac{a}{|u'|^4}\frac{|u'|^2}{a}\bmGamma(\lambda)^2_{0,\infty}
        \int_0^v||a^{\frac{k}{2}}\mathcal{D}^{k}\TiPsi_{1,2,3}||^2_{L^2_{sc}(\mathcal{S}_{u',v'})}
        +\int_0^v\frac{\bmGamma(\rho)_{0,\infty}^2}{|u|^2}\int_{u_{\infty}}^u\frac{a}{|u'|^2}
        ||a^{\frac{k}{2}}\mathcal{D}^{k}\Psi_4||^2_{L^2_{sc}(\mathcal{S}_{u',v'})} \\
  &\leq\frac{1}{|u|}\bmGamma(\lambda)^2_{0,\infty}\bm\Psi[\TiPsi_{1,2,3}]^2
        +\frac{1}{|u|^2}\bmGamma(\rho)^2_{0,\infty}\underline{\bm\Psi}[\Psi_4]^2
        \lesssim1.
\end{align*}
For~$I_{23}$ we have that 
\begin{align*}
I_{23}\leq\int_0^v\int_{u_{\infty}}^u\frac{a}{|u'|^2}\left(\frac{a^{\frac{i_2}{2}}}{|u'|^{i_2+1}}\mathcal{O}^{i_2}
(a^{\frac{1}{2}}\bmGamma(\Psi_0)+\bmGamma(\Psi_{II}))\frac{|u|}{a^{\frac{1}{2}}}\bmGamma(\lambda)\right)^2 
\leq\frac{a}{|u|}\leq1,
\end{align*}
where we have made use of~$\bmGamma(\Psi_0,\lambda)\lesssim1$.

\smallskip
\noindent
For the term~$I_3$ one has that the worst term in~$\Gamma\bm\varphi^2$
is~$\mu\varphi_0^2$. One then has that 
\begin{align*}
I_{3}\leq\int_0^v\int_{u_{\infty}}^u\frac{a}{|u'|^2}\left(\frac{a^{\frac{i_2}{2}}}{|u'|^{i_2+2}}\mathcal{O}^{i_2}
\frac{|u'|^2}{a}\mathcal{O}a\mathcal{O}^2\right)^2 
\leq\frac{a}{|u|}\leq1,
\end{align*}
where we have made use of~$\bmGamma(\mu)\lesssim1$
and~$\bmGamma(\varphi_0)\lesssim\bm\varphi[\varphi_1]+1\lesssim1$.

\smallskip
\noindent
\textbf{$I_4$}: This integral can also be bounded by 1. 

Collecting the estimates for~$I_1,\ldots, I_4$ we conclude
that~$J\lesssim1$.

\smallskip
\noindent
\textbf{Summary.} From the estimates for~$M$ and~$N$ we have that 
\begin{align*}
||(a^{\frac{1}{2}}\meth)^k\Psi_{I}||^2_{L^2_{sc}(\mathcal{N}_u(0,v))}
+||(a^{\frac{1}{2}}\meth)^k\Psi_{II}||^2_{L^2_{sc}(\mathcal{N}'_v(u_{\infty},u))} 
\lesssim\mathcal{I}_0^2+1.
\end{align*} 
Now, the rest of the strings of derivatives in~$\mathcal{D}^{k_i}$
give rise to the same result and hence we have that 
\begin{align*}
||(a^{\frac{1}{2}}\mathcal{D})^k\Psi_{I}||^2_{L^2_{sc}(\mathcal{N}_u(0,v))}
+||(a^{\frac{1}{2}}\mathcal{D})^k\Psi_{II}||^2_{L^2_{sc}(\mathcal{N}'_v(u_{\infty},u))} 
\lesssim\mathcal{I}_0^2+1.
\end{align*} 

\end{proof}

Combining the estimates in last three sections we have improved the
bootstrap assumption
\begin{align*}
\bmGamma+\bm\Psi+\bmvarphi\leq\mathcal{O}
\end{align*}
by 
\begin{align*}
\bmGamma+\bm\Psi+\bmvarphi\lesssim(\mathcal{I}_0)^2+\mathcal{I}_0+1. 
\end{align*}
Hence, following the discussion in Section~\ref{BootstrapAssumption}
we obtain our main existence result. Namely, one has that

\begin{theorem}
  [\textbf{\em Existence result}]
  \label{Theorem:ExistenceDomain}
Given~$\mathcal{I}$, there exists a sufficiently large
$a_0=a_0(\mathcal{I})$, such that for~$0\leq a_0 \leq a$ and initial
data 
\begin{align*}
\mathcal{I}_0\equiv\sum_{j=0}^1\sum_{i=0}^{15}\frac{1}{a^{\frac{1}{2}}}
||\mthorn^j(|u_{\infty}|\mathcal{D})^i(\sigma,\varphi_0)||_{L^{2}(\mathcal{S}_{u_{\infty},v})}\leq\mathcal{I},
\end{align*}
along the outgoing initial null hypersurface~$u=u_{\infty}$ and
Minkowskian initial data along the ingoing initial null
hypersurface~$v=0$, then the Einstein-scalar field system has a
solution in the causal diamond
\begin{align*}
\mathbb{D}=\big\{(u,v)|u_{\infty}\leq u\leq -a/4, \ \ 0\leq v\leq 1\big\}.
\end{align*}
\end{theorem}

\section{Trapped surface formation}
\label{TrappedSurface}

We are now in a position to provide the main argument for the
formation of trapped surfaces in the existence domain~$\mathscr{D}$.

We fix~$(x^2,x^3)\in\mathcal{S}$ and carry out a pointwise
analysis. From the estimate of~$\ulomega$,
\begin{align*}
||\ulomega||_{L^{\infty}(\mathcal{S}_{u,v})}\lesssim\frac{a}{|u|^3}
\end{align*}
and~$\mthorn'Q=\ulomega Q$, we have that 
\begin{align*}
|Q|\lesssim1+\int_{u_{\infty}}^u||\mthorn'Q||_{L^{\infty}(\mathcal{S}_{u,v})} ,
\end{align*}
so that
\begin{align*}
|Q|-1\lesssim\frac{\mathcal{O}}{|u|}.
\end{align*}
This means~$Q$ and~$Q^{-1}$ are both close to~$1$.

For~$P^{\mathcal{A}}$, we use the equation
\begin{align*}
\mthorn P^{\mathcal{A}}=\rho
P^{\mathcal{A}}+\sigma\bar{P}^{\mathcal{A}}
\end{align*}
and obtain that  
\begin{align*}
|P^{\mathcal{A}}(u,v)|\lesssim&|P^{\mathcal{A}}(u,0)|+\int_0^v|\mthorn P^{\mathcal{A}}| \\
\leq&|P^{\mathcal{A}}(u,0)|+\frac{a^{\frac{1}{2}}}{|u|}\mathcal{O}.
\end{align*}
This means that in the course of the evolution~$P^{\mathcal{A}}$
remains close to its initial value.

For~$C^{\mathcal{A}}$ we use the equation
\begin{align*}
  \mthorn
  C^{\mathcal{A}}=(\bar\tau+\pi)P^{\mathcal{A}}+(\tau+\bar\pi)\bar{P}^{\mathcal{A}},
\end{align*}
so that 
\begin{align*}
|C^{\mathcal{A}}|\lesssim\int_0^v|\mthorn C^{\mathcal{A}}|\lesssim\frac{a^{\frac{1}{2}}}{|u|^2}\mathcal{O}.
\end{align*}

Now, defining
\begin{align*}
F\equiv |u|^2Q^{-1}(\sigma\bar\sigma+3\varphi_0^2).
\end{align*}
and using~$\bmn=\partial_u+C^{\mathcal{A}}\partial_{\mathcal{A}}$, we
have
\begin{align*}
\partial_uF=\mthorn'F-C^{\mathcal{A}}\partial_{\mathcal{A}}F\equiv I_1+I_2.
\end{align*}
We now proceed to estimate the terms~$I_1$ and~$I_2$.

\smallskip
\noindent
\textbf{Estimate for~$I_1$.} For~$\sigma$, $\varphi_0$ and~$Q$ we
make use of their~$\mthorn'$ equations
\begin{align*}
\mthorn'\sigma&=\meth\tau-\mu\sigma-\bar\lambda\rho-\tau^2+\ulomega\sigma
-3\varphi_{1}^2, \\
\mthorn'\varphi_{0}&=\meth\bar\varphi_{1}+(\ulomega-\mu)\varphi_{0}
+\rho\varphi_{2}-\bar\tau\varphi_{1}-\tau\bar\varphi_{1}, \\
\mthorn'Q&=\ulomega Q,
\end{align*}
so that
\begin{align*}
\mthorn'F&=|u|^2Q^{-1}[-\frac{2\sigma\bar\sigma+6\varphi_0^2}{|u|}+\ulomega(\sigma\bar\sigma+3\varphi_0^2)
+\sigma\mthorn'\bar\sigma+\bar\sigma\mthorn'\sigma+6\varphi_0\mthorn'\varphi_0] \\
&=|u|^2Q^{-1}(-2\Timu\sigma\bar\sigma-6\Timu\varphi_0^2+6\rho\varphi_0\varphi_2
-3\sigma\bar\varphi_1^2-\lambda\rho\sigma-3\bar\sigma\varphi_1^2-\bar\lambda\rho\bar\sigma
-6\tau\varphi_0\bar\varphi_1 \\
&+\bar\sigma\tau^2-6\varphi_0\varphi_1\bar\tau-\sigma\bar\tau^2+9\ulomega\varphi_0^2
+3\ulomega\sigma\bar\sigma+\bar\sigma\meth\tau+6\varphi_0\meth'\varphi_1+\sigma\meth'\bar\tau) \\
&=|u|^2Q^{-1}M,
\end{align*}
where we have made use of~$\Timu=\mu+1/|u|$. Next, we make use of the
results in Section~\ref{L2estimate} and obtain
\begin{align*}
|M|\lesssim\frac{a}{|u|^4}.
\end{align*}
It follows from the latter that 
\begin{align*}
|I_1|\lesssim\frac{a}{|u|^2}.
\end{align*}

\smallskip
\noindent
\textbf{Estimate for~$I_2$.} In this case we make use of the
constraint~$\delta Q=(\tau-\bar\pi)Q$ and that on~$\mathcal{S}_{u,v}$
one has
\begin{align*}
|\partial_{\mathcal{A}}f|\leq||\mathcal{D}f||_{L^{\infty}(\mathcal{S}_{u,v})},
\end{align*}
so that 
\begin{align*}
|I_2|&\leq|u|^2|C^{\mathcal{A}}|\cdot
|\partial_{\mathcal{A}}(Q^{-1}(\sigma\bar\sigma+3\varphi_0^2))| \\
&\leq|u|^2|C^{\mathcal{A}}|\cdot(||\mathcal{D}Q^{-1}||_{L^{\infty}(\mathcal{S}_{u,v})}
||(\sigma\bar\sigma+3\varphi_0^2)||_{L^{\infty}(\mathcal{S}_{u,v})}
+|Q^{-1}|\cdot||\mathcal{D}(\sigma\bar\sigma+3\varphi_0^2)||_{L^{\infty}(\mathcal{S}_{u,v})} ) \\
&\leq|u|^2|C^{\mathcal{A}}|\cdot|Q^{-1}|\cdot
(||\tau,\pi||_{L^{\infty}(\mathcal{S}_{u,v})}||(\sigma\bar\sigma+3\varphi_0^2)||_{L^{\infty}(\mathcal{S}_{u,v})}
+||\mathcal{D}(\sigma\bar\sigma+3\varphi_0^2)||_{L^{\infty}(\mathcal{S}_{u,v})}) \\
&\lesssim|u|^2\frac{a^{\frac{1}{2}}}{|u|^2}\mathcal{O}(\frac{a^{\frac{3}{2}}}{|u|^4}+\frac{a}{|u|^3})
\leq\frac{a^{\frac{3}{2}}}{|u|^3}\mathcal{O}
\leq\frac{a^{\frac{3}{2}}}{|u|^2}.
\end{align*}

Combining the results above we obtain
\begin{align*}
  |\partial_uF|\leq\frac{a^{\frac{3}{2}}}{|u|^2}\ll\frac{a^{\frac{7}{4}}}{|u|^2},
\end{align*}
so that
\begin{align*}
-\frac{a^{\frac{7}{4}}}{|u|^2}\leq\partial_uF\leq\frac{a^{\frac{7}{4}}}{|u|^2} .
\end{align*}
The latter, in turn, leads to
\begin{align*}
|u|^2Q^{-1}(\sigma\bar\sigma+3\varphi_0^2)|_{u,v}-|u_{\infty}|^2(\sigma\bar\sigma+3\varphi_0^2)|_{u=u_{\infty},v}
\geq-\frac{a^{\frac{7}{4}}}{|u|}+\frac{a^{\frac{7}{4}}}{|u_{\infty}|}.
\end{align*}
Hence, for sufficiently large~$a$, we have
\begin{align*}
  & \int_0^1\left( |u|^2Q^{-1}(\sigma\bar\sigma+3\varphi_0^2)|\right)(u,v',x^2,x^3) \mathrm{d}v'\\
&\quad\geq\int_0^1\left( |u_{\infty}|^2(\sigma\bar\sigma+3\varphi_0^2)|\right)(u_{\infty},v',x^2,x^3) \mathrm{d}v'-\frac{a^{\frac{7}{4}}}{|u|} \\
&\geq a-\frac{a^{\frac{7}{4}}}{|u|}
\geq a-\frac{4a^{\frac{7}{4}}}{a}
\geq \frac{a}{2}.
\end{align*}
Consequently, when~$u=-\frac{a}{4}$ we have
\begin{align*}
&\quad\frac{a^2}{16}\int_0^1\left(Q^{-1}(\sigma\bar\sigma+3\varphi_0^2)|\right)(-\frac{a}{4},v',x^2,x^3) \mathrm{d}v'\geq \frac{a}{2} \\
&\Rightarrow 
\int_0^1\left(Q^{-1}(\sigma\bar\sigma+3\varphi_0^2)|\right)(-\frac{a}{4},v',x^2,x^3) \mathrm{d}v'\geq \frac{8}{a}.
\end{align*}

\smallskip
\noindent
\textbf{Behaviour of the expansion.} In this step we make use of the
structure equation for~$\rho$ ---namely,
\begin{align*}
\mthorn\rho=\rho^2+\sigma\bar\sigma+3\varphi_{0}^2.
\end{align*}
Recall that the outgoing expansion~$\theta_{\bml}$ is defined by
\begin{align*}
\theta_{\bml}\equiv\sigma^{ab}\nabla_a\l_b=-\rho-\bar\rho=-2\rho.
\end{align*}
Its initial value is Minkowskian, namely we have that 
\begin{align*}
\rho(u,0,x^2,x^3)=-\frac{1}{|u|}. 
\end{align*}

We now show that the outgoing expansion on~$\mathcal{N}_{u_{\infty}}$
is positive.  From~$\bml=Q\partial_v=\partial_v$, observing that~$Q=1$
on~$\mathcal{N}_{u_{\infty}}$, we have
\begin{align*}
\rho(u_{\infty},v,x^2,x^3)
&=\rho(u_{\infty},0,x^2,x^3)+\int_0^1\frac{\partial\rho}{\partial v} \mathrm{d}v 
=-\frac{1}{|u_{\infty}|}+\int_0^v\mthorn\rho \mathrm{d}v' \\
&= -\frac{1}{|u_{\infty}|}+\int_0^v(\sigma\bar\sigma+3\varphi_{0}^2)(-\frac{a}{4},v',x^2,x^3) \mathrm{d}v'  \\
  &\lesssim-\frac{1}{|u_{\infty}|}+\frac{a}{|u_{\infty}|^2}<0,
\end{align*}
where we have made use of the estimates for~$\sigma$, $\varphi_0$
and~$a<|u_{\infty}|$. It follows that the outgoing expansion on the
initial outgoing lightcone is positive.

Turning now the attention to~$\mathcal{N}_{-a/4}$ we have that
\begin{align*}
\rho(-\frac{a}{4},1,x^2,x^3)=&\rho(-\frac{a}{4},0,x^2,x^3)+\int_0^1\frac{\partial\rho}{\partial v} \mathrm{d}v 
=-\frac{4}{a}+\int_0^1Q^{-1}\mthorn\rho \mathrm{d}v' \\
\geq& -\frac{4}{a}+\int_0^1Q^{-1}(\sigma\bar\sigma+3\varphi_{0}^2)(-\frac{a}{4},v',x^2,x^3) \mathrm{d}v'  \\
\geq&-\frac{4}{a}+\frac{8}{a}=\frac{4}{a}>0.
\end{align*}
It follows then that the outgoing expansion on~$\mathcal{N}_{-a/4}$ is
negative.

Recall now that the ingoing expansion~$\theta_{\bmn}$ is defined by
\begin{align*}
\theta_{\bmn}\equiv\sigma^{ab}\nabla_an_b=\mu+\bar\mu=2\mu.
\end{align*}
From the estimate 
\begin{align*}
||\Timu||_{L^{\infty}(\mathcal{S}_{u,v})}=
||\mu+\frac{1}{|u|}||_{L^{\infty}(\mathcal{S}_{u,v})}\leq\frac{1}{|u|^2},
\end{align*}
we have
\begin{align*}
\mu(-\frac{a}{4},v,x^2,x^3)<0, \quad
\mu(u_{\infty},v,x^2,x^3)<0
\end{align*}
for any~$(x^2,x^3)$. This means that the ingoing expansion is always
negative.

In conclusion, for any point on~$\mathcal{S}_{-\frac{a}{4},1}$,
both~$\theta_{\bml}$ and~$\theta_{\bmn}$ are negative. Consequently
$\mathcal{S}_{-\frac{a}{4},1}$ is a trapped surface. Summarising, we
have proven the following:

\begin{theorem}
[\textbf{\em Trapped surface formation}]
Given~$\mathcal{I}$, there exists a sufficiently large
$a=a(\mathcal{I})$ such that if~$0\leq a_0 \leq a$, and the initial
data satisfies
\begin{align*}
\mathcal{I}_0\equiv\sum_{j=0}^1\sum_{i=0}^{15}\frac{1}{a^{\frac{1}{2}}}
||\mthorn^j(|u_{\infty}|\mathcal{D})^i(\sigma,\varphi_0)||_{L^{2}(\mathcal{S}_{u_{\infty},v})}\leq\mathcal{I}
\end{align*}
along the outgoing initial null hypersurface~$u=u_{\infty}$ and is
Minkowskian initial data along ingoing initial null
hypersurface~$v=0$, with
\begin{align*}
\int_0^1\left( |u_{\infty}|^2(\sigma\bar\sigma+3\varphi_0^2)|\right)(u_{\infty},v') \mathrm{d}v'\geq a,
\end{align*}
uniformly for any point on the outgoing initial null hypersurface
$u=u_{\infty}$, then
we have that~$\mathcal{S}_{-a/4,1}$ is a trapped surface.
\end{theorem}

\appendix

\section{The Einstein field equations in the NP formalism}

In this appendix we provide an overview of the main geometric objects
and equations used in the analysis in the main text.

\subsection{The NP field equations}

Given a NP frame~$\{l^a,\ n^a,\ m^a,\ \bar{m}^a \}$, we define the
complex spin connection coefficients in the usual manner as
\begin{align*}
&\kappa\equiv -m^al^b\nabla_bl_a,\ \ \rho\equiv
   -m^a\bar{m}^b\nabla_bl_a, \ \ \sigma\equiv-m^am^b\nabla_bl_a, \ \
   \tau\equiv -m^an^b\nabla_bl_a, \\
&\nu\equiv \bar{m}^an^b\nabla_bn_a, \ \ \mu\equiv
   \bar{m}^am^b\nabla_bn_a, \ \ \lambda\equiv \bar{m}^a\bar{m}^b\nabla_bn_a, \ \ \pi=\bar{m}^al^b\nabla_bn_a, \\
&\alpha\equiv
   \frac{1}{2}(l^a\bar{m}^b\nabla_bn_a-m^a\bar{m}^b\nabla_b\bar{m}_a),
   \ \ \beta\equiv \frac{1}{2}(\bar{m}^am^b\nabla_bm_a-n^am^b\nabla_bl_a),  \\
&\epsilon\equiv
   \frac{1}{2}(\bar{m}^al^b\nabla_bm_a-n^al^b\nabla_bl_a), \ \
   \gamma\equiv \frac{1}{2}(l^an^b\nabla_bn_a-m^an^b\nabla_b\bar{m}_a).
\end{align*}
We also write
\begin{align*}
  D\equiv l^a\nabla_a, \qquad  \Delta\equiv n^a\nabla_a, \qquad
  \delta\equiv m^a\nabla_a,  \qquad \bar{\delta}=\bar{m}^a\nabla_a.
\end{align*}
Then the commutators of the above directional covariant derivatives
satisfy
\begin{subequations}
\begin{align}
& (\Delta D - D\Delta) \psi = \big( (\gamma+\bar{\gamma}) D +
   (\epsilon+\bar{\epsilon})\Delta -(\bar{\tau} + \pi)\delta
   -(\tau+\bar{\pi})\bar{\delta}\big)\psi, \label{NPCommutator1}\\
& (\delta D - D \delta) \psi =\big( (\bar{\alpha}+\beta -\bar{\pi})D
   +\kappa\Delta -(\bar{\rho}+\epsilon-\bar{\epsilon})\delta
   -\sigma\bar{\delta}  \big)\psi, \label{NPCommutator2}\\
& (\delta \Delta -\Delta \delta)\psi = \big( -\bar{\nu} D +
   (\tau-\bar{\alpha}-\beta)\Delta + (\mu -\gamma +\bar{\gamma})\delta
   +\bar{\lambda} \bar{\delta}\big)\psi, \label{NPCommutator3}\\
& (\bar{\delta}\delta - \delta \bar{\delta})\psi = \big(
  (\bar{\mu}-\mu)D + (\bar{\rho}-\rho) \Delta +
  (\alpha-\bar{\beta})\delta -(\bar{\alpha}-\beta)\bar{\delta} \big)\psi.
  \label{NPCommutator4},
\end{align}
\end{subequations}
where~$\psi$ is any scalar field.   

We use the same notation in~\cite{CFEBook} to denote the components of
Weyl spinor~$\Psi_{ABCD}$ and tracefree Ricci spinor
$\Phi_{AA^{\prime}BB^{\prime}}$ ---namely \{$\Psi_0$, $\Psi_1$,
$\Psi_2$, $\Psi_3$, $\Psi_4$\} and \{$\Phi_{00}$, $\Phi_{01}$,
$\Phi_{02}$, $\Phi_{11}$, $\Phi_{12}$, $\Phi_{22}$, $\Lambda$\}. These
coefficients are defined as follows:
\begin{align*}
&\Psi_0\equiv C_{abcd}l^am^bl^cm^d, \ \ \Psi_1\equiv
   C_{abcd}l^an^bl^cm^d, \ \ \Psi_2\equiv \frac{1}{2}C_{abcd}l^an^b(l^cn^d-m^c\bar{m}^d) \\
&\Psi_3\equiv C_{abcd}n^al^bn^c\bar{m}^d, \ \ \Psi_4\equiv C_{abcd}n^a\bar{m}^bn^c\bar{m}^d, \\
&\Phi_{00}\equiv \frac{1}{2}R_{\{ab\}}l^al^b, \ \ \Phi_{01}\equiv
   \frac{1}{2}R_{\{ab\}}l^am^b, \ \ \Phi_{02}\equiv \frac{1}{2}R_{\{ab\}}m^am^b, \\
&\Phi_{11}\equiv \frac{1}{4}R_{\{ab\}}(l^an^b+m^a\bar{m}^b), \ \
   \Phi_{12}\equiv \frac{1}{2}R_{\{ab\}}n^am^b, \\
&\Phi_{22}\equiv \frac{1}{2}R_{\{ab\}}n^an^b, \ \ \Lambda\equiv \frac{R}{24},
\end{align*}
where~$R_{\{ab\}}=R_{ab}-\frac{1}{4}Rg_{ab}$. The NP structure
equations and the Bianchi identities can be found in
e.g.~\cite{HilValZha20}.

\subsection{The massless scalar equations}
\label{Appendix:MasslessScalarEquations}

Given a massless scalar field~$\varphi$ let
\begin{align*}
\varphi_0=D\varphi, \qquad  \varphi_2=\Delta\varphi,\qquad
\varphi_1=\delta\varphi,\qquad  \bar\varphi_1=\bar\delta\varphi. 
\end{align*}
Then the equation~$\nabla^a\nabla_a\varphi=0$ can be expressed in
terms of the NP frame. One finds that
\begin{align*}
&D\varphi_2+\Delta\varphi_0-\delta\bar\varphi_1-\bar\delta\varphi_1+\bar\varphi_1(\bar\alpha-\beta)
+\varphi_1(\alpha-\bar\beta) \\
&-\varphi_0(\gamma+\bar\gamma)+\varphi_2(\epsilon+\bar\epsilon)
+\varphi_0(\mu+\bar\mu)-\bar\varphi_1\bar\pi-\varphi_1\pi \\
&-\varphi_2(\rho+\bar\rho)+\varphi_1\bar\tau+\bar\varphi_1\tau=0.
\end{align*}
Applying the commutator
relations~\eqref{NPCommutator1}-\eqref{NPCommutator4}, one has that
\begin{subequations}
\begin{align}
\Delta\varphi_0-\bar\delta\varphi_1&=
(\gamma+\bar\gamma-\bar\mu)\varphi_0
+\rho\varphi_2+(\bar\beta-\alpha-\bar\tau)\varphi_1
-\tau\bar\varphi_1, \label{EOMMasslessScalarNP1} \\
\Delta\varphi_1-\delta\varphi_2&=
\bar\nu\varphi_0
+(\bar\alpha+\beta-\tau)\varphi_2
+(\gamma-\bar\gamma-\mu)\varphi_1
-\bar\lambda\bar\varphi_1, \label{EOMMasslessScalarNP2} \\
D\bar\varphi_1-\bar\delta\varphi_0&=
(\pi-\alpha-\bar\beta)\varphi_0
-\bar\kappa\varphi_2+\bar\sigma\varphi_1
+(\bar\epsilon-\epsilon+\rho)\bar\varphi_1, \label{EOMMasslessScalarNP3} \\
D\varphi_2-\bar\delta\varphi_1&=
-\bar\mu\varphi_0
+(\rho-\epsilon-\bar\epsilon)\varphi_2
+(\bar\beta-\alpha+\pi)\varphi_1
+\bar\pi\bar\varphi_1, \label{EOMMasslessScalarNP4} \\
\delta\bar\varphi_1-\bar\delta\varphi_1&=
(\mu-\bar\mu)\varphi_0
+(\rho-\bar\rho)\varphi_2
+(\bar\beta-\alpha)\varphi_1
+(\bar\alpha-\beta)\bar\varphi_1. \label{EOMMasslessScalarNP5} 
\end{align}
\end{subequations}
Now, from the Einstein-scalar field equation 
\begin{align*}
R_{ab}=6\nabla_a\varphi\nabla_b\varphi
\end{align*}
one obtains the relations 
\begin{align*}
&\Phi_{00}=3\varphi_0^2, \quad \Phi_{01}=3\varphi_0\varphi_1, \quad \Phi_{02}=3\varphi_1^2, \\
&\Phi_{11}=\frac{3}{2}(\varphi_0\varphi_2+\varphi_1\bar\varphi_1), \quad 
\Phi_{12}=3\varphi_2\varphi_1, \quad \Phi_{22}=3\varphi_2^2, \\
&\Lambda=\frac{1}{2}(-\varphi_0\varphi_2+\varphi_1\bar\varphi_1).
\end{align*}
Now, making use of the definitions of the renormalised components of the
Weyl tensor ---namely
\begin{align*}
\tilde\Psi_1\equiv&\Psi_1-\Phi_{01}=\Psi_1-3\varphi_0\varphi_1,  \\
\tilde\Psi_2\equiv&\Psi_2+2\Lambda=\Psi_2-\varphi_0\varphi_2+\varphi_1\bar\varphi_1,  \\
\tilde\Psi_3\equiv&\Psi_3-\Phi_{21}=\Psi_3-3\varphi_2\bar\varphi_1,
\end{align*}
and substituting into the structure equations and Bianchi identities
one can obtain the NP form of the Einstein-Massless scalar system.
The structure equations take the form
\begin{subequations}
\begin{align}
\Delta\epsilon-D\gamma= &
-\tilde\Psi_2-3\varphi_0\varphi_2 
+\epsilon( 2\gamma+\bar{\gamma})+ \gamma\bar{\epsilon}+\kappa\nu-\beta\pi-\alpha\bar{\pi}-\alpha\tau-\pi\tau-\beta\bar{\tau} \label{MasslessScalarstructureeq1},\\
\Delta\kappa-D\tau=&
-\tilde\Psi_1-6\varphi_0\varphi_1
+3 \gamma\kappa+ \bar{\gamma}\kappa-\bar{\pi}\rho-\pi\sigma-\epsilon\tau+\bar{\epsilon}\tau-\rho\tau-\sigma\bar{\tau} \label{MasslessScalarstructureeq2},\\
\Delta\pi-D\nu=&
-\tilde\Psi_3-6\varphi_2\bar\varphi_1
+3\epsilon\nu+\bar{\epsilon} \nu-\gamma\pi+\bar{\gamma}\pi-\mu\pi-\lambda\bar{\pi}-\lambda\tau-\mu\bar{\tau} \label{MasslessScalarstructureeq3},\\
\delta\gamma-\Delta\beta=&
3\varphi_2\varphi_1
-\bar{\alpha}\gamma-2\beta\gamma+\beta\bar{\gamma}+\alpha\bar{\lambda}+\beta\mu-\epsilon\bar{\nu}-\nu\sigma+\gamma\tau+\mu\tau \label{MasslessScalarstructureeq4},\\
\delta\epsilon-D\beta=&
-\tilde\Psi_{1}-3\varphi_0\varphi_1
+\bar{\alpha}\epsilon+\beta\bar{\epsilon}+\gamma\kappa+\kappa\mu-\epsilon\bar{\pi}-\beta\bar{\rho}-\alpha\sigma-\pi\sigma \label{MasslessScalarstructureeq5},\\
\delta\kappa-D\sigma=&
-\Psi_{0}+\bar{\alpha}\kappa+3\beta\kappa-\kappa\bar{\pi}-3\epsilon\sigma+\bar{\epsilon}\sigma-\rho\sigma-\bar{\rho}\sigma+\kappa\tau \label{MasslessScalarstructureeq6},\\
\delta\nu-\Delta\mu=&
3\varphi_2^2+\lambda\bar{\lambda}+\gamma\mu+\bar{\gamma}\mu+\mu^2-\bar{\alpha}\nu-3\beta\nu-\bar{\nu}\pi+\nu\tau \label{MasslessScalarstructureeq7},\\
\delta\pi-D\mu=&
-\tilde\Psi_2
+\epsilon\mu+\bar{\epsilon}\mu+\kappa\nu+\bar{\alpha}\pi-\beta\pi-\pi\bar{\pi}-\mu\bar{\rho}-\lambda\sigma \label{MasslessScalarstructureeq8},\\
\delta\tau-\Delta\sigma=&3\varphi_1^2
-\kappa\bar{\nu}+\bar{\lambda}\rho-3\gamma\sigma+\bar{\gamma}\sigma+\mu\sigma-\bar{\alpha}\tau+\beta\tau+\tau^2 \label{MasslessScalarstructureeq9},\\
\bar{\delta}\beta-\delta\alpha=&-3\bar\varphi_1\varphi_1+\tilde\Psi_2
-\alpha\bar{\alpha}+2\alpha\beta-\beta\bar{\beta}-\epsilon\mu+\epsilon\bar{\mu}-\gamma\rho-\mu\rho+\gamma\bar{\rho}+\lambda\sigma \label{MasslessScalarstructureeq10},\\
\bar{\delta}\gamma-\Delta\alpha=&\tilde\Psi_3+3\varphi_2\bar\varphi_1
-\bar{\beta}\gamma-\alpha\bar{\gamma}+\beta\lambda+\alpha\bar{\mu}-\epsilon\nu-\nu\rho+\lambda\tau+\gamma\bar{\tau} \label{MasslessScalarstructureeq11},\\
\bar{\delta}\epsilon-D\alpha=&-3\varphi_0\bar\varphi_1
+2\alpha\epsilon+\bar{\beta}\epsilon-\alpha\bar{\epsilon}+\gamma\bar{\kappa}+\kappa\lambda-\epsilon\pi-\alpha\rho-\pi\rho-\beta\bar{\sigma} \label{MasslessScalarstructureeq12},\\
\bar{\delta}\kappa-D\rho=&-3\varphi_0^2
+3\alpha\kappa+\bar{\beta}\kappa-\kappa\pi-\epsilon\rho-\bar{\epsilon}\rho-\rho^2-\sigma\bar{\sigma}+\bar{\kappa}\tau \label{MasslessScalarstructureeq13},\\
\bar{\delta}\mu-\delta\lambda=&\tilde\Psi_3
-\bar{\alpha}\lambda+3\beta\lambda-\alpha\mu-\bar{\beta}\mu-\mu\pi+\bar{\mu}\pi-\nu\rho+\nu\bar{\rho} \label{MasslessScalarstructureeq14},\\
\bar{\delta}\nu-\Delta\lambda=&\Psi_4
+3\gamma\lambda-\bar{\gamma}\lambda+\lambda\mu+\lambda\bar{\mu}-3\alpha\nu-\bar{\beta}\nu-\nu\pi+\nu\bar{\tau} \label{MasslessScalarstructureeq15},\\
\bar{\delta}\pi-D\lambda=&-3\bar\varphi_1^2
+3\epsilon\lambda-\bar{\epsilon}\lambda+\bar{\kappa}\nu-\alpha\pi+\bar{\beta}\pi-\pi^2-\lambda\rho-\mu\bar{\sigma} \label{MasslessScalarstructureeq16},\\
\bar{\delta}\sigma-\delta\rho=&\tilde\Psi_1
-\kappa\mu+\kappa\bar{\mu}-\bar{\alpha}\rho-\beta\rho+3\alpha\sigma-\bar{\beta}\sigma-\rho\tau-\bar{\rho}\tau \label{MasslessScalarstructureeq17},\\
\bar{\delta}\tau-\Delta\rho=&\tilde\Psi_2
-\kappa\nu-\gamma\rho-\bar{\gamma}\rho+\bar{\mu}\rho+\lambda\sigma+\alpha\tau-\bar{\beta}\tau+\tau\bar{\tau} \label{MasslessScalarstructureeq18}.
\end{align}
\end{subequations}

The Bianchi identities written in terms of the renormalised components
of the Weyl tensor take the following decoupled form:
\begin{subequations}
\begin{align}
\Delta\Psi_0-\delta\tilde\Psi_1=&6\varphi_0\delta\varphi_1
+(4\gamma-\mu)\Psi_0-(2\beta+4\tau)\tilde\Psi_1
+3\sigma\tilde\Psi_2 \nonumber\\
&+6\varphi_0\varphi_1\bar\alpha+6\varphi_0\varphi_2\sigma
-3\varphi_0^2\bar\lambda-3\varphi_1^2\bar\rho-6\bar\varphi_1\varphi_1\sigma
-6\varphi_0\varphi_1(\beta+2\tau), \label{BianchiMasslessScalarNP1} \\
D\tilde\Psi_1-\bar\delta\Psi_0=&-6\varphi_0\delta\varphi_0
+(\pi-4\alpha)\Psi_0+(2\epsilon+4\rho)\tilde\Psi_1
-3\kappa\tilde\Psi_2 \nonumber \\
&+3\varphi_0^2(2\bar\alpha+2\beta-\bar\pi)+6\varphi_1\bar\varphi_1\kappa
+3\varphi_1^2\bar\kappa+6\varphi_0\varphi_1(2\rho-\bar\rho)
+6\varphi_0\bar\varphi_1\sigma, \label{BianchiMasslessScalarNP2} \\
\Delta\tilde\Psi_1-\delta\tilde\Psi_2=&-6\varphi_1\bar\delta\varphi_1
+\nu\Psi_0+(2\gamma-2\mu)\tilde\Psi_1-3\tau\tilde\Psi_2+3\sigma\tilde\Psi_3 \nonumber \\
&+6\varphi_2\bar\varphi_1\sigma+6\varphi_1(\varphi_0\bar\mu-\varphi_0\mu+\bar\varphi_1\tau)
+3\varphi_1^2(2\alpha-2\bar\beta+\bar\tau) \nonumber \\
&-3\varphi_0^2\bar\nu-6\varphi_2\varphi_1\rho, \label{BianchiMasslessScalarNP3} \\
D\tilde\Psi_2-\bar\delta\tilde\Psi_1=&6\varphi_1\bar\delta\varphi_0
-\lambda\Psi_0+(2\pi-2\alpha)\tilde\Psi_1+3\rho\tilde\Psi_2-2\kappa\tilde\Psi_3 \nonumber \\
&3\varphi_0^2\bar\mu+6\varphi_0\varphi_1\pi+3\varphi_1^2\bar\sigma
-6\varphi_0\varphi_1(\alpha+\bar\beta)-6\varphi_2\bar\varphi_1\kappa, \label{BianchiMasslessScalarNP4} \\
\Delta\tilde\Psi_2-\delta\tilde\Psi_3=&6\bar\varphi_1\delta\varphi_2
+2\nu\tilde\Psi_1-3\mu\tilde\Psi_2+(2\beta-2\tau)\tilde\Psi_3+\sigma\Psi_4 \nonumber \\
&+6\varphi_2\bar\varphi_1(\bar\alpha+\beta-\tau)+6\varphi_0\varphi_1\nu
-3\bar\varphi_1^2\bar\lambda-3\varphi_2^2\bar\rho, \label{BianchiMasslessScalarNP5} \\
D\tilde\Psi_3-\bar\delta\tilde\Psi_2=&6\bar\varphi_1\delta\bar\varphi_1
-2\lambda\tilde\Psi_1+3\pi\tilde\Psi_2+(2\rho-2\epsilon)\tilde\Psi_3-\kappa\Psi_4 \nonumber \\
&+3\bar\varphi_1^2(2\bar\alpha-2\beta-\bar\pi)+3\varphi_2^2\bar\kappa+6\varphi_0\bar\varphi_1\mu
+6\varphi_2\bar\varphi_1(\rho-\bar\rho)\nonumber \\
&-6\varphi_0\varphi_1\lambda
+6\bar\varphi_1\varphi_1\pi,  \label{BianchiMasslessScalarNP6} \\
\Delta\tilde\Psi_3-\delta\Psi_4=&-6\varphi_2\bar\delta\varphi_2
+3\nu\tilde\Psi_2-(2\gamma+4\mu)\tilde\Psi_3+(4\beta-\tau)\Psi_4 \nonumber \\
&+6\varphi_2\varphi_1\lambda+6\varphi_2\bar\varphi_1(\bar\mu-2\mu)
+3\varphi_2^2(\bar\tau-2\alpha-2\bar\beta)-6\varphi_1\bar\varphi_1\nu
-3\bar\varphi_1^2\bar\nu, \label{BianchiMasslessScalarNP7} \\
D\Psi_4-\bar\delta\tilde\Psi_3=&6\varphi_2\bar\delta\bar\varphi_1
-3\lambda\tilde\Psi_2+(2\alpha+4\pi)\tilde\Psi_3+(\rho-4\epsilon)\Psi_4 \nonumber \\
&6\varphi_2\bar\varphi_1(\alpha-\bar\beta+2\pi)+6\varphi_1\bar\varphi_1\lambda+3\bar\varphi_1^2\mu
+3\varphi_2^2\bar\sigma-6\varphi_0\varphi_2\lambda. \label{BianchiMasslessScalarNP8}
\end{align}
\end{subequations}

\section{NP T-weight equations}
\label{NPT-weightEq}

Using the T-weighted operators the structure equations can be
rewritten as
\begin{subequations}
\begin{align}
\mthorn\rho-\meth'\kappa&=\rho^2+\sigma\bar\sigma+\omega\rho-\bar\kappa\tau
-(2\bar\vartheta-\pi)\kappa+\Phi_{00}, \label{T-weightStructureEq1}\\
\mthorn\sigma-\meth\kappa&=(\rho+\bar\rho)\sigma+\omega\sigma
-(\tau-\bar\pi+2\vartheta)\kappa+\Psi_0, \label{T-weightStructureEq2}\\
\mthorn\mu-\meth\pi&=\bar\rho\mu+\sigma\lambda+\pi\bar\pi
-\omega\mu-\nu\kappa+\Psi_2+2\Lambda, \label{T-weightStructureEq3}\\
\mthorn\lambda-\meth'\pi&=\rho\lambda+\bar\sigma\mu+\pi^2-\nu\bar\kappa
-\omega\lambda+\Phi_{20}, \label{T-weightStructureEq4}\\
\mthorn'\rho-\meth'\tau&=-\rho\bar\mu-\sigma\lambda-\bar\tau\tau+\ulomega\rho
-\nu\kappa-\Psi_2-2\Lambda, \label{T-weightStructureEq5}\\
\mthorn'\sigma-\meth\tau&=-\mu\sigma-\bar\lambda\rho-\tau^2+\ulomega\sigma
+\kappa\bar\nu-\Phi_{02}, \label{T-weightStructureEq6}\\
\mthorn'\mu-\meth\nu&=-\mu^2-\lambda\bar\lambda-\ulomega\mu+\bar\nu\pi
+(2\vartheta-\tau)\nu-\Phi_{22}, \label{T-weightStructureEq7}\\
\mthorn'\lambda-\meth'\nu&=-(\mu+\bar\mu)\lambda-\ulomega\lambda
+(2\bar\vartheta+\pi-\bar\tau)\nu-\Psi_4, \label{T-weightStructureEq8}\\
\mthorn\tau-\mthorn'\kappa&=(\tau+\bar\pi)\rho+(\bar\tau+\pi)\sigma
-2\ulomega\kappa+\Psi_1+\Phi_{01}, \label{T-weightStructureEq9}\\
\mthorn'\pi-\mthorn\nu&=-(\pi+\bar\tau)\mu-(\bar\pi+\tau)\lambda+2\omega\nu
-\Psi_3-\Phi_{21},  \label{T-weightStructureEq10}\\
\mthorn\ulomega-\mthorn'\omega&=(\tau+\bar\pi)\bar\vartheta+(\bar\tau+\pi)\vartheta-2\ulomega\omega
+\tau\pi+\bar\tau\bar\pi \nonumber\\ 
&\quad -\nu\kappa-\bar\nu\bar\kappa+\Psi_2+\bar\Psi_2+2\Phi_{11}-2\Lambda, \label{T-weightStructureEq11}\\
\mthorn\vartheta-\meth\omega&=(\bar\vartheta+\pi)\sigma+(\vartheta+\bar\pi)\bar\rho+(\bar\pi-\vartheta)\omega \nonumber\\
&\quad -(\mu+\ulomega)\kappa-\bar\kappa\bar\lambda+\Psi_1+\Phi_{01},  \label{T-weightStructureEq12}\\
\mthorn'\vartheta-\meth\ulomega&=-(\bar\tau+\bar\vartheta)\bar\lambda-(\tau+\vartheta)\mu-(\tau-\vartheta)\ulomega \nonumber \\
&\quad +(\bar\rho+\omega)\bar\nu+\sigma\nu-\bar\Psi_3-\Phi_{12}, \label{T-weightStructureEq13}\\
\meth\rho-\meth'\sigma&=\vartheta\rho-\bar\vartheta\sigma+(\rho-\bar\rho)\tau 
+(\mu-\bar\mu)\kappa-\Psi_1+\Phi_{01},  \label{T-weightStructureEq14}\\
\meth\lambda-\meth'\mu&=\bar\vartheta\mu-\vartheta\lambda+(\rho-\bar\rho)\nu
+(\mu-\bar\mu)\pi-\Psi_3+\Phi_{21}, \label{T-weightStructureEq15}
\end{align}
\end{subequations}
while for the Bianchi identities one has
\begin{subequations}
\begin{align}
\mthorn\Psi_1-\meth'\Psi_0-\mthorn\Phi_{01}+\meth\Phi_{00}=&(\pi-2\bar\vartheta)\Psi_0+(\omega+4\rho)\Psi_1
-(\omega+2\bar\rho)\Phi_{01}+(2\vartheta-\bar\pi)\Phi_{00} \nonumber\\
&-3\kappa\Psi_2-2\sigma\Phi_{10}+2\kappa\Phi_{11}+\bar\kappa\Phi_{02}, \label{T-weightBianchi1}\\
\mthorn'\Psi_0-\meth\Psi_1+\mthorn\Phi_{02}-\meth\Phi_{01}=&(2\ulomega-\mu)\Psi_0-(4\tau+\vartheta)\Psi_1
+\bar\rho\Phi_{02}+(2\pi-\vartheta)\Phi_{01} \nonumber \\
&+3\sigma\Psi_2+2\sigma\Phi_{11}-2\kappa\Phi_{12}-\bar\lambda\Phi_{00}, \label{T-weightBianchi2}\\
\mthorn\Psi_2-\meth'\Psi_1+\mthorn'\Phi_{00}-\meth'\Phi_{01}+2\mthorn\Lambda=&3\rho\Psi_2+(2\pi-\bar\vartheta)\Psi_1
+(2\ulomega-\bar\mu)\Phi_{00}-(2\bar\tau+\bar\vartheta)\Phi_{01} \nonumber\\
&-\lambda\Psi_0-2\kappa\Psi_3-2\tau\Phi_{10}+2\rho\Phi_{11}+\bar\sigma\Phi_{02}, \label{T-weightBianchi3}\\
\mthorn'\Psi_1-\meth\Psi_2-\mthorn'\Phi_{01}+\meth'\Phi_{02}-2\meth\Lambda=&(\ulomega-2\mu)\Psi_1-3\tau\Psi_2
+(2\bar\mu-\ulomega)\Phi_{01}+\bar\tau\Phi_{02} \nonumber \\
&+\nu\Psi_0+2\sigma\Psi_3+2\tau\Phi_{11}-2\rho\Phi_{12}-\bar\nu\Phi_{00}, \label{T-weightBianchi4}\\
\mthorn\Psi_3-\meth'\Psi_2-\mthorn\Phi_{21}+\meth\Phi_{20}-2\meth'\Lambda=&(2\rho-\omega)\Psi_3+3\pi\Psi_2
+(\omega-2\bar\rho)\Phi_{21}-\bar\pi\Phi_{20} \nonumber \\
&-2\lambda\Psi_1-\kappa\Psi_4-2\pi\Phi_{11}+2\mu\Phi_{10}+\bar\kappa\Phi_{22}, \label{T-weightBianchi5} \\
\mthorn'\Psi_2-\meth\Psi_3+\mthorn\Phi_{22}-\meth\Phi_{21}+2\mthorn'\Lambda=&-3\mu\Psi_2
+(\vartheta-2\tau)\Psi_3+(\bar\rho-2\omega)\Phi_{22}+(2\bar\pi+\vartheta)\Phi_{21} \nonumber \\
&-2\nu\Psi_1+\sigma\Psi_4+2\pi\Phi_{12}-2\mu\Phi_{11}-\bar\lambda\Phi_{20},  \label{T-weightBianchi6}\\
\mthorn\Psi_4-\meth'\Psi_3+\mthorn'\Phi_{20}-\meth'\Phi_{21}=&(\rho-2\omega)\Psi_4+(4\pi+\bar\vartheta)\Psi_3
+(\bar\vartheta-2\bar\tau)\Phi_{21} \nonumber \\
&-\bar\mu\Phi_{20}-3\lambda\Psi_2-2\lambda\Phi_{11}+2\nu\Phi_{10}+\bar\sigma\Phi_{22}, \label{T-weightBianchi7}\\
\mthorn'\Psi_3-\meth\Psi_4+\meth'\Phi_{22}-\mthorn'\Phi_{21}=&(2\vartheta-\tau)\Psi_4-(4\mu+\ulomega)\Psi_3
+(\bar\tau-2\bar\vartheta)\Phi_{22} \nonumber \\
&+(2\bar\mu+\ulomega)\Phi_{21}+3\nu\Psi_2+2\lambda\Phi_{12}-2\nu\Phi_{11}-\bar\nu\Phi_{20},  \label{T-weightBianchi8}\\
\mthorn\Phi_{11}-\meth\Phi_{10}-\meth'\Phi_{01}+\mthorn'\Phi_{00}+3\mthorn\Lambda=&(2\ulomega-\mu-\bar\mu)\Phi_{00}
+(\pi-\bar\vartheta-2\bar\tau)\Phi_{01}+(\bar\pi-\vartheta-2\tau)\Phi_{10} \nonumber \\
&+2(\rho+\bar\rho)\Phi_{11}+\bar\sigma\Phi_{02}+\sigma\Phi_{20}-\bar\kappa\Phi_{12}-\kappa\Phi_{21}, \label{T-weightBianchi9}\\
\mthorn\Phi_{12}-\meth\Phi_{11}-\meth'\Phi_{02}+\mthorn'\Phi_{01}+3\meth\Lambda=&(\pi-\bar\tau)\Phi_{02}
+(\bar\rho+2\rho-\omega)\Phi_{12}+2(\bar\pi-\tau)\Phi_{11} \nonumber \\
&+(\ulomega-2\bar\mu-\mu)\Phi_{01}+\bar\nu\Phi_{00}-\bar\lambda\Phi_{10}+\sigma\Phi_{21}-\kappa\Phi_{22}, \label{T-weightBianchi10}\\
\mthorn\Phi_{22}-\meth\Phi_{21}-\meth'\Phi_{12}+\mthorn'\Phi_{11}+3\mthorn'\Lambda=&(\bar\rho+\rho-2\omega)\Phi_{22}
+(\bar\vartheta+2\pi-\bar\tau)\Phi_{12}+(\vartheta+2\bar\pi-\tau)\Phi_{21} \nonumber \\
&-2(\mu+\bar\mu)\Phi_{11}+\nu\Phi_{01}+\bar\nu\Phi_{10}-\bar\lambda\Phi_{20}-\lambda\Phi_{02}. \label{T-weightBianchi11}
\end{align}
\end{subequations}

The commutator relations for any quantity~$f$ with T-weight~$s$ are
given by
\begin{subequations}
\begin{align}
(\mthorn\mthorn'-\mthorn'\mthorn)f=&s(\Psi_2-\bar\Psi_2+\bar\kappa\bar\nu-\kappa\nu+\pi\tau-\bar\pi\bar\tau)f \nonumber \\
&-\ulomega\mthorn f-\omega\mthorn'f+(\pi+\bar\tau)\meth f+(\bar\pi+\tau)\meth'f, \label{T-weightCommutator1}\\
(\mthorn\meth-\meth\mthorn)f=&s(-\Phi_{01}+\Psi_1+\bar\kappa\bar\lambda-\bar\pi\bar\rho-\kappa\mu+\pi\sigma)f \nonumber \\
&+(\bar\pi-\vartheta)\mthorn f-\kappa\mthorn'f+\bar\rho\meth f+\sigma\meth'f, \label{T-weightCommutator2}\\
(\mthorn'\meth-\meth\mthorn')f=&s(-\Phi_{12}+\bar\Psi_3+\nu\sigma-\mu\tau-\bar\nu\bar\rho+\bar\lambda\bar\tau)f \nonumber \\
&+\bar\nu\mthorn f+(\vartheta-\tau)\mthorn'f-\mu\meth f-\bar\lambda\meth'f, \label{T-weightCommutator3}\\
(\mthorn\meth'-\meth'\mthorn)f=&s(\Phi_{10}-\bar\Psi_1-\kappa\lambda+\pi\rho+\bar\kappa\bar\mu-\bar\pi\bar\sigma)f \nonumber \\
&+(\bar\vartheta-\pi)\mthorn f-\bar\kappa\mthorn'f+\bar\sigma\meth f+\rho\meth'f, \label{T-weightCommutator4}\\
(\mthorn'\meth'-\meth'\mthorn')f=&s(\Phi_{21}-\Psi_3-\bar\nu\bar\sigma+\bar\mu\bar\tau+\nu\rho-\lambda\tau)f \nonumber \\
&+\nu\mthorn f+(\bar\vartheta-\bar\tau)\mthorn'f-\lambda\meth f-\bar\mu\meth'f, \label{T-weightCommutator5}\\
(\meth\meth'-\meth'\meth)f=&sKf+(\mu-\bar\mu)\mthorn f+(\rho-\bar\rho)\mthorn'f, \label{T-weightCommutator6}
\end{align}
\end{subequations}
where~$K$ is the Gauss curvature
\begin{align*}
K\equiv \mu\rho+\bar\mu\bar\rho-\lambda\sigma-\bar\lambda\bar\sigma+2\Lambda+2\Phi_{11}-\Psi_2-\bar\Psi_2.
\end{align*}
Making use of the structure equations~\eqref{T-weightStructureEq14}
and~\eqref{T-weightStructureEq15}, the commutator
relations~\eqref{T-weightCommutator2}-\eqref{T-weightCommutator5} can
be rewritten as
\begin{subequations}
\begin{align}
(\mthorn\meth-\meth\mthorn)f=&s(-\meth\rho+\meth'\sigma+\vartheta\rho-\bar\vartheta\sigma+(\rho-\bar\rho)\tau-\bar\mu\kappa
+\bar\kappa\bar\lambda-\bar\pi\bar\rho+\pi\sigma)f \nonumber \\
&+(\bar\pi-\vartheta)\mthorn f-\kappa\mthorn'f+\bar\rho\meth f+\sigma\meth'f,  \label{T-weightCommutatorAlt2}\\
(\mthorn'\meth-\meth\mthorn')f=&s(-\meth'\bar\lambda+\meth\bar\mu+\vartheta\bar\mu-\bar\vartheta\bar\lambda+(\bar\mu-\mu)\bar\pi-\rho\bar\nu
+\nu\sigma-\mu\tau+\bar\lambda\bar\tau)f \nonumber \\
&+\bar\nu\mthorn f+(\vartheta-\tau)\mthorn'f-\mu\meth f-\bar\lambda\meth'f, \label{T-weightCommutatorAlt3}\\
(\mthorn\meth'-\meth'\mthorn)f=&s(\meth'\bar\rho-\meth\bar\sigma-\bar\vartheta\bar\rho+\vartheta\bar\sigma+(\rho-\bar\rho)\bar\tau+\mu\bar\kappa
-\kappa\lambda+\pi\rho-\bar\pi\bar\sigma)f \nonumber \\
&+(\bar\vartheta-\pi)\mthorn f-\bar\kappa\mthorn'f+\bar\sigma\meth f+\rho\meth'f, \label{T-weightCommutatorAlt4}\\
(\mthorn'\meth'-\meth'\mthorn')f=&s(\meth\lambda-\meth'\mu-\bar\vartheta\mu+\vartheta\lambda+\bar\rho\nu+(\bar\mu-\mu)\pi
-\bar\nu\bar\sigma+\bar\mu\bar\tau-\lambda\tau)f \nonumber \\
&+\nu\mthorn f+(\bar\vartheta-\bar\tau)\mthorn'f-\lambda\meth f-\bar\mu\meth'f. \label{T-weightCommutatorAlt5}
\end{align}
\end{subequations}

Observe that the T-weight of~$f$ determines whether there are
derivatives of~$\rho$, $\sigma$ and $\mu$, $\lambda$ on the right-hand
side of the commutators or not.

\subsection{The Einstein-scalar field system in the T-weight formalism}
\label{ScalarNPEq}

We conclude by providing the full form of the Einstein-scalar field
system in the T-weight formalism.

We various components of the auxiliary field~$\varphi_a$ encoding the
derivatives of~$\varphi$ are assigned T-weights as per the following
table:
\begin{table}[h!]
\begin{center}
\label{ScalarT-weight}
\begin{tabular}{|c|c|c|c|c|c|}
\hline
~ & $\varphi$ & $\varphi_{0}$ & $\varphi_{2}$ & $\varphi_{1}$ & $\bar\varphi_{1}$ \\
\hline
s& 0 & 0 & 0 & -1 & 1\\
\hline
\end{tabular}
\end{center}
\end{table}

Recalling that~$\varphi_a\equiv \nabla_a\varphi$, it follows then that
\begin{subequations}
\begin{align}
\mthorn'\varphi_{0}-\meth'\varphi_{1}&=(\ulomega-\bar\mu)\varphi_{0}
+\rho\varphi_{2}-\bar\tau\varphi_{1}-\tau\bar\varphi_{1}, \label{EOMMasslessScalarT-weight1} \\
\mthorn\varphi_{1}-\meth\varphi_{0}&=(\bar\pi-\vartheta)\varphi_{0}
-\kappa\varphi_{2}+\sigma\bar\varphi_{1}+\bar\rho\varphi_{1}, \label{EOMMasslessScalarT-weight2} \\
\mthorn'\varphi_{1}-\meth\varphi_{2}&=\bar\nu\varphi_{0}
+(\vartheta-\tau)\varphi_{2}-\mu\varphi_{1}-\bar\lambda\bar\varphi_{1}, \label{EOMMasslessScalarT-weight3} \\
\mthorn\varphi_{2}-\meth'\varphi_{1}&=-\bar\mu\varphi_{0}
+(\rho-\omega)\varphi_{2}+\pi\varphi_{1}+\bar\pi\bar\varphi_{1}, \label{EOMMasslessScalarT-weight4} \\
\meth\bar\varphi_{1}-\meth'\varphi_{1}&=
(\mu-\bar\mu)\varphi_{0}
+(\rho-\bar\rho)\varphi_{2}. \label{EOMMasslessScalarT-weight5} 
\end{align}
\end{subequations}

Now, recall that components of the definition of the renormalised Weyl
tensor are given by
\begin{subequations}
\begin{align}
\tilde\Psi_1\equiv&\Psi_1-\Phi_{01}=\Psi_1-3\varphi_{0}\varphi_{1},  \label{NewPsi1}\\
\tilde\Psi_2\equiv&\Psi_2+2\Lambda=\Psi_2-\varphi_{0}\varphi_{2}+\varphi_{1}\bar\varphi_{1},  \label{NewPsi2} \\
\tilde\Psi_3\equiv&\Psi_3-\Phi_{21}=\Psi_3-3\varphi_{2}\bar\varphi_{1}. \label{NewPsi3}
\end{align}
\end{subequations}
Substituting the latter into the structure equations and Bianchi
identities one obtains the final form for the coupled
Einstein-Massless scalar field system. In particular, the structure
equations take the form
\begin{subequations}
\begin{align}
\mthorn\rho-\meth'\kappa&=\rho^2+\sigma\bar\sigma+\omega\rho-\bar\kappa\tau
-(2\bar\vartheta-\pi)\kappa+3\varphi_{0}^2, \label{T-weightMasslessStructureEq1}\\
\mthorn\sigma-\meth\kappa&=(\rho+\bar\rho)\sigma+\omega\sigma
-(\tau-\bar\pi+2\vartheta)\kappa+\Psi_0, \label{T-weightMasslessStructureEq2}\\
\mthorn\mu-\meth\pi&=\bar\rho\mu+\sigma\lambda+\pi\bar\pi
-\omega\mu-\nu\kappa+\TiPsi_2, \label{T-weightMasslessStructureEq3}\\
\mthorn\lambda-\meth'\pi&=\rho\lambda+\bar\sigma\mu+\pi^2-\nu\bar\kappa
-\omega\lambda+3\bar\varphi_1^2, \label{T-weightMasslessStructureEq4}\\
\mthorn'\rho-\meth'\tau&=-\rho\bar\mu-\sigma\lambda-\bar\tau\tau+\ulomega\rho
-\nu\kappa-\TiPsi_2, \label{T-weightMasslessStructureEq5}\\
\mthorn'\sigma-\meth\tau&=-\mu\sigma-\bar\lambda\rho-\tau^2+\ulomega\sigma
+\kappa\bar\nu-3\varphi_{1}^2, \label{T-weightMasslessStructureEq6}\\
\mthorn'\mu-\meth\nu&=-\mu^2-\lambda\bar\lambda-\ulomega\mu+\bar\nu\pi
+(2\vartheta-\tau)\nu-3\varphi_{2}^2, \label{T-weightMasslessStructureEq7}\\
\mthorn'\lambda-\meth'\nu&=-(\mu+\bar\mu)\lambda-\ulomega\lambda
+(2\bar\vartheta+\pi-\bar\tau)\nu-\Psi_4, \label{T-weightMasslessStructureEq8}\\
\mthorn\tau-\mthorn'\kappa&=(\tau+\bar\pi)\rho+(\bar\tau+\pi)\sigma
-2\ulomega\kappa+\TiPsi_1+6\varphi_0\varphi_1, \label{T-weightMasslessStructureEq9}\\
\mthorn'\pi-\mthorn\nu&=-(\pi+\bar\tau)\mu-(\bar\pi+\tau)\lambda+2\omega\nu
-\TiPsi_3-6\varphi_{2}\bar\varphi_{1},  \label{T-weightMasslessStructureEq10}\\
\mthorn\ulomega-\mthorn'\omega&=(\tau+\bar\pi)\bar\vartheta+(\bar\tau+\pi)\vartheta-2\ulomega\omega
+\tau\pi+\bar\tau\bar\pi \nonumber\\ 
&\quad -\nu\kappa-\bar\nu\bar\kappa+\TiPsi_2+\bar\TiPsi_2+6\varphi_{0}\varphi_{2}, \label{T-weightMasslessStructureEq11}\\
\mthorn\vartheta-\meth\omega&=(\bar\vartheta+\pi)\sigma+(\vartheta+\bar\pi)\bar\rho+(\bar\pi-\vartheta)\omega \nonumber\\
&\quad -(\mu+\ulomega)\kappa-\bar\kappa\bar\lambda+\TiPsi_1+6\varphi_{0}\varphi_{1},  \label{T-weightMasslessStructureEq12}\\
\mthorn'\vartheta-\meth\ulomega&=-(\bar\tau+\bar\vartheta)\bar\lambda-(\tau+\vartheta)\mu-(\tau-\vartheta)\ulomega \nonumber \\
&\quad +(\bar\rho+\omega)\bar\nu+\sigma\nu-\bar\TiPsi_3-6\varphi_{2}\varphi_{1}, \label{T-weightMasslessStructureEq13}\\
\meth\rho-\meth'\sigma&=\vartheta\rho-\bar\vartheta\sigma+(\rho-\bar\rho)\tau 
+(\mu-\bar\mu)\kappa-\TiPsi_1,  \label{T-weightMasslessStructureEq14}\\
\meth\lambda-\meth'\mu&=\bar\vartheta\mu-\vartheta\lambda+(\rho-\bar\rho)\nu
+(\mu-\bar\mu)\pi-\TiPsi_3, \label{T-weightMasslessStructureEq15}
\end{align}
\end{subequations}
while the Bianchi identities are given by 
\begin{subequations}
\begin{align}
\mthorn\TiPsi_1-\meth'\Psi_0=&-6\varphi_0\meth\varphi_0+(\pi-2\bar\vartheta)\Psi_0
+(4\rho-\omega)\TiPsi_1-3\kappa\TiPsi_2
 \nonumber\\
&+3\varphi_0^2(2\vartheta-\bar\pi)+6\varphi_1\bar\varphi_1\kappa+3\varphi_1^2\bar\kappa
+6\varphi_0\varphi_1(2\rho-\bar\rho)-6\varphi_0\bar\varphi_1\sigma, \label{T-weightMasslessBianchi1}\\
\mthorn'\Psi_0-\meth\TiPsi_1=&6\varphi_0\meth\varphi_1+(2\ulomega-\mu)\Psi_0
-(4\tau+\vartheta)\TiPsi_1+3\sigma\TiPsi_2
\nonumber \\
&+6\varphi_0\varphi_2\sigma-3\varphi_0^2\bar\lambda-3\varphi_1^2\bar\rho
-6\varphi_1\bar\varphi_1\sigma-12\varphi_0\varphi_1\tau, \label{T-weightMasslessBianchi2}\\
\mthorn\TiPsi_2-\meth'\TiPsi_1=&6\varphi_1\meth’\varphi_0-\lambda\Psi_0
+(2\pi-\bar\vartheta)\TiPsi_1+3\rho\TiPsi_2-2\kappa\TiPsi_3
 \nonumber\\
&+3\varphi_0^2\bar\mu+6\varphi_0\varphi_1\pi+3\varphi_1^2\bar\sigma
-6\varphi_0\varphi_1\bar\vartheta-6\varphi_2\bar\varphi_1\kappa, \label{T-weightMasslessBianchi3}\\
\mthorn'\TiPsi_1-\meth\TiPsi_2=&-6\varphi_1\meth’\varphi_1+\nu\Psi_0
+(\ulomega-2\mu)\TiPsi_1-3\tau\TiPsi_2+2\sigma\TiPsi_3
\nonumber \\
&+6\varphi_2\bar\varphi_1\sigma+6\varphi_1(\varphi_0\bar\mu-\varphi_0\mu+\bar\varphi_1\tau)
+3\varphi_1^2\bar\tau-3\varphi_0^2\bar\nu-6\varphi_2\varphi_1\rho, \label{T-weightMasslessBianchi4}\\
\mthorn\TiPsi_3-\meth'\TiPsi_2=&-6\bar\varphi_1\meth\bar\varphi_1-2\lambda\TiPsi_1
+3\pi\TiPsi_2+(2\rho-\omega)\TiPsi_3-\kappa\Psi_4. \nonumber \\
&-3\bar\varphi_1^2\bar\pi+3\varphi_2^2\bar\kappa+6\varphi_0\bar\varphi_1\mu
+6\varphi_2\bar\varphi_1(\rho-\bar\rho)-6\varphi_0\varphi_1\lambda-6\varphi_1\bar\varphi_1\pi \label{T-weightMasslessBianchi5} \\
\mthorn'\TiPsi_2-\meth\TiPsi_3=&6\bar\varphi_1\meth\varphi_2+2\nu\TiPsi_1-3\mu\TiPsi_2
+(\vartheta-2\tau)\TiPsi_3+\sigma\Psi_4 \nonumber \\
&+6\varphi_2\bar\varphi_1(\vartheta-\tau)+6\varphi_0\varphi_1\nu
-3\bar\varphi_1^2\bar\lambda-3\varphi_2^2\bar\rho,  \label{T-weightMasslessBianchi6}\\
\mthorn\Psi_4-\meth'\TiPsi_3=&6\varphi_2\meth’\bar\varphi_1-3\lambda\TiPsi_2
+(4\pi+\bar\vartheta)\TiPsi_3+(\rho-2\omega)\Psi_4
 \nonumber \\
&+12\varphi_2\bar\varphi_1\pi+6\varphi_1\bar\varphi_1\lambda+3\bar\varphi_1^2\mu
+3\varphi_2^2\bar\sigma-6\varphi_0\varphi_2\lambda, \label{T-weightMasslessBianchi7}\\
\mthorn'\TiPsi_3-\meth\Psi_4=&-6\varphi_2\meth’\varphi_2+3\nu\TiPsi_2
-(4\mu+\ulomega)\TiPsi_3+(2\vartheta-\tau)\Psi_4
 \nonumber \\
&+6\varphi_2\varphi_1\lambda+6\varphi_2\bar\varphi_1(\bar\mu-2\mu)
+3\varphi_2^2(\bar\tau-2\bar\vartheta)-6\varphi_1\bar\varphi_1\nu-3\bar\varphi_1^2\bar\nu,  \label{T-weightMasslessBianchi8}
\end{align}
\end{subequations}
Finally, commutators take the form
\begin{subequations}
\begin{align}
(\mthorn\mthorn'-\mthorn'\mthorn)f=&s(\TiPsi_2-\bar\TiPsi_2+\bar\kappa\bar\nu-\kappa\nu+\pi\tau-\bar\pi\bar\tau)f \nonumber \\
&-\ulomega\mthorn f-\omega\mthorn'f+(\pi+\bar\tau)\meth f+(\bar\pi+\tau)\meth'f, \label{T-weightMasslessCommutator1}\\
(\mthorn\meth-\meth\mthorn)f=&s(-\meth\rho+\meth'\sigma+\vartheta\rho-\bar\vartheta\sigma+(\rho-\bar\rho)\tau-\bar\mu\kappa
+\bar\kappa\bar\lambda-\bar\pi\bar\rho+\pi\sigma)f \nonumber \\
&+(\bar\pi-\vartheta)\mthorn f-\kappa\mthorn'f+\bar\rho\meth f+\sigma\meth'f  \label{T-weightMasslessCommutator2}\\
(\mthorn'\meth-\meth\mthorn')f=&s(-\meth'\bar\lambda+\meth\bar\mu+\vartheta\bar\mu-\bar\vartheta\bar\lambda+(\bar\mu-\mu)\bar\pi-\rho\bar\nu
+\nu\sigma-\mu\tau+\bar\lambda\bar\tau)f \nonumber \\
&+\bar\nu\mthorn f+(\vartheta-\tau)\mthorn'f-\mu\meth f-\bar\lambda\meth'f, \label{T-weightMasslessCommutator3}\\
(\mthorn\meth'-\meth'\mthorn)f=&s(\meth'\bar\rho-\meth\bar\sigma-\bar\vartheta\bar\rho+\vartheta\bar\sigma+(\rho-\bar\rho)\bar\tau+\mu\bar\kappa
-\kappa\lambda+\pi\rho-\bar\pi\bar\sigma)f \nonumber \\
&+(\bar\vartheta-\pi)\mthorn f-\bar\kappa\mthorn'f+\bar\sigma\meth f+\rho\meth'f, \label{T-weightMasslessCommutator4}\\
(\mthorn'\meth'-\meth'\mthorn')f=&s(\meth\lambda-\meth'\mu-\bar\vartheta\mu+\vartheta\lambda+\bar\rho\nu+(\bar\mu-\mu)\pi
-\bar\nu\bar\sigma+\bar\mu\bar\tau-\lambda\tau)f \nonumber \\
&+\nu\mthorn f+(\bar\vartheta-\bar\tau)\mthorn'f-\lambda\meth f-\bar\mu\meth'f, \label{T-weightMasslessCommutator5} \\
(\meth\meth'-\meth'\meth)f=&sKf+(\mu-\bar\mu)\mthorn f+(\rho-\bar\rho)\mthorn'f. \label{T-weightMasslessCommutator6}
\end{align}
\end{subequations}
where~$K$ is the Gauss curvature
\begin{align*}
K=\mu\rho+\bar\mu\bar\rho-\lambda\sigma-\bar\lambda\bar\sigma-\TiPsi_2-\bar\TiPsi_2+6\varphi_1\bar\varphi_1.
\end{align*}

\section{Scaling argument}
\label{ScalingArgue}

Consider a conformal transformation from a
spacetime~$(\mathcal{M},\bmg)$ to another
spacetime~$(\mathcal{M}',\bmg')$ of the form
\begin{align}
\bmg'=\delta^2\bmg,
\label{ConformalRescaling}
\end{align}
where~$\delta$ is a constant ---i.e. a \emph{homotety}. This conformal
transformation leaves the scalar field invariant. Now, define a
coordinate system~$(u',v',x'^{\mathcal{A}})$ on~$\mathcal{M}'$ via the
relations
\begin{align*}
x'^{\mu}(p')\equiv\delta x^{\mu}(p)
\Rightarrow 
u'\equiv\delta u,\quad
v'\equiv\delta v,\quad
x'^{\mathcal{A}}\equiv\delta x^{\mathcal{A}},
\end{align*}
where~$p'$ is the image of~$p$ under the conformal
transformation. Under this transformation the domain
\begin{align*}
  \mathscr{D}=\bigg\{p\in\mathcal{M}\;|\;u_{\infty}
  \leq u\leq-\frac{a}{4},\;\;0\leq v\leq1\bigg\}
\end{align*}
is mapped  to 
\begin{align*}
\mathscr{D}'=\bigg\{p'\in\mathcal{M}'\;|\;\delta u_{\infty}\leq
u'\leq-\frac{\delta a}{4},\;\;0\leq v'\leq\delta \bigg\}.
\end{align*}

As the scalar field is not derived from the metric, we can choose it
to be invariant ---i.e. we set
\begin{align*}
\varphi(x')=\varphi(x).
\end{align*}
Moreover, we also have that 
\begin{align*}
Q'=Q, \qquad
C'^{\mathcal{A}}=C^{\mathcal{A}}, \qquad
P'^{\mathcal{A}}=P^{\mathcal{A}}.
\end{align*}
On the other hand, the rescaling~\eqref{ConformalRescaling} implies
for the NP frame that 
\begin{align*}
\bml'=\delta^{-1}\bml, \quad
\bmn'=\delta^{-1}\bmn, \quad
\bmm'=\delta^{-1}\bmm, \quad
\bar\bmm'=\delta^{-1}\bar{\bmm'}.
\end{align*}
In addition, we also have that~$\nabla'_a=\nabla_a$. Then from the
definition of the NP connection coefficients we have
that~$\Gamma'=\delta^{-1}\Gamma$. For the curvature we
have~$R'^a_{\phantom{'a}bcd}=R^a_{\phantom{a}bcd}$, so
that~$R'_{abcd}=\delta^2R_{abcd}$ which, in turn, leads
to~$\Psi'=\delta^{-2}\Psi$. Taking into account the rescaling
properties of the frame it follows that the components of the
derivative of the scalar field satisfy
\begin{align*}
\varphi'_0=\delta^{-1}\varphi_0, \qquad
\varphi'_1=\delta^{-1}\varphi_1, \qquad
\varphi'_2=\delta^{-1}\varphi_2.
\end{align*}

From the estimates in the main text and considering
that~$\frac{a\delta}{4}\leq|u'|\leq\delta|u_{\infty}|$, one can
readily see that

\smallskip
\noindent
\textbf{$\Gamma'$}:
\begin{align*}
|\rho'(x')|&=\delta^{-1}|\rho(x)|\leq\delta^{-1}\frac{1}{|u|}=\frac{1}{|u'|}, \quad
|\sigma'(x')|=\delta^{-1}|\sigma(x)|\leq\delta^{-1}\frac{a^{\frac{1}{2}}}{|u|}=\frac{a^{\frac{1}{2}}}{|u'|}, \\
|\mu'(x')|&=\delta^{-1}|\mu(x)|\leq\delta^{-1}\frac{1}{|u|}=\frac{1}{|u'|}, \quad
|\Timu'(x')=\mu'+\frac{1}{|u'|}|=\delta^{-1}|\Timu(x)|\leq\delta^{-1}\frac{1}{|u|^2}=\frac{\delta}{|u'|^2}, \\
|\lambda'(x')|&=\delta^{-1}|\lambda(x)|\leq\delta^{-1}\frac{a^{\frac{1}{2}}}{|u|^2}=\frac{\delta a^{\frac{1}{2}}}{|u'|^2},\quad
|\ulomega'(x')|=\delta^{-1}|\ulomega(x)|\leq\delta^{-1}\frac{a}{|u|^3}=\frac{\delta^2a}{|u'|^3}\leq\frac{\delta a^{\frac{1}{2}}}{|u'|^2}, \\
|\tau'(x')|&=\delta^{-1}|\tau(x)|\leq\delta^{-1}\frac{a^{\frac{1}{2}}}{|u|^2}=\frac{\delta a^{\frac{1}{2}}}{|u'|^2}, \quad
|\pi'(x')|=\delta^{-1}|\pi(x)|\leq\delta^{-1}\frac{a^{\frac{1}{2}}}{|u|^2}=\frac{\delta a^{\frac{1}{2}}}{|u'|^2}.
\end{align*}

\smallskip
\noindent
\textbf{$\Psi'$}:
\begin{align*}
|\Psi'_0(x')|&=\delta^{-2}|\Psi_0(x)|\leq\delta^{-2}\frac{a^{\frac{1}{2}}}{|u|}=\frac{\delta^{-1}a^{\frac{1}{2}}}{|u'|}, \quad
|\Psi'_1(x')|=\delta^{-2}|\Psi_1(x)|\leq\delta^{-2}\frac{a^{\frac{1}{2}}}{|u|^2}=\frac{a^{\frac{1}{2}}}{|u'|^2}, \\
|\Psi'_2(x')|&=\delta^{-2}|\Psi_2(x)|\leq\delta^{-2}\frac{a}{|u|^3}=\frac{\delta a}{|u'|^3}, \quad
|\Psi'_3(x')|=\delta^{-2}|\Psi_3(x)|\leq\delta^{-2}\frac{a^{\frac{3}{2}}}{|u|^4}
=\frac{\delta^{2}a^{\frac{3}{2}}}{|u'|^4}\leq\frac{\delta a^{\frac{1}{2}}}{|u'|^3}, \\
|\Psi'_4(x')|&=\delta^{-2}|\Psi_4(x)|\leq\delta^{-2}\frac{a^2}{|u|^5}=\frac{\delta^{3}a^{2}}{|u'|^5}.
\end{align*}

\smallskip
\noindent
\textbf{$\bmvarphi'$}:
\begin{align*}
|\varphi'(x')|&=|\varphi(x)|\leq\frac{a^{\frac{1}{2}}}{|u|}=\frac{\delta a^{\frac{1}{2}}}{|u'|},\quad
|\varphi'_0(x')|=\delta^{-1}|\varphi_0(x)|\leq\delta^{-1}\frac{a^{\frac{1}{2}}}{|u|}=\frac{a^{\frac{1}{2}}}{|u'|}, \\
|\varphi'_1(x')|&=\delta^{-1}|\varphi_1(x)|\leq\delta^{-1}\frac{a^{\frac{1}{2}}}{|u|^2}=\frac{\delta a^{\frac{1}{2}}}{|u'|^2}, \quad
|\varphi'_2(x')|=\delta^{-1}|\varphi_2(x)|\leq\delta^{-1}\frac{a^{\frac{1}{2}}}{|u|^2}=\frac{\delta a^{\frac{1}{2}}}{|u'|^2}.
\end{align*}



\end{document}